\documentclass[
hidelinks,onefignum,onetabnum]{siamart220329}

\newcommand{\ie}{\text{i.\,e.}\ }
\newcommand{\eg}{\emph{e.\,g.}\ }

\newcommand{\ar}{\operatorname{ar}}
\newcommand{\Npos}{\mathbb N_{> 0}}

\newcommand{\F}{\mathcal{F}}

\newcommand{\x}{\textbf{x}}

\newcommand{\rot}{\operatorname{ROT}}
\newcommand{\dist}{\operatorname{dist}}
\newcommand{\indexSetRotation}{[D]^2}
\newcommand{\indexSetH}{([D]^2)^2}
\newcommand{\sampler}{\textsc{EstimateFrequencies}}
\newcommand{\setOfMaxCliques}[1][G]{\mathcal{K}^{#1}}
\newcommand{\maxcl}{\operatorname{maxcl}}
\renewcommand{\phi}{\varphi}
\renewcommand{\mod}{\mathrel{\mathrm{mod}}}

\newcommand*\circled[1]{\tikz[baseline=(char.base)]{
		\node[shape=circle,draw,inner sep=1pt] (char) {#1};}}
\DeclareMathOperator{\zigzag}{\circled{{\rm z}}}

\newcommand{\margin}[1]{$\bullet$\marginpar{\raggedright\footnotesize #1}}

\providecommand{\abs}[1]{\left\lvert#1\right\rvert}

\newcommand{\other}[1]{\tilde{#1}}
\newcommand{\struc}[1]{#1}
\newcommand{\univ}[1]{U(#1)}
\newcommand{\classStruc}[1]{\mathcal{#1}}
\newcommand{\rel}[2]{#1(#2)}
\newcommand{\gaifman}[1]{G(#1)}
\newcommand{\graphProp}{\mathcal{P}_{\operatorname{graph}}}
\newcommand{\graphFormula}{\psi_{\operatorname{graph}}}
\newcommand{\underlyingGraph}[1]{\underline{G}(#1)}
\newcommand{\vet}[1]{\ensuremath{\overline{#1}}}

\usepackage{tikz}
\usepackage{amssymb}
\usepackage{hyperref}
\usetikzlibrary{shapes,calc,positioning,arrows,hobby}
\usepackage{subcaption}

\usepackage[symbols,nogroupskip,record]{glossaries-extra}

\allowdisplaybreaks

\definecolor{C1}{RGB}{1,1,1}
\definecolor{C2}{RGB}{0,0,170}
\definecolor{C3}{RGB}{251,86,4}
\definecolor{C4}{RGB}{50,180,110}

\newtheorem{fact}{Fact}
\newtheorem{example}{Example}
\newtheorem{question}{Open Question}
\newtheorem{observation}{Observation}

\usepackage{amsfonts}
\usepackage{graphicx}
\usepackage{epstopdf}
\usepackage{algorithmic}
\ifpdf
\DeclareGraphicsExtensions{.eps,.pdf,.png,.jpg}
\else
\DeclareGraphicsExtensions{.eps}
\fi

\newsiamremark{remark}{Remark}
\newsiamremark{hypothesis}{Hypothesis}
\crefname{hypothesis}{Hypothesis}{Hypotheses}
\newsiamthm{claim}{Claim}

\glsxtrnewsymbol[description={the set of natural numbers including $0$}]{NN}{$\mathbb{N}$}
\glsxtrnewsymbol[description={the set of all positive natural numbers}]{NNP}{$\Npos$}
\glsxtrnewsymbol[description={the set $\{0,\dots,n-1\}$}]{[n]}{$[n]$}
\glsxtrnewsymbol[description={disjoint union}]{DU}{$\sqcup$}
\glsxtrnewsymbol[description={symmetric difference}]{SD}{$\triangle$}
\glsxtrnewsymbol[description={the neighbourhood of fixed radius $r$ around $x$, up to isomorphism}]{NBI}{$\tau_r(x)$}
\glsxtrnewsymbol[description={asymtotic upper-bounds}]{big-O}{$\mathcal{O}(\cdot), o(\cdot)$}
\glsxtrnewsymbol[description={zig-zag product}]{zz}{$\zigzag$}
\glsxtrnewsymbol[description={expansion ratio of $G$}]{exp}{$h(G)$}
\glsxtrnewsymbol[description={edges crossing $S$ and $T$ in $G$}]{croEd}{$\langle S,T\rangle_G$}
\glsxtrnewsymbol[description={signature}]{s}{$\sigma$}
\glsxtrnewsymbol[description={signature with one binary relation symbol $E$}]{sg}{$\sigma_{\textrm{graph}}$}
\glsxtrnewsymbol[description={degree of vertex $v$ in $G$}]{deg}{$\deg_G(v)$}
\glsxtrnewsymbol[description={distance between two vertices $v$ and $w$ in $A$}]{dist}{$\dist_A(v,w)$}
\glsxtrnewsymbol[description={distance between structure $A$ and property $\classStruc{P}$}]{distSP}{$\dist(A,\classStruc{P})$}
\glsxtrnewsymbol[description={rotation map of $G$}]{rot}{$\rot_{G}$}

\glsxtrnewsymbol[description={the $(u,v)$ entry of the normalised adjacency matrix $M$ of $G$}]{nam}{$M_{u,v}$}

\glsxtrnewsymbol[description={substracture of $A$ induced by $S$}]{induce}{$A[S]$}
\glsxtrnewsymbol[description={the $r$-neighborhood of $a$ in structure $A$}]{rNBH}{$N_r^A(a)$}

\glsxtrnewsymbol[description={arity of relation $R$}]{relation}{$\ar(R)$}
\glsxtrnewsymbol[description={class of graphs of bounded degree $d$}]{CCn}{$\mathcal{C}_d$}
\glsxtrnewsymbol[description={class of $\sigma$-structures of bounded degree $d$}]{Cn}{$\classStruc{C}_{\sigma,d}$}
\glsxtrnewsymbol[description={property defined by formula $\varphi$}]{PPhi}{$\classStruc{P}_\phi$}

\glsxtrnewsymbol[description={prefix classes with $i-1$ quantifier alterations}]{PC}{$\Sigma_i$, $\Pi_i$, $\Delta_i$}
\glsxtrnewsymbol[description={is a model of}]{models}{$\models$}
\glsxtrnewsymbol[description={equivalence of FO-formulas}]{equiv}{$\equiv$}
\glsxtrnewsymbol[description={equivalence of FO-formulas on structures of bounded degree $d$}]{dEquiv}{$\equiv_d$}
\glsxtrnewsymbol[description={logical negation, conjunction, disjunction, implication and biimplication}]{l}{$\lnot$, $\land$, $\lor$, $\rightarrow$, $\leftrightarrow$}
\glsxtrnewsymbol[description={existential and universal quantifier}]{quant}{$\exists$, $\forall$}

\glsxtrnewsymbol[description={answer to query $q$ to structure $\mathcal{A}$}]{query}{$\operatorname{ans}_{\mathcal{A}}(q)$}
\glsxtrnewsymbol[description={binary relation symbols for $i,j\in [D]^2$ and $k\in ([D]^2)^2$}]{binary}{$E_{i,j}$, $F_{k}$, $R$, $L_k$}

\glsxtrnewsymbol[description={underlying graph of a model $\struc{A}$ of $\varphi_{\zigzag}$}]{underlying}{$\underlyingGraph{\struc{A}}$}

\glsxtrnewsymbol[description={the formula definded in Section \ref{sec: definitionFormula} whose underlying graphs are expanders }]{expander}{$\varphi_{\zigzag}$}

\glsxtrnewsymbol[description={a set of models of some sentence in $\Sigma_2$ }]{modelset}{$\mathfrak{M}$}

\glsxtrnewsymbol[description={a maximal set of pairwise non-isomorphic structures in $\classStruc{C}_{\sigma,d}$}]{maximal}{$\mathfrak{H}$}

\glsxtrnewsymbol[description={a conjunction of atomic (resp. negated) formulas containing both variables from tuples $\overline{x}$ and $\overline{y}$}]{positive}{$\operatorname{pos}^i(\overline{x},\overline{y}),\operatorname{neg}^i(\overline{x},\overline{y})$}

\glsxtrnewsymbol[description={the number of maximal $i$-cliques containing $v$}]{maxcl}{$\maxcl^G(v,i)$}
\glsxtrnewsymbol[description={the set of maximal cliques in $G$}]{setclique}{$\setOfMaxCliques$}
\glsxtrnewsymbol[description={the $r$-type distribution of  $\struc{A}$}]{rho}{$\rho_{\struc{A},r}$}

\glsxtrnewsymbol[description={the sampling distance of depth $r$ between two $\sigma$-structures $\struc{A}$ and $\struc{B}$ }]{samplingdist}{$\delta_{\odot}^r(\struc{A},\struc{B})$}

\glsxtrnewsymbol[description={Hanf normal form}]{HNF}{HNF}
\glsxtrnewsymbol[description={proximity oblivious tester}]{POT}{POT}
\glsxtrnewsymbol[description={generalised subgraph freeness}]{GSF}{GSF}

\headers{On Testability of First-Order Properties in Bounded-Degree Graphs}{A. Adler, N. Köhler, P. Peng}

\title{On Testability of First-Order Properties in Bounded-Degree Graphs and Connections to Proximity-Oblivious Testing
}

\author{Isolde Adler\thanks{University of Bamberg, Bamberg, Germany
		(\email{isolde.adler@uni-bamberg.de}
		).}
	\and Noleen Köhler\thanks{Universit\'{e} Paris-Dauphine, PSL University, CNRS UMR7243, LAMSADE, Paris, France
		(\email{noleen.kohler@dauphine.psl.eu}).}
	\and Pan Peng\thanks{University of Science and Technology of China, Hefei, China
		(\email{ppeng@ustc.edu.cn}).}}

\usepackage{amsopn}

\ifpdf
\hypersetup{
  pdftitle={On Testability of First-Order Properties in Bounded-Degree Graphs},
  pdfauthor={A. Adler, N. Köhler, P. Peng}
}
\fi

\begin{document}

\maketitle

\begin{abstract}
 We study property testing of properties that are definable in first-order logic (FO) in the
bounded-degree graph and relational structure models. We show that any FO property that is defined by a formula with quantifier prefix $\exists^*\forall^*$ is testable (i.e., testable with constant query complexity), while there exists an FO property that is expressible by a formula with quantifier prefix  $\forall^*\exists^*$ that is not testable. In the dense graph model, a similar picture is long known (Alon,
Fischer, Krivelevich, Szegedy, Combinatorica 2000), despite the very different nature of the two models. In particular, we obtain our lower bound by an FO formula that defines a class of bounded-degree expanders,
based on zig-zag products of graphs. We expect this to be of independent
interest.

We then use our class of FO definable bounded-degree expanders to answer a long-standing open problem for \emph{proximity-oblivious testers (POTs)}. POTs are a class of particularly simple testing algorithms, where a basic test is performed a number of times that may depend on the proximity parameter, but the basic test itself is independent of the proximity parameter.

In their seminal work, Goldreich and Ron [STOC 2009; SICOMP 2011] show that
the graph properties that are constant-query proximity-oblivious testable in the bounded-degree model are precisely the properties that can be expressed as a \emph{generalised subgraph freeness (GSF)} property that satisfies the \emph{non-propagation} condition.
It is left open whether the non-propagation condition is necessary. Indeed, calling
properties expressible as a generalised subgraph freeness property \emph{GSF-local properties}, they ask
whether all GSF-local properties are non-propagating.
We give a negative answer by showing that our 
is GSF-local and propagating. Hence in particular, our property does not
admit a POT, despite being GSF-local.
For this result we establish a new connection between FO properties and GSF-local properties via neighbourhood profiles.

Finally, motivated by our lower bound and by Hanf-locality of FO, we explore testability of properties that speak about isomorphism types of neighbourhoods.
\end{abstract}



\section{Introduction}
Graph property testing is a framework for studying sampling-based algorithms that solve a relaxation of classical decision problems on graphs. Given a graph $G$ and a property $\classStruc{P}$ (e.\,g.\ triangle-freeness), the goal of a property testing algorithm, called a \emph{property tester}, is to distinguish if a graph satisfies $\classStruc{P}$ or is \emph{far} from satisfying $\classStruc{P}$, where the definition of \emph{far} depends on the model. The general notion of property testing was first proposed by Rubinfeld and Sudan \cite{rubinfeld1996robust}, with the motivation for the study of program checking. Goldreich, Goldwasser and Ron \cite{goldreich1998property} then introduced the property testing for combinatorial objects and graphs. They formalized the \emph{dense graph model} for testing graph properties, in which the algorithm can query if any pair of vertices of the input graph $G$ with $n$ vertices are adjacent or not, and the goal is to distinguish, with probability at least $2/3$, the case of $G$ satisfying a property $\classStruc{P}$ from the case that one has to modify (delete or insert) more than $\varepsilon n^2$ edges to make it satisfy $\classStruc{P}$, for any specified proximity parameter $\varepsilon\in (0,1]$.
A property $\classStruc{P}$ is called testable (in the dense graph model), if it can be tested with constant query complexity, i.e., the number of queries made by the tester is bounded by a function of $\varepsilon$ and is independent of the size of the input graph. Since \cite{goldreich1998property}, much effort has been made on the testability of graph properties in this model, culminating in the work by Alon et al.~\cite{alon2009combinatorial}, who showed that a property is testable if and only if it can be reduced to testing for a finite number of regular partitions. 

Since Goldreich and Ron's seminal work~\cite{GoldreichRon2002} introducing property testing on bounded-degree graphs, much 
attention has been paid to property testing in sparse graphs. 
Nevertheless, our understanding of testability of properties in such graphs is still limited. In the \emph{bounded-degree graph model}~\cite{GoldreichRon2002}, the algorithm has oracle access to the input graph $G$ with maximum degree $d$, which is assumed to be a constant, and is allowed to perform \emph{neighbour queries} to the oracle. That is, for any specified vertex $v$ and index $i\leq d$, the oracle returns the $i$-th neighbour of $v$ if it exists or a special symbol $\bot$ otherwise in constant time. 
A graph $G$ with $n$ vertices is called \emph{$\varepsilon$-far} from satisfying a property $\classStruc{P}$, if one needs to modify more than $\varepsilon dn$ edges to make it satisfy $\classStruc{P}$. 
The goal now becomes to distinguish, with probability at least $2/3$, if $G$ satisfies a property $\classStruc{P}$ or is $\varepsilon$-far from satisfying $\classStruc{P}$, for any specified proximity parameter $\varepsilon\in (0,1]$. 
Again, a property $\classStruc{P}$ is testable in the bounded-degree model, if it can be tested with constant query complexity, where the constant can depend on $\varepsilon, d$ while being  independent of $n$. So far, it is known that some properties are testable, including subgraph-freeness, $k$-edge connectivity, cycle-freeness, being Eulerian, degree-regularity~\cite{GoldreichRon2002}, minor-freeness~\cite{benjamini2010every,hassidim2009local,kumar2019random}, hyperfinite properties \cite{NewmanSohler2013}, $k$-vertex connectivity~\cite{yoshida2012property,forster2019computing}, and subdivision-freeness~\cite{kawarabayashi2013testing}. We now discuss the contributions of this paper. 

\subsection{Our contributions}
\subsubsection{Non-testability of first-order logic} We study the testability of properties definable in first-order  logic (FO) in the bounded-degree graph model. Recall that formulas of first-order logic on graphs are built from predicates for the edge relation and equality, using
Boolean connectives $\vee,
\wedge,\neg$ and universal and existential quantifiers $\forall,\exists$, where the variables represent graph vertices. First-order logic can e.\,g.\ express subgraph-freeness (i.\,e., no isomorphic copy of some fixed graph $H$ appears as a subgraph) and subgraph containment (i.\,e., an isomorphic copy of some fixed $H$ appears as a subgraph).  
Note however, that there are constant-query testable properties, such as connectivity and cycle-freeness, that cannot be expressed in FO. We study the question of which first-order properties are testable in the bounded-degree graph model. 
Our study extends to the bounded-degree \emph{relational structure} model \cite{AdlerH18}, while we focus on the classes of relational structures  
with binary relations, i.e., edge-coloured directed graphs. In this model for relational structures, one can perform neighbour queries,
 querying for both in- and out-neighbours and the edge colour that connects them.
This model is natural in the context of relational databases, where each
(edge-)relation is given by a list of the tuples it contains. 

We consider the testability of first-order 
properties in the bounded-degree model according to 
quantifier alternation, inspired by a similar study for dense graphs by Alon et al.~\cite{alon2000efficient}. On relational structures of bounded-degree over a fixed finite signature, 
we have the following simple observation: Any 
first-order property definable by a sentence \emph{without} quantifier alternations is testable.
This means the sentence either consists of a quantifier prefix of the form $\exists^*$ (any
finite number of existential quantifications), followed by a quantifier-free formula, or it consists of
a quantifier prefix of the form $\forall^*$ (any finite number of universal quantifications), 
followed by a quantifier-free formula.
Basically, every property of the form $\exists^*$ is testable because the structure required by 
the quantifier-free part of the formula can be planted
with a small number of tuple modifications if the input structure is large enough (depending on the formula), and
we can use an exact algorithm to determine the answer in constant time otherwise.
Every property of the form $\forall^*$ is testable because a formula of the form $\forall \bar x \phi(\bar x)$,
where $\phi$ is quantifier-free, is logically equivalent to a formula of the form 
$\neg \exists \bar x \psi(\bar x)$, where $\psi$ is quantifier-free. Testing $\neg \exists \bar x \psi(\bar x)$
then amounts to testing for the absence of a finite number of induced substructures, which can be done
similar to testing subgraph freeness~\cite{GoldreichRon2002}. 
The testability of a property becomes less clear if it is defined by a sentence \emph{with} quantifier alternations. 
Formally, we let $\Pi_2$ (resp. $\Sigma_2$) denote the set of properties that can be expressed by a formula in the $\forall^*\exists^*$-prefix (resp. $\exists^*\forall^*$-prefix) class. 
We obtain the following.
\\[-0.3cm]

\textit{Every first-order property in $\Sigma_2$ is testable in the bounded-degree model (Theorem \ref{thm:sigma2}). On the other hand, there is a first-order property in $\Pi_2$, that is not testable in the bounded-degree model (Theorem~\ref{thm:pi2}).}\\[-0.3cm]

The theorems that we refer to in the above statement speak about relational structures, while we also give a lower bound on graphs (Theorem~\ref{thm:simpleDelta2}),
so the statement also holds when restricted to FO on graphs. 
Interestingly, the above dividing line is the
same as for FO properties in dense graph model~\cite{alon2000efficient},
despite the very different nature of the two models.
Our proof uses a number of new proof techniques, combining graph theory, combinatorics and logic. 

We remark that our lower bound, i.e., the existence of a property in $\Pi_2$ that is not testable, is somewhat 
astonishing (on an intuitive level)  
due to the following two reasons. Firstly, it is proven by constructing a 
first-order definable class of structures that encode a class of expander graphs, which highlights that FO is surprisingly expressive on bounded-degree graphs, despite its locality~\cite{Hanf1965,Gaifman82,FaginStockmeyerVardi1995}.   
Secondly, it is known that property testing algorithms in the bounded-degree model proceed by sampling vertices from the input graph and exploring their local neighbourhoods, and FO can only express `local' properties, while our lower bound shows that this is not sufficient for testability. We elaborate on this in the following.
On one hand, Hanf's Theorem~\cite{Hanf1965} gives insight into first-order logic on graphs of bounded-degree and 
implies a strong normal form, called \emph{Hanf Normal Form} (HNF) in~\cite{BolligKuske2012}, which we briefly sketch.
For a graph $G$ of maximum degree $d$ and a vertex $x$ in $G$, the
neighbourhood of fixed radius $r$ around $x$ in $G$ can be described by a first-order formula $\tau_r(x)$, up to isomorphism.
A \emph{Hanf sentence} is a first-order sentence of the form `there are at least $\ell$ vertices $x$ of 
neighbourhood (isomorphism) type $\tau_r(x)$'. A  first-order sentence is in HNF, if it is a Boolean combination of Hanf sentences.
By Hanf's Theorem, every first-order sentence is equivalent to a sentence in HNF on bounded-degree graphs~\cite{Hanf1965,FaginStockmeyerVardi1995,EF95}. Note that Hanf sentences only speak about local 
neighbourhoods. Hence this theorem gives evidence that first-order logic can only express local
properties. 
On the other hand, 
if a property is constant-query testable in the bounded-degree graph model, then it can be tested by approximating the distribution of local neighbourhoods (see \cite{CzumajPS16} and~\cite{goldreich2011proximity}). That is, a constant-query tester can essentially only test properties that are close to being defined by a distribution of local neighbourhoods. For these reasons\footnote{Furthermore, previously, typical FO properties were all known to be testable, including degree-regularity for a fixed given degree, containing a $k$-clique and a dominating set of size $k$ for fixed $k$ (which are trivially testable), and the aforementioned subgraph-freeness and subgraph containment (see e.g.~\cite{goldreich2017introduction}).}, a priori, it could be true that every property that can be expressed in first-order logic is testable in the bounded-degree model. Indeed, the validity of this statement was raised as an open question in~\cite{AdlerH18}. However, our lower bound gives a negative answer to this question.

\subsubsection{GSF-locality is not sufficient for proximity oblivious testing}
Typical property testers make decisions regarding the global property of the graph based on local views only. In the extreme case, a tester could make the size 
of the local views independent of the distance $\varepsilon$ to a predetermined set of graphs. Motivated by this, Goldreich and Ron \cite{goldreich2011proximity} initiated the study of (one-sided error) \emph{proximity-oblivious testers (POTs)} for graphs, where a tester simply repeats a basic test for a number of times that depends on the proximity parameter, while the basic tester is oblivious of the proximity parameter. They gave characterizations of graph properties that can be tested with constant query complexity by a POT in both dense graph model and the bounded-degree model. In each model, it is known that the class of properties that have constant-query POTs is a strict subset of the class of properties that are testable (by standard testers).

Informally, 
a (one-sided error) POT for a property $\mathcal{P}$ is a tester that always accepts a graph $G$ if it satisfies $\mathcal{P}$, and rejects $G$ with probability that is a monotonically increasing function of the distance of $G$ from the property $\mathcal{P}$. We say $\mathcal{P}$ is \emph{proximity-oblivious testable} if such a tester exists for $\mathcal{P}$ with constant query complexity. To characterise the class of proximity-oblivious testable properties in the bounded-degree model, Goldreich and Ron \cite{goldreich2011proximity} introduced a notion of generalized subgraph freeness (GSF), that extends the notions of induced subgraph freeness and (non-induced) subgraph freeness. A graph property is called a \emph{GSF-local} property if it is expressible as a GSF property. It has been shown in \cite{goldreich2011proximity} that a graph property is constant-query proximity-oblivious testable if and only if it is a GSF-local property that satisfies a so-called \emph{non-propagation} condition. 
Informally, a GSF-local property $\mathcal{P}$ is non-propagating if repairing a graph $G$ that does not satisfy $\mathcal{P}$ does not trigger a global ``chain reaction'' of necessary modifications.  
We refer the reader to Section~\ref{sec:gsf_preliminaries} for the formal definitions. 

A major question that is left open in~\cite{goldreich2011proximity}
 is whether every GSF-local property satisfies the non-propagation condition.
By using the aforementioned non-testable FO property and establishing a new connection between FO properties and GSF-local properties, we resolve this question 
by showing the following negative result.  \\[-0.3cm]

\textit{	There exists a GSF-local property of graphs of degree at most $3$ that is not testable in the bounded-degree model. Thus, not all GSF-local properties are non-propagating (Theorem~\ref{thm:existenceLocalNonTestableProperty}).}\\[-0.3cm]

We expect this result will shed some light on a full characterisation of testable properties in the bounded-degree model. Indeed, in a recent work by Ito, Khoury and Newman \cite{ito2019characterization}, the authors gave a characterization of testable \emph{monotone}
graph properties and testable \emph{hereditary}
graph properties with one-sided error in the bounded-degree graph model; and they asked the open question ``\emph{is every property that is defined by a set of forbidden configurations testable?}''. Since their definition of a property defined by a set of ``forbidden configuration'' is equivalent to a GSF-local property, 
our result above also gives a negative answer to their question.

We complete the picture by showing the following. \\[-0.3cm]

\textit{Every GSF-local property of graphs of degree at most $2$ is non-propagating (Theorem~\ref{thm:degreeTwoCase}).}

\subsubsection{Neighbourhood freeness and neighbourhood regularity.} 
Motivated by our lower bounds, we turn back to FO sentences in Hanf-normal form. While Hanf sentences are testable (they are in $\Sigma_2$) we  ask whether properties defined by negated Hanf sentences are testable. Towards this, we  
give testers with constant query complexity for some first-order 
properties that speak about isomorphism types of neighbourhoods. Given a bounded-degree graph, an \emph{$r$-ball} around a vertex $x$ is the neighbourhood of radius $r$ around $x$ in the graph. We call the isomorphism types of $r$-balls \emph{$r$-types}. 
We consider two basic such properties, called \emph{$\tau$-neighbourhood regularity} and \emph{$\tau$-neighbourhood-freeness}, that correspond to ``all vertices have $r$-type $\tau$'' and ``no vertex has $r$-type $\tau$'', respectively. (Neighbourhood-regularity can be seen as a generalisation of degree-regularity, which is known to be testable \cite{goldreich2017introduction} and testing neighbourhood freeness corresponds to testing a negated Hanf sentence.) 
As we show in Lemma \ref{lem:existencesigma2}, there exist $1$-types $\tau,\tau'$ such that neither $\tau$-neighbourhood-freeness nor $\tau'$-neighbourhood regularity can be defined by a formula in $\Sigma_2$.  
Thus, our previous tester for $\Sigma_2$ cannot be applied to these properties. We give constant-query testers for them under certain conditions on 
	$\tau$ (Theorem~\ref{thm:dNeighbourhoodFreeness}, \ref{thm:1NeighbourhoodFreeness} and \ref{thm:neighbourhoodRegularity}).
	Both $\tau$-neighbourhood-freeness and $\tau$-neighbourhood regularity can be defined by formulas in $\Pi_2$ for any neighbourhood type $\tau$. Thus, our results imply that there are properties defined by formulas in $\Pi_2\setminus \Sigma_2$ which are testable.

\subsection{Our techniques}
\subsubsection{On the testability and non-testability of FO properties} 
For showing that every property $\classStruc{P}$ defined by a formula $\varphi$ in $\Sigma_2$ (i.e.\ of the form $\exists^*\forall^*$) is testable, we show that $\classStruc{P}$ is equivalent to 
{the union of properties $\classStruc{P}_i$, each of which is} `indistinguishable' from a property $\classStruc{Q}_i$ that is defined by a formula of form $\forall^*$. Here the indistinguishability means we can transform any structure satisfying $\classStruc{P}_i$, into a structure satisfying $\classStruc{Q}_i$ by modifying a small fraction of the tuples of the structure and vice versa. This allows us to reduce the problem of testing $\classStruc{P}$ to testing properties defined by $\forall^*$ formulas.  
Then the testability of $\classStruc{P}$ follows, as any property of the form $\forall^*$ is testable and testable properties are closed under union \cite{goldreich2017introduction}. The main challenge here is to deal with the interactions between existentially quantified variables and universally quantified variables.
Intuitively, the degree bound limits the structure that can be imposed by the universally quantified variables. Using this, we are able to deal with the existential variables together with these interactions by `planting' a required
constant size substructure in such a way, that we are only a constant number of modifications `away' from a formula of the form $\forall^*$.

Complementing this, we use Hanf's theorem to observe that every FO property on degree-regular structures is in $\Pi_2$ (see Lemma \ref{lem:d-regHNF}). Thus to prove that there exists a property defined by a formula in $\Pi_2$ which is not testable, it suffices to show the existence of an FO property that is not testable and degree-regular.  
For the latter, 
we note that it suffices to construct a formula $\phi$, that defines a class of relational structures with binary relations only (edge-coloured directed graphs) whose underlying undirected graphs are expander graphs. To see this, we use an earlier result that if a property is constant-query testable, then the distance between the local (constant-size) neighbourhood distributions of a relational structure $\struc{A}$  satisfying the property $\phi$ and a relational structure $\struc{B}$ that is $\varepsilon$-far from having the property must be relatively large (see \cite{AdlerH18} which in turn is built upon the so-called ``canonical testers'' for bounded-degree graphs in \cite{CzumajPS16,goldreich2011proximity}). We then exploit a result of Alon (see Proposition 19.10 in~\cite{LovaszBook2012}), that the neighbourhood distribution of an arbitrarily large relational structure $\struc{A}$
can be approximated by the neighbourhood distribution of a structure $\struc{H}$ of small constant size. Thus, for any $\struc{A}$ in $\phi$, by taking the union of ``many'' disjoint copies of the ``small'' structure $\struc{H}$, we obtain another structure $\struc{B}$ such that the local neighbourhood distributions of $\struc{A}$ and $\struc{B}$ have small distance. If the underlying undirected graphs of the structures in $\phi$ are expander graphs, it immediately follows that $\struc{B}$ is far from the property defined by the formula $\phi$, from which we can conclude that the property $\phi$ is not testable. We remark that for simple undirected graphs, it was known before that any property that only consists of expander graphs is not testable~\cite{fichtenberger2019every}.

Now we construct a formula $\phi$, that defines a class of relational structures with binary relations only whose underlying
undirected graphs are expander graphs, arising from the zig-zag product by Reingold, Vadhan and Wigderson~\cite{Reingold00entropywaves}. For expressibility in FO, we hybridise the zig-zag construction of 
expanders with a tree structure. Roughly speaking, we start with a small graph
$H$, which is a good expander, and the formula $\phi$ expresses that each model\footnote{When the context is clear, we use ``model'' to indicate that a structure satisfies some formula. This should not be confused with the names for our computational models, e.g., the bounded-degree model.}
looks like a rooted $k$-ary tree (for a suitable fixed $k$), where level $0$ consists of the root only,  
level $1$ contains $G_1:=H^2$, and level $i$ contains the zig-zag product
of $G_{i-1}^2$ with $H$. The class of trees is not definable in FO. 
However, we achieve that every finite model of our formula is connected and looks like a 
$k$-ary tree with the desired graphs on the levels.
This structure is
obtained by a recursive `copying-inflating' mechanism, 
to mimic the expander construction locally between consecutive levels.
For this we use a constant number of edge-colours, one set of colours for the edges of the tree, and  
another for the edges of the `level' graphs $G_i$.
On the way, many technicalities need to be tackled, such as encoding the zig-zag construction
into the local copying mechanism (and achieving the right degrees), and finally proving connectivity.
We then show that the underlying 
undirected graphs of the models of $\phi$ are expander graphs. 
Using a hardness reduction which inserts carefully designed gadgets to encode the different edge-colours, we finally 
obtain a non-testable property of undirected $3$-regular graphs.

\subsubsection{On GSF-locality and POTs} We then proceed to showing that this property of $3$-regular graphs is GSF-local.
For this, we first study the relation between locality of first-order logic and GSF-locality.
Hanf's Theorem \cite{Hanf1965} implies that we can understand locality of FO as prescribing upper and lower bounds for the number of occurrences of certain local neighbourhood (isomorphism) types. 
On the other hand, a GSF-local property as defined in \cite{goldreich2011proximity} prescribes the absence of some constant-size \emph{marked} graphs, where a marked graph $F$ specifies an induced subgraph and how it `interacts' with the rest of the graph (see Definition \ref{def:gsf}). Intuitively, such a property just specifies a condition that the local neighbourhoods of a graph $G$ should satisfy, i.e., certain types of local neighbourhoods cannot occur in $G$, or equivalently, these types have $0$ occurrences. However, it does not follow that every GSF-local property is FO definable, because the set of forbidden marked graphs depends on the size $n$ 
of the graphs in the class. Indeed, it is not hard to come up with
undecidable properties that are GSF-local.

To establish a connection between FO properties and GSF-local properties, we first encode the bounds on the number of occurrences of local neighbourhood types into what we call \emph{neighbourhood profiles}, and characterise FO definable properties of bounded-degree relational structures as finite unions of properties defined by neighbourhood profiles (Lemma~\ref{lemma:FO-neighbourhood}). We then show that every FO formula defined by a non-trivial finite union of properties each of which is defined by a \emph{$0$-profile}, \ie the prescribed lower bounds are all $0$, is GSF-local (Theorem~\ref{thm:subsetOfFOIsLocal}). Given the fundamental roles of local properties in graph theory, graph limits \cite{LovaszBook2012}, we believe this new connection is of independent interest. 

For technical reasons, we make use of the property defined by our formula $\phi$ above, which is a property of \emph{relational structures} that is not testable in the bounded-degree model, instead of directly using our non-testable graph property of $3$-regular graphs. We prove that the property defined by $\phi$
can actually be defined by $0$-profiles (Lemma \ref{lem:neighbouhoodProfilOfPZigZag}). 
We then derive that our non-testable graph property of $3$-regular graphs is also GSF-local (Lemma~\ref{lemma:graphproperty_gsf_local}), by showing that the reduction 
maintains definability by $0$-profiles.  

\subsubsection{On testing neighbourhood regularity, neighbourhood-freeness} In order to obtain our testers for $\tau$-neighbourhood regularity and $\tau$-neighbourhood-freeness, we show that if a graph $G$ is $\varepsilon$-far from having the property, it contains a linear fraction of constant-size neighbourhoods certifying that $G$ does not satisfy the property. Such a statement may be intuitively true, but it is tricky to prove. Assume we want to test for $\tau$-freeness, for some fixed $r$-neighbourhood type $\tau$,
	and assume a graph $G$ has one vertex $x$ with forbidden neighbourhood of type $\tau$.
	Changing the $r$-neighbourhood of $x$  
	by edge modifications, in order to remove $\tau$, might introduce new forbidden 
	neighbourhoods around vertices close to $x$, triggering
	a `chain reaction' of necessary modifications.
	This means that a graph might be $\epsilon$-far from being $\tau$-free, but we do not see it by sampling constantly many neighbourhoods in the graph. Such a subtle difficulty has already been observed for testing degree-regularity (see Claim 8.5.1 in \cite{goldreich2017introduction}). We show that under appropriate assumptions, such a `chain reaction' can be bypassed 
	by carefully fixing the neighbourhood of $x$ without changing the neighbourhood type of the vertices surrounding $x$. 
	Though fairly simple, it provides non-trivial analysis, handling the subtle difficulty of relating local distance to global distance without triggering a `chain reaction'.

\subsection{Other related work}
Besides the aforementioned works on testing properties with constant query complexity in the bounded-degree graph model, Goldreich and Ron~\cite{goldreich2011proximity} have obtained a characterisation for a class of properties that are testable by a constant-query proximity-oblivious tester in bounded-degree graphs (and dense graphs). Such a class is a rather restricted subset of the class of all constant-query testable properties. 
Fichtenberger et al.~\cite{fichtenberger2019every} showed that every testable property is either finite or contains an infinite hyperfinite subproperty (see Definition~\ref{def:hyperfinite}). Ito et al.~\cite{ito2019characterization} gave characterisations of one-sided error (constant-query) testable monotone graph properties, and one-sided error testable hereditary graph properties in the bounded-degree (directed and undirected) graph model.

In the bounded-degree graph model,   
there are also properties (e.g.\ bipartiteness, expansion, $k$-clusterability) that require $\Omega(\sqrt{n})$  
queries, and properties (e.g.\ $3$-colorability) that require $\Omega(n)$ queries. We refer the reader to Goldreich's recent book~\cite{goldreich2017introduction}.

Property testing on relational structures was recently motivated by the application in databases. Besides the aforementioned work \cite{AdlerH18}, Chen and Yoshida  \cite{chen2019testability} studied the testability of relational database queries for each relational structure in the framework of property testing.

The notion of POT was implicitly defined in~\cite{blum1993self}. Goldreich and Shinkar~\cite{goldreich2016two}  studied two-sided error POTs for both dense graph and bounded-degree graph models. Goldreich and Kaufman~\cite{goldreichkauf2011proximity} investigated the relation between local conditions that are invariant in an adequate sense and properties that have a constant-query proximity-oblivious testers. 
 
This paper combines and extends the results of two conference papers, \cite{AdlerKP21} 
and~\cite{AdlerK021}.
In this paper we modified the
property for the lower bound in~\cite{AdlerKP21} slightly so that it is GSF-local, which allows
us to use it both as a non-testable $\Pi_2$-property and as a GSF-local property that is
propagating. We also give an improved reduction from relational structures to undirected graphs which reduces the original (large) degree bound to $3$.
Finally, we complete the picture by showing that all GSF-local properties of degree at most $2$
are non-propagating. 
\subsection{Structure of the paper}
Section~\ref{sec:preliminaries} contains the preliminaries, including logic,
property testing and the zig-zag construction of expander graphs.
In Section~\ref{sec: definitionFormula} we construct the FO formula $\phi$ and prove properties 
of its models.
In Section~\ref{sec:FOnontestability}, we prove that there is a $\Pi_2$-property that is not testable, by proving that the property defined by $\phi$ on bounded-degree structures is not constant-query testable. Using a reduction, in Section~\ref{sec:reduction_graphs} we then 
provide a $\Pi_2$-property of undirected graphs of degree at most $3$ that is non-testable. In Section \ref{sec:testableSigma2}, we show that all $\Sigma_2$ properties are testable.
In Section~\ref{sec:GSFlocality} we then turn to POTs, showing that our $\Pi_2$-property of undirected graphs of degree at most $3$ is GSF-local and propagating. We then show that 
all GSF-local properties of degree at most $2$ are non-propagating.
In Section~\ref{sec:freeness} we give positive results for some first-order properties that
speak about isomorphism types of neighbourhoods. We conclude in Section \ref{sec:conclusion}. 

\section{Preliminaries}\label{sec:preliminaries} 
We let $\mathbb{N}$ denote the set of natural numbers including $0$, and $\Npos:=\mathbb N\setminus\{0\}$. For $n\in \mathbb{N}$ we let $[n]:=\{0,1,\dots,n-1\}$ denote the set of the first $n$ natural numbers. For a set $S$ and $k\in \mathbb{N}$ we denote the Cartesian product $S\times\dots \times S$ of $k$ copies of $S$ by $S^k$. We use $\binom{S}{2}$ 
to denote the set of all two-element subsets of $S$, we denote the disjoint union of sets by $\sqcup$ and the symmetric difference by $\triangle$.
\subsection{Undirected graphs}\label{app:undirectedGraphs}
Unless otherwise specified we allow graphs to have self-loops and parallel edges. We represent an undirected graph $G$ as a triple $(V,E,f)$, where $V$ is the set of vertices, $E$ is the set of edges and $f:E\rightarrow V\cup \binom{V}{2}$ is the incidence map. An isomorphism from $G_1=(V_1,E_1,f_1)$ to $G_2=(V_2,E_2,f_2)$ is a pair of bijective maps $(h_V,h_E)$, where $h_V:V_1\rightarrow V_2$ and $h_E:E_1\rightarrow E_2$, such that $h_V(f_1(e))=f_2(h_E(e))$ for any $e\in E_1$, where $h_V(X):=\{h_V(x)\mid x\in X\}$ for any set $X\subseteq V_1$. Undirected graphs without self-loops and parallel edges are called \emph{simple}. 
For a simple graph $G$, we also represent $G$ as a tuple $G=(V(G),E(G))$, where $V(G)$ is the vertex set and
$E(G)\subseteq \binom{V}{2}$. 
The \emph{degree} $\deg_G(v)$ of a vertex $v$ in a graph $G$ is the number of edges to which $v$ is incident. In particular, self-loops contribute one to the degree. We will say that a graph $G$ is \emph{$d$-regular} for some $d\in \mathbb{N}$ if every vertex in $G$ has degree $d$. We specify paths in graphs by tuples of vertices. We further let all paths and cycles be simple, \ie no vertex appears twice.   
The \emph{length} of a path on $n$ vertices is $n-1$.  We define the distance between two vertices $v$ and $w$ in a graph $G$, denoted $\dist_G(v,w)$, as the length of a shortest path from $v$ to $w$ or $\infty$ if there is no path from $v$ to $w$ in $G$. Any subset $S\subseteq V$ of vertices \emph{induces} a graph  $G[S]:=(S,\{e\in E\mid f(e)\in S\cup \binom{S}{2}\},f|_{S})$. A \emph{connected component} of $G$ is a graph induced by a maximal set $S$, such that each pair $v,w\in S$ has finite distance in $G$. A graph is connected if it has only one connected component. We refer the reader to~\cite{diestelBook}
for the basic notions of graph theory.

We also consider rooted undirected trees. By specifying a root we can uniquely direct the edges away from the root. This allows us to use the terminology of \emph{children} and \emph{parents} for undirected rooted trees. We call a  tree \emph{$k$-ary} if every vertex has either none or exactly $k$ children and we call it complete if, for every $i\in \mathbb N$, there are either exactly $k^i$ or no vertices of distance $i$ to the root of the tree.

\subsection{Relational structures and first-order logic}
We will briefly introduce structures and first-order logic and point the reader to~\cite{EF95} for a more detailed introduction.
A  (relational) \textit{signature} is a finite set $\sigma =\{R_1,\dots,R_\ell\}$ of relation symbols $R_i$. Every relation symbol $R_i$, $1\leq i\leq \ell$  has an arity  $\ar(R_i)\in \Npos$.
A \textit{$\sigma$-structure} is a tuple $\struc{A}=(\univ{A},\rel{R_1}{\struc{A}},\dots,\rel{R_\ell}{\struc{A}})$, where $\univ{A}$ is a \emph{finite} set, called the \emph{universe} of $\struc{A}$ and $\rel{R_i}{\struc{A}}\subseteq \univ{A}^{\ar(R_i)}$ is an $\ar(R_i)$-ary relation on $\univ{A}$. Note that if $\sigma=\{E_1,\ldots,E_{\ell}\}$ is a signature where each $E_i$ is a binary relation
symbol, then $\sigma$-structures are directed graphs 
with $\ell$ edge-colours. Let $\sigma_{\operatorname{graph}}:=\{E\}$ be a signature with one binary relation symbol $E$. Then we can understand undirected simple graphs as $\sigma_{\operatorname{graph}}$-structures for which the relation $E$ is symmetric (every undirected edge is represented by two tuples) and irreflexive. Using this we can transfer all notions defined below to simple graphs. Typically we name graphs $G,H,F$, we denote the set of vertices of a graph $G$ by $V(G)$, the set of edges by $E(G)$ and vertices are typically named $u,v,w,u',v',w',\dots$. In contrast when we talk about a general relational structure we use $A,B$ and $a,b,a',b',\dots$ to denote elements from the universe.

In the following we let $\sigma$ be a relational signature.
Two $\sigma$-structures $\struc{A}$ and $\struc{B}$ are \emph{isomorphic} if there is a bijective map from $\univ{A}$ to $\univ{B}$ that preserves all relations. 
For a $\sigma$-structure $\struc{A}$ and a subset $S\subseteq \univ{A}$, we let $ \struc{A}[S]$
denote the \emph{substructure} of $ \struc{A}$ \emph{induced} by $S$, i.\,e.\ $ \struc{A}[S]$ has universe $S$ and $\rel{R}{\struc{A}[S]}:=\rel{R}{ \struc{A}}\cap S^{\text{ar}(R)}$ for all $R\in \sigma$.
The \emph{degree} of an element $a\in \univ{A}$ denoted by $\deg_{\struc{A}}(a)$ is defined to be the number of tuples in $\struc{A}$ 
containing $a$.
We define the \textit{degree} of $\struc{A}$, denoted by $\deg(\struc{A})$, to be the maximum degree of its elements. A structure $\struc{A}$ is \emph{$d$-regular} for some $d\in \mathbb{N}$ if every element $a\in \univ{A}$ has degree $d$.
Given a signature $\sigma$ and a constant $d$, we let $\classStruc{C}_{\sigma,d}$ be the class of all $\sigma$-structures of degree at most $d$, and let $\mathcal{C}_d$ the set of all graphs of degree at most $d$. Note that the degree of a graph differs by exactly a factor $2$ from the degree of the corresponding $\sigma_{\operatorname{graph}}$-structure. Let $\classStruc{C}$ be any class of $\sigma$-structures which is closed under isomorphism. A \emph{property} $\classStruc{P}$ in $\classStruc{C}$ is a subset of $\classStruc{C}$ which is closed under isomorphism. We say that a structure $\struc{A}$ has property $\classStruc{P}$ if $\struc{A}\in \classStruc{P}$. 

Syntax and semantic of FO is defined in the usual way (see \eg \cite{EF95}).  
We use $\exists^{\geq m}x\,\phi$ (and $\exists^{= m}x\,\phi$, $\exists^{\leq m}x\,\phi$, respectively)
as a shortcut for the FO formula expressing that  the number of witnesses $x$ satisfying $\phi$
is at least $m$ (exactly $m$, at most $m$, respectively).
We say that a variable occurs \emph{freely} in an FO formula if at least one of its occurrences is not bound by any quantifier.
We use $\varphi(x_1,\dots,x_k)$ to express that the set of variables which occur  freely in the FO formula $\varphi$ is a subset of $\{x_1,\dots,x_k\}$. For a formula $\varphi(x_1,\dots,x_k)$, a $\sigma$-structure $\struc{A}$ and $a_1,\dots,a_k\in \univ{A}$ we write $\struc{A}\models \varphi(a_1,\dots,a_k)$ if $\varphi$ evaluates to true after assigning $a_i$ to $x_i$, for $1\leq i\leq k$. A \emph{sentence} of FO is a formula with no free variables. For an FO sentence $\varphi$ we say that $\struc{A}$ is a \emph{model} of $\varphi$ or $\struc{A}$ satisfies $\varphi$ if $\struc{A}\models \varphi$. Let $\classStruc{C}$ be a class of $\sigma$-structures closed under isomorphism. Every FO-sentence $\phi$ over $\sigma$ defines a property $\classStruc{P}_\phi\subseteq \classStruc{C}$ on $\classStruc{C}$, where
$\classStruc{P}_\phi:=\{\struc{A}\in \classStruc{C}\mid \struc{A}\models \varphi\}$. 

\paragraph{Hanf normal form}
The \textit{Gaifman graph} of a $\sigma$-structure $\struc{A}$ is the undirected graph $\gaifman{\struc{A}}=(\univ{A},E)$, where $\{v,w\}\in E$, if $v\not=w$ and there is an $R\in \sigma$ and a tuple $\overline{a}=(a_1,\dots,a_{\ar(R)})\in \rel{R}{\struc{A}}$, such that $v=a_j$ and $w=a_k$ for some $1\leq k,j\leq \ar(R)$. We use $\gaifman{\struc{A}}$ to apply graph theoretic notions to relational structures. Note that for any graph the Gaifman graph of the corresponding  symmetric $\sigma_{\operatorname{graph}}$-structure is the graph itself.
We say that a $\sigma$-structure $\struc{A}$ is \emph{connected} if its Gaifman graph $\gaifman{\struc{A}}$ is connected.
For two elements $a,b\in \univ{A}$, we define the \emph{distance} between $a$ and $b$ in $\struc{A}$, denoted by
$\dist_{\struc{A}}(a,b)$, as the length of a shortest path from $a$ to $b$ in $\gaifman{\struc{A}}$, or $\infty$ if there is no such path. 
For $r\in \mathbb{N}$ and   $a\in \univ{A}$, the \textit{$r$-neighbourhood} of $a$ is the set $N_r^{\struc{A}}(a):=\{b\in \univ{A}: \dist_{\struc{A}}(a,b)\leq r\}$. We define $\mathcal{N}_r^{\struc{A}}(a):=\struc{A}[N_r^{\struc{A}}(a)]$ to be the substructure of $\struc{A}$ induced by the $r$-neighbourhood of $a$. 
For $r\in \mathbb{N}$ an \emph{$r$-ball} is a tuple $(\struc{B},b)$, where $\struc{B}$ is a $\sigma$-structure, $b\in \univ{B}$ and $\univ{B}=N_r^{\struc{B}}(b)$, \ie $\struc{B}$ has radius $r$ and $b$ is the centre.
Note that by definition $(\mathcal{N}_r^{\struc{A}}(a),a)$ is an $r$-ball for any $\sigma$-structure $\struc{A}$ and $a\in \univ{A}$. Two $r$-balls $(\struc{B},b),(\struc{B}',b')$ are isomorphic if there is an isomorphism of $\sigma$-structure from $\struc{B}$ to $\struc{B}'$ that maps $b$ to $b'$.
We call the isomorphism classes of $r$-balls \emph{$r$-types}. 
For an $r$-type $\tau$ and an element $a\in \univ{A}$ we say that $a$ \emph{has} ($r$-)type $\tau$  if $(\mathcal{N}_r^{\struc{A}}(a),a)\in \tau$. 
Moreover, given such an $r$-type $\tau$, there is a formula $\phi_{\tau}(x)$ such that 
for every $\sigma$-structure $\struc{A}$ and for every $a\in \univ{A}$, $\struc{A}\models\phi_{\tau}(a)$ iff
$(\mathcal{N}_r^{\struc{A}}(a),a)\in \tau$.
A \emph{Hanf-sentence} is a sentence of the form $\exists ^{\geq m} x \phi_{\tau}(x)$, for some $m\in\Npos$, where $\tau$ is an 
$r$-type.  
An FO sentence is in \emph{Hanf normal form}, if it is a Boolean combination\footnote{By Boolean combination we 
	always mean \emph{finite} Boolean combination.} of Hanf sentences.
Two formulas $\phi(x_1,\dots,x_k)$ and $\psi(x_1,\dots,x_k)$ of signature $\sigma$ are called
\emph{$d$-equivalent}, denoted by $\phi(x_1,\dots,x_k) \equiv_d\psi(x_1,\dots,x_k)$, if they are equivalent on $\classStruc{C}_{\sigma,d}$, i.\,e.\ for all $\struc{A}\in \classStruc{C}_{\sigma,d}$ and all 
$(a_1,\dots,a_k) \in \univ{A}^{k}$  we have 
$\struc{A}\models\phi(a_1,\dots,a_k)$ iff $\struc{A}\models\psi(a_1,\dots,a_k)$.
Hanf's locality theorem for first-order logic~\cite{Hanf1965} implies the following.

\begin{theorem}[Hanf~\cite{Hanf1965}]\label{thm:Hanf}
	Let $d\in\mathbb N$. Every sentence of first-order logic is $d$-equivalent to a 
	sentence in Hanf normal form.
\end{theorem} 
\paragraph{Quantifier alternations of first-order formulas}
Let $\sigma$ be any relational signature.
We use the following recursive definition, classifying first-order formulas according to the number of quantifier alterations in their quantifier prefix. Let $\Sigma_0=\Pi_0$ be the class of all quantifier free first-order formulas over $\sigma$. Then for every $i\in \Npos$ we let $\Sigma_i$ be the set of all FO formulas $\varphi(y_1,\dots,y_\ell)$ for which there is $k\in \mathbb{N}$ and a formula $\psi(x_1,\dots,x_k,y_1,\dots,y_\ell)\in \Pi_{i-1}$ such that 
\[\varphi\equiv \exists x_1\dots \exists x_k\psi(x_1,\dots,x_k,y_1,\dots,y_\ell).\] Analogously, $\Pi_i$ consists of all FO formulas $\varphi(y_1,\dots,y_\ell)$ for which there is $k\in \mathbb{N}$ and a formula $\psi(x_1,\dots,x_k,y_1,\dots,y_\ell)\in \Sigma_{i-1}$ such that 
\[\varphi\equiv \forall x_1\dots \forall x_k\psi(x_1,\dots,x_k,y_1,\dots,y_\ell).\] 
We further say that a property $\classStruc{P}\subseteq \classStruc{C}$ is in $\Sigma_i$ or $\Pi_i$ if there is an FO-sentence $\varphi$ in $\Sigma_i$ or $\Pi_i$, respectively, such that $\classStruc{P}=\classStruc{P}_{\varphi}$.
\begin{example}[Substructure freeness]
	Let $\struc{B}$ be a $\sigma$-structure, and let $d\in \mathbb N$. The property \[\classStruc{P}:=\{\struc{A}\in \classStruc{C}_{\sigma,d}\mid \struc{A} \text{ does not contain }\struc{B}\text{ as substructure}\}\]
	is in $\Pi_1$.
\end{example}

\subsection{Property testing}\label{sec:boundeddeg_model}
In the following, we give definitions of two models for property testing - the bounded-degree model for simple graphs introduced in ~\cite{GoldreichRon2002} and a bounded-degree model for relational structures similar to the model introduced in~\cite{AdlerH18}. The model for relational structures described here is chosen to simplify notation. It differs from the model in ~\cite{AdlerH18} in the way the query access is defined, however, they are equivalent in the sense that testability in either model implies testability in the other model. This can be easily seen using a local reduction as defined in Section~\ref{sec:localreduction}. The bounded-degree model for relational structures extends the bounded-degree model for undirected graphs introduced in ~\cite{GoldreichRon2002} and conforms with the bidirectional model of ~\cite{CzumajPS16}. 

For notational convenience, $\classStruc{C}$ will either denote a class of graphs of bounded-degree $d$ closed under isomorphism, or a class of $\sigma$-structures of bounded-degree $d$ closed under isomorphism for some signature $\sigma$ and some $d\in \mathbb{N}$. Let $\classStruc{P}$ be a property on $\classStruc{C}$. We will further refer to both graphs and $\sigma$-structures as structures. Let $\classStruc{P}_n$ be the subset of $\classStruc{P}$ with $n$ vertices/elements. Thus $\classStruc{P}=\cup_{n\in \mathbb{N}}\classStruc{P}_n$. 
We define the distance of a structure $\struc{A}$ on $n$ vertices/elements to a property $\mathcal{P}=\bigcup_{n\in \mathbb{N}}\mathcal{P}_n$  as $$\dist(\struc{A},\classStruc{P}):=\min_{B\in \classStruc{P}_n}\frac{\sum_{R\in \sigma}|\rel{R}{A}\triangle \rel{R}{B}|}{dn}.$$
For $\epsilon\in (0,1)$ we say that a structure $\struc{A}$ on $n$ vertices/elements is \emph{$\epsilon$-close} to $\classStruc{P}$  if $\dist(\struc{A},\classStruc{P})\leq \epsilon$, that is one can modify $\struc{A}$ into a structure in $\classStruc{P}$ by adding/deleting at most $\epsilon d n$ tuples of $\struc{A}$. We say that $\struc{A}$ is $\epsilon$-far from $\classStruc{P}$ if $\struc{A}$ is not $\epsilon$-close to $\classStruc{P}$.

An algorithm that processes a structure $\struc{A}\in \classStruc{C}$
does not obtain an encoding of $\struc{A}$ as a bit string in the usual way. Instead, we assume that the algorithm receives the number $n$ of elements/vertices of $\struc{A}$, and that the elements/vertices of $\struc{A}$ are numbered $1,2,\ldots,n$. In addition, the algorithm has direct access to $\struc{A}$ using an \emph{oracle} which answers \emph{neighbour queries}
in $\struc{A}$ in constant time. A \emph{query} to a $\sigma$-structure  $\struc{A}$ of bounded-degree $d$ has the form $(a,i)$ for an element $a\in \univ{A}$, $i\in \{1,\dots,d\}$ and is answered by $\operatorname{ans}(a,i):=(R,a_1,\dots,a_{\ar(R)})$ where $(a_1,\dots,a_{\ar(R)})$ is the $i$-th tuple (according to some fixed ordering) containing $a$ and $(a_1,\dots,a_{\ar(R)})\in \rel{R}{\struc{A}}$.  A \emph{query} to a graph $G$ of bounded-degree $d$ has the form $(v,i)$ for  $v\in V(G)$, $i\in \{1,\dots,d\}$ and is answered by $\operatorname{ans}(v,i):=w$ where $w$ is the $i$-th neighbour of $v$.   

 Now we give the formal definitions of standard property testing and proximity-oblivious testing.
\begin{definition}[(Standard) property testing]
	Let $\classStruc{P}=\cup_{n\in \mathbb{N}}\classStruc{P}_{n}$ be a property. An \emph{$\epsilon$-tester} for $\classStruc{P}_n$ is a probabilistic algorithm which, given query access to a structure $\struc{A}\in \classStruc{C}$ with $n$ vertices/elements, 
	\begin{itemize}
		\item accepts $\struc{A}$ with probability $2/3$ if $\struc{A}\in \classStruc{P}_n$.
		\item rejects $\struc{A}$ with probability $2/3$ if $\struc{A}$ is $\epsilon$-far from $\classStruc{P}_n$.
	\end{itemize} 
	We say that a property $\classStruc{P}$ is \emph{testable} if for every $n\in \mathbb{N}$ and $\epsilon\in (0,1)$, there exists an $\epsilon$-tester for $\classStruc{P}_n$ that makes at most $q=q(\epsilon,d)$ queries. We say the property $\classStruc{P}$ is testable with \emph{one-sided error} if the $\epsilon$-tester always accepts $\struc{A}$ if $\struc{A}\in \classStruc{P}$. 
\end{definition}
We introduce below the formal definition of proximity-oblivious testers. 
\begin{definition}[Proximity-oblivious testing (with one-sided error)]
	Let $\classStruc{P}=\cup_{n\in \mathbb{N}}\classStruc{P}_{n}$ be a property. Let $\eta:(0,1]\to (0,1]$   
	be a monotonically non-decreasing  
	function. A \emph{proximity-oblivious tester (POT)} with detection probability $\eta$ for $\classStruc{P}_n$ is a probabilistic algorithm which, given query access to a structure $\struc{A}\in \classStruc{C}$ with $n$ vertices/elements,
	\begin{itemize}
		\item accepts $\struc{A}$ with probability $1$ if $\struc{A}\in \classStruc{P}_n$.
		\item rejects $\struc{A}$ with probability at least $\eta(\dist(\struc{A},\classStruc{P}_n))$ if $\struc{A}\notin\classStruc{P}_n$. 
	\end{itemize} 
	We say that a property $\classStruc{P}$ is \emph{proximity-oblivious testable} if for every $n\in \mathbb{N}$, there exists a POT for $\classStruc{P}_n$ of constant query complexity with detection probability $\eta$. 
\end{definition}
We remind the reader of the following which we argued in the introduction.
\begin{remark}\label{rem:Sigma1Pi1}
	Let $d\in \mathbb N$.
	Every propery definable in $\Sigma_1$ is testable on $C_d$, and every property definable in
	$\Pi_1$ is testable on $C_d$. 
\end{remark}

\subsection{Expansion and the zig-zag product}\label{sec:zigZagProduct}

In this section we  recall a construction of a class of expanders introduced in ~\cite{Reingold00entropywaves}. This construction uses undirected graphs with parallel edges and self-loops. 

Let $G=(V,E,f)$ be an undirected $D$-regular graph on $N$ vertices.  We follow the convention that each self-loop counts $1$ towards the degree. Let $I$ be a set of size $D$. 
Then a \emph{rotation map} of $G$ is a function $\rot_G:V\times I\rightarrow V\times I$ such that for every two not necessarily different vertices $u,v\in V$
\begin{displaymath}
|\{(i,j)\in I\times I \mid \rot_G(u,i)=(v,j)\}|=2|\{e\in E\mid f(e)=\{u,v\} \}|
\end{displaymath}
and $\rot_{G}$ is self inverse, i.e. $\rot_G(\rot_G(v,i))=(v,i)$ for all $v\in V, i\in I$. A rotation map is a representation of a graph that additionally fixes for every vertex $v$ an order on all edges incident to $v$. We let the normalised adjacency matrix $M$ of $G$ be 
defined by 
\begin{displaymath}
M_{u,v}:=\frac{1}{D}\cdot|\{e \mid f(e) = \{u,v\}\}|.
\end{displaymath}
Since $M$ is real, symmetric, contains no negative entries and all columns sum up to $1$, all its eigenvalues are  in the real  interval $[-1,1]$. Let $1=\lambda_1\geq \lambda_2\geq \dots\geq \lambda_N\geq -1$ denote the eigenvalues of $M$. We let $\lambda(G):=\max\{|\lambda_2|,|\lambda_N|\}$. Note that these notions do not depend on the rotation map. We say that a graph is an $(N,D,\lambda)$-graph, if $G$ has $N$ vertices, is $D$-regular and $\lambda(G)\leq \lambda$.
We will use the following lemma.
\begin{lemma}[\cite{Hoory06expandergraphs}]\label{lem:connectedBipartiteEigenvalues}
The graph $G$ is connected if and only if $\lambda_2<1$. Furthermore, if $G$ is connected, then $G$ is bipartite if and only if $\lambda_N=-1$.
\end{lemma}
For any subsets $S,T\subseteq V$ let $\langle S,T\rangle_G:=\{e\in E\mid f(e)\cap S\not=\emptyset,f(e)\cap T\not=\emptyset\}$ be the set of edges \emph{crossing} between $S$ and $T$. 
\begin{definition}\label{def:expansionRatio}
	For any set $S\subseteq V$, we let $h(S):=\frac{|\langle S,\overline{S}\rangle_G|}{|S|}$ be 
	the \emph{expansion} of $S$. We let $h(G)$ be the \emph{expansion ratio} of $G$ defined by 
	$h(G):=\min_{\{S\subset V \mid |S|\leq N/2\}}h(S)$.
\end{definition}

For any constant $\epsilon>0$ we call a sequence $\{G_m\}_{m\in\Npos}$ of graphs of increasing number of vertices a \emph{family of $\epsilon$-expanders}, if  $h(G_m)\geq \epsilon$ for all $m\in \Npos$. We say that a  family of graphs is a family of expanders if it is a family of $\epsilon$-expanders for some constant $\epsilon>0$. We further often call  a graph from a family of expanders an expander.
There exists the following connection between $h(G)$ and $\lambda(G)$.
\begin{theorem}[\cite{Dod1984difference,AM1985lambda1}]\label{thm:boundExpansionRatioInTermsOfLambda}
	Let $G$ be a $D$-regular graph. Then it holds that
	$h(G)\geq {D(1-\lambda(G))}/{2}.$
\end{theorem}

This implies that for a sequence of graphs $\{G_m\}_{m\in\Npos}$ of increasing number of vertices, if there is a constant $\epsilon<1$ such that $\lambda(G_m)\leq \epsilon$ for all $m\in \Npos$, then the sequence $\{G_m\}_{m\in \Npos}$ is a family of ${D(1-\varepsilon)}/{2}$-expanders. 
\begin{definition}
	Let $G$ be a $D$-regular graph on $N$ vertices with rotation map $\rot_{G}:V\times I\rightarrow V\times I$ and $I$ a set of size $D$. Then the \emph{square of $G$}, denoted by $G^2$, is a $D^2$-regular graph on $N$ vertices with rotation map
	$\rot_{G^2}(u,(k_1,k_2)):=(w,(\ell_2,\ell_1)),\text{ where}$ 
	\begin{align*}
	\rot_{G}(u,k_1)=&(v,\ell_1) \text{ and }
	\rot_{G}(v,k_2)=(w,\ell_2),
	\end{align*}
	and $u,v,w\in V$, $k_1,k_2,\ell_1,\ell_2\in I$.
\end{definition}
Note that the edges of $G^2$ correspond to walks of length $2$ in $G$ and the adjacency matrix of $G^2$ is the square of the adjacency matrix of $G$.
Note here that if $G$ is bipartite then $G^2$ is not connected, which can be easily seen by using Lemma~\ref{lem:connectedBipartiteEigenvalues}.
\begin{lemma}[\cite{Reingold00entropywaves}]\label{lem:expansionOfSquaring}
	If $G$ is a $(N,D,\lambda)$-graph then $G^2$ is a $(N,D^2,\lambda^2)$-graph.	
\end{lemma}
\begin{definition}
	Let $G_1=(V_1,E_1,f_1)$ be a $D_1$-regular graph on $N_1$ vertices, $I_1$ a set of size $D_1$ and $\rot_{G_1}:V_1\times I_1\rightarrow V_1\times I_1$ a rotation map of $G_1$. Let $G_2=(I_1,E_2,f_2)$ be a $D_2$-regular graph, let $I_2$ be a set of size $D_2$ and $\rot_{G_2}:I_1\times I_2\rightarrow I_1\times I_2$ be a rotation map of $G_2$. Then the \emph{zig-zag product of $G_1$ and $G_2$}, denoted by $G_1\zigzag G_2$, is the $D_2^2$-regular graph on vertex set $V_1\times I_1$ with rotation map given by
	$\rot_{G_1\zigzag G_2}((v,k),(i,j)):=((w,\ell),(j',i')),\text{ where}$
	\begin{align*}
	&\rot_{G_2}(k,i)=(k',i'), \,
	\rot_{G_1}(v,k')=(w,\ell'),\text{ and}
	\rot_{G_2}(\ell',j)=(\ell,j'),
	\end{align*}
	and $v,w\in V_1$, $k,k',\ell,\ell'\in I_1$, $i,i',j,j'\in I_2$.
\end{definition}

The zig-zag product $G_1\zigzag G_2$ can be seen as the result of the following construction. First pick some numbering of the vertices of $G_2$. Then replace every vertex in $G_1$ by a copy of $G_2$ where we colour edges from $G_1$, say, red, and edges from $G_2$ blue. We do this in such a way that the $i$-th edge in $G_1$ of a vertex $v$ will be incident to vertex $i$ of the to-$v$-corresponding-copy of $G_2$. Then for every red edge $(v,w)$ and for every tuple $(i,j)\in I_2\times I_2$ we add an edge to the zig-zag product $G_1\zigzag G_2$ connecting $v'$ and $w'$ where $v'$ is the vertex reached from $v$ by taking its $i$-th blue edge and $w'$ can be reached from $w$ by taking its $j$-th blue edge. 
Figure~\ref{fig:zigZagProduct} shows an example, where in the graph on the right hand side we show the $4$ edges that are added to the zig-zag product for the highlighted edge of the graph on the left hand side.

\begin{figure*}
	\centering
	\begin{tikzpicture}[level distance=7mm,scale = 0.45]
	
	\tikzstyle{ns1}=[line width=1]
	\node[draw,circle,fill=black,inner sep=0pt, minimum width=3pt] (1) at (-5,-3.464) {};
	\node[draw,circle,fill=black,inner sep=0pt, minimum width=3pt] (2) at (-1,-3.464) {};
	\node[draw,circle,fill=black,inner sep=0pt, minimum width=3pt] (3) at (1,-3.464) {};
	\node[draw,circle,fill=black,inner sep=0pt, minimum width=3pt] (4) at (5,-3.464) {};	
	\node[draw,circle,fill=black,inner sep=0pt, minimum width=3pt] (5) at (-4,-1.732) {};
	\node[draw,circle,fill=black,inner sep=0pt, minimum width=3pt] (6) at (-2,-1.732) {};
	\node[draw,circle,fill=black,inner sep=0pt, minimum width=3pt] (7) at (2,-1.732) {};
	\node[draw,circle,fill=black,inner sep=0pt, minimum width=3pt] (8) at (4,-1.732) {};
	\node[draw,circle,fill=black,inner sep=0pt, minimum width=3pt] (9) at (-5,0) {};
	\node[draw,circle,fill=black,inner sep=0pt, minimum width=3pt] (10) at (-1,0) {};
	\node[draw,circle,fill=black,inner sep=0pt, minimum width=3pt] (11) at (1,0) {};
	\node[draw,circle,fill=black,inner sep=0pt, minimum width=3pt] (12) at (5,0) {};
	\node[draw,circle,fill=black,inner sep=0pt, minimum width=3pt] (13) at (-4,1.732) {};
	\node[draw,circle,fill=black,inner sep=0pt, minimum width=3pt] (14) at (-2,1.732) {};
	\node[draw,circle,fill=black,inner sep=0pt, minimum width=3pt] (15) at (2,1.732) {};
	\node[draw,circle,fill=black,inner sep=0pt, minimum width=3pt] (16) at (4,1.732) {};
	\node[draw,circle,fill=black,inner sep=0pt, minimum width=3pt] (17) at (-5,3.464) {};
	\node[draw,circle,fill=black,inner sep=0pt, minimum width=3pt] (18) at (-1,3.464) {};
	\node[draw,circle,fill=black,inner sep=0pt, minimum width=3pt] (19) at (1,3.464) {};
	\node[draw,circle,fill=black,inner sep=0pt, minimum width=3pt] (20) at (5,3.464) {};

	\draw[ns1,C3]
	(-4.5,-4.33)--(1)--(5)--(6)--(2)--(3)--(7)--(8)--(4)--(4.5,-4.33)(-1.5,-4.33)--(2)(1.5,-4.33)--(3)(-5.5,-3.464)--(1)(-5.5,0)--(9)(-5.5,3.464)--(17)(5)--(9)--(13)(8)--(12)--(16)(-4.5,4.33)--(17)--(13)--(14)--(18)--(19)--(15)--(16)--(20)--(4.5,4.33)(-1.5,4.33)--(18)(1.5,4.33)--(19)(5.5,-3.464)--(4)(5.5,0)--(12)(5.5,3.464)--(20)(6)--(10)--(14)(7)--(11)--(15);
	\draw[ns1,C1] (10)--(11);	
	
	\node[minimum height=10pt,inner sep=0,font=\small] at (6.5,0) {$\zigzag$};

	\node[draw,circle,fill=black,inner sep=0pt, minimum width=3pt] (41) at (7.5,-0.866) {};
	\node[draw,circle,fill=black,inner sep=0pt, minimum width=3pt] (42) at (9.5,-0.866) {};
	\node[draw,circle,fill=black,inner sep=0pt, minimum width=3pt] (43) at (8.5,0.866) {};
	\draw[ns1,C2] (41)--(42)--(43)--(41);
	
	\node[minimum height=10pt,inner sep=0,font=\small] at (10.5,0) {$=$};	
	
	\node[draw,circle,fill=black,inner sep=0pt, minimum width=3pt] (21) at (11.5,-3.464) {};
	\node[draw,circle,fill=black,inner sep=0pt, minimum width=3pt] (22) at (12.25,-3.031) {};
	\node[draw,circle,fill=black,inner sep=0pt, minimum width=3pt] (23) at (12.25,-3.897) {};
	\node[draw,circle,fill=black,inner sep=0pt, minimum width=3pt] (24) at (16.5,-3.464) {};
	\node[draw,circle,fill=black,inner sep=0pt, minimum width=3pt] (25) at (15.75,-3.031) {};
	\node[draw,circle,fill=black,inner sep=0pt, minimum width=3pt] (26) at (15.75,-3.897) {};
	\node[draw,circle,fill=black,inner sep=0pt, minimum width=3pt] (27) at (17.5,-3.464) {};
	\node[draw,circle,fill=black,inner sep=0pt, minimum width=3pt] (28) at (18.25,-3.031) {};
	\node[draw,circle,fill=black,inner sep=0pt, minimum width=3pt] (29) at (18.25,-3.897) {};
	\node[draw,circle,fill=black,inner sep=0pt, minimum width=3pt] (30) at (22.5,-3.464) {};
	\node[draw,circle,fill=black,inner sep=0pt, minimum width=3pt] (31) at (21.75,-3.031) {};
	\node[draw,circle,fill=black,inner sep=0pt, minimum width=3pt] (32) at (21.75,-3.897) {};
	\node[draw,circle,fill=black,inner sep=0pt, minimum width=3pt] (33) at (13.5,-1.732) {};
	\node[draw,circle,fill=black,inner sep=0pt, minimum width=3pt] (34) at (12.75,-1.209) {};
	\node[draw,circle,fill=black,inner sep=0pt, minimum width=3pt] (35) at (12.75,-2.165) {};
	\node[draw,circle,fill=black,inner sep=0pt, minimum width=3pt] (36) at (14.5,-1.732) {};
	\node[draw,circle,fill=black,inner sep=0pt, minimum width=3pt] (37) at (15.25,-1.209) {};
	\node[draw,circle,fill=black,inner sep=0pt, minimum width=3pt] (38) at (15.25,-2.165) {};
	\node[draw,circle,fill=black,inner sep=0pt, minimum width=3pt] (39) at (19.5,-1.732) {};
	\node[draw,circle,fill=black,inner sep=0pt, minimum width=3pt] (40) at (18.75,-1.209) {};
	\node[draw,circle,fill=black,inner sep=0pt, minimum width=3pt] (41) at (18.75,-2.165) {};
	\node[draw,circle,fill=black,inner sep=0pt, minimum width=3pt] (42) at (20.5,-1.732) {};
	\node[draw,circle,fill=black,inner sep=0pt, minimum width=3pt] (43) at (21.25,-1.209) {};
	\node[draw,circle,fill=black,inner sep=0pt, minimum width=3pt] (44) at (21.25,-2.165) {};
	\node[draw,circle,fill=black,inner sep=0pt, minimum width=3pt] (45) at (11.5,0) {};
	\node[draw,circle,fill=black,inner sep=0pt, minimum width=3pt] (46) at (12.25,0.433) {};
	\node[draw,circle,fill=black,inner sep=0pt, minimum width=3pt] (47) at (12.25,-0.433) {};
	\node[draw,circle,fill=black,inner sep=0pt, minimum width=3pt] (48) at (16.5,0) {};
	\node[draw,circle,fill=black,inner sep=0pt, minimum width=3pt] (49) at (15.75,0.433) {};
	\node[draw,circle,fill=black,inner sep=0pt, minimum width=3pt] (50) at (15.75,-0.433) {};
	\node[draw,circle,fill=black,inner sep=0pt, minimum width=3pt] (51) at (17.5,0) {};
	\node[draw,circle,fill=black,inner sep=0pt, minimum width=3pt] (52) at (18.25,0.433) {};
	\node[draw,circle,fill=black,inner sep=0pt, minimum width=3pt] (53) at (18.25,-0.433) {};
	\node[draw,circle,fill=black,inner sep=0pt, minimum width=3pt] (54) at (22.5,0) {};
	\node[draw,circle,fill=black,inner sep=0pt, minimum width=3pt] (55) at (21.75,0.433) {};
	\node[draw,circle,fill=black,inner sep=0pt, minimum width=3pt] (56) at (21.75,-0.433) {};
	\node[draw,circle,fill=black,inner sep=0pt, minimum width=3pt] (57) at (13.5,1.732) {};
	\node[draw,circle,fill=black,inner sep=0pt, minimum width=3pt] (58) at (12.75,1.209) {};
	\node[draw,circle,fill=black,inner sep=0pt, minimum width=3pt] (59) at (12.75,2.165) {};
	\node[draw,circle,fill=black,inner sep=0pt, minimum width=3pt] (60) at (14.5,1.732) {};
	\node[draw,circle,fill=black,inner sep=0pt, minimum width=3pt] (61) at (15.25,1.209) {};
	\node[draw,circle,fill=black,inner sep=0pt, minimum width=3pt] (62) at (15.25,2.165) {};
	\node[draw,circle,fill=black,inner sep=0pt, minimum width=3pt] (63) at (19.5,1.732) {};
	\node[draw,circle,fill=black,inner sep=0pt, minimum width=3pt] (64) at (18.75,1.209) {};
	\node[draw,circle,fill=black,inner sep=0pt, minimum width=3pt] (65) at (18.75,2.165) {};
	\node[draw,circle,fill=black,inner sep=0pt, minimum width=3pt] (66) at (20.5,1.732) {};
	\node[draw,circle,fill=black,inner sep=0pt, minimum width=3pt] (67) at (21.25,1.209) {};
	\node[draw,circle,fill=black,inner sep=0pt, minimum width=3pt] (68) at (21.25,2.165) {};
	\node[draw,circle,fill=black,inner sep=0pt, minimum width=3pt] (69) at (11.5,3.464) {};
	\node[draw,circle,fill=black,inner sep=0pt, minimum width=3pt] (70) at (12.25,3.031) {};
	\node[draw,circle,fill=black,inner sep=0pt, minimum width=3pt] (71) at (12.25,3.897) {};
	\node[draw,circle,fill=black,inner sep=0pt, minimum width=3pt] (72) at (16.5,3.464) {};
	\node[draw,circle,fill=black,inner sep=0pt, minimum width=3pt] (73) at (15.75,3.031) {};
	\node[draw,circle,fill=black,inner sep=0pt, minimum width=3pt] (74) at (15.75,3.897) {};
	\node[draw,circle,fill=black,inner sep=0pt, minimum width=3pt] (75) at (17.5,3.464) {};
	\node[draw,circle,fill=black,inner sep=0pt, minimum width=3pt] (76) at (18.25,3.031) {};
	\node[draw,circle,fill=black,inner sep=0pt, minimum width=3pt] (77) at (18.25,3.897) {};
	\node[draw,circle,fill=black,inner sep=0pt, minimum width=3pt] (78) at (22.5,3.464) {};
	\node[draw,circle,fill=black,inner sep=0pt, minimum width=3pt] (79) at (21.75,3.031) {};
	\node[draw,circle,fill=black,inner sep=0pt, minimum width=3pt] (80) at (21.75,3.897) {};
	
	\foreach \x in {21,24,...,78}{
		\pgfmathtruncatemacro{\xx}{\x+1};
		\pgfmathtruncatemacro{\xxx}{\x+2};
		\draw[dotted,ns1,C2!50](\x)--(\xx)--(\xxx)--(\x);	
	}
	
	\draw[dashed,ns1,C3!50]
	(12.5,-4.33)--(23)(22)--(35)(33)--(36)(38)--(25)(24)--(27)(28)--(41)(39)--(42)(44)--(31)(32)--(21.5,-4.33)(15.5,-4.33)--(26)(18.5,-4.33)--(29)(34)--(47)(46)--(58)(12.5,4.33)--(71)(70)--(59)(57)--(60)(62)--(73)(72)--(75)(76)--(65)(63)--(66)(68)--(79)(80)--(21.5,4.33)(15.5,4.33)--(74)(18.5,4.33)--(77)(43)--(56)(55)--(67)(37)--(50)(49)--(61)(40)--(53)(52)--(64)(48)--(51);
	\draw[ns1,C1] (49).. controls (16.7,0.7) and (17.3,0.7) ..(52);
	\draw[ns1,C1] (50).. controls (16.7,-0.7) and (17.3,-0.7) ..(53);
	
	\draw[ns1,C1] (49).. controls (16,-0.5) and (17,-0.2) .. (53);
	\draw[ns1,C1] (50).. controls (15,0) and (15,2) ..(52);
	\end{tikzpicture}
	\caption{Zig-zag product of a $3$-regular grid with a triangle} \label{fig:zigZagProduct}
\end{figure*}
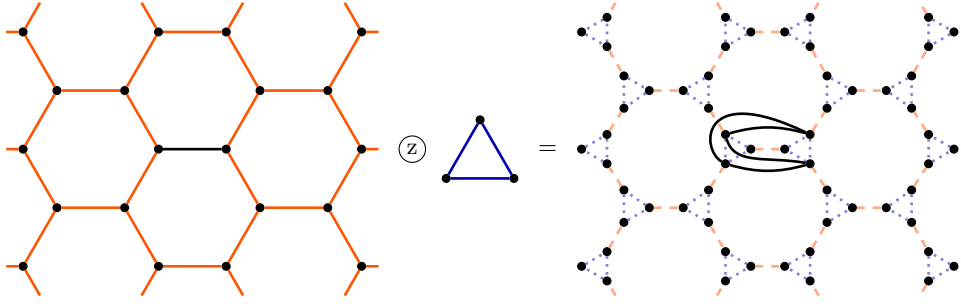

\begin{theorem}[\cite{Reingold00entropywaves}]\label{thm:expansionOfZigZag}
	If $G_1$ is an $(N_1,D_1,\lambda_1)$-graph and $G_2$ is a $(D_1,D_2,\lambda_2)$-graph then $G_1\zigzag G_2$ is a $(N_1\cdot D_1,D_2^2,g(\lambda_1,\lambda_2))$-graph, where 
	\begin{displaymath}
		g(\lambda_1,\lambda_2)=\frac{1}{2}(1-\lambda_2^2)\lambda_1+\frac{1}{2}\sqrt{(1-\lambda_2^2)^2\lambda_1+4\lambda_2^2}.
	\end{displaymath}
	This function has the following properties.
	\begin{enumerate}
		\item If both $\lambda_1<1$ and $\lambda_2<1$ then $g(\lambda_1,\lambda_2)< 1$. 
		\item $g(\lambda_1,\lambda_2)<\lambda_1+\lambda_2$.
	\end{enumerate}	
\end{theorem}

\begin{definition}[\cite{Hoory06expandergraphs}]\label{dfn:expanders}
	Let $D$ be a sufficiently large prime power (e.g. $D=2^{16}$). Let $H$ be a  $(D^4,D,{1}/{4})$ expander (an explicit constructions for $H$ exist, cf.~\cite{Reingold00entropywaves}.) 
	We define $\{G_{m}\}_{m\in\Npos}$ by 
	\begin{eqnarray}
	G_1:=H^2, \hspace{20pt} G_m:=G_{m-1}^2\zigzag H \text{ for }m >1. \label{eqn:zigzagconstruction}
	\end{eqnarray}
\end{definition}
\begin{proposition}[\cite{Hoory06expandergraphs}]\label{prop:recursiveConstruction}
	For any $m\in\Npos$, the graph $G_m$ is a $(D^{4m},D^2,1/2)$-graph.
\end{proposition}
In the next section we will use the following lemma whose proof is deferred to Appendix~\ref{app:A}.
\begin{lemma}\label{lem:nonBipartitenessConnectedness}
	Let $G$ be a $D$-regular graph and $S$ be the set of vertices of a connected component of $G^2$. Then $\lambda(G^2[S])< 1$. 
\end{lemma} 
\begin{proof} 
	Let $1=\lambda_1\geq \lambda_2\geq \dots\geq \lambda_N$ be the eigenvalues of $G^2[S]$. Since $G^2[S]$ is connected, Lemma \ref{lem:connectedBipartiteEigenvalues} implies that $\lambda_1>\lambda_2$. Now assume that $-1$ is an eigenvalue of $G^2[S]$ with eigenvector $\overline{v}$. Then the vector $\overline{v}'$ defined by $\overline{v}'_v=\overline{v}_v$ for all $v\in S$ and $\overline{v}'_v=0$ otherwise is the eigenvector for eigenvalue $-1$ of the graph $G^2$. But $G^2$ can not have a negative eigenvalue as every eigenvalue of $G^2$ is a square of a real number. Therefore $\lambda_1\not= \lambda_N$ and $\lambda(G^2[S])<1$ as claimed. 
\end{proof}

\section{A class of expanders definable in FO}\label{sec: definitionFormula}
In this section we define a formula such that the underlying graphs of its models are expanders.
We start with a high-level description of the formula. 
Let $\{G_m\}_{m \in \Npos}$ be as in Definition~\ref{dfn:expanders}. 
Loosely speaking, each model of our formula is a structure which consists of the disjoint union of $G_1,\dots,G_n$ for some $n\in \Npos$ 
with some underlying tree structure connecting $G_{m-1}$ to $G_{m}$ for all $m\in \{2,\dots,n\}$. For illustration see Figure~\ref{fig:modelOfFormula}. 
The tree structure enables us to provide an FO-checkable certificate for the construction of expanders.
The tree structure is a  $D^4$-ary tree, that is used to connect a vertex $v$ of $G_{m-1}$ to every vertex of the copy of $H$ which will replace $v$ in $G_{m}$. 
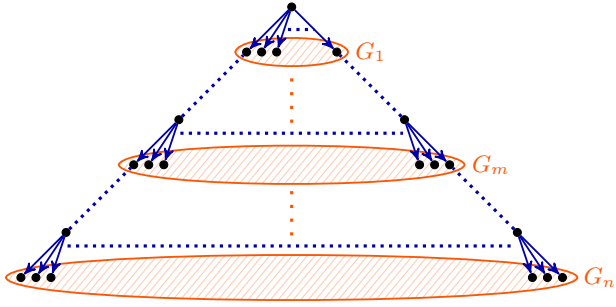
\begin{figure}
	\centering
	\begin{tikzpicture}
	\tikzstyle{ns1}=[line width=0.7]
	\tikzstyle{ns2}=[line width=1.2]
	\def \depth {3.6}
	\def \middle {2.1}
	\def \slope {1}
	\def \r {0.3}
	\def \levelDist {0.6}
	\begin{scope}
	\clip[postaction={fill=white,fill opacity=0.2}] (\depth*\slope+0.2,-\depth)..controls(\depth*\slope+0.1,-\depth+0.4)and(-\depth*\slope-0.1,-\depth+0.4)..(-\depth*\slope-0.2,-\depth)..controls(-\depth*\slope-0.1,-\depth-0.4)and(\depth*\slope+0.1,-\depth-0.4)..(\depth*\slope+0.2,-\depth);			
	\foreach \x in {-7.1,-7,...,15.2}%
	\draw[C3!30](\x, -5)--+(12,14.4);
	\end{scope}
	\draw[ns1,C3](\depth*\slope+0.2,-\depth)..controls(\depth*\slope+0.1,-\depth+0.4)and(-\depth*\slope-0.1,-\depth+0.4)..(-\depth*\slope-0.2,-\depth)..controls(-\depth*\slope-0.1,-\depth-0.4)and(\depth*\slope+0.1,-\depth-0.4)..(\depth*\slope+0.2,-\depth);
	\begin{scope}
	\clip[postaction={fill=white,fill opacity=0.2}] (\middle*\slope+0.2,-\middle)..controls(\middle*\slope+0.1,-\middle+0.34)and(-\middle*\slope-0.1,-\middle+0.34)..(-\middle*\slope-0.2,-\middle)..controls(-\middle*\slope-0.1,-\middle-0.34)and(\middle*\slope+0.1,-\middle-0.34)..(\middle*\slope+0.2,-\middle);		
	\foreach \x in {-7.1,-7,...,15.2}%
	\draw[C3!30](\x, -5)--+(12,14.4);
	\end{scope}
	\draw[ns1,C3](\middle*\slope+0.2,-\middle)..controls(\middle*\slope+0.1,-\middle+0.34)and(-\middle*\slope-0.1,-\middle+0.34)..(-\middle*\slope-0.2,-\middle)..controls(-\middle*\slope-0.1,-\middle-0.34)and(\middle*\slope+0.1,-\middle-0.34)..(\middle*\slope+0.2,-\middle);
	\begin{scope}
	\clip[postaction={fill=white,fill opacity=0.2}] (\levelDist*\slope+0.15,-\levelDist)..controls(\levelDist*\slope+0.1,-\levelDist+0.25)and(-\levelDist*\slope-0.1,-\levelDist+0.25)..(-\levelDist*\slope-0.15,-\levelDist)..controls(-\levelDist*\slope-0.1,-\levelDist-0.25)and(\levelDist*\slope+0.1,-\levelDist-0.25)..(\levelDist*\slope+0.15,-\levelDist);			
	\foreach \x in {-7.1,-7,...,15.2}%
	\draw[C3!30](\x, -5)--+(12,14.4);
	\end{scope}
	\draw[ns1,C3](\levelDist*\slope+0.15,-\levelDist)..controls(\levelDist*\slope+0.1,-\levelDist+0.25)and(-\levelDist*\slope-0.1,-\levelDist+0.25)..(-\levelDist*\slope-0.15,-\levelDist)..controls(-\levelDist*\slope-0.1,-\levelDist-0.25)and(\levelDist*\slope+0.1,-\levelDist-0.25)..(\levelDist*\slope+0.15,-\levelDist);
	\node[draw,circle,fill=black,inner sep=0pt, minimum width=3pt] (1) at (0,0) {};
	\node[draw,circle,fill=black,inner sep=0pt, minimum width=3pt] (2) at (\depth*\slope,-\depth) {};
	\node[draw,circle,fill=black,inner sep=0pt, minimum width=3pt] (3) at (-\depth*\slope,-\depth) {};
	\node[draw,circle,fill=black,inner sep=0pt, minimum width=3pt] (4) at (-\levelDist*\slope,-\levelDist) {};
	\node[draw,circle,fill=black,inner sep=0pt, minimum width=3pt] (5) at (-\levelDist*\slope+0.2,-\levelDist) {};
	\node[draw,circle,fill=black,inner sep=0pt, minimum width=3pt] (6) at (-\levelDist*\slope+0.4,-\levelDist) {};
	\node[draw,circle,fill=black,inner sep=0pt, minimum width=3pt] (7) at (\levelDist*\slope,-\levelDist) {};
	\node[draw,circle,fill=black,inner sep=0pt, minimum width=3pt] (8) at (\depth*\slope-\levelDist*\slope,-\depth+\levelDist) {};
	\node[draw,circle,fill=black,inner sep=0pt, minimum width=3pt] (9) at (-\depth*\slope+\levelDist*\slope,-\depth+\levelDist) {};
	\node[draw,circle,fill=black,inner sep=0pt, minimum width=3pt] (10) at (-\depth*\slope+0.2,-\depth) {};
	\node[draw,circle,fill=black,inner sep=0pt, minimum width=3pt] (11) at (-\depth*\slope+0.4,-\depth) {};
	\node[draw,circle,fill=black,inner sep=0pt, minimum width=3pt] (12) at (\depth*\slope-0.2,-\depth) {};
	\node[draw,circle,fill=black,inner sep=0pt, minimum width=3pt] (13) at (\depth*\slope-0.4,-\depth) {};
	
	\node[draw,circle,fill=black,inner sep=0pt, minimum width=3pt] (14) at (\middle*\slope-\levelDist*\slope,-\middle+\levelDist) {};
	\node[draw,circle,fill=black,inner sep=0pt, minimum width=3pt] (15) at (-\middle*\slope+\levelDist*\slope,-\middle+\levelDist) {};
	\node[draw,circle,fill=black,inner sep=0pt, minimum width=3pt] (16) at (\middle*\slope,-\middle) {};
	\node[draw,circle,fill=black,inner sep=0pt, minimum width=3pt] (17) at (-\middle*\slope,-\middle) {};
	\node[draw,circle,fill=black,inner sep=0pt, minimum width=3pt] (18) at (\middle*\slope-0.2,-\middle) {};
	\node[draw,circle,fill=black,inner sep=0pt, minimum width=3pt] (19) at (-\middle*\slope+0.2,-\middle) {};
	\node[draw,circle,fill=black,inner sep=0pt, minimum width=3pt] (20) at (\middle*\slope-0.4,-\middle) {};
	\node[draw,circle,fill=black,inner sep=0pt, minimum width=3pt] (21) at (-\middle*\slope+0.4,-\middle) {};

	\draw[ns2,C2,dotted] (15)--(4);
	\draw[ns2,C2,dotted] (9)--(17);
	\draw[ns2,C2,dotted] (8)--(16);
	\draw[ns2,C2,dotted] (14)--(7);
	\draw[->, >=stealth',ns1,C2](1)--(4);
	\draw[->, >=stealth',ns1,C2](1)--(5);
	\draw[->, >=stealth',ns1,C2](1)--(6);
	\draw[->, >=stealth',ns1,C2](1)--(7);
	\draw[->, >=stealth',ns1,C2](8)--(2);
	\draw[->, >=stealth',ns1,C2](9)--(3);
	\draw[->, >=stealth',ns1,C2](9)--(10);
	\draw[->, >=stealth',ns1,C2](9)--(11);
	\draw[->, >=stealth',ns1,C2](8)--(12);
	\draw[->, >=stealth',ns1,C2](8)--(13);
	\draw[->, >=stealth',ns1,C2](14)--(16);
	\draw[->, >=stealth',ns1,C2](14)--(18);
	\draw[->, >=stealth',ns1,C2](14)--(20);
	\draw[->, >=stealth',ns1,C2](15)--(17);
	\draw[->, >=stealth',ns1,C2](15)--(19);
	\draw[->, >=stealth',ns1,C2](15)--(21);
	\draw[dotted,ns2,C2](-0.05,-0.3)--(0.3,-0.3);
	\draw[dotted,ns2,C2](-\depth*\slope+0.7*\levelDist*\slope+0.2,-\depth+0.7*\levelDist)--(\depth*\slope-0.7*\levelDist*\slope-0.2,-\depth+0.7*\levelDist);
	\draw[dotted,ns2,C2](-\middle*\slope+0.7*\levelDist*\slope+0.2,-\middle+0.7*\levelDist)--(\middle*\slope-0.7*\levelDist*\slope-0.2,-\middle+0.7*\levelDist);
	\draw[loosely dotted,ns2,C3](0,-\levelDist-0.35)--(0,-\middle+0.5);
	\draw[loosely dotted,ns2,C3](0,-\middle-0.35)--(0,-\depth+0.5);
	\node[minimum height=10pt,inner sep=0,font=\small,C3] at (1.05,-\levelDist) {$G_1$};
	\node[minimum height=10pt,inner sep=0,font=\small,C3] at (2.65,-\middle) {$G_m$};
	\node[minimum height=10pt,inner sep=0,font=\small,C3] at (4.1,-\depth) {$G_{n}$
	};

	\end{tikzpicture}
	\caption{Schematic representation of a model of $\varphi_{\zigzag}$, where the parts in red (grey) only contain relations from $E$ and relations in $F$ are blue (black). Relation $R$ and $L$ are omitted.}\label{fig:modelOfFormula}
\end{figure}
We use $D^4$ relations $\{F_{k}\}_{k \in \indexSetH}$ to enforce an ordering on the $D^4$ children of each vertex. We use additional relations to encode rotation maps. 
For $i,j\in \indexSetRotation$ 
let  $E_{i,j}$ be a binary relation. For every pair $i,j\in \indexSetRotation$ we represent an  edge  $\{v,w\}$ in  $G_m$ by the two tuples $(v,w)\in \rel{E_{i,j}}{\struc{A}}$ 
and $(w,v)\in \rel{E_{j,i}}{\struc{A}}$.
This allows us to encode the relationship $\rot_{G_m}(v,i)=(w,j)$ in first-order logic using the formula $`E_{i,j}(v,w)$'.

We use auxiliary relations $R$ and $L_{k}$ for $k\in \indexSetH$, to force the models to be degree regular. The relation $R$ contains the tuple $(r,r)$ for the  root $r$ of the tree, and $L_{k}$ will contain the tuple $(v,v)$ for every leaf $v$ of the tree. 

We now give the precise definition of the formula. We use $[n]:=\{0,1,\dots,n-1\}$ for $n\in \mathbb{N}$. Let 
\begin{displaymath}
\sigma:=\big\{ \{E_{i,j}\}_{i,j \in \indexSetRotation},\{F_{k}\}_{k \in \indexSetH},R,\{L_k\}_{k\in \indexSetH}\big\},
\end{displaymath} 
where $E_{i,j}$, $F_{k}$, $R$ and $L_k$ are binary relation symbols for $i,j\in [D]^2$ and $k\in ([D]^2)^2$.  
For convenience
we introduce auxiliary relations $E$ and $F$ with the property that for every $\sigma$-structure $\struc{A}$ we have $\rel{E}{\struc{A}}:=\bigcup_{i,j\in \indexSetRotation}\rel{E_{i,j}}{\struc{A}}$ and $\rel{F}{\struc{A}}:=\bigcup_{k\in \indexSetH}\rel{F_k}{\struc{A}}$. 
In any formula we can reverse using these auxiliary relations by replacing formulas of the form 
``$E(x,y)$'' by ``$\bigvee_{i,j\in\indexSetRotation}E_{i,j}(x,y)$'' and formulas of the form ``$F(x,y)$'' by
``$\bigvee_{k\in \indexSetH}F_{k}(x,y)$'' below. 

We use the following formula 
$\varphi_{\operatorname{root}}(x):= \forall y \lnot F(y,x)$ and we say that an element $a\in \univ{A}$ is a root of a structure $\struc{A}$ if $\struc{A}\models \varphi_{\operatorname{root}}(a)$.

We now define a formula $\varphi_{\operatorname{tree}}$, which expresses that any model restricted to the relation $F$ locally looks like a $D^4$-ary tree. More precisely, the formula defines that the structure has no more than one root, that every other vertex has exactly one parent and every vertex has either no children or exactly one child for each of the $D^4$ relations $F_k$. It also defines the self-loops used to make the structure degree regular. 
\begin{align*}
\varphi_{\operatorname{tree}}&:= \exists^{\leq 1} x \varphi_{\operatorname{root}}(x)\land \forall x \Big(\big(\varphi_{\operatorname{root}}(x)\land R(x,x)\big)\lor \\& 
\big(\exists^{=1} y F(y,x)\land \lnot \exists y R(x,y)\land \lnot \exists y R(y,x)\big)\Big)\land
\\&\forall x \bigg(\Big[\neg\exists y F(x,y)\land \bigwedge_{k\in \indexSetH} L_k(x,x)\land \forall y \big(y\not= x \rightarrow 
\\&
\bigwedge_{k\in \indexSetH}\lnot L_k(x,y) \wedge \bigwedge_{k\in \indexSetH}\lnot L_k(y,x)\big)\Big]
\\&\lor \Big[\lnot\exists y \bigvee_{k\in \indexSetH}\big(L_k(x,y)\lor L_k(y,x)\big)\land 
\bigwedge_{k\in \indexSetH}\exists y_{k} \Big(x\not=y_{k}\land F_{k}(x,y_{k})\land \\&
(\bigwedge_{k'\in \indexSetH,k'\not=k}\lnot F_{k'}(x,y_k))\land \forall y(y\not=y_k\rightarrow \lnot F_{k}(x,y))\Big)\Big]\bigg).
\end{align*}

\bigskip
The formula $\varphi_{\operatorname{rotationMap}}$ will define the properties the relations in $E$ need to have in order to encode rotation maps of $D^2$-regular graphs.  For this we  make sure that the edge colours encode a map, i.e. for any pair of a vertex $x$ and index $i\in \indexSetRotation$ there is only one pair of vertex $y$ and index $j\in \indexSetRotation$ such that $E_{i,j}(x,y)$ holds and that the map is self inverse, i.e. if $E_{i,j}(x,y)$ then $E_{j,i}(y,x)$.
\begin{align*}
\varphi_{\operatorname{rotationMap}}:=& \forall x \forall y \Big(\bigwedge_{i,j\in \indexSetRotation}(E_{i,j}(x,y)\rightarrow E_{j,i}(y,x))\Big)
\land
\\&\forall x \Big(\bigwedge_{i\in \indexSetRotation}\Big(\bigvee_{j\in \indexSetRotation}\big(\exists^{=1}y E_{i,j}(x,y)\land \bigwedge_{\substack{j'\in \indexSetRotation\\ j'\not=j}}\lnot \exists y E_{i,j'}(x,y)\big)\Big)\Big)
\end{align*}

We now define a formula $\varphi_{\operatorname{base}}$ which expresses that every root $x$ of a structure has a self-loop $(x,x)$ in each relation $E_{i,j}$ and that the $D^4$ children of a root form $G_1$. Let $H$ be the $(D^4,D,1/4)$-graph from Definition~\ref{dfn:expanders}. We assume that  $H$ has vertex set $\indexSetH$. We then identify vertex $k\in\indexSetH$ with the element $y$ such that $(x,y)\in \rel{F_k}{\struc{A}}$ for each root $x$. Let  $\rot_H:\indexSetH \times[D]\rightarrow \indexSetH\times[D]$ be any rotation map of $H$. Fixing a rotation map for $H$ fixes the rotation map for $H^2$. Recall that $G_1:=H^2$. We can define $G_1$ by a conjunction over all edges of $G_1$.  
\begin{align*}
\varphi_{\operatorname{base}}:=&\forall x  \Big(\varphi_{\operatorname{root}}(x)\rightarrow\Big[\bigwedge_{i,j\in \indexSetRotation}\Big(E_{i,j}(x,x)\land \forall y \Big(x\not= y\rightarrow \\& \big(\lnot E_{i,j}(x,y)\land \lnot E_{i,j}(y,x)\big)\Big)\Big)\land \\& \bigwedge_{\substack{ \rot_{H^2}(k,i)=(k',i')\\k,k'\in \indexSetH\\i,i'\in \indexSetRotation }}\exists y \exists y'\big(F_k(x,y)\land F_{k'}(x,y')\land E_{i,i'}(y,y')\big)\Big]\Big)
\end{align*}

We will now define a formula $\varphi_{\operatorname{recursion}}$ which will ensure that level $m$ of the tree contains $G_m$. Recall that  $G_m:= G_{m-1}^2\zigzag H$. We therefore express that if there is a path of length two between two vertices $x,z$ then for every pair $i,j\in [D]$ there is an edge connecting the corresponding children of $x$ and $z$ according to the definition of the zig-zag product. Here it is important that $x$ and $z$ either both have no children in the underlying tree structure or they both have children. This will also be encoded in the formula.
\begin{align*}
\varphi_{\operatorname{recursion}}:=& \forall x \forall z\bigg[\Big(\lnot \exists y F(x,y)\land \lnot \exists y F(z,y)\Big)\lor
\\&
\bigwedge_{\substack{k_1',k_2'\in \indexSetRotation\\\ell_1',\ell_2'\in \indexSetRotation}}\Big(\exists y \big[E_{k_1',\ell_1'}(x,y)\land E_{k_2',\ell_2'}(y,z)\big]\rightarrow 
\\&
\bigwedge_{\substack{i,j,i',j'\in [D], k,\ell\in \indexSetH\\\rot_H(k,i)=((k_1', k_2'),i')\\ \rot_H((\ell_2', \ell_1'),j)=(\ell,j')}}\exists x'\exists z'\big[ F_k(x,x')\land F_\ell(z,z')\land
E_{(i,j),(j',i')}(x',z')\big]\Big)\bigg]
\end{align*}

We finally let $\varphi_{\zigzag}:=\varphi_{\operatorname{tree}}\land \varphi_{\operatorname{rotationMap}}\land \varphi_{\operatorname{base}}\land \varphi_{\operatorname{recursion}}$.
This concludes defining the formula. 

\subsection{Proving expansion} 
In this section we prove that the formula $\varphi_{\zigzag}$ defines a property of expanders on bounded-degree relational structures.

Let $d:=2D^2+D^4+1$, which is chosen in such a way to allow for any element of a $\sigma$-structure in $\classStruc{C}_{\sigma,d}$ to be in $2D^2$ $E$-relations ($G_m$ is $D^2$ regular and every edge of $G_m$ is modelled by two tuples), to have either $D^4$ $F$-children or $D^4$ $L$-self-loops and to either have one $F$-parent or be in one $R$-self-loop. 

To each model $\struc{A}$ of $\varphi_{\zigzag}$ we will associate an undirected (with parallel edges and self loops) graph $\underlyingGraph{\struc{A}}$ with vertex set $\univ{A}$. For every tuple in each of the relations of $\struc{A}$, the graph $\underlyingGraph{\struc{A}}$ will have an edge. We will define $\underlyingGraph{\struc{A}}$ by a rotation map, which extends the rotation map encoded by the relation $E$. For this let $I:=\{0\}\sqcup \indexSetH\sqcup \indexSetRotation$ be an index set. Formally, we define the \emph{underlying graph} $\underlyingGraph{\struc{A}}$ of a model $\struc{A}$ of $\varphi_{\zigzag}$ to be the undirected graph with vertex set $\univ{A}$ given by the rotation map $\rot_{\underlyingGraph{\struc{A}}}: A\times I\rightarrow A\times I$ defined by 
\begin{align*}
\rot_{\underlyingGraph{\struc{A}}}(v,i):=\begin{cases} 
(v,0) & \text{if }i=0\text{ and }(v,v)\in \rel{R}{\struc{A}} \\
(w,j) & \text{if }i=0\text{ and }(w,v)\in \rel{F_j}{\struc{A}}\\
(w,0) & \text{if }i\in \indexSetH \text{ and }(v,w)\in \rel{F_i}{\struc{A}}\\
(v,i) & \text{if }i\in \indexSetH \text{ and }(v,v)\in \rel{L_i}{\struc{A}}\\
(w,j) & \text{if }i\in \indexSetRotation \text{ and }(v,w)\in \rel{E_{i,j}}{\struc{A}}. 
\end{cases}
\end{align*} 
We can understand this rotation map as labelling the tuples containing an element $v$ as follows: $(v,v)\in \rel{R}{\struc{A}}$ or $(w,v)\in \rel{F_k}{\struc{A}}$ respectively is labelled by  $0$, $(v,w)\in \rel{F_k}{\struc{A}}$ or $(v,v)\in \rel{L_k}{\struc{A}}$ respectively is labelled by $k$ and $(v,w)\in \rel{E_{i,j}}{\struc{A}}$ is labelled by $i$.
Note that $\underlyingGraph{\struc{A}}$ is $(D^2+D^4+1)$-regular. We chose the notion of an underlying graph here instead of the Gaifman graph, 
and it is  more convenient in particular for using results from \cite{Reingold00entropywaves}. However the Gaifman graph can be obtained from the underlying graph by ignoring self-loops and multiple edges.  

\begin{theorem}\label{thm:expansionOfModels}
	There is an $\epsilon>0$ such that the class $\{\underlyingGraph{\struc{A}}\mid \struc{A}\models \varphi_{\zigzag}\}$ is a family of $\epsilon$-expanders.
\end{theorem}

In the rest of this section, we give the proof of Theorem \ref{thm:expansionOfModels}.
Let $\struc{A}$ be a model of $\varphi_{\zigzag}$. 
Let  $\struc{A}|_F$ be the $\{(F_{k})_{k\in \indexSetH}\}$-structure $(\univ{A},(\rel{F_k}{\struc{A}})_{k\in \indexSetH})$. Recall that we denote the Gaifman graph of $\struc{A}|_F$ by $G(\struc{A}|_F)$. 
Let $\struc{A}|_E$ be the $\{(E_{i,j})_{i,j\in \indexSetRotation}\}$-structure $(\univ{A},(\rel{E_{i,j}}{\struc{A}})_{i,j\in \indexSetRotation})$.  We further define the \emph{underlying graph} $\underlyingGraph{\struc{A}|_E}$ of $\struc{A}|_E$ as the undirected graph specified by the rotation map  $\rot_{\underlyingGraph{\struc{A}|_E}}$ which is defined by $\rot_{\underlyingGraph{\struc{A}|_E}}(v,i):=(w,j)$ if $(v,w)\in \rel{E_{i,j}}{\struc{A}}$. This is well defined as $\struc{A}\models \varphi_{\operatorname{rotationMap}}$. 
We use the substructures $G(\struc{A}|_F)$ and $\underlyingGraph{\struc{A}|_E}$ to express the structural properties of models of $\varphi_{\zigzag}$. More precisely, we want to prove that $G(\struc{A}|_F)$ is a rooted complete tree and $\underlyingGraph{\struc{A}|_E}$ is the disjoint union of the expanders $G_1,\dots,G_{n}$ for some $n\in \mathbb{N}$ (Lemma~\ref{lem:exactFormOfModels}). To prove this we use two technical lemmas (Lemma~\ref{lem:recursion} and Lemma~\ref{lem:connected}). Lemma~\ref{lem:recursion} intuitively shows that the children in $G(\struc{A}|_F)$ of each connected part of $\underlyingGraph{\struc{A}|_E}$  form the  zig-zag product with $H$ of the square of the connected part. Lemma~\ref{lem:connected} shows that $G(\struc{A}|_F)$ is connected. To prove Theorem \ref{thm:expansionOfModels} we use that a tree with an expander on each level has good expansion. Loosely speaking, this is true because cutting the tree `horizontally' takes many edge deletions and for cutting the tree `vertically' we  cut many expanders. 

\begin{lemma}\label{lem:recursion}
	Let $\struc{A}$ be a model of $\varphi_{\zigzag}$ and assume $S$ is the set of all vertices belonging to a  connected component  
	of $(\underlyingGraph{\struc{A}|_E})^2$ not containing a root and let $S':=\{w\in \univ{A}\mid (v,w)\in \rel{F}{\struc{A}}, v\in S\}$.  
	If $S'\not=\emptyset$ then  $\underlyingGraph{\struc{A}|_E}[S']$ is a connected component of $\underlyingGraph{\struc{A}|_E}$ and  $\underlyingGraph{\struc{A}|_E}[S']\cong((\underlyingGraph{\struc{A}|_E})^2[S])\zigzag H$. 
\end{lemma} 
We use connected components of $(\underlyingGraph{\struc{A}|_E})^2$ as the square of a connected component of $\underlyingGraph{\struc{A}|_E}$ may not be connected, in which case the zig-zag product with $H$ of the square of the connected component cannot be connected. 
\begin{proof}[Proof of Lemma \ref{lem:recursion}] Assume that $S'\not=\emptyset$. We first show that $\underlyingGraph{\struc{A}|_E}[S']\cong((\underlyingGraph{\struc{A}|_E})^2[S])\zigzag H$.  
	For this we use the following two claims. 
	\begin{claim}\label{claim:edgeInducesPathOfLenghTwoOnParents}
If $$\rot_{(\underlyingGraph{\struc{A}|_E})^2[S]\zigzag H}((u,k),(i,j))=((w,\ell),(j',i'))$$ for some $u,w\in S$, $k,\ell\in \indexSetH$, $i,j,i',j'\in [D]$  then there is $v\in S$ such that $(u,v)\in \rel{E_{k_1',\ell_1'}}{\struc{A}}$ and $(v,w)\in \rel{E_{k_2',\ell_2'}}{\struc{A}}$ where $\rot_{H}(k,i)=((k_1', k_2'),i')$ and $\rot_{H}((\ell_2', \ell_1'),j)=(\ell,j')$. 
	\end{claim}
	\begin{proof}
		the precondition of the Claim and the definition of the zig-zag product, we have that $\rot_{(\underlyingGraph{\struc{A}|_E})^2[S]}(u,(k_1', k_2'))=(w,(\ell_2', \ell_1'))$ for   $\rot_{H}(k,i)=((k_1', k_2'),i')$ and $\rot_{H}((\ell_2', \ell_1'),j)=(\ell,j')$. 
		
		Since $\rot_{(\underlyingGraph{\struc{A}|_E})^2[S]}$ is equal to $\rot_{(\underlyingGraph{\struc{A}|_E})^2}$ restricted to elements of the set $S$, we have that $\rot_{(\underlyingGraph{\struc{A}|_E})^2}(u,(k_1', k_2'))=(w,(\ell_2', \ell_1'))$. Consequently, by the definition of the square of a graph $\rot_{(\underlyingGraph{\struc{A}|_E})^2}(u,(k_1', k_2'))=(w,(\ell_2', \ell_1'))$ implies that there is $v$ such that $\rot_{\underlyingGraph{\struc{A}|_E}}(u,k_1')=(v,\ell_1')$  and $\rot_{\underlyingGraph{\struc{A}|_E}}(v,k_2')=(w,\ell_2')$.
	\end{proof}
	
	\begin{claim}\label{claim:pathOfLenghTwoCausesNoneOrTwoChildren}
		If  $(u,v)\in \rel{E_{k_1',\ell_1'}}{\struc{A}}$ and $(v,w)\in \rel{E_{k_2',\ell_2'}}{\struc{A}}$ for $u,v,w\in \univ{A}$, $k_1',k_2',\ell_1',\ell_2'\in \indexSetH$ and there is $u'\in \univ{A}$ with $(u,u')\in \rel{F}{\struc{A}}$ then there is $w'\in \univ{A}$ such that $(w,w')\in \rel{F}{\struc{A}}$. Furthermore for any $i,i',j,j'\in [D]$ there are $\tilde{u},\tilde{w}\in \univ{A}$, $k,\ell\in \indexSetH$ such that $(\tilde{u},\tilde{w})\in \rel{E_{(i,j),(j'i')}}{\struc{A}}$ for $(u,\tilde{u})\in \rel{F_{k}}{\struc{A}}$ and $(w,\tilde{w})\in \rel{F_{\ell}}{\struc{A}}$ where $\rot_{H}(k,i)=((k_1', k_2'),i')$ and $\rot_{H}((\ell_2', \ell_1'),j)=(\ell,j')$.
	\end{claim}
	\begin{proof}
		We only use that $\struc{A}\models \varphi_{\operatorname{recursion}}$. Since $\varphi_{\operatorname{recursion}}$ has the form $\forall x\forall z \psi(x,z)$ for some formula $\psi(x,z)$ we know that $\struc{A}\models \psi(u,w)$.
		Since $(u,u')\in \rel{F}{\struc{A}}$ we have  $\struc{A}\not\models  \lnot \exists y F(u,y)\land \lnot \exists y F(w,y)$. Since  $\struc{A}\models \exists y \big[E_{k_1',\ell_1'}(u,y)\land E_{k_2',\ell_2'}(w,z)\big]$  
		\begin{align*}
		&\struc{A}\models \bigwedge_{\substack{i,j,i',j'\in [D], k,\ell\in \indexSetH\\\rot_H(k,i)=((k_1', k_2'),i')\\ \rot_H((\ell_2', \ell_1'),j)=(\ell,j')}}\exists x'\exists z'\big[ F_k(u,x')\land 
		F_\ell(w,z')\land E_{(i,j),(j',i')}(x',z')\big]
		\end{align*}
		Since $H$ is $D$-regular, for every $k'_1,k'_2 \in [D]^2$ and $i,i' \in [D]$, there is $k \in ([D]^2)^2$ such that $\rot_H(k,i) = ((k'_1,k'_2,i')$ (and the same for $\ell'_1,\ell'_2,j,j'$). Thus, the above conjunction is not empty. This further implies that for any $i,i',j,j'\in [D]$ there are $\tilde{u},\tilde{w}\in \univ{A}$, $k,\ell\in \indexSetH$ as claimed. In particular there is $w' \in \univ{A}$ such that $(w,w') \in \rel{F}{\struc{A}}$. 
	\end{proof}

	We will argue that for every element $w\in S$ there is a $w'\in S'$ such that $(w,w')\in \rel{F}{\struc{A}}$.  For this pick any $u'\in S'$. Let $u\in S$ be the element such that  $(u,u')\in \rel{F}{\struc{A}}$. By combining Lemma \ref{lem:nonBipartitenessConnectedness}, Theorem \ref{thm:expansionOfZigZag} and Lemma \ref{lem:connectedBipartiteEigenvalues} it follows that $((\underlyingGraph{\struc{A}|_E})^2[S])\zigzag H$ is connected. Therefore, there exists a path $(u'_0,\dots,u'_m)$ in $((\underlyingGraph{\struc{A}|_E})^2[S])\zigzag H$ from $u'_0=(u,(k_1,k_2))$ to $u'_m=(w,(\ell_1,\ell_2))$ for some $k_1,k_2$,$\ell_1$,$\ell_2$ $\in \indexSetRotation$. By Claim \ref{claim:edgeInducesPathOfLenghTwoOnParents} there is a path $(u_0,v_0,u_1,v_1,\dots u_{m-1},v_{m-1},u_m)$ in $\underlyingGraph{\struc{A}|_E}$ from $u_0=u$ to $u_m=w$. By inductively using Claim \ref{claim:pathOfLenghTwoCausesNoneOrTwoChildren} on the path we find $w'$ such that $(w,w')\in \rel{F}{\struc{A}}$.

	Combining this with $\struc{A}\models \varphi_{\operatorname{tree}}$ implies that the map $f:S\times \indexSetH \rightarrow S'$, given by $f(v,k)=u$ if $(v,u)\in \rel{F_{k}}{\struc{A}}$, is well-defined. Furthermore, by Claim~\ref{claim:edgeInducesPathOfLenghTwoOnParents} and \ref{claim:pathOfLenghTwoCausesNoneOrTwoChildren}, we have that if it holds that  $\rot_{(\underlyingGraph{\struc{A}|_E})^2[S]\zigzag H}((u,k),(i,j))=((w,\ell),(j',i'))$ then \[\rot_{(\underlyingGraph{\struc{A}|_E})[S']}(f((u,k)),(i,j))=(f((w,\ell)),(j',i')).\] This proves that $f$ maps each edge in $((\underlyingGraph{\struc{A}|_E})^2[S])\zigzag H$ injectively to an edge in $\underlyingGraph{\struc{A}|_E}[S']$. Then the map $f$ together with the corresponding edge map is an isomorphism from $((\underlyingGraph{\struc{A}|_E})^2[S])\zigzag H$ to $\underlyingGraph{\struc{A}|_E}$ as both are $D^2$-regular.

	Moreover, $\underlyingGraph{\struc{A}|_E}[S']\cong((\underlyingGraph{\struc{A}|_E})^2[S])\zigzag H$ implies that  $\underlyingGraph{\struc{A}|_E}[S']$ is connected and $D^2$-regular. Since $\struc{A}\models \varphi_{\operatorname{rotationMap}}$ enforces that $\underlyingGraph{\struc{A}|_E}$ is $D^2$-regular, no vertex $v\in S'$ can have  neighbours which are not in $S'$ and therefore $\underlyingGraph{\struc{A}|_E}[S']$ is a connected component of $\underlyingGraph{\struc{A}|_E}$.
\end{proof}

\begin{lemma}\label{lem:connected}
	Let $\struc{A}\in \classStruc{C}_{\sigma,d}$ be a model  of $\varphi_{\zigzag}$. Then every connected component of $G(\struc{A}|_F)$  contains a root of $\struc{A}$. In particular for every model $\struc{A}\in \classStruc{C}_{\sigma,d}$ of $\varphi_{\zigzag}$ the graph $G(\struc{A}|_F)$ is connected.
\end{lemma}
Note that the connectivity of $G(\struc{A}|_F)$ for a model $\struc{A}\in \classStruc{C}_{\sigma,d}$  of $\varphi_{\zigzag}$ implies that $\struc{A}$ is connected as $G(\struc{A}|_F)$ is a subgraph of the Gaifman graph of $\struc{A}$ containing the same set of vertices. Hence the following corollary follows immediately from Lemma~\ref{lem:connected}.
\begin{corollary}
	Any model $\struc{A}\in \classStruc{C}_{\sigma,d}$ of $\varphi_{\zigzag}$ is connected.
\end{corollary}
\begin{proof}[Proof of Lemma~\ref{lem:connected}]
	Assume that there is a connected component of $G(\struc{A}|_F)$  which contains no root of $\struc{A}$ and let $G'$ to be a connected component of $G(\struc{A}|_F)$ with vertex set $V\subseteq \univ{A}$ such that   $\struc{A}\not\models \varphi_{\operatorname{root}}(v)$ for every $v\in V$.  
For the next claim we should have in mind that $(\struc{A}|_F)[V]$ can be understood as a directed graph in which every vertex has in-degree $1$ and the corresponding undirected graph $G'$ is connected. Hence $(\struc{A}|_F)[V]$ must consist of a set of disjoint directed trees whose roots form a directed cycle. Consequently $G'$ has the structure as given in the following claim. 

	\begin{claim}\label{claim:containsCycle}
		$G'$ contains a cycle $(c_0,\dots,c_{\ell-1})$ and for every vertex $v$ of $G'$ there is exactly one path $(p_0,\dots,p_m)$ in $G'$ with $p_0=v$, $p_m$ on the cycle and $p_i$ not on the cycle for all $i\in [m]$.
	\end{claim}
	\begin{proof}
		Let $v_0$ be any vertex in $G'$ and let $S_0=\{v_0\}$. We will now recursively define  $v_i$  to be the vertex of $G'$ such that $(v_i,v_{i-1})\in \rel{F}{\struc{A}}$. Such a vertex always exists by the choice of $G'$ (i.e. that no root is in $G'$) and the fact that $\struc{A} \models \varphi_{\operatorname{tree}}$. Furthermore, such a vertex is unique as $\struc{A}\models \varphi_{\operatorname{tree}}$. We also let $S_i:=S_{i-1}\cup \{v_i\}$. Since $\univ{A}$ is finite the chain $S_0\subseteq S_1\subseteq \dots \subseteq S_i\subseteq \dots$ must become stationary at some point. Let $i\in \mathbb{N}$ be the minimum index such that $S_{i-1}=S_i$ and let $j< i$ be such that $v_i=v_j$. Then $(v_i,v_{i-1},\dots,v_{j+1},v_j)$ is a cycle in $G'$ as by construction $(v_k,v_{k-1})\in \rel{F}{\struc{A}}$ which implies that $\{v_k,v_{k-1}\}$ is an edge in the Gaifman graph $G(\struc{A}|_F)$.
		Let $C=\{c_0,\dots,c_{\ell-1}\}$ be the vertices of the cycle. 
		Since $G'$ is connected a path that satisfies the property as described in the assertion of the claim always exists. So let us argue that such a path is unique. Assume  there are two different such paths $(p_0,\dots,p_m)$ and $(p'_0,\dots,p'_{m'})$ and assume that $p_m=c_i$ and $p'_{m'}=c_j$. Let $k\leq \min\{m,m'\}$ be the minimum index such that $p_k\not=p'_k$. Such an index must exist as the paths are different and as $p_0=p'_0=v$ we also know that $k\geq 1$. Since $\struc{A}\models \varphi_{\operatorname{tree}}$ for every vertex $w$ of $G'$ there can only be one vertex $w'$ of $G'$ such that $(w',w)\in \rel{F}{\struc{A}}$. As $p_{m-1}\notin C$ and $(c_{(i-1)\,\,\mod\,\, \ell},p_m)\in \rel{F}{\struc{A}}$ it follows that $(p_{m},p_{m-1})\in \rel{F}{\struc{A}}$. Applying the argument inductively we get that $(p_k,p_{k-1})\in \rel{F}{\struc{A}}$. The same argument works for the path $(p'_0,\dots, p'_{m'})$ and therefore $(p'_k,p'_{k-1})\in \rel{F}{\struc{A}}$. By the choice of $k$ we know that $p_{k-1}=p'_{k-1}$ and $p_k\not=p'_k$ which contradicts $\struc{A}\models \varphi_{\operatorname{tree}}$.
	\end{proof}	
	
	Let $S_0$ be the vertex set of the connected component of $\underlyingGraph{\struc{A}|_E}$ with $c_0\in S_0$. Note  that $S_0$ might not be contained in $G'$.  
	
	We now recursively define the infinite sequence of sets $S_i:=\{w\in \univ{A}\mid (v,w)\in \rel{F}{\struc{A}}, v\in S_{i-1}\}$ for each $i\in\Npos$.  Let $m_i:=\max_{v\in S_i\cap V}\min_{j\in \{0,\dots,\ell-1\}}\{\dist_{G'}(c_j,v)\}$ and let $v_i\in S_i\cap V$ be a vertex of distance $m_i$ from $C$ in $G'$. Note here that $m_i$ is well defined as $c_{i\,\, \mod \,\, \ell}\in S_i$. 
	
	\begin{claim}\label{claim:characterisationOfG|_E[S_i]}
		$\underlyingGraph{\struc{A}|_E}[S_{i}]=(\underlyingGraph{\struc{A}|_E}[S_{i-1}])^2\zigzag H$.
		
	\end{claim}
	\begin{proof}
		We  show  the stronger statement that $\underlyingGraph{\struc{A}|_E}[S_i]$ is a connected component of $\underlyingGraph{\struc{A}|_E}$, $(\underlyingGraph{\struc{A}|_E}[S_{i}])^2\zigzag H=$ $\underlyingGraph{\struc{A}|_E}[S_{i+1}]$ and  $\lambda(\underlyingGraph{\struc{A}|_E}[S_i])<1$ for $i\in \mathbb{N}$ by induction.
		
		$\underlyingGraph{\struc{A}|_E}[S_0]$ is a connected component of $\underlyingGraph{\struc{A}|_E}$ by choice of $S_0$.	
		Let $\tilde{S}:=\{w\in \univ{A}\mid (w,v)\in \rel{F}{\struc{A}}, v\in S_0 \}$. 
		 
		We now argue that $(\underlyingGraph{\struc{A}|_E})^2[\tilde{S}]$ is a connected component of $(\underlyingGraph{\struc{A}|_E})^2$. Assuming the contrary, either a connected component of  $(\underlyingGraph{\struc{A}|_E})^2$ contains vertices from  both $\tilde{S}$ and $A\setminus \tilde{S}$ or $(\underlyingGraph{\struc{A}|_E})^2[\tilde{S}]$ splits into more than one connected component. Let $S'$ be the vertices of a connected component as in the first case. Then $|S'|>1$ and hence $S'$ can not contain any root as a root is not in any $E$-relation with any element different from itself. Hence by Lemma~\ref{lem:recursion} we get a connected component of $\underlyingGraph{\struc{A}|_E}$ on the children of $S'$ containing vertices both from $S_0$ and from $\univ{A}\setminus S_0$, which  contradicts $S_0$ being a connected component of $\underlyingGraph{\struc{A}|_E}$. Now let $S'$ be a connected component as in the second case, and pick $S'$ such that it does not contain a root (this is possible as there is at most one root). 
		Then by Lemma~\ref{lem:recursion} $S_0$ must have a non-empty intersection with at least two connected components of $\underlyingGraph{\struc{A}|_E}$, which is a contradiction. 
		
		Thus, by Lemma \ref{lem:nonBipartitenessConnectedness}  $\lambda((\underlyingGraph{\struc{A}|_E})^2[\tilde{S}])<1$. But by Lemma \ref{lem:recursion} $\underlyingGraph{\struc{A}|_E}[S_0]=((\underlyingGraph{\struc{A}|_E})^2[\tilde{S}])\zigzag H$. Then Theorem \ref{thm:expansionOfZigZag} and  $\lambda(H)<1$ ensure that $\lambda(\underlyingGraph{\struc{A}|_E}[S_0])<1$.
		
For $i>1$, by induction it holds that $\lambda(\underlyingGraph{\struc{A}|_E}[S_{i-1}])<1$, which, together with Lemma~\ref{lem:expansionOfSquaring} and Lemma~\ref{lem:connectedBipartiteEigenvalues}, implies that $(\underlyingGraph{\struc{A}|_E}[S_{i-1}])^2$ is a connected component\footnote{We remark that the statement that $(\underlyingGraph{\struc{A}|_E}[S_{i-1}])^2$ is a connected component does not directly follow from the fact that $\underlyingGraph{\struc{A}|_E}[S_{i-1}]$ is a connected component of $\underlyingGraph{\struc{A}|_E}$, as the square of a connected bipartite graph is not necessarily connected.} of $(\underlyingGraph{\struc{A}|_E})^2$ and that $(\underlyingGraph{\struc{A}|_E})^2[S_{i-1}]=(\underlyingGraph{\struc{A}|_E}[S_{i-1}])^2$. Since $c_{i\,\,\mod\,\, \ell}\in S_i$, by Lemma \ref{lem:recursion}, we have that $\underlyingGraph{\struc{A}|_E}[S_i]$ is a connected component of $\underlyingGraph{\struc{A}|_E}$ and $\underlyingGraph{\struc{A}|_E}[S_i]= (\underlyingGraph{\struc{A}|_E}[S_{i-1}])^2\zigzag H$.  Furthermore, using Lemma \ref{lem:expansionOfSquaring} and Theorem \ref{thm:expansionOfZigZag}, this proves $\lambda(\underlyingGraph{\struc{A}|_E}[S_i])<1$.
	\end{proof}
	\begin{claim}
		For every $v\in S_i$ there is $w\in V$ such that $(v,w)\in \rel{F}{\struc{A}}$.
	\end{claim}
	\begin{proof}
By Claim \ref{claim:characterisationOfG|_E[S_i]} we have that $\underlyingGraph{\struc{A}|_E}[S_{i+1}]=(\underlyingGraph{\struc{A}|_E}[S_i])^2\zigzag H$. This means that by definition of squaring and the zig-zag product we know that $|S_{i+1}|=D^4\cdot |S_i|$. But as in addition $\struc{A}\models \varphi_{\operatorname{tree}}$  we know that every element $v\in S_i$ will contribute to no more then $D^4$ elements to $S_{i+1}$. This means by construction of $S_{i+1}$ that for every element in $S_i$ there must be $w\in V$ such that $(v,w)\in \rel{F}{\struc{A}}$. 
	\end{proof}
	Therefore, for every $i\in \mathbb{N}_{>0}$ there is $w_i\in V$ such that $(v_i,w_i)\in \rel{F}{\struc{A}}$ where $v_i$ is the vertex of distance $m_i$ from $C$ in $G'$ picked above. Let $(u_0,\dots,u_{m_i})$ be the path in $G'$ from $u_0=v_i$ to $u_{m_i}\in C$. Note that it is impossible that $w_i = u_1$. This is true as for the path $(u_0,...,u_{m_i})$, we have that $(u_{j+1},u_{j})\in \rel{F}{\struc{A}}$ for all $j\in [m_i]$. Furthermore, since $v_i=u_0\not=u_1$, assuming that $w_i=u_1$ would imply $(v_i,u_1),(u_2,u_1)\in \rel{F}{\struc{A}}$,  which contradicts $\struc{A}\models \varphi_{\operatorname{tree}}$.
		 Then $(w_i,u_0,\dots,u_{m_i})$ is a path in $G'$ from $w_i$ to $C$. 
		Since $w_i\in S_{i+1}$ by construction, Claim~\ref{claim:containsCycle}  implies that $m_{i+1}\geq m_i+1$. Therefore $m_i\geq i+m_0$ inductively. But this yields a contradiction, because $\ell+m_0\leq m_\ell=m_0$ and the length of the cycle $\ell>0$. See Figure~\ref{fig:illustrationOfCycleFreenessOfModelsOfFi} for an illustration. Therefore  every connected component of $G(\struc{A}|_F)$  must contain a root of $\struc{A}$.  Furthermore, since every connected component of $G(\struc{A}|_F)$ must contain a root  and since $\struc{A}\models \exists^{\leq 1}x\varphi_{\operatorname{root}}(x)$ there can not be more than one root,  $G(\struc{A}|_F)$ is connected.
\end{proof}	

\begin{figure}
	\centering
	\begin{tikzpicture}
	\tikzstyle{ns1}=[line width=0.7]
	\tikzstyle{ns2}=[line width=1.2]
	\def \radius {2}
	\def \rad {4.5}
	\def \margin {2.7} 
	\begin{scope}
	\clip[postaction={fill=white,fill opacity=0.2}] (254:1.6)..controls(248:2.60)and(283:3.1)..(258:4.3)..controls(250:4.76)..(246:4.07)..controls(244:3.955)..(247:3.38)..controls(250:3.035)and(228:2.69)..(245:1.5)..controls(250:1.4)..(254:1.6);				
	\foreach \x in {-7.1,-7,...,15.2}%
	\draw[C3!20](\x, -5)--+(12,14.4);
	\end{scope}
	\draw[ns1,C3!50] (254:1.6)..controls(248:2.60)and(283:3.1)..(258:4.3)..controls(250:4.76)..(246:4.07)..controls(244:3.955)..(247:3.38)..controls(250:3.035)and(228:2.69)..(245:1.5)..controls(250:1.4)..(254:1.6);	
	\begin{scope}
	\clip[postaction={fill=white,fill opacity=0.2}] (253:1.6)..controls(250:2.27)and(270:2.54)..(255:2.9)..controls(250:3.08)..(246:2.81)..controls(245:2.765)..(247:2.54)..controls(250:2.405)and(230:2.27)..(246:1.5)..controls(250:1.4)..(253:1.6);			
	\foreach \x in {-7.1,-7,...,15.2}%
	\draw[C3!30](\x, -5)--+(12,14.4);
	\end{scope}
	\draw[ns1,C3] (253:1.6)..controls(250:2.27)and(270:2.54)..(255:2.9)..controls(250:3.08)..(246:2.81)..controls(245:2.765)..(247:2.54)..controls(250:2.405)and(230:2.27)..(246:1.5)..controls(250:1.4)..(253:1.6);
	\begin{scope}
	\clip[postaction={fill=white,fill opacity=0.2}] (293:1.6)..controls(290:2.36)and(310:2.72)..(295:3.12)..controls(290:3.344)..(286:3.08)..controls(285:3.02)..(287:2.72)..controls(290:2.54)and(270:2.36)..(286:1.5)..controls(290:1.4)..(293:1.6);			
	\foreach \x in {-7.1,-7,...,15.2}%
	\draw[C3!30](\x, -5)--+(12,14.4);
	\end{scope}
	\draw[ns1,C3] (293:1.6)..controls(290:2.36)and(310:2.72)..(295:3.12)..controls(290:3.344)..(286:3.08)..controls(285:3.02)..(287:2.72)..controls(290:2.54)and(270:2.36)..(286:1.5)..controls(290:1.4)..(293:1.6);
	\begin{scope}
	\clip[postaction={fill=white,fill opacity=0.2}] (333:1.6)..controls(330:2.45)and(350:2.9)..(335:3.5)..controls(330:3.8)..(326:3.35)..controls(325:3.275)..(327:2.9)..controls(330:2.675)and(310:2.45)..(326:1.5)..controls(330:1.4)..(333:1.6);				
	\foreach \x in {-7.1,-7,...,15.2}%
	\draw[C3!30](\x, -5)--+(12,14.4);
	\end{scope}
	\draw[ns1,C3] (333:1.6)..controls(330:2.45)and(350:2.9)..(335:3.5)..controls(330:3.8)..(326:3.35)..controls(325:3.275)..(327:2.9)..controls(330:2.675)and(310:2.45)..(326:1.5)..controls(330:1.4)..(333:1.6);	
	\begin{scope}
	\clip[postaction={fill=white,fill opacity=0.2}] (214:1.6)..controls(210:2.6)and(232:3.2)..(216:4)..controls(210:4.2)..(206:3.8)..controls(204:3.7)..(207:3.2)..controls(210:2.9)and(188:2.6)..(205:1.5)..controls(210:1.4)..(214:1.6);			
	\foreach \x in {-7.1,-7,...,15.2}%
	\draw[C3!30](\x, -5)--+(12,14.4);
	\end{scope}
	\draw[ns1,C3] (214:1.6)..controls(210:2.6)and(232:3.2)..(216:4)..controls(210:4.2)..(206:3.8)..controls(204:3.7)..(207:3.2)..controls(210:2.9)and(188:2.6)..(205:1.5)..controls(210:1.4)..(214:1.6);
	\node[draw,circle,fill=black,inner sep=0pt, minimum width=5pt] (1) at (210:\radius) {};
	\node[draw,circle,fill=black,inner sep=0pt, minimum width=5pt] (2) at (250:\radius) {};
	\node[draw,circle,fill=black,inner sep=0pt, minimum width=5pt] (3) at (290:\radius) {};
	\node[draw,circle,fill=black,inner sep=0pt, minimum width=5pt] (4) at (330:\radius) {};
	\draw[->, >=stealth',ns1,C2] ({210+\margin}:\radius) arc ({210+\margin}:{250-\margin}:\radius);
	\draw[->, >=stealth',ns1,C2] ({250+\margin}:\radius) arc ({250+\margin}:{290-\margin}:\radius);
	\draw[->, >=stealth',ns1,C2] ({290+\margin}:\radius) arc ({290+\margin}:{330-\margin}:\radius);
	\draw[->, >=stealth',ns1,C2] ({330+\margin}:\radius) arc ({330+\margin}:{370-\margin}:\radius);
	\draw[dashed,->, >=stealth',ns1,C2] ({-30+\margin}:\radius) arc ({-30+\margin}:{210-\margin}:\radius);
	\draw[->, >=stealth',ns1,C2] (250:1.43)--(290:1.45);
	\draw[->, >=stealth',ns1,C2] (250:1.43)--(285:1.52);
	\draw[->, >=stealth',ns1,C2] (250:1.43)--(283:1.61);
	\draw[loosely dotted,ns2,C2] (268:1.6)--(272:2.8);
	\draw[->, >=stealth',ns1,C2] (251:3.02)--(289:3.27);
	\draw[->, >=stealth',ns1,C2] (251:3.02)--(287:3.12);
	\draw[->, >=stealth',ns1,C2] (251:3.02)--(285.2:2.97);
	\draw[->, >=stealth',ns1,C2] (290:1.43)--(330:1.45);
	\draw[->, >=stealth',ns1,C2] (290:1.43)--(325:1.52);
	\draw[->, >=stealth',ns1,C2] (290:1.43)--(323:1.61);
	\draw[loosely dotted,ns2,C2] (308:1.65)--(312:3);
	\draw[->, >=stealth',ns1,C2] (290.5:3.28)--(330:3.69);
	\draw[->, >=stealth',ns1,C2] (290.5:3.28)--(327.3:3.49);
	\draw[->, >=stealth',ns1,C2] (290.5:3.28)--(325.8:3.34);
	\draw[ns1,C2] (330:1.43)--(355:1.45);
	\draw[ns1,C2] (330:1.43)--(352:1.52);
	\draw[ns1,C2] (330:1.43)--(349:1.59);
	\draw[ns1,C2] (330.5:3.7)--(341:3.86);
	\draw[ns1,C2] (330.5:3.7)--(338.3:3.75);
	\draw[ns1,C2] (330.5:3.7)--(335.8:3.65);
	\draw[->, >=stealth',ns1,C2] (210:1.43)--(250:1.45);
	\draw[->, >=stealth',ns1,C2] (210:1.43)--(245:1.52);
	\draw[->, >=stealth',ns1,C2] (210:1.43)--(243:1.61);
	\draw[loosely dotted,ns2,C2] (228:1.65)--(232:3.8);
	\draw[->, >=stealth',ns1,C2] (211:4.13)--(251:4.6);
	\draw[->, >=stealth',ns1,C2] (211:4.13)--(248:4.35);
	\draw[->, >=stealth',ns1,C2] (211:4.13)--(246.5:4.15);
	\draw[->, >=stealth',ns1,C2] (185:1.43)--(208:1.43);
	\draw[->, >=stealth',ns1,C2] (185:1.52)--(205:1.53);
	\draw[->, >=stealth',ns1,C2] (185:1.61)--(202:1.63);
	\draw[->, >=stealth',ns1,C2] (200:4.13)--(210:4.13);
	\draw[->, >=stealth',ns1,C2] (200:4.03)--(208:4);
	\draw[->, >=stealth',ns1,C2] (200:3.93)--(206.5:3.87);
	
	\node[minimum height=10pt,inner sep=0,font=\small] at (0:0) {$C$};
	\node[minimum height=10pt,inner sep=0,font=\small] at (248:\radius-0.2) {$c_{0}$};
	\node[minimum height=10pt,inner sep=0,font=\small] at (252:\radius+0.6) {$S_{0}$};	
	\node[minimum height=10pt,inner sep=0,font=\small] at (292:\radius+0.9) {$S_{1}$};
	\node[minimum height=10pt,inner sep=0,font=\small] at (332:\radius+1.2) {$S_{2}$};
	\node[minimum height=10pt,inner sep=0,font=\small] at (212:\radius+1.6) {$S_{\ell-1}$};
	\node[minimum height=10pt,inner sep=0,font=\small] at (254:\radius+2) {$S_{\ell}=S_0$};
	\end{tikzpicture}
	\caption{Illustration of the proof of Lemma \ref{lem:connected}.}\label{fig:illustrationOfCycleFreenessOfModelsOfFi}
\end{figure}
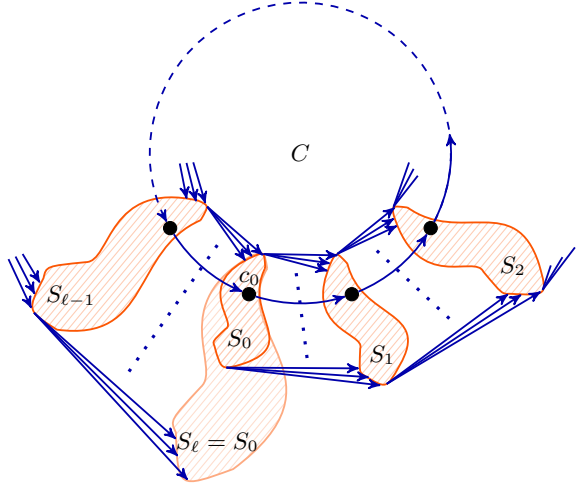

We let $\classStruc{P}_{\zigzag}:=\classStruc{P}_{\varphi_{\zigzag}}$ for the formula $\varphi_{\zigzag}$ from Section \ref{sec: definitionFormula}. 

\begin{lemma}\label{lem:exactFormOfModels}
	Any (finite) model $\struc{A}\in \classStruc{C}_{\sigma,d}$ of $\varphi_{\zigzag}$ has the following structure. 
	\begin{itemize}
		\item Either $\univ{A}=\emptyset$ or $|\univ{A}|=\sum_{m=0}^{n}D^{4m}$ for some $n\in \mathbb{N}_{n\geq 1}$.
		\item $G(\struc{A}|_F)$ is a  $D^4$-ary complete rooted tree, where the root is the unique element $r\in \univ{A}$ for which $\struc{A}\models \varphi_{\operatorname{root}}(r)$.
		\item $\underlyingGraph{\struc{A}|_E}[T_m]\cong G_m$ where $G_m$ is defined as in Definition~\ref{dfn:expanders} and  $T_m$ is the set of vertices of distance $m$ to $r$ in the tree $G(\struc{A}|_F)$ for any $m\in \{1,\dots,n\}$.
	\end{itemize} 
Furthermore for every $n\in \mathbb{N}_{\geq 1}$ there is a model of $\varphi_{\zigzag}$ of size $\sum_{m=0}^{n}D^{4m}$.
\end{lemma}
\begin{proof}
	First note that the empty structure $\struc{A}_{\emptyset}\in \classStruc{P}_{\zigzag}$ as $\struc{A}_{\emptyset}\models \exists^{\leq 1}x\varphi_{\operatorname{root}}(x)$ and  therefore $\struc{A}_{\emptyset}\models \varphi_{\zigzag}$ as $\varphi_{\zigzag}$ is a conjunction of $\exists^{\leq 1}x\varphi_{\operatorname{root}}(x)$ and universally quantified formulas. Hence $\univ{A}=\emptyset$ is possible. Now assume that $\struc{A}$ is a model of $\varphi_{\zigzag}$ and $\univ{A}\not=\emptyset$. Then 
	Lemma  \ref{lem:connected} implies  that  $G(\struc{A}|_F)$ is connected.	Combining this with $\struc{A}\models \varphi_{\operatorname{tree}}$ proves that $G(\struc{A}|_F)$ is a rooted  tree. Let $n$ be the maximum distance of any vertex in $G(\struc{A}|_F)$ to the root and let $T_m$ be the vertices of distance  $m$ to the root for $m\leq n$. Then $\underlyingGraph{\struc{A}|_E}[T_1]\cong G_1$ because $\struc{A}\models \varphi_{\operatorname{base}}$. Now assume towards an inductive proof that $\underlyingGraph{\struc{A}|_E}[T_m]\cong G_m$ for some fixed $m\in \Npos$. Since $\lambda(G_m)<1$   by Lemma~\ref{lem:expansionOfSquaring} and Lemma~\ref{lem:connectedBipartiteEigenvalues} we get that $(\underlyingGraph{\struc{A}|_E})^2[T_m]$  is a connected component of $(\underlyingGraph{\struc{A}|_E})^2$. Hence by Lemma \ref{lem:recursion} we get that $\underlyingGraph{\struc{A}|_E}[T_{m+1}]\cong G_{m+1}$. Since $G_m$ has $D^{4m}$ vertices this also proves that $\struc{A}$ has $\sum_{m=0}^{n}D^{4m}$ vertices. Furthermore, for $n\in \mathbb{N}$ the existence of a  model of $\varphi_{\zigzag}$ of size $\sum_{m=0}^{n}D^{4m}$ is straightforward by the construction of the formula $\varphi_{\zigzag}$.
\end{proof}

Now we are ready to prove Theorem \ref{thm:expansionOfModels}.
\begin{proof}[Proof of Theorem \ref{thm:expansionOfModels}]
	We will prove that for $\epsilon={D^2}/{12}$ the claimed is true.
	Let $\struc{A}$ be the model of $\varphi_{\zigzag}$ of size $\sum_{m=0}^{n}D^{4m}$ and $S\subseteq \univ{A}$ with  $|S|\leq(\sum_{m=0}^{n}D^{4m})/{2}$. 
	Let $T_m$ be the vertices of distance $m$ to the root of the tree $G(\struc{A}|_F)$ and let $S_m:=T_m\cap S$.
	
	We can assume that $|S|>1$ as every vertex has degree at least $\epsilon$.
	Let us first assume that $|S_m|\leq {D^{4m}}/{2}$ for all $m\in [n]$. Then because  $G_m$ is a ${D^2}/{4}$-expander (this follows directly from Theorem \ref{thm:boundExpansionRatioInTermsOfLambda} as $\lambda(G_m)\leq {1}/{2}$ by Proposition~\ref{prop:recursiveConstruction}) and $\underlyingGraph{\struc{A}|_E}[T_m]\cong G_m$ 
	we know that 
	\begin{align*}
	|\langle S, \overline{S}\rangle_{\underlyingGraph{\struc{A}}} |\geq \sum_{m=1}^{n}\frac{D^2}{4} |S_m|\geq \frac{D^2}{12}\sum_{m=0}^{n}|S_m|=\frac{D^2}{12}|S|.
	\end{align*} 
	Now assume the opposite and choose $m'$ to be the largest index such that 
	\begin{align}\label{eq:choiceOfm'}
		|S_{m'}|>\frac{|T_{m'}|}{2}=\frac{D^{4m'}}{2}.
	\end{align}  
	We will use the following claim.
	\begin{claim}\label{claim:relationOfSizesOfLevels}
		$ \sum_{m=0}^{\tilde{m}-1}|T_{m}|\leq \frac{1}{2}|T_{\tilde{m}}|$ for all $\tilde{m}\leq n$.
	\end{claim}
	\begin{proof}
		Inductively, we argue that 
		\[ \sum_{m=0}^{\tilde{m}-1}|T_m|= \sum_{m=0}^{\tilde{m}-2}|T_m|+|T_{\tilde{m}-1}|\leq \frac{1}{2}(3|T_{\tilde{m}-1}|) \leq \frac{1}{2}|T_{\tilde{m}}|.
		\]
	\end{proof}
	Claim \ref{claim:relationOfSizesOfLevels} implies that $\frac{3}{4}\cdot |T_{n}|\geq \frac{1}{2}|T_{n}|+\frac{1}{2}\sum_{m=0}^{n-1}|T_m| =  \frac{1}{2}|A| \geq|S|\geq |S_{n}|$. In the case that $m'=n$, using that $G_n$ is a ${D^2}/{4}$-expander we get 
	\begin{align*}
	|\langle S,\overline{S}\rangle_{\underlyingGraph{\struc{A}}}|\geq \frac{D^2}{4}(|T_{n}|-|S_{n}|)\geq \frac{D^2}{16}|T_{n}|\geq \frac{D^2}{12}|S|.
	\end{align*}
	\begin{figure}
		\centering
		\begin{tikzpicture}
		\tikzstyle{ns1}=[line width=0.7]
		\tikzstyle{ns2}=[line width=1.2]
		\definecolor{C5}{RGB}{180,180,180}
		\def \depth {4}
		\def \middle {1.9}
		\def \oneBeforeMiddle {2.5}
		\def \slope {1}
		\def \r {0.3}
		\def \levelDist {0.5}
		\draw[ns1,white] (\depth*\slope+0.2,-\depth) arc [start angle=0,end angle=360,x radius=\depth*\slope cm+0.2cm,y radius=0.3cm]
		node (1a) [pos=0.3125,inner sep=0pt,minimum width=0pt] {} 
		node (2a) [pos=.695,inner sep=0pt,minimum width=0pt] {};
		\begin{scope}
		\clip[postaction={fill=white,fill opacity=0.2}](2a) arc [start angle=250,end angle=360,x radius=\depth*\slope cm+0.2cm,y radius=0.3cm]arc [start angle=0,end angle=110,x radius=\depth*\slope cm+0.2cm,y radius=0.3cm]--(2a);			
		\foreach \x in {-7.1,-7,...,15.2}%
		\draw[C5!50](\x, -5)--+(12,14.4);
		\end{scope}
		\draw[ns1,C5] (2a) arc [start angle=250,end angle=360,x radius=\depth*\slope cm+0.2cm,y radius=0.3cm]arc [start angle=0,end angle=110,x radius=\depth*\slope cm+0.2cm,y radius=0.3cm]--(2a);
		\begin{scope}
		\clip[postaction={fill=white,fill opacity=0.2}](1a) arc [start angle=112.5,end angle=247.5,x radius=\depth*\slope cm+0.2cm,y radius=0.3cm]--(1a);		
		\foreach \x in {-7.1,-7,...,15.2}%
		\draw[C3!20](\x, -5)--+(-12,14.4);
		\end{scope}
		\draw[ns1,C3!50] (1a) arc [start angle=112.5,end angle=247.5,x radius=\depth*\slope cm+0.2cm,y radius=0.3cm]--(1a);
		\draw[ns1,white](\oneBeforeMiddle*\slope+0.2,-\oneBeforeMiddle) arc [start angle=0,end angle=360,x radius=\oneBeforeMiddle*\slope cm+0.2cm,y radius=0.28cm]
		node (1b) [pos=0.35,inner sep=0pt,minimum width=0pt] {} 
		node (2b) [pos=0.6625,inner sep=0pt,minimum width=0pt] {};
		\begin{scope}
		\clip[postaction={fill=white,fill opacity=0.2}](2b) arc [start angle=237.7,end angle=360,x radius=\oneBeforeMiddle*\slope cm+0.2cm,y radius=0.28cm]arc [start angle=0,end angle=122.3,x radius=\oneBeforeMiddle*\slope cm+0.2cm,y radius=0.28cm]--(2b);
		\foreach \x in {-7.1,-7,...,15.2}%
		\draw[C5!50](\x, -5)--+(12,14.4);
		\end{scope}
		\draw[ns1,C5](2b) arc [start angle=237.7,end angle=360,x radius=\oneBeforeMiddle*\slope cm+0.2cm,y radius=0.28cm]arc [start angle=0,end angle=122.3,x radius=\oneBeforeMiddle*\slope cm+0.2cm,y radius=0.28cm]--(2b);
		\begin{scope}
		\clip[postaction={fill=white,fill opacity=0.2}](1b) arc [start angle=126,end angle=234,x radius=\oneBeforeMiddle*\slope cm+0.2cm,y radius=0.28cm]--(1b);
		\foreach \x in {-7.1,-7,...,15.2}%
		\draw[C3!20](\x, -5)--+(-12,14.4);
		\end{scope}
		\draw[ns1,C3!50](1b) arc [start angle=126,end angle=234,x radius=\oneBeforeMiddle*\slope cm+0.2cm,y radius=0.28cm]--(1b);
		\draw[ns1,C2](0.3*\middle*\slope,-\middle)--(0.3*\oneBeforeMiddle*\slope,-\oneBeforeMiddle);
		\draw[ns1,C2](0.3*\middle*\slope,-\middle)--(0.3*\oneBeforeMiddle*\slope-0.15,-\oneBeforeMiddle);
		\draw[ns1,C2](0.3*\middle*\slope,-\middle)--(0.3*\oneBeforeMiddle*\slope+0.15,-\oneBeforeMiddle);
		\draw[ns1,C2](0.15*\middle*\slope,-\middle)--(0.15*\oneBeforeMiddle*\slope,-\oneBeforeMiddle);
		\draw[ns1,C2](0.15*\middle*\slope,-\middle)--(0.15*\oneBeforeMiddle*\slope-0.15,-\oneBeforeMiddle);
		\draw[ns1,C2](0.15*\middle*\slope,-\middle)--(0.15*\oneBeforeMiddle*\slope+0.15,-\oneBeforeMiddle);
		\draw[ns1,C2](-0*\middle*\slope,-\middle)--(-0*\middle*\slope,-\oneBeforeMiddle);
		\draw[ns1,C2](-0*\middle*\slope,-\middle)--(-0*\middle*\slope-0.15,-\oneBeforeMiddle);
		\draw[ns1,C2](-0*\middle*\slope,-\middle)--(-0*\middle*\slope+0.15,-\oneBeforeMiddle);
		\draw[ns1,C2](-0.15*\middle*\slope,-\middle)--(-0.15*\oneBeforeMiddle*\slope,-\oneBeforeMiddle);
		\draw[ns1,C2](-0.15*\middle*\slope,-\middle)--(-0.15*\oneBeforeMiddle*\slope-0.15,-\oneBeforeMiddle);
		\draw[ns1,C2](-0.15*\middle*\slope,-\middle)--(-0.15*\oneBeforeMiddle*\slope+0.15,-\oneBeforeMiddle);
		\draw[ns1,C2](-0.3*\middle*\slope,-\middle)--(-0.3*\oneBeforeMiddle*\slope,-\oneBeforeMiddle);
		\draw[ns1,C2](-0.3*\middle*\slope,-\middle)--(-0.3*\oneBeforeMiddle*\slope-0.15,-\oneBeforeMiddle);
		\draw[ns1,C2](-0.3*\middle*\slope,-\middle)--(-0.3*\oneBeforeMiddle*\slope+0.15,-\oneBeforeMiddle);
		\draw[ns1,C2](-0.45*\middle*\slope,-\middle)--(-0.45*\oneBeforeMiddle*\slope,-\oneBeforeMiddle);
		\draw[ns1,C2](-0.45*\middle*\slope,-\middle)--(-0.45*\oneBeforeMiddle*\slope-0.15,-\oneBeforeMiddle);
		\draw[ns1,C2](-0.45*\middle*\slope,-\middle)--(-0.45*\oneBeforeMiddle*\slope+0.15,-\oneBeforeMiddle);
		\draw[ns1,white](\middle*\slope+0.2,-\middle) arc [start angle=0,end angle=360,x radius=\middle*\slope cm+0.2cm,y radius=0.25cm]
		node (1c) [pos=0.1875,inner sep=0pt,minimum width=0pt] {}
		node (2c) [pos=0.825,inner sep=0pt,minimum width=0pt] {};
		\begin{scope}
		\clip[postaction={fill opacity=0.2}](2c) arc [start angle=295,end angle=360,x radius=\middle*\slope cm+0.2cm,y radius=0.25cm]arc [start angle=0,end angle=65,x radius=\middle*\slope cm+0.2cm,y radius=0.25cm]--(2c);	
		\foreach \x in {-7.1,-7,...,15.2}%
		\draw[C5!50](\x, -5)--+(-12,14.4);
		\end{scope}
		\draw[ns1,C5](2c) arc [start angle=298,end angle=360,x radius=\middle*\slope cm+0.2cm,y radius=0.25cm]arc [start angle=0,end angle=62,x radius=\middle*\slope cm+0.2cm,y radius=0.25cm]--(2c);
		\begin{scope}
		\clip[postaction={fill opacity=0.2}](1c) arc [start angle=67.5,end angle=292.5,x radius=\middle*\slope cm+0.2cm,y radius=0.25cm]--(1c);		
		\foreach \x in {-7.1,-7,...,15.2}%
		\draw[C3!20](\x, -5)--+(12,14.4);
		\end{scope}
		\draw[ns1,C3!50](1c) arc [start angle=67.5,end angle=292.5,x radius=\middle*\slope cm+0.2cm,y radius=0.25cm]--(1c);

		\node[inner sep=0pt, minimum width=3pt] (1) at (\oneBeforeMiddle*\slope+0.2,-\oneBeforeMiddle) {};
		\node[inner sep=0pt, minimum width=0pt] (2) at (\depth*\slope+0.2,-\depth) {};
		\node[inner sep=0pt, minimum width=0pt] (3) at (-\depth*\slope-0.2,-\depth) {};
		\node[inner sep=0pt, minimum width=3pt] (4) at (-\oneBeforeMiddle*\slope-0.2,-\oneBeforeMiddle) {};
		\node[inner sep=0pt, minimum width=0pt] (8) at (\depth*\slope-\levelDist*\slope+0.2,-\depth+\levelDist) {};
		\node[inner sep=0pt, minimum width=0pt] (9) at (-\depth*\slope+\levelDist*\slope-0.2,-\depth+\levelDist) {};
		\node[inner sep=0pt, minimum width=3pt] (10) at (-\depth*\slope,-\depth) {};
		\node[inner sep=0pt, minimum width=3pt] (11) at (-\depth*\slope+0.2,-\depth) {};
		\node[inner sep=0pt, minimum width=3pt] (12) at (\depth*\slope,-\depth) {};
		\node[inner sep=0pt, minimum width=3pt] (13) at (\depth*\slope-0.2,-\depth) {};
		
		\node[inner sep=0pt, minimum width=3pt] (14) at (\middle*\slope-0.7*\levelDist*\slope+0.1,-\middle+1*\levelDist) {};
		\node[inner sep=0pt, minimum width=3pt] (15) at (-\middle*\slope+0.7*\levelDist*\slope-0.1,-\middle+1*\levelDist) {};
		\node[inner sep=0pt, minimum width=3pt] (16) at (\middle*\slope+0.2,-\middle) {};
		\node[inner sep=0pt, minimum width=3pt] (17) at (-\middle*\slope-0.2,-\middle) {};
		\node[inner sep=0pt, minimum width=3pt] (18) at (\middle*\slope,-\middle) {};
		\node[inner sep=0pt, minimum width=3pt] (19) at (-\middle*\slope,-\middle) {};
		\node[inner sep=0pt, minimum width=3pt] (20) at (\middle*\slope-0.2,-\middle) {};
		\node[inner sep=0pt, minimum width=3pt] (21) at (-\middle*\slope+0.2,-\middle) {};
		\node[inner sep=0pt, minimum width=3pt] (22) at (\oneBeforeMiddle*\slope,-\oneBeforeMiddle) {};
		\node[inner sep=0pt, minimum width=3pt] (23) at (\oneBeforeMiddle*\slope-0.2,-\oneBeforeMiddle) {};
		\node[inner sep=0pt, minimum width=3pt] (24) at (-\oneBeforeMiddle*\slope,-\oneBeforeMiddle) {};
		\node[inner sep=0pt, minimum width=3pt] (25) at (-\oneBeforeMiddle*\slope+0.2,-\oneBeforeMiddle) {};

		\draw[ns2,C5,dotted] (9)--(4);
		\draw[ns2,C5,dotted] (8)--(1);
		\draw[ns2,C5,dotted] (14)--(\middle*\slope-2*\levelDist*\slope+0.3,-\middle+2*\levelDist) ;
		\draw[ns2,C5,dotted] (15)--(-\middle*\slope+2*\levelDist*\slope-0.3,-\middle+2*\levelDist) ;
		\draw[ns1,C5](8)--(2);
		\draw[ns1,C5](9)--(3);
		\draw[ns1,C5](9)--(10);
		\draw[ns1,C5](9)--(11);
		\draw[ns1,C5](8)--(12);
		\draw[ns1,C5](8)--(13);
		\draw[ns1,C5](14)--(16);
		\draw[ns1,C5](14)--(18);
		\draw[ns1,C5](14)--(20);
		\draw[ns1,C5](15)--(17);
		\draw[ns1,C5](15)--(19);
		\draw[ns1,C5](15)--(21);
		\draw[ns1,C5](16)--(1);
		\draw[ns1,C5](16)--(22);
		\draw[ns1,C5](16)--(23);
		\draw[ns1,C5](17)--(4);
		\draw[ns1,C5](17)--(24);
		\draw[ns1,C5](17)--(25);
		\draw[ns1,C3](-1.6,-\depth+0.28)--(-1.4,-\depth+0.28);
		\draw[ns1,C3](-1.6,-\depth+0.21)--(-1.4,-\depth+0.28);
		\draw[ns1,C3](-1.6,-\depth+0.21)--(-1.4,-\depth+0.21);
		\draw[ns1,C3](-1.6,-\depth+0.14)--(-1.4,-\depth+0.14);
		\draw[ns1,C3](-1.6,-\depth+0.14)--(-1.4,-\depth-0.14);
		\draw[ns1,C3](-1.6,-\depth+0.14)--(-1.4,-\depth);
		\draw[ns1,C3](-1.6,-\depth+0.07)--(-1.4,-\depth+0.07);
		\draw[ns1,C3](-1.6,-\depth)--(-1.4,-\depth);
		\draw[ns1,C3](-1.6,-\depth-0.07)--(-1.4,-\depth-0.07);
		\draw[ns1,C3](-1.6,-\depth-0.14)--(-1.4,-\depth-0.14);
		\draw[ns1,C3](-1.6,-\depth-0.14)--(-1.4,-\depth-0.07);
		\draw[ns1,C3](-1.6,-\depth-0.21)--(-1.4,-\depth-0.21);
		\draw[ns1,C3](-1.6,-\depth-0.28)--(-1.4,-\depth-0.28);
		\draw[ns1,C3](-1.6,-\depth-0.07)--(-1.4,-\depth-0.28);
		\draw[ns1,C3](-1.6,-\oneBeforeMiddle+0.22)--(-1.4,-\oneBeforeMiddle+0.23);
		\draw[ns1,C3](-1.6,-\oneBeforeMiddle+0.14)--(-1.4,-\oneBeforeMiddle+0.23);
		\draw[ns1,C3](-1.6,-\oneBeforeMiddle+0.14)--(-1.4,-\oneBeforeMiddle+0.14);
		\draw[ns1,C3](-1.6,-\oneBeforeMiddle+0.07)--(-1.4,-\oneBeforeMiddle+0.07);
		\draw[ns1,C3](-1.6,-\oneBeforeMiddle+0.07)--(-1.4,-\oneBeforeMiddle+0.23);
		\draw[ns1,C3](-1.6,-\oneBeforeMiddle)--(-1.4,-\oneBeforeMiddle);
		\draw[ns1,C3](-1.6,-\oneBeforeMiddle-0.07)--(-1.4,-\oneBeforeMiddle-0.07);
		\draw[ns1,C3](-1.6,-\oneBeforeMiddle-0.07)--(-1.4,-\oneBeforeMiddle+0.14);
		\draw[ns1,C3](-1.6,-\oneBeforeMiddle-0.14)--(-1.4,-\oneBeforeMiddle-0.14);
		\draw[ns1,C3](-1.6,-\oneBeforeMiddle-0.14)--(-1.4,-\oneBeforeMiddle-0.23);
		\draw[ns1,C3](-1.6,-\oneBeforeMiddle-0.22)--(-1.4,-\oneBeforeMiddle-0.23);
		\draw[ns1,C3](-1.6,-\oneBeforeMiddle)--(-1.4,-\oneBeforeMiddle-0.14);
		\draw[ns1,C3](-1.6,-\oneBeforeMiddle+0.22)--(-1.4,-\oneBeforeMiddle+0.23);
		\draw[ns1,C3](-1.6,-\oneBeforeMiddle+0.22)--(-1.4,-\oneBeforeMiddle+0.23);
		\draw[ns1,C3](0.8,-\middle+0.22)--(0.97,-\middle+0.22);
		\draw[ns1,C3](0.8,-\middle+0.22)--(0.97,-\middle+0.14);
		\draw[ns1,C3](0.8,-\middle+0.14)--(0.97,-\middle+0.14);
		\draw[ns1,C3](0.8,-\middle+0.07)--(0.97,-\middle+0.22);
		\draw[ns1,C3](0.8,-\middle+0.07)--(0.97,-\middle+0.07);
		\draw[ns1,C3](0.8,-\middle)--(0.97,-\middle+0.07);
		\draw[ns1,C3](0.8,-\middle)--(0.97,-\middle);
		\draw[ns1,C3](0.8,-\middle+0.07)--(0.97,-\middle-0.07);
		\draw[ns1,C3](0.8,-\middle-0.07)--(0.97,-\middle-0.07);
		\draw[ns1,C3](0.8,-\middle-0.14)--(0.97,-\middle-0.14);
		\draw[ns1,C3](0.8,-\middle-0.07)--(0.97,-\middle-0.22);
		\draw[ns1,C3](0.8,-\middle-0.22)--(0.97,-\middle-0.22);

		\node[minimum height=10pt,inner sep=0,font=\small,C3] at (-1.4,-\middle) {$S_{m'}$};
		\node[minimum height=10pt,inner sep=0,font=\small,C3] at (-2.15,-\oneBeforeMiddle-0.03) {$S_{m'+1}$};
		\node[minimum height=10pt,inner sep=0,font=\small,C3] at (-2.7,-\depth) {$S_{n}$};

		\end{tikzpicture}
		\caption[Schematic representation of $S$ crossing edges.]{Schematic representation of $S$ crossing edges (orange and blue) in the underlying undirected graph in the case of $m'<n$.}\label{fig:expansionOfModels}
	\end{figure}
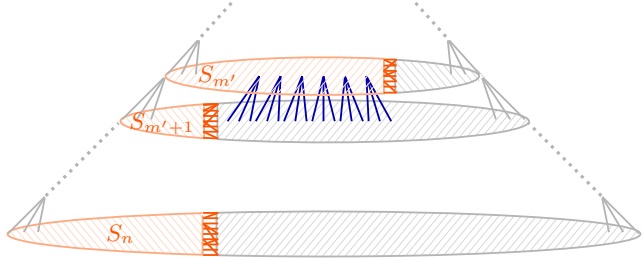
	Assume now that $m'<n$.  Since $S$ is the disjoint union of all $S_m$ we know that the set $\langle S,\overline{S}\rangle_{\underlyingGraph{\struc{A}}}$ contains the disjoint sets $\langle S_m,T_m\setminus S_m  \rangle_{\underlyingGraph{\struc{A}}}$, $\langle T_{m'}\setminus S_{m'},T_{m'} \rangle_{\underlyingGraph{\struc{A}}}$ and  $\langle S_{m'},T_{m'+1}\setminus S_{m'+1} \rangle_{\underlyingGraph{\struc{A}}} $ for all $m\in \{m',\dots,n\}$. Since every vertex in $T_{m'}$ has $D^4$ neighbours in $T_{m'+1}$ and on the other hand every vertex in $T_{m'+1}$ has one neighbour in $T_{m'}$ we know that $|\langle S_{m'},T_{m'+1}\setminus S_{m'+1} \rangle_{\underlyingGraph{\struc{A}}}|=|\langle S_{m'},T_{m'+1} \rangle_{\underlyingGraph{\struc{A}}} |-|\langle S_{m'},S_{m'+1} \rangle_{\underlyingGraph{\struc{A}}} | \geq D^4|S_{m'}|-|S_{m'+1}|\geq D^4(|S_{m'}|-{D^{4m'}}/{2})$. Since additionally  
	$G_m$ is an ${D^2}/{4}$-expander for every $m$ we get 
	\begin{align*}
	&|\langle S,\overline{S}\rangle_{\underlyingGraph{\struc{A}}}|\geq \sum_{m>m'}\frac{D^2}{4}|S_m|+\frac{D^2}{4}|T_{m'}\setminus S_{m'}|+D^4\Big(|S_{m'}|-\frac{D^{4m'}}{2}\Big)\\
	&=\frac{D^2}{4}\sum_{m>m'}|S_m|+\Big(D^4-\frac{D^2}{2}\Big)|S_{m'}|-\Big(D^4-\frac{D^2}{2}\Big)\frac{|T_{m'}|}{2}+\frac{D^2}{8}|S_{m'}|+\frac{D^2}{8}|S_{m'}|\\
	&\stackrel{\text{Equation }\ref{eq:choiceOfm'}}{\geq} \frac{D^2}{4}\sum_{m>m'}|S_m|+\frac{D^2}{8}|S_{m'}|+\frac{D^2}{8}\Big(\frac{|T_{m'}|}{2}\Big)\\
	&\stackrel{\text{Claim }\ref{claim:relationOfSizesOfLevels} }{\geq} \frac{D^2}{4}\sum_{m>m'}|S_m|+\frac{D^2}{8}|S_{m'}|+\frac{D^2}{8}\sum_{m<m'}|T_m|\\
	&\stackrel{|T_{m}|\geq |S_{m}|}{\geq} \frac{D^2}{12}|S|.
	\end{align*}
	By the choice of $\epsilon$ this shows that the models of $\varphi_{\zigzag}$ are a class of $\epsilon$-expanders.
	
\end{proof}

\section{On the non-testability of a $\Pi_2$-property}\label{sec:FOnontestability} 
In this section we prove that there exists an FO property on relational structures in $\Pi_2$ that is not testable. To do so, we first prove that the property $P_{\varphi_{\zigzag}}$ defined by the formula $\varphi_{\zigzag}$ in Section \ref{sec: definitionFormula} is not testable. Later we prove that $\varphi_{\zigzag}$ is equivalent to a sentence in $\Pi_2$. 

\subsection{Non-testability}
Recall that $r$-types are the isomorphism classes of $r$-balls and that restricted to the class $\classStruc{C}_{\sigma,d}$ there are finitely many $r$-types. Let $\tau_1,\dots,\tau_t$ be a list of all $r$-types of bounded degree $d$. We let $\rho_{\struc{A},r}$ be the $r$-type distribution of  $\struc{A}$, \ie   \[\rho_{\struc{A},r}(X)_:=\frac{\sum_{\tau\in X}|\{a\in \univ{A}\mid \mathcal{N}_r^{\struc{A}}(a)\in \tau\}|}{|\univ{A}|}\] for any $X\subseteq \{\tau_1,\dots,\tau_t\}$. For two $\sigma$-structures $\struc{A}$ and $\struc{B}$ we define the sampling distance of depth $r$ as $\delta_{\odot}^r(\struc{A},\struc{B}):=\sup_{X\subseteq \{\tau_1,\dots,\tau_t\}}|\rho_{\struc{A},r}(X)-\rho_{\struc{B},r}(X)|$. Note that $\delta_{\odot}^r(\struc{A},\struc{B})$ is just the total variance distance between $\rho_{\struc{A},r},\rho_{\struc{B},r}$, and it holds that  $$\delta_{\odot}^r(\struc{A},\struc{B})=\frac12\sum_{i=1}^{t}|\rho_{\struc{A},r}(\{\tau_i\})-\rho_{\struc{B},r}(\{\tau_i\})|.$$ Then the sampling distance of $\struc{A}$ and $\struc{B}$ is defined as $$\delta_\odot (\struc{A},\struc{B}):=\sum_{r=0}^{\infty}\frac{1}{2^r}\cdot\delta_\odot^r(\struc{A},\struc{B}).$$

The following theorem was proven for simple graphs and easily extends to $\sigma$-structures. 

\begin{theorem}[\cite{LovaszBook2012}]\label{thm:approximatingNeighbourhoodDistributionBySmallGraph} 
	For every $\lambda>0$ there is a positive integer $n_0$ such that for every $\sigma$-structure $\struc{A}\in \classStruc{C}_{\sigma,d}$ there is a $\sigma$-structure $\struc{H}\in \classStruc{C}_{\sigma,d}$ such that $|H|\leq n_0$ and $\delta_{\odot}(\struc{A},\struc{H})\leq \lambda$.
\end{theorem}
We make use of the following definition of repairable properties. 

\begin{definition}[\cite{AdlerH18}]
	Let $\epsilon\in (0,1]$. A property $\classStruc{P}\subseteq \classStruc{C}_{\sigma,d}$ is \emph{$\epsilon$-repairable}\footnote{In \cite{AdlerH18}, the notion of repairability is called locality.} on $\classStruc{C}_{\sigma,d}$ if there are numbers $r:=r(\epsilon)\in \mathbb{N}$, $\lambda:=\lambda(\epsilon) >0$ and $n_0:=n_0(\epsilon)\in \mathbb{N}$ such that for any $\sigma$-structure $\struc{A}\in \classStruc{P}$ and $\struc{B}\in \classStruc{C}_{\sigma,d}$ both on $n\geq n_0$ vertices, if $\sum_{i=1}^{t}|\rho_{\struc{A},r}(\{\tau_i\})-\rho_{\struc{B},r}(\{\tau_i\})|<\lambda $  then $\struc{B}$ is $\epsilon$-close to $P$, where $\tau_1,\dots,\tau_t$ is a list of all $r$-types of bounded degree $d$.

	The property $\classStruc{P}$ is repairable on $\classStruc{C}_{\sigma,d}$ if it is $\epsilon$-repairable on $\classStruc{C}_{\sigma,d}$ for every $\epsilon\in (0,1]$. 
\end{definition}
The following theorem relating testable properties and repairable properties was proven in \cite{AdlerH18}.
\begin{theorem}[\cite{AdlerH18}]\label{thm:Locality} 
	For every property $\classStruc{P}\in \classStruc{C}_{\sigma,d}$, $\classStruc{P}$ is testable if and only if $\classStruc{P}$ is repairable on $\classStruc{C}_{\sigma,d}$.
\end{theorem}
We recall that $\classStruc{P}_{\zigzag}:=\classStruc{P}_{\varphi_{\zigzag}}$where $\varphi_{\zigzag}$ is the formula from Section \ref{sec: definitionFormula}. We also let $\sigma$, $D$ and $d$ be as defined in Section~$\ref{sec: definitionFormula}$.

\begin{theorem}\label{thm:nonTestabilityForStructures} 
	$\classStruc{P}_{\zigzag}$ is not testable on $\classStruc{C}_{\sigma,d}$.
\end{theorem}
\begin{proof}
	
	We prove non-repairability for $\classStruc{P}_{\zigzag}$ and  get non-testability with Theorem \ref{thm:Locality}. 
	Let $\epsilon:={1}/(144D^2)$ and let $r\in \mathbb{N}$, $\lambda >0$ and $n_0\in \mathbb{N}$ be arbitrary. We set $\lambda':={\lambda}/( t 2^{r+1})$, where $\tau_1,\dots,\tau_t$ are all $r$-types of bounded degree $d$, and let $n_0'$ be the positive integer from Theorem \ref{thm:approximatingNeighbourhoodDistributionBySmallGraph} corresponding to   $\lambda'$. We now pick $n\in \mathbb{N}$ such that $n=\sum_{i=0}^{k}D^{4i}$ for some $k\in \mathbb{N}$, $n\geq 4n_0$ and $n\geq 4({n_0'}/{\lambda})$. Let $\struc{A}\in \classStruc{C}_{\sigma,d}$ be a model of $\varphi_{\zigzag}$ on $n$ elements. By Theorem \ref{thm:approximatingNeighbourhoodDistributionBySmallGraph} there is a structure $\struc{H}\in \classStruc{C}_{\sigma,d}$ on $m\leq n_0'$ elements such that $\delta_{\odot}(\struc{A},\struc{H})\leq \lambda$. Let $\struc{B}$ be the structure consisting of $\lfloor{n}/{m}\rfloor$ copies of $\struc{H}$ and $n\mod m$ isolated elements (elements not being contained in any tuple). Note that we picked $\struc{B}$ such that $|A|=|B|$.

	We will first argue that $\struc{B}$ is in fact $\epsilon$-far from having the property $\classStruc{P}_{\zigzag}$. 
	First we rename the elements from $\univ{B}$ in such a way that $\univ{A}=\univ{B}$ and the number $\sum_{\tilde{R}\in \sigma}|\rel{\tilde{R}}{\struc{A}}\triangle \rel{\tilde{R}}{\struc{B}}|$ of edge modifications to turn $\struc{A}$ and $\struc{B}$ into the same structure is minimal. Pick a partition $\univ{A}=\univ{B}=S\sqcup S'$  in such a way  that $(S\times S')\cap \rel{\tilde{R}}{\struc{B}}=\emptyset$, $(S'\times S)\cap \rel{\tilde{R}}{\struc{B}}=\emptyset$ for any $\tilde{R}\in \sigma$ and $||S|-|S'||$ is minimal among all such partitions. Assume that $|S|\leq |S'|$. Since the connected components of $\gaifman{\struc{B}}$ are of size at most $m$ we know that $||S|-|S'||\leq m$. This is because otherwise we can get a partition $\univ{B}=T\sqcup T'$ with $||T|-|T'||<||S|-|S'||$ by picking all elements of any connected component of $G(\mathcal{B})$, which is contained in $S'$, and moving these elements from $S'$  to $S$. Since $|S|\leq |S'|$ and $m\leq {n}/{4}$ we know that $ {n}/{4}\leq |S|\leq {n}/{2}$.
	Since $(S\times S')\cap \tilde{R}^\mathcal{B}=\emptyset$ we know that $\mathcal{A}$ and $\mathcal{B}$ must differ in at least all tuples that correspond to an $S$ and $S'$ crossing edge in $U(\mathcal{A})$ \ie an edge in $\langle S, S'\rangle_{U(\mathcal{A})}$. Hence
	\begin{align*}
	\sum_{\tilde{R}\in \sigma}|\rel{\tilde{R}}{\struc{A}}\triangle \rel{\tilde{R}}{\struc{B}}|&\geq|\langle S,S'\rangle_{\underlyingGraph{\struc{A}}}|\stackrel{\text{Def }\ref{def:expansionRatio}}{\geq}|S|\cdot  h(\struc{A})\\&\stackrel{\text{Thm }\ref{thm:expansionOfModels}}{\geq}\frac{n}{4}\cdot\frac{D^2}{12}=\frac{1}{48} D^2n\geq \frac{1}{144 D^2}dn.
	\end{align*}
	Therefore $\struc{B}$ is $\epsilon$-far from 
	being in $\classStruc{P}_{\zigzag}$.
	
	However, the neighbourhood distributions of $\struc{A}$ and $\struc{B}$ are similar as the following shows, proving that $\classStruc{P}_{\zigzag}$ is not repairable.

	\begin{align*}
	&\sum_{i=1}^{t}|\rho_{\struc{A},r}(\{\tau_i\})-\rho_{\struc{B},r}(\{\tau_i\})|\\
	&=\sum_{i=1}^{t}\Big|\rho_{\struc{A},r}(\{\tau_i\})
	-\frac{n\mod m}{n}\cdot \rho_{K_1,r}(\{\tau_i\})-\Big\lfloor\frac{n}{m}\Big\rfloor \cdot\frac{m}{n}\cdot\rho_{\struc{H},r}(\{\tau_i\})\Big|
	\\&\leq\sum_{i=1}^{t}\Big|\rho_{\struc{A},r}(\{\tau_i\})-\rho_{\struc{H},r}(\{\tau_i\})\Big|+\sum_{i=1}^{t}\Big|\frac{n\mod m}{n}\cdot \rho_{K_1,r}(\{\tau_i\})\Big|\\&+\sum_{i=1}^{t}\Big|\rho_{\struc{H},r}(\{\tau_i\})-\Big\lfloor\frac{n}{m}\Big\rfloor \cdot\frac{m}{n}\cdot\rho_{\struc{H},r}(\{\tau_i\})\Big|\\
	&\leq\sum_{i=1}^{t}\Big|\rho_{\struc{A},r}(\{\tau_i\})-\rho_{\struc{H},r}(\{\tau_i\})\Big|+\frac{2m}{n}\\
	&\leq t\cdot\sup_{X\subseteq \struc{B}_r}\abs{\rho_{\struc{A},r}(X)-\rho_{\struc{H},r}(X)}+\frac{2m}{n}\\
	&\leq t\cdot 2^r\cdot \delta_{\odot}(\struc{A},\struc{H})+\frac{2m}{n}\\
	&\leq \frac{\lambda}{2}+\frac{\lambda}{2}=\lambda.
	\end{align*}
	The last inequality holds by choice of $\lambda'$ and Theorem \ref{thm:approximatingNeighbourhoodDistributionBySmallGraph}.
\end{proof}

\subsection{Every FO property on degree-regular structures is in $\Pi_2$}
We start with the following observation. 
\begin{observation}\label{ex:delta2}
	A Hanf sentence $\exists ^{\geq m} x\, \phi_{\tau}(x)$
	is short for
	\[\exists x_1\ldots \exists x_m  \big(\bigwedge_{1\leq i,j\leq m, i\neq j} x_i\neq x_j\wedge\bigwedge_{1\leq i\leq m} \phi_{\tau}(x_i)\big),\]
	and $\phi_{\tau}(x_i)$ can be expressed by an $\exists^*\forall$-formula, where the existential
	quantifiers ensure the existence of the desired $r$-neighbourhood with all tuples in relations / not in relations as required by $\tau$, 
	and the universal quantifier is used to express that there are no other elements in
	the $r$-neighbourhood of $x_i$.
\end{observation}
Note that by the above, any Hanf sentence is in~$\Sigma_2$. We now show the following lemma. 

\begin{lemma}\label{lem:d-regHNF}	
	Let $d\in \mathbb N$ and let $\phi$ be an FO sentence.
	If every model of $\varphi$ is $d$-regular, then $\varphi$ is $d$-equivalent to a $\Pi_2$ sentence.
\end{lemma}
The lemma can be equivalently stated by the following syntactic formulation. Let $\varphi^d_{\operatorname{reg}}$ be the FO-sentence expressing that every element has degree $d$. Then for every FO-sentence $\varphi$ the sentence $\varphi \land \varphi^d_{\operatorname{reg}}$ is $d$-equivalent to a sentence in $\Pi_2$. 
\begin{proof}
	Before we begin, let us define an $r$-type 
	$\tau$ to be \emph{$d$-regular}, if for all structures $\struc{A}$ and all elements 
	$a\in \univ{A}$ of $r$-type $\tau$, every $b\in \univ{A}$ with 
	$\dist(a,b)<r$ has $\deg_{\struc{A}}(b)=d$.
	
	We first prove the following claim.	
	\begin{claim}\label{claim:Pi2}
		Let $d\in \mathbb N$, let $\phi$ be an FO sentence, and let $\psi$ be in HNF with 
		$\psi\equiv_d\phi$ such that $\psi$ is in DNF, where the literals
		are Hanf sentences or negated Hanf sentences. Furthermore, assume that the neighbourhood types in all positive Hanf sentences of $\psi$ are $d$-regular. Then $\phi$ is $d$-equivalent to a sentence in $\Pi_2$.
	\end{claim}
	
	\begin{proof}
		Assume $\psi$ is of the form $\exists^{\geq m} x\, \phi_{\tau}(x)$, where $\tau$ is
		$d$-regular.
		As in Observation~\ref{ex:delta2}, we may assume $\phi_{\tau}(x)$ is an 
		$\exists^*\forall$-formula, which arises from a conjunction of an
		$\exists^*$-formula $\phi'_{\tau}(x)$ (expressing that
		$x$ has an `induced sub-neighbourhood' of type $\tau$) and a 
		universal formula saying that
		there are no further elements in the neighbourhood.
		We now have that $\psi\equiv_d\exists^{\geq m} x \,\phi'_{\tau}(x)$. To 
		see this, let
		$\struc{A}\models \exists^{\geq m} x \phi'_{\tau}(x)$ and 
		$\deg(\struc{A})\leq d$. Then 
		$\struc{A}\models \exists^{\geq m} x \phi_{\tau}(x)$ because $\tau$ is 
		$d$-regular. The converse is obvious.
		
		If $\psi$ is of form $\neg \exists^{\geq m} x\, \phi_{\tau}(x)$, where
		$\phi_{\tau}(x)$ is an	$\exists^*\forall$-formula, then 
		$\neg \exists^{\geq m} x\, \phi_{\tau}(x)$ is equivalent to a formula in $\Pi_2$.
		Since $\Pi_2$ is closed under disjunction and conjunction, this proves the claim.
	\end{proof}
	Now the proof follows from Claim~\ref{claim:Pi2}, because if $\phi$ only has $d$-regular models, then by Hanf's theorem there is a formula $\psi\equiv_d \phi$ satisfying the assumptions of the claim.
\end{proof}

\paragraph{Existence of a non-testable $\Pi_2$-property.}
With Lemma \ref{lem:d-regHNF} and Theorem \ref{thm:nonTestabilityForStructures}, we are ready to prove the following theorem.

\begin{theorem}\label{thm:pi2}
	There is a degree bound $d\in \mathbb{N}$ and a signature $\sigma$ such that there exists a property on $\classStruc{C}_{\sigma,d}$ definable by a formula in $\Pi_2$ that is not testable.
\end{theorem}

\begin{proof}
	Pick $d=2D^2+D^4+1$ for any large prime power $D$. Then using the construction from \cite{Reingold00entropywaves} we can find a $(D^4,D,1/4)$-graph $H$. By Theorem~\ref{thm:nonTestabilityForStructures}, using this base expander $H$ for the construction of the formula $\varphi_{\zigzag}$ we get a property which is not testable on $\classStruc{C}_{\sigma,d}$. Since all models of $\varphi_{\zigzag}$ are $d$-regular by construction, Lemma \ref{lem:d-regHNF} gives us that $\varphi_{\zigzag}$ is $d$-equivalent to a formula in $\Pi_2$.
\end{proof}

\section{Reducing to simple undirected graphs}\label{sec:reduction_graphs}
By our previous argument, to show the existence of a non-testable $\Pi_2$-property for simple graphs, \ie undirected graphs without parallel edges and without self-loops, it suffices to construct a non-testable FO graph property of degree regular graphs. 
To do so, we reduce testing the $\sigma$-structure property $\classStruc{P}_{\zigzag}$ from the previous sections to testing a property $\graphProp$ of simple graphs of bounded degree $3$. 
To construct the reduction we carefully translate the edge-coloured directed graphs ($\sigma$-structures) of our previous 
example in Section \ref{sec: definitionFormula} to simple graphs. 
We encode  
$\sigma$-structures by representing each type of directed edge by a constant size graph gadget, maintaining the degree regularity. We then translate the formula $\varphi_{\zigzag}$  into a formula $\graphFormula$ defining the graph property $\graphProp$. This proves the following result. 
\begin{theorem}\label{thm:simpleDelta2}
	There exists an FO property of simple graphs of bounded degree $3$ definable by a formula in $\Pi_2$ that is not testable. 
\end{theorem}
In the rest of this section, we prove the above theorem via local reductions from a structural property to a graph property, and the non-testable $\Pi_2$-property in Theorem \ref{thm:pi2}. This technique will also be in the proofs in Section \ref{sec:GSFlocality}. We remark that there is an alternative proof of the above theorem, which might be of independent interest. That is, we can prove that the property $\graphProp$ is a class of $\xi$-expanders for some constant $\xi>0$ (Lemma~\ref{lemma:undirected_expander}), 
which yields that there are classes of (simple undirected) expanders which are definable in FO (Theorem~\ref{thm:expandingClassDefinableInFO}). 
Then we make use of a result that no property of expanders is testable which is a corollary of the main result of \cite{fichtenberger2019every}. Details are outlined in Appendix~\ref{app:C}. 

\subsection{Local reductions}
 We first introduce the following notion of local reduction between two property testing models. In the following, when the context is clear, we will use $\mathcal{C}$ to denote both a class of structures and the corresponding property testing model, which can be either the bounded-degree model for graphs or bounded-degree model for relational structures.   
\begin{definition}[Local reduction]
	Let $\mathcal{C},\mathcal{C'}$ be two property testing models 
	and let $\mathcal{P}\subseteq\mathcal{C}$, $\mathcal{P}'\subseteq\mathcal{C'}$ be two properties. We say that a function $f:\mathcal{C}\rightarrow \mathcal{C'}$ is a local reduction from $\mathcal{P}$ to $\mathcal{P}'$ if there are constants $c_1,c_2\in \mathbb{N}_{\geq 1}$ such that for every $X\in \mathcal{C}$ the following properties hold.
	\begin{enumerate}
		\item If $X\in \mathcal{P}$ then $f(X)\in \mathcal{P}'$.
		\item If $X$ is $\epsilon$-far from $\mathcal{P}$ then $f(X)$ is $(\epsilon/c_1)$-far from $\mathcal{P}'$.
		\item For every query to $f(X)$ we can adaptively\footnote{By adaptively computing queries we mean that the selection of the next query may depend on the answer to the previous query.} compute $c_2$ queries to $X$ such that the answer to the query to $f(X)$ can be computed  from the answers to the $c_2$ queries to $X$.
	\end{enumerate}
\end{definition}
The following lemma is known. 
\begin{lemma}[Theorem 7.14 in \cite{goldreich2017introduction}]\label{lem:localReduction}
	Let $\mathcal{C},\mathcal{C'}$ be two property testing models, $\mathcal{P}\subseteq\mathcal{C}$, $\mathcal{P}'\subseteq\mathcal{C'}$ be two properties	and $f$ a local reduction from $\mathcal{P}$ to $\mathcal{P}'$. If $\mathcal{P}'$ is testable then so is $\mathcal{P}$.
\end{lemma}
\subsection{Constructing the local reduction}\label{sec:localreduction}
Now we construct a property $\graphProp$ of $3$-regular graphs from the property $\classStruc{P}_{\zigzag}$. We obtain this graph property as $f(\classStruc{P}_{\zigzag})$ by defining a map $f:\classStruc{C}_{\sigma,d}\rightarrow \mathcal{C}_3$. 
To define $f$ we  introduce a distinct arrow-graph gadget for every relation in $\sigma$ (\ie for every edge colour). The map $f$ then replaces every tuple in a certain relation (every coloured, directed edge) by the respective arrow-graph gadget. Here all arrow gadgets are designed to allow for $3$-regularity of the reduced graph. To obtain $3$-regularity we additionally replace every element of a structure in $\classStruc{P}_{\zigzag}$ by a cycle  
of length $d$ such that each arrow-graph gadget can be incident to a unique vertex of the circle. We further prove that this replacement operation defines a local reduction $f$ from $\classStruc{P}_{\zigzag}$ to $\graphProp$. Recall that a local reduction is a function maintaining distance that can be simulated locally by queries. Since by Lemma~\ref{lem:localReduction} local reductions preserve testability, we use the local reduction from $\classStruc{P}_{\zigzag}$ to $\graphProp$ to obtain non-testability of the property  $\graphProp$ from the non-testability of $\classStruc{P}_{\zigzag}$. We will now define $f$ formally.

We first define building blocks which will be combined to different arrow-graph gadgets.
Let $H_1(u,v)$ be the graph with vertex set $\{u=u_0,\dots,v=u_5\}$ and edge set $\{\{u_i,u_{i+3}\}\mid i\in \{0,1,2\}\}$. Next we let $H_2(u,v)$
be the graph with vertex set $\{u=u_0,\dots, v=u_5\}$ and edge set $\{\{u_0,u_6\},\{u_i,u_{i+2}\}\mid i\in \{1,2\}\}$. 
Let $H_3(u,v)$
be the graph with vertex set $\{u=u_0,\dots, v=u_9\}$ and edge set $\{\{u_0,u_9\},\{u_i,u_{i+2}\}\mid i\in \{1,2,5,6\}\}$. Let $H_4(u)$
be the graph with vertex set $\{u=u_0,\dots,u_4\}$ and edge set $\{\{u_0,u_3\},\{u_1,u_{4}\},\{u_2,u_4\}\}$.  See Figure~\ref{fig:buildingBlocks} for illustration.
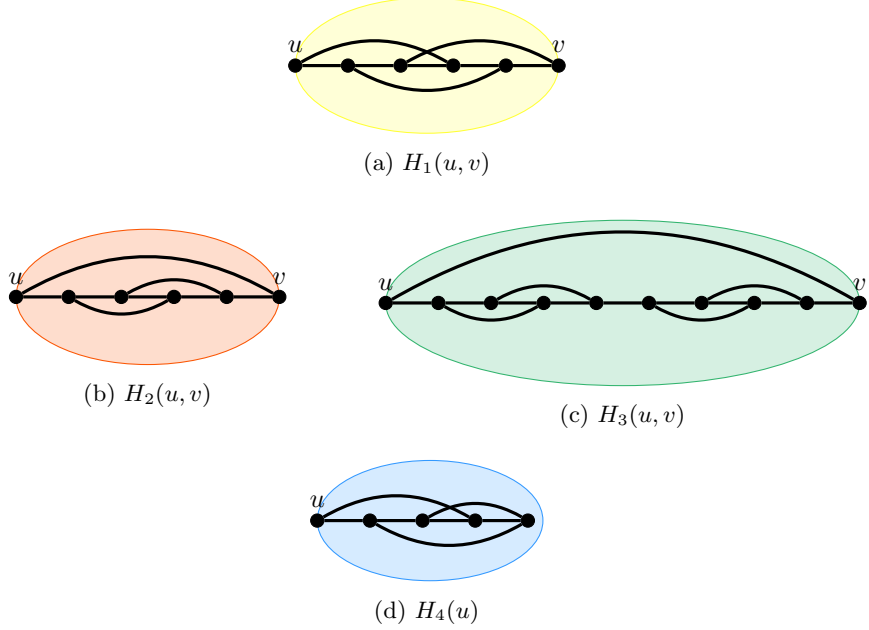
\begin{figure*}
	\definecolor{C1}{RGB}{1,1,1}
	\definecolor{C2}{RGB}{255,255,55}
	\definecolor{C3}{RGB}{251,86,4}
	\definecolor{C4}{RGB}{50,180,110}
	\definecolor{C5}{RGB}{51,153,255}
	\centering
	\centering
	\begin{minipage}{0.4\linewidth}
		\centering
		\begin{tikzpicture}
		\tikzstyle{ns1}=[line width=1.2]
		\tikzstyle{ns2}=[line width=1.2]
		\def \dist {0.7}
		\def \heightWriting {0.2}
		\draw [C2, fill=C2!20](2.5*\dist,0) ellipse (1.75cm and 0.9cm);
		\node[draw,circle,fill=black,inner sep=0pt, minimum width=5pt] (0) at (0,0) {};
		\node[draw,circle,fill=black,inner sep=0pt, minimum width=5pt] (1) at (\dist,0) {};	
		\node[draw,circle,fill=black,inner sep=0pt, minimum width=5pt] (2) at (2*\dist,0) {};
		\node[draw,circle,fill=black,inner sep=0pt, minimum width=5pt] (3) at (3*\dist,0) {};
		\node[draw,circle,fill=black,inner sep=0pt, minimum width=5pt] (4) at (4*\dist,0) {};
		\node[draw,circle,fill=black,inner sep=0pt, minimum width=5pt] (5) at (5*\dist,0) {};
		\draw[ns1] (0)--(1)--(2)--(3)--(4)--(5);
		\path[ns1] (0) edge   [bend left] (3);	
		\path[ns1] (1) edge   [bend right] (4);
		\path[ns1] (2) edge   [bend left] (5);

		\node[minimum height=10pt,inner sep=0] at (0,\heightWriting+0.05) {$u$};		
		\node[minimum height=10pt,inner sep=0] at (5*\dist,\heightWriting+0.05) {$v$};		
		
		\end{tikzpicture}
		\subcaption{$H_1(u,v)$} 	
	\end{minipage}
	\hfill
	\begin{minipage}{0.4\linewidth}
		\centering
		\begin{tikzpicture}
		\tikzstyle{ns1}=[line width=1.2]
		\tikzstyle{ns2}=[line width=1.2]
		\def \dist {0.7}
		\def \heightWriting {0.2}
		\draw [C3, fill=C3!20](2.5*\dist,0) ellipse (1.75cm and 0.9cm);
		\node[draw,circle,fill=black,inner sep=0pt, minimum width=5pt] (0) at (0,0) {};
		\node[draw,circle,fill=black,inner sep=0pt, minimum width=5pt] (1) at (\dist,0) {};	
		\node[draw,circle,fill=black,inner sep=0pt, minimum width=5pt] (2) at (2*\dist,0) {};
		\node[draw,circle,fill=black,inner sep=0pt, minimum width=5pt] (3) at (3*\dist,0) {};
		\node[draw,circle,fill=black,inner sep=0pt, minimum width=5pt] (4) at (4*\dist,0) {};
		\node[draw,circle,fill=black,inner sep=0pt, minimum width=5pt] (5) at (5*\dist,0) {};
		\draw[ns1] (0)--(1)--(2)--(3)--(4)--(5);
		\path[ns1] (0) edge   [bend left] (5);	
		\path[ns1] (1) edge   [bend right] (3);
		\path[ns1] (2) edge   [bend left] (4);

		\node[minimum height=10pt,inner sep=0] at (0,\heightWriting+0.05) {$u$};		
		\node[minimum height=10pt,inner sep=0] at (5*\dist,\heightWriting+0.05) {$v$};		
		
		\end{tikzpicture}
		\subcaption{$H_2(u,v)$} 	
	\end{minipage}
	\begin{minipage}{0.57\linewidth}
		\centering
		\begin{tikzpicture}
		\tikzstyle{ns1}=[line width=1.2]
		\tikzstyle{ns2}=[line width=1.2]
		\def \dist {0.7}
		\def \heightWriting {0.2}
		\draw [C4, fill=C4!20](4.5*\dist,0) ellipse (3.15cm and 1.1cm);
		\node[draw,circle,fill=black,inner sep=0pt, minimum width=5pt] (0) at (0,0) {};
		\node[draw,circle,fill=black,inner sep=0pt, minimum width=5pt] (1) at (\dist,0) {};	
		\node[draw,circle,fill=black,inner sep=0pt, minimum width=5pt] (2) at (2*\dist,0) {};
		\node[draw,circle,fill=black,inner sep=0pt, minimum width=5pt] (3) at (3*\dist,0) {};
		\node[draw,circle,fill=black,inner sep=0pt, minimum width=5pt] (4) at (4*\dist,0) {};
		\node[draw,circle,fill=black,inner sep=0pt, minimum width=5pt] (5) at (5*\dist,0) {};
		\node[draw,circle,fill=black,inner sep=0pt, minimum width=5pt] (6) at (6*\dist,0) {};
		\node[draw,circle,fill=black,inner sep=0pt, minimum width=5pt] (7) at (7*\dist,0) {};
		\node[draw,circle,fill=black,inner sep=0pt, minimum width=5pt] (8) at (8*\dist,0) {};
		\node[draw,circle,fill=black,inner sep=0pt, minimum width=5pt] (9) at (9*\dist,0) {};
		\draw[ns1] (0)--(1)--(2)--(3)--(4)--(5)--(6)--(7)--(8)--(9);
		\path[ns1] (0) edge   [bend left] (9);	
		\path[ns1] (1) edge   [bend right] (3);
		\path[ns1] (2) edge   [bend left] (4);
		\path[ns1] (5) edge   [bend right] (7);
		\path[ns1] (6) edge   [bend left] (8);

		\node[minimum height=10pt,inner sep=0] at (0,\heightWriting+0.05) {$u$};		
		\node[minimum height=10pt,inner sep=0] at (9*\dist,\heightWriting+0.05) {$v$};		
		
		\end{tikzpicture}
		\subcaption{$H_3(u,v)$} 	
	\end{minipage}
\hfill
	\begin{minipage}{0.35\linewidth}
		\centering
		\begin{tikzpicture}
		\tikzstyle{ns1}=[line width=1.2]
		\tikzstyle{ns2}=[line width=1.2]
		\def \dist {0.7}
		\def \heightWriting {0.2}
		\draw [C5, fill=C5!20](2*\dist+0.1,0) ellipse (1.5cm and 0.8cm);
		\node[draw,circle,fill=black,inner sep=0pt, minimum width=5pt] (0) at (0,0) {};
		\node[draw,circle,fill=black,inner sep=0pt, minimum width=5pt] (1) at (\dist,0) {};	
		\node[draw,circle,fill=black,inner sep=0pt, minimum width=5pt] (2) at (2*\dist,0) {};
		\node[draw,circle,fill=black,inner sep=0pt, minimum width=5pt] (3) at (3*\dist,0) {};
		\node[draw,circle,fill=black,inner sep=0pt, minimum width=5pt] (4) at (4*\dist,0) {};
		\draw[ns1] (0)--(1)--(2)--(3)--(4);
		\path[ns1] (0) edge   [bend left] (3);	
		\path[ns1] (1) edge   [bend right] (4);
		\path[ns1] (2) edge   [bend left] (4);

		\node[minimum height=10pt,inner sep=0] at (0,\heightWriting+0.05) {$u$};			
		
		\end{tikzpicture}
		\subcaption{$H_4(u)$} 	
	\end{minipage}
	
	\caption{Illustration of the different building blocks used to define the arrow gadgets.}\label{fig:buildingBlocks}
\end{figure*}

Let $\ell$ be the number of relations (the number of edge colours) in $\sigma$. 
We now introduce the different types of arrow-graph gadgets we need to define the local reduction.
For $1\leq k\leq \ell$, we let $H_{\rightarrow}^k(u_0,v_{2\ell})$ be the graph consisting of $2\ell-1$ vertex disjoint copies $H_1(u_0,v_0), \dots, H_1(u_{k-1},v_{k-1}), H_1(u_{k+1},v_{k+1})$, $\dots, H_1(u_{2\ell-1},v_{2\ell-1})$, one copy $H_2(u_k,v_k)$, one copy $H_3(u_{2\ell},v_{2\ell})$ and additional edges $\{v_i,u_{i+1}\}$ for each $i\in [2\ell]$ connecting the respective copies.  
Note that $H_{\rightarrow}^k(u_0,v_{2\ell})$ has $12\ell +10$ vertices and every vertex apart from $u_0,v_{2\ell}$ has degree $3$. We call $H_{\rightarrow}^k(u_0,v_{2\ell})$ a \emph{$k$-arrow}.  For any graph $G$ and vertices $u,v\in V(G)$, we say that there is a $k$-arrow from $u$ to $v$, denoted $u\xrightarrow{k}v$, if there are $12\ell+8$ vertices $w_1,\dots,w_{12\ell+8}\in V(G)$ and an isomorphism  $g:H_{\rightarrow}^k(u_0,v_{2\ell})\rightarrow \mathcal{N}_1^G(w_1,\dots,w_{12\ell+8})$ such that $g(u_0)=u$ and $g(v_{2\ell})=v$. Note that requiring an isomorphism with these properties guarantees that no vertex contained in a $k$-arrow has neighbours not contained in the $k$-arrow with the exception of the end vertices $u$ and $v$. For any collection $w_1,\dots,w_{12\ell+10}$ of vertices we let $E_{\rightarrow}^k(w_1,\dots,w_{12\ell+10})$ be a set of edges such that there is a graph isomorphism $f:H_{\rightarrow}^k(u_0,v_{2\ell})\rightarrow \big(\{w_1,\dots,w_{12\ell+10}\},E_{\rightarrow}^k(w_1,\dots,w_{12\ell+10})\big)$ with $f(u_0)=w_1$ and $f(v_{2\ell})=w_{12\ell+10}$.  

We now define a second arrow gadget. For $1\leq k\leq \ell$, let $H_{\circlearrowleft}^k(u_0)$ be the graph consisting of $\ell-1$ vertex disjoint copies $H_1(u_0,v_0), \dots$, $H_1(u_{k-1},v_{k-1})$, $H_1(u_{k+1},v_{k+1})$, $\dots, H_1(u_{\ell-1},v_{\ell-1})$, one copy $H_2(u_k,v_k)$, one copy $H_4(u_{\ell})$ and edges $\{v_i,u_{i+1}\}$ for each $i\in [\ell-1]$.  
Note that $H_{\circlearrowleft}^k(u_0)$ has $6\ell +5$ vertices and every vertex apart from $u_0$ has degree $3$. We call $H_{\circlearrowleft}^k$ a \emph{$k$-loop}.  For any graph $G$ and vertex $u\in V(G)$, we say that there is a $k$-loop at $u$, denoted $u\xrightarrow{k}u$, if there are $6\ell+4$ vertices $w_1,\dots,w_{6\ell+4}\in V(G)$ and an isomorphism  $g:H_{\circlearrowleft}^k(u_0)\rightarrow \mathcal{N}_1^G(w_1,\dots,w_{6\ell+4})$ such that $g(u_0)=u$. For any collection $w_1,\dots,w_{6\ell+5}$ vertices we let $E_{\circlearrowleft}^k(w_1\dots,w_{6\ell+5})$ be a set of edges for which there is an isomorphism $f:H_{\circlearrowleft}^k(u_0)\rightarrow \big(\{w_1,\dots,w_{6\ell+5}\},E_{\circlearrowleft}^k(w_1,\dots,w_{6\ell+5})\big)$ for which $f(u_0)=w_1$.

Finally, let $H_{\bot}(u_0)$ be the graph consisting of $\ell$ vertex disjoint copies $H_1(u_0,v_0)$, $\dots$, $H_1(u_{\ell-1},v_{\ell-1})$, one copy $H_4(u_{\ell})$ and additional edges $\{v_i,u_{i+1}\}$ for each $i\in [\ell-1]$.  
Note that $H_\bot(u_0)$ has $6\ell +5$ vertices and every vertex apart from $u_0$ has degree $3$. We call $H_{\bot}$ a \emph{non-arrow}.  For any graph $G$ and vertex $u\in V(G)$, we say that there is a non-arrow at $u$, denoted $u\not\rightarrow$, if there are $6\ell+4$ vertices $w_1,\dots,w_{6\ell+4}\in V(G)$ and an isomorphism  $g:H_\bot\rightarrow \mathcal{N}_1^G(w_1,\dots,w_{6\ell+4})$ such that $g(u_0)=u$. For any collection $w_1,\dots,w_{6\ell+5}$ vertices we let $E_{\bot}(w_1\dots,w_{6\ell+5})$ be a set of edges for which there is an isomorphism $f:H_{\bot}(u_0)\rightarrow \big(\{w_1,\dots,w_{6\ell+5}\},E_{\circlearrowleft}^k(w_1,\dots,w_{6\ell+5})\big)$ for which $f(u_0)=w_1$.  

We now define a function $f:\classStruc{C}_{\sigma,d}\rightarrow \mathcal{C}_3$ by	$f(\struc{A}):=G_\struc{A}$, where $G_{\struc{A}}$ is the graph on vertex set $V(G_{\struc{A}}):= \{u_{a,i},v^k_{a,i}\mid  1\leq i\leq d, a\in \univ{A},1\leq k\leq 6\ell+5\}$ and edge set $E(G_{\struc{A}})$ defined by
\begin{align*}
&\Big\{\{u_{a,i},v^1_{a,i}\}\mid a\in \univ{A}, 1\leq i\leq d\Big\}\\
\cup&\Big\{\{u_{a,d},u_{a,1}\},\{u_{a,i},u_{a,i+1}\}\mid a\in \univ{A}, 1\leq i\leq d-1\Big\} \\
\cup &\bigcup_{ \operatorname{ans}(a,i)=\operatorname{ans}(b,j)=(k,a,b)\atop{a\not=b}}E_{\rightarrow}^k\Big(v_{a,i}^1,\dots,v_{a,i}^{6\ell+5},v_{b,j}^{6\ell+5},\dots,v_{b,j}^1\Big) \\
\cup & \bigcup_{\operatorname{ans}(a,i)=(k,a,a)}E_{\circlearrowleft}^k\Big(v_{a,i}^1,\dots,v_{a,i}^{6\ell+5} \Big)  \\ 
\cup& \bigcup_{\operatorname{ans}(a,i)=\bot}E_{\bot}\Big(v_{a,i}^1,\dots,v_{a,i}^{6\ell+5}\Big),
\end{align*}
where $\operatorname{ans}(a,i)=(k,a,b)$ denotes that the $i$-th tuple of $a$ is $(a,b)$ and is in the $k$-th relation. 
Hence $G_{\struc{A}}$ is defined in such a way that every element $a\in \univ{A}$ is represented by an induced cycle $(u_{a,1},\dots,u_{a,d},u_{a,1})$ and if $(a,b)$ is a tuple in the $k$-th relation of $\sigma$ in $\struc{A}$, then $u_{a,i}\xrightarrow{k}u_{b,j}$ in $G_{\struc{A}}$ for some $1\leq i,j\leq d$,  and $u_{a,i}$ has a non-arrow for every $i$ satisfying that $\operatorname{ans}(a,i)=\bot$ for every $k$. Note that $G_{\struc{A}}$ is $3$ regular by construction for every $\struc{A}\in \classStruc{C}_{\sigma,d}$. For illustration see Figure~\ref{fig:arrowGraphGadgets}.
In the following we refer to vertices of $G_\struc{A}$ of the form $u_{a,i}$ by \emph{element-vertices} while we call vertices of the form $v_{a,i}^j$ \emph{relation-vertices}. 
The following is easy to observe from the construction and from the fact that $d=2D^2+D^4+1<3D^4+1=|\sigma|=\ell$ for some large prime power $D$ (see Section \ref{sec: definitionFormula} for definitions).  
\begin{fact}\label{rem:elementVerts}
	For every $u\in V(G_\struc{A})$, 
	$u$ is an element-vertex iff $u$ is contained in a cycle 
	of length $d$. Furthermore, two vertices $u,v\in V(G_\struc{A})$ correspond to the same element $a$ of $\struc{A}$ (\ie there are $i,j \in \{1,\dots,d\}$ such that $u=u_{a,i}$ and $v=u_{a,j}$) iff there is a cycle of length $d$ containing both $u$ and $v$.
\end{fact}
Note that we do not need to ask for cycles of length $d$ to be induced because the structure we obtain does not allow for cycles of length $d$ apart from the cycles corresponding to elements.

Now we  define property  $\graphProp:=\{f(\struc{A})\mid \struc{A}\in \classStruc{P}_{\zigzag}\}\subseteq \mathcal{C}_3$.
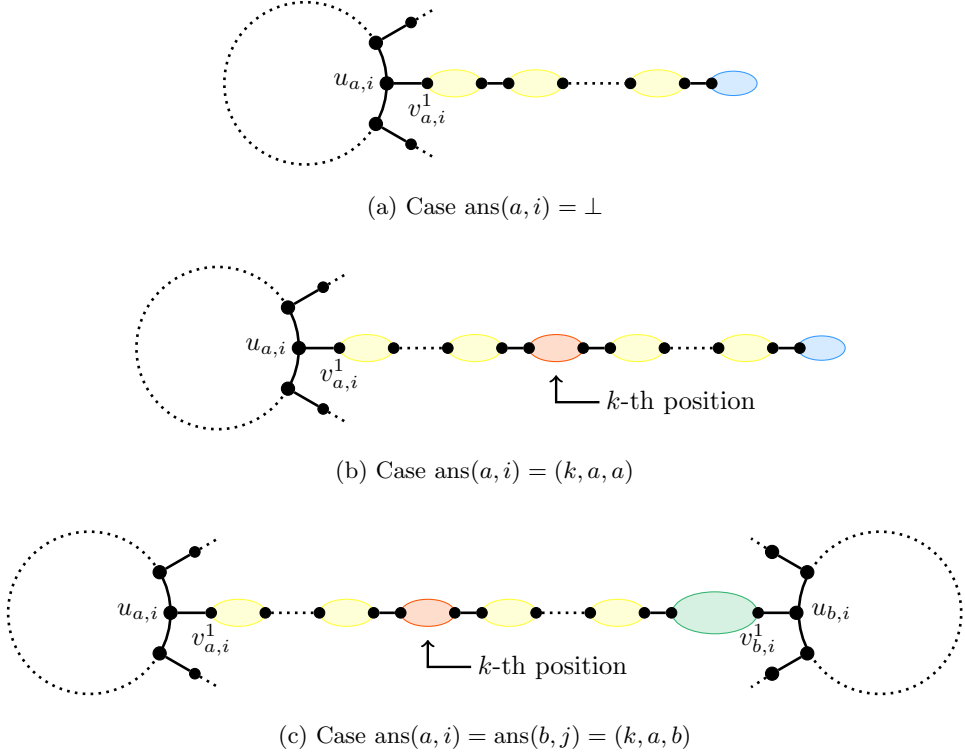
\begin{figure*}
	\definecolor{C1}{RGB}{1,1,1}
	\definecolor{C2}{RGB}{255,255,55}
	\definecolor{C3}{RGB}{251,86,4}
	\definecolor{C4}{RGB}{50,180,110}
	\definecolor{C5}{RGB}{51,153,255}
	\begin{minipage}{\linewidth}
		\centering
		\begin{tikzpicture}[scale = 0.9]
		\tikzstyle{ns1}=[line width=1]
		\tikzstyle{ns2}=[line width=0.9]
		\def \dist {0.4}
		\def \lengthEllipse {0.8}
		\def \heightWriting {0.4}
		\def \radius {1.2}
		\def \margin {0.4}
		\draw [C2, fill=C2!20](\radius+1.5*\dist+0.5*\lengthEllipse,0) ellipse (0.4cm and 0.2cm);
		\draw [C2, fill=C2!20](\radius+2.5*\dist+1.5*\lengthEllipse,0) ellipse (0.4cm and 0.2cm);
		\draw [C2, fill=C2!20](\radius+5*\dist+2.5*\lengthEllipse,0) ellipse (0.4cm and 0.2cm);
		\draw [C5, fill=C5!20](\radius+6*\dist+3.4*\lengthEllipse,0) ellipse (0.35cm and 0.18cm);
		\node[draw,circle,fill=black,inner sep=0pt, minimum width=5pt] (1) at (30:\radius) {};
		\node[draw,circle,fill=black,inner sep=0pt, minimum width=4pt] (1a) at (30:\radius+1.5*\dist) {};
		\node[draw,circle,fill=black,inner sep=0pt, minimum width=5pt] (2) at (0:\radius) {};
		\node[draw,circle,fill=black,inner sep=0pt, minimum width=5pt] (3) at (330:\radius) {};
		\node[draw,circle,fill=black,inner sep=0pt, minimum width=4pt] (3a) at (330:\radius+1.5*\dist) {};
		\draw[ns1,C1] ({30+\margin}:\radius) arc ({30+\margin}:{0-\margin}:\radius);
		\draw[ns1,C1] ({0+\margin}:\radius) arc ({0+\margin}:{-30-\margin}:\radius);
		\draw[ns1,C1,dotted] ({30+\margin}:\radius) arc ({30+\margin}:{330-\margin}:\radius);
		\node[circle,fill=white,inner sep=0pt, minimum width=0pt] (15) at (0,-0.3) {};
		\node[draw,circle,fill=black,inner sep=0pt, minimum width=4pt] (4) at (\radius+1.5*\dist,0) {};
		\node[draw,circle,fill=black,inner sep=0pt, minimum width=4pt] (5) at (\radius+1.5*\dist+\lengthEllipse,0) {};
		\node[draw,circle,fill=black,inner sep=0pt, minimum width=4pt] (6) at (\radius+2.5*\dist+\lengthEllipse,0) {};
		\node[draw,circle,fill=black,inner sep=0pt, minimum width=4pt] (7) at (\radius+2.5*\dist+2*\lengthEllipse,0) {};
		\node[draw,circle,fill=black,inner sep=0pt, minimum width=4pt] (8) at (\radius+5*\dist+2*\lengthEllipse,0) {};
		\node[draw,circle,fill=black,inner sep=0pt, minimum width=4pt] (9) at (\radius+5*\dist+3*\lengthEllipse,0) {};
		\node[draw,circle,fill=black,inner sep=0pt, minimum width=4pt] (10) at (\radius+6*\dist+3*\lengthEllipse,0) {};
		
		\draw[ns1] (4)--(2)(5)--(6)(9)--(10)(1)--(1a)(3)--(3a);
		\draw[ns1,dotted] (7)--(8)(1a)--(30:\radius+2.5*\dist)(3a)--(330:\radius+2.5*\dist);

		\node[minimum height=10pt,inner sep=0] at (0.6*\radius,0) {$u_{a,i}$};
		\node[minimum height=10pt,inner sep=0] at (\radius+1.5*\dist,-\heightWriting) {$v_{a,i}^1$};

		\end{tikzpicture}
		\subcaption{Case $\operatorname{ans}(a,i)=\bot$} 	
	\end{minipage}\\

	\begin{minipage}{\linewidth}
		\centering
		\begin{tikzpicture}[scale = 0.9]
		\tikzstyle{ns1}=[line width=1]
		\tikzstyle{ns2}=[line width=0.9]
		\def \dist {0.4}
		\def \lengthEllipse {0.8}
		\def \heightWriting {0.4}
		\def \radius {1.2}
		\def \margin {0.4}
		\draw [C2, fill=C2!20](\radius+1.5*\dist+0.5*\lengthEllipse,0) ellipse (0.4cm and 0.2cm);
		\draw [C2, fill=C2!20](\radius+3.5*\dist+1.5*\lengthEllipse,0) ellipse (0.4cm and 0.2cm);
		\draw [C3, fill=C3!20](\radius+4.5*\dist+2.5*\lengthEllipse,0) ellipse (0.4cm and 0.2cm);
		\draw [C2, fill=C2!20](\radius+5.5*\dist+3.5*\lengthEllipse,0) ellipse (0.4cm and 0.2cm);
		\draw [C2, fill=C2!20](\radius+7.5*\dist+4.5*\lengthEllipse,0) ellipse (0.4cm and 0.2cm);
		\draw [C5, fill=C5!20](\radius+8.5*\dist+5.4*\lengthEllipse,0) ellipse (0.35cm and 0.18cm);
		\node[draw,circle,fill=black,inner sep=0pt, minimum width=5pt] (1) at (30:\radius) {};
		\node[draw,circle,fill=black,inner sep=0pt, minimum width=4pt] (1a) at (30:\radius+1.5*\dist) {};
		\node[draw,circle,fill=black,inner sep=0pt, minimum width=5pt] (2) at (0:\radius) {};
		\node[draw,circle,fill=black,inner sep=0pt, minimum width=5pt] (3) at (330:\radius) {};
		\node[draw,circle,fill=black,inner sep=0pt, minimum width=4pt] (3a) at (330:\radius+1.5*\dist) {};
		\draw[ns1,C1] ({30+\margin}:\radius) arc ({30+\margin}:{0-\margin}:\radius);
		\draw[ns1,C1] ({0+\margin}:\radius) arc ({0+\margin}:{-30-\margin}:\radius);
		\draw[ns1,C1,dotted] ({30+\margin}:\radius) arc ({30+\margin}:{330-\margin}:\radius);
		\node[circle,fill=white,inner sep=0pt, minimum width=0pt] (15) at (0,-0.3) {};
		\node[draw,circle,fill=black,inner sep=0pt, minimum width=4pt] (4) at (\radius+1.5*\dist,0) {};
		\node[draw,circle,fill=black,inner sep=0pt, minimum width=4pt] (5) at (\radius+1.5*\dist+\lengthEllipse,0) {};
		\node[draw,circle,fill=black,inner sep=0pt, minimum width=4pt] (6) at (\radius+3.5*\dist+\lengthEllipse,0) {};
		\node[draw,circle,fill=black,inner sep=0pt, minimum width=4pt] (7) at (\radius+3.5*\dist+2*\lengthEllipse,0) {};
		\node[draw,circle,fill=black,inner sep=0pt, minimum width=4pt] (8) at (\radius+4.5*\dist+2*\lengthEllipse,0) {};
		\node[draw,circle,fill=black,inner sep=0pt, minimum width=4pt] (9) at (\radius+4.5*\dist+3*\lengthEllipse,0) {};
		\node[draw,circle,fill=black,inner sep=0pt, minimum width=4pt] (10) at (\radius+5.5*\dist+3*\lengthEllipse,0) {};
		\node[draw,circle,fill=black,inner sep=0pt, minimum width=4pt] (11) at (\radius+5.5*\dist+4*\lengthEllipse,0) {};
		\node[draw,circle,fill=black,inner sep=0pt, minimum width=4pt] (12) at (\radius+7.5*\dist+4*\lengthEllipse,0) {};
		\node[draw,circle,fill=black,inner sep=0pt, minimum width=4pt] (13) at (\radius+7.5*\dist+5*\lengthEllipse,0) {};
		\node[draw,circle,fill=black,inner sep=0pt, minimum width=4pt] (14) at (\radius+8.5*\dist+5*\lengthEllipse,0) {};
		
		\draw[ns1] (4)--(2)(7)--(8)(9)--(10)(13)--(14)(1)--(1a)(3)--(3a);
		\draw[ns1,dotted] (5)--(6)(11)--(12)(1a)--(30:\radius+2.5*\dist)(3a)--(330:\radius+2.5*\dist);

		\node[minimum height=10pt,inner sep=0] at (0.6*\radius,0) {$u_{a,i}$};
		\node[minimum height=10pt,inner sep=0] at (\radius+1.5*\dist,-\heightWriting) {$v_{a,i}^1$};	
		\node[minimum height=10pt,inner sep=0] at (\radius+4.5*\dist+4.8*\lengthEllipse,-2*\heightWriting) {$k$-th position};	
		\draw[ns2,<-] (\radius+4.5*\dist+2.5*\lengthEllipse,-\heightWriting)--(\radius+4.5*\dist+2.5*\lengthEllipse,-2*\heightWriting)--(\radius+4.5*\dist+3.3*\lengthEllipse,-2*\heightWriting);		
		
		\end{tikzpicture} 
		\subcaption{Case $\operatorname{ans}(a,i)=(k,a,a)$}
	\end{minipage}\\

	\begin{minipage}{\linewidth}
		\centering
		\begin{tikzpicture}[scale = 0.9]
		\tikzstyle{ns1}=[line width=1]
		\tikzstyle{ns2}=[line width=0.9]
		\def \dist {0.4}
		\def \lengthEllipse {0.8}
		\def \heightWriting {0.4}
		\def \radius {1.2}
		\def \margin {0.4}
		\draw [C2, fill=C2!20](\radius+1.5*\dist+0.5*\lengthEllipse,0) ellipse (0.4cm and 0.2cm);
		\draw [C2, fill=C2!20](\radius+3.5*\dist+1.5*\lengthEllipse,0) ellipse (0.4cm and 0.2cm);
		\draw [C3, fill=C3!20](\radius+4.5*\dist+2.5*\lengthEllipse,0) ellipse (0.4cm and 0.2cm);
		\draw [C2, fill=C2!20](\radius+5.5*\dist+3.5*\lengthEllipse,0) ellipse (0.4cm and 0.2cm);
		\draw [C2, fill=C2!20](\radius+7.5*\dist+4.5*\lengthEllipse,0) ellipse (0.4cm and 0.2cm);
		\draw [C4, fill=C4!20](\radius+8.5*\dist+5.8*\lengthEllipse,0) ellipse (0.63cm and 0.31cm);
		\node[draw,circle,fill=black,inner sep=0pt, minimum width=5pt] (1) at (30:\radius) {};
		\node[draw,circle,fill=black,inner sep=0pt, minimum width=4pt] (1a) at (30:\radius+1.5*\dist) {};
		\node[draw,circle,fill=black,inner sep=0pt, minimum width=5pt] (2) at (0:\radius) {};
		\node[draw,circle,fill=black,inner sep=0pt, minimum width=5pt] (3) at (330:\radius) {};
		\node[draw,circle,fill=black,inner sep=0pt, minimum width=4pt] (3a) at (330:\radius+1.5*\dist) {};
		\draw[ns1,C1] ({30+\margin}:\radius) arc ({30+\margin}:{0-\margin}:\radius);
		\draw[ns1,C1] ({0+\margin}:\radius) arc ({0+\margin}:{-30-\margin}:\radius);
		\draw[ns1,C1,dotted] ({30+\margin}:\radius) arc ({30+\margin}:{330-\margin}:\radius);
		\node[circle,fill=white,inner sep=0pt, minimum width=0pt] (15) at (0,-0.3) {};
		\node[draw,circle,fill=black,inner sep=0pt, minimum width=4pt] (4) at (\radius+1.5*\dist,0) {};
		\node[draw,circle,fill=black,inner sep=0pt, minimum width=4pt] (5) at (\radius+1.5*\dist+\lengthEllipse,0) {};
		\node[draw,circle,fill=black,inner sep=0pt, minimum width=4pt] (6) at (\radius+3.5*\dist+\lengthEllipse,0) {};
		\node[draw,circle,fill=black,inner sep=0pt, minimum width=4pt] (7) at (\radius+3.5*\dist+2*\lengthEllipse,0) {};
		\node[draw,circle,fill=black,inner sep=0pt, minimum width=4pt] (8) at (\radius+4.5*\dist+2*\lengthEllipse,0) {};
		\node[draw,circle,fill=black,inner sep=0pt, minimum width=4pt] (9) at (\radius+4.5*\dist+3*\lengthEllipse,0) {};
		\node[draw,circle,fill=black,inner sep=0pt, minimum width=4pt] (10) at (\radius+5.5*\dist+3*\lengthEllipse,0) {};
		\node[draw,circle,fill=black,inner sep=0pt, minimum width=4pt] (11) at (\radius+5.5*\dist+4*\lengthEllipse,0) {};
		\node[draw,circle,fill=black,inner sep=0pt, minimum width=4pt] (12) at (\radius+7.5*\dist+4*\lengthEllipse,0) {};
		\node[draw,circle,fill=black,inner sep=0pt, minimum width=4pt] (13) at (\radius+7.5*\dist+5*\lengthEllipse,0) {};
		\node[draw,circle,fill=black,inner sep=0pt, minimum width=4pt] (14) at (\radius+8.5*\dist+5*\lengthEllipse,0) {};
		\node[draw,circle,fill=black,inner sep=0pt, minimum width=4pt] (15) at (\radius+8.5*\dist+6.6*\lengthEllipse,0) {};
		
		\node[draw,circle,fill=black,inner sep=0pt, minimum width=5pt][shift={(2*\radius+7*\dist+6.6*\lengthEllipse,0)}] (16) at (150:\radius) {};
		\node[draw,circle,fill=black,inner sep=0pt, minimum width=5pt][shift={(2*\radius+7*\dist+6.6*\lengthEllipse,0)}] (16a) at (150:\radius+1.5*\dist) {};
		\node[draw,circle,fill=black,inner sep=0pt, minimum width=5pt][shift={(2*\radius+7*\dist+6.6*\lengthEllipse,0)}] (17) at (180:\radius) {};
		\node[draw,circle,fill=black,inner sep=0pt, minimum width=5pt][shift={(2*\radius+7*\dist+6.6*\lengthEllipse,0)}] (18) at (210:\radius) {};
		\node[draw,circle,fill=black,inner sep=0pt, minimum width=5pt][shift={(2*\radius+7*\dist+6.6*\lengthEllipse,0)}] (18a) at (210:\radius+1.5*\dist) {};
		\draw[ns1,C1][shift={(2*\radius+10*\dist+6.6*\lengthEllipse,0)}] ({150+\margin}:\radius) arc ({150+\margin}:{180-\margin}:\radius);
		\draw[ns1,C1][shift={(2*\radius+10*\dist+6.6*\lengthEllipse,0)}] ({180+\margin}:\radius) arc ({180+\margin}:{210-\margin}:\radius);
		\draw[ns1,C1,dotted][shift={(2*\radius+10*\dist+6.6*\lengthEllipse,0)}] ({210+\margin}:\radius) arc ({210+\margin}:{510-\margin}:\radius);
		
		\draw[ns1] (4)--(2)(7)--(8)(9)--(10)(13)--(14)(15)--(17)(1)--(1a)(3)--(3a)(16)--(16a)(18)--(18a);
		\draw[ns1,dotted] (5)--(6)(11)--(12)(1a)--(30:\radius+2.5*\dist)(3a)--(330:\radius+2.5*\dist);
		\draw[ns1,dotted][shift={(2*\radius+10*\dist+6.6*\lengthEllipse,0)}] (16a)--(150:\radius+2.5*\dist)(18a)--(210:\radius+2.5*\dist);

		\node[minimum height=10pt,inner sep=0] at (0.6*\radius,0) {$u_{a,i}$};
		\node[minimum height=10pt,inner sep=0] at (1.4*\radius+10*\dist+6.6*\lengthEllipse,0) {$u_{b,i}$ };
		\node[minimum height=10pt,inner sep=0] at (\radius+1.5*\dist,-\heightWriting) {$v_{a,i}^1$};
		\node[minimum height=10pt,inner sep=0] at (1*\radius+8.5*\dist+6.6*\lengthEllipse,-\heightWriting) {$v_{b,i}^1$};
		\node[minimum height=10pt,inner sep=0] at (\radius+4.5*\dist+4.8*\lengthEllipse,-2*\heightWriting) {$k$-th position};	
		\draw[ns2,<-] (\radius+4.5*\dist+2.5*\lengthEllipse,-\heightWriting)--(\radius+4.5*\dist+2.5*\lengthEllipse,-2*\heightWriting)--(\radius+4.5*\dist+3.3*\lengthEllipse,-2*\heightWriting);	
		\end{tikzpicture} 
		\subcaption{Case $\operatorname{ans}(a,i)=\operatorname{ans}(b,j)=(k,a,b)$}
	\end{minipage}
	\caption{Different types of arrows in $G_{\struc{A}}$. Here different coloured ellipses represent a copy of $H_1(u,v),H_2(u,v),H_3(u,v)$ or $H_4(u)$ respectively (see Figure~\ref{fig:buildingBlocks} for details).}\label{fig:arrowGraphGadgets}
	
\end{figure*}

\begin{lemma}\label{lem:local_reduction}
	The map	$f$ is a local reduction from $\classStruc{P}_{\zigzag}$ to $\graphProp$.
\end{lemma}
\begin{proof}
	First note that for any $\struc{A}\in \classStruc{P}_{\zigzag}$, we have that $f(\struc{A})\in \graphProp$ by definition. 
	
	Now let $c_1=2d^2(1+6\ell+5)$. We prove that  if $\struc{A}\in \classStruc{C}_{\sigma,d}$ is $\epsilon$-far from $\classStruc{P}_{\zigzag}$ then $f(\struc{A})$ is $\epsilon/c_1$-far from $\graphProp$ by contraposition. Therefore assume that $f(\struc{A})=:G_{\struc{A}}$ is not $\epsilon/c_1$-far from $\graphProp$ for some $\struc{A}\in \classStruc{C}_{\sigma,d}$. Then there is a set $E\subseteq \{e\subseteq V(G_{\struc{A}})\mid |e|=2 \}$ of size at most $\epsilon d |V(G_{\struc{A}})|/c_1$,  and a graph $G\in \graphProp$ such that  $G$ is obtained from $G_{\struc{A}}$ by modifying the tuples in $E$. By definition of $\graphProp$, there is a structure $\struc{A}_G\in \classStruc{P}_{\zigzag}$ such that $f(\struc{A}_G)=G$. First note that $|\univ{A_G}|=|\univ{A}|$, as $d(1+6\ell+5)|\univ{A}|=|V(G_{\struc{A}})|=|V(G)|=d(1+6\ell+5)|\univ{A_G}|$. Hence there must be a set $R$ of tuples that need to be modified to make $\struc{A}$ isomorphic to $\struc{A}_G$.  
	First note that $R$ cannot contain a tuple $(a,b)$ where $\{u_{a,i},v^k_{a,i},u_{b,i},v_{b,i}^k\mid  1\leq i\leq d, 1\leq k\leq \ell\}\cap e=\emptyset$ for every $e\in E$. This is because if $(a,b)$ is a tuple in the $k$-th relation of $\struc{A}$, then $u_{a,i}\xrightarrow{k}u_{b,j}$ in $G_{\struc{A}}$ for some $i,j\in \{1,\dots,d\}$. But since $\{u_{a,i},v^k_{a,i},u_{b,i},v_{b,i}^k\mid  1\leq i\leq d, 1\leq k\leq \ell\}\cap e=\emptyset$ for every $e\in E$, we have that $u_{a,i}\xrightarrow{k}u_{b,j}$ in $G$. Further, $(u_{a,1},\dots,u_{a,d},u_{a,1})$ and $(u_{b,1},\dots,u_{b,d},u_{b,1})$ are cycles of length $d$ in $G$. Hence by \ref{rem:elementVerts} there are elements $a,b$ in $\struc{A}_G$ corresponding to $(u_{a,1},\dots,u_{a,d},u_{a,1})$ and $(u_{b,1},\dots,u_{b,d},u_{b,1})$ such that $(a,b)$ is  a tuple in the $k$-th relation of $\struc{A}_G$, and hence $(a,b)$ cannot be in $R$. The same argument works when assuming that $(a,b)$ is a tuple in $\struc{A}_G$. Since for every $e\in E$, there  is at most $2d$ tuples $(a,b)$ such that  $\{u_{a,i},v^k_{a,i},u_{b,i},v_{b,i}^k\mid  1\leq i\leq d, 1\leq k\leq \ell\}\cap e\not=\emptyset$, we get that 
	\begin{displaymath}
	|R|\leq 2d \epsilon d |V(G_{\struc{A}})|/c_1=2d^2(1+6\ell+5)\epsilon d|\univ{A}|/c_1=\epsilon d |\univ{A}|.
	\end{displaymath}
	Hence $\struc{A}$ is not $\epsilon$-far to being in $\classStruc{P}_{\zigzag}$.

	Let $c_2:= d+1$. Let $\struc{A}\in \classStruc{C}_{\sigma,d}$ and $G_{\struc{A}}:=f(\struc{A})$. 
	First it is important to observe that we can pick an ordering of the vertices of $G_{\struc{A}}$ such that the position of each vertex depends solely on the number of elements of $\struc{A}$. Hence we can assume that for any element $a$ of $\struc{A}$ we can decide for any vertex $v\in V(G_\struc{A})$ whether $v$ is of the form $u_{a,i}$ and whether $v$ is of the form $v_{a,i}^k$. Now we argue how we can determine the answer to any neighbour query in $G_\struc{A}$. 
	First note that for any $a\in \univ{A}$ and $i\in \{1,\dots,d\}$ the vertex $u_{a,i}$  is adjacent in $G_{\struc{A}}$ to $v_{a,i}^1$ and the two neighbouring vertices on the cycle $(u_{a,1},\dots,u_{a,d},u_{a,1})$. Hence any neighbour query in $G_{\struc{A}}$ to $u_{a,i}$ 
	can be answered without querying $\struc{A}$.  Assume $v\in \{v_{a,i}^k\mid 1\leq k\leq \ell \}$ for some $a\in \univ{A}$ and some $1\leq i\leq d$. Then we can determine all neighbours of $v$ by querying $(a,i)$ and further if  $\operatorname{ans}(a,i)\not=\bot$ and $\operatorname{ans}(a,i)=(k,a,b)$, then we need to query $(b,j)$ for every $1\leq j\leq d$ to find out for which $j$ we have $\operatorname{ans}(b,j)=(k,a,b)$. Hence we can determine the answer to any query to $G_{\struc{A}}$ by making $c_2$ queries to $\struc{A}$. This proves that $f$ is a local reduction  from $\classStruc{P}_{\zigzag}$ to $\graphProp$.
\end{proof}
\subsection{The property of graphs is definable in FO}
In this section we find an FO sentence $\graphFormula$ which defines the property $\graphProp$. We do this by defining a formula expressing for two vertices $u,v$ that $u\xrightarrow{k}v$, a formula expressing for vertex $u$ that $u\xrightarrow{k}u$ and a formula expressing for vertex $u$ that $u\not\rightarrow$ and replacing formulas of the form $R(u,v),R(v,v)$ and $\lnot R(u,v)$ for $R\in  \sigma$ by the new formulas appropriately. We additionally restrict the scope of the quantifiers.
In the previous subsection we already defined $\ell:=|\sigma|$. We further rename the relations in $\sigma$ in an arbitrary way such that for this section we can assume that  $\sigma=\{R_1,\dots,R_\ell\}$.

We now translate the formula $\varphi_{\zigzag}$ into a formula $\graphFormula$
in the language of undirected graphs using the FO formulas defined in the following.
We let $\alpha(x)$ be a formula saying `$x$ is an element-vertex' and $\beta(x,y)$ be a formula saying `$x$ and $y$ represent the same element of $\mathcal{A}$', which is easy to do by Fact~\ref{rem:elementVerts}. We further let $\gamma(x)$ be a formula saying `$x$ is an internal vertex of either a $k$-arrow, a $k$-loop for any $k\in \{1,\dots,\ell\}$ or a non-arrow'. Here an `internal vertex' of an arrow refers to any vertex on this arrow except the two endpoints.  
Let $\delta^k_\rightarrow(x,y)$ 
denote `$x\xrightarrow{k} y$' for any $k\in \{1,\dots,\ell\}$, similarly, let $\delta^k_{\circlearrowleft}(x)$ denote `$x\xrightarrow{k}x$' for any $k\in \{1,\dots,\ell\}$. Given $\varphi_{\zigzag}$, formula $\graphFormula$ is obtained as follows. In $\varphi_{\zigzag}$ we replace each
expression $R_{k}(x,x)$ by $\delta^k_{\circlearrowleft}(x,x)$ and each expression $R_{k}(x,y)$ by $\delta^k_{\rightarrow}(x,y)$ (for $x\not=y$). In addition, we relativise all quantifiers in the following way. We replace every expression of the form $\exists x \,\chi (x,x_1,\dots,x_m)$ by $\exists x\,(\alpha(x)\wedge \chi(x,x_1,\dots,x_m))$ and every expression of the form $\forall x\, \chi(x,x_1,\dots,x_m)$ by $\forall x\,(\alpha(x)\rightarrow \exists y \beta(x,y)\land \chi(y,x_1,\dots,x_m))$. Let us call the resulting formula $\psi$. Then we set $\graphFormula$ to be the conjunction of the formula $\psi$ and the formula $\forall x \Big((\lnot \alpha(x)\rightarrow \gamma(x))\land (\alpha(x)\rightarrow \exists y \gamma(y)\land E(x,y))\Big)$.
\begin{lemma}\label{lem:correspondenceOfModels}
	For any $\struc{A}\in \classStruc{C}_{\sigma,d}$ the following proposition is true. $\struc{A}\models \varphi_{\zigzag}$ if and only if $f(\struc{A})\models \graphFormula$. Additionally we have that if $G\in \mathcal{C}_3$ is a model of $\graphFormula$ then $G\cong f(\struc{A})$ for some $\struc{A}\in \classStruc{A}_{\sigma,d}$.
\end{lemma}
\begin{proof}
	First assume that $\struc{A}\models \varphi_{\zigzag}$. First observe that by construction of $G_\struc{A}:=f(\struc{A})$ and $\psi$, we get that $\struc{A}\models \psi_{\zigzag}$ if and only if $G_\struc{A}\models \psi$.  Note that for this statement it is important that the set of $k$-arrows and $k$-loops for all $k\in  \{1,\dots,\ell\}$ is a set of pairwise non-isomorphic graphs.
	In the construction of $G_\struc{A}$, every vertex is either an element-vertex $u_{a,i}$ in which case it is adjacent to the relation-vertex $v_{a,i}^1$, or is  an internal vertex of some $k$-arrow, $k$-loop or non-arrow. Hence  we get that $G_\struc{A}\models \forall x \Big((\lnot \alpha(x)\rightarrow \gamma(x))\land (\alpha(x)\rightarrow \exists y \gamma(y)\land E(x,y))\Big)$, which completes the proof of the first statement.
	
	Towards proving the second statement of Lemma~\ref{lem:correspondenceOfModels}, let us assume that some graph $G\in \mathcal{C}_3$ is a model of $\graphFormula$. Then $G\models \forall x \Big((\lnot \alpha(x)\rightarrow \gamma(x))\land (\alpha(x)\rightarrow \exists y \gamma(y)\land E(x,y))\Big)$. Hence $G$ consists of a set of element-vertices that are connected according to $\psi$ with $k$-arrow, $k$-loops or non-arrows. Hence we can reverse the local reduction to obtain $\struc{A}_G$ which is the corresponding model of $\varphi_{\zigzag}$ for which $f(\struc{A}_G)\cong G$ by the following construction. For any maximal set of vertices $X\subseteq V(G)$ such that $\beta(u,v)$ holds for every pair $u,v\in X$, we introduce an element $a_X$. For $X,Y\subseteq V(G)$, we add a tuple $(a_X,a_Y)$ to the relation $\rel{R_k}{\struc{A}_G}$ if there are $u\in X$ and $v\in Y$ such that $u\xrightarrow{k}v$ in $G$. With a similar argument as above, we get that $\struc{A}_G$ is a model of $\varphi_{\zigzag}$ by the construction of $\psi$. Additionally we get for some  ordering of the neighbours of each element of $\struc{A}_G$ that $f(\struc{A}_G)\cong G$ (this ordering has to be consistent with the order of $k$-arrows along the cycle of element-vertices).
\end{proof}
\begin{proof}[Proof of Theorem~\ref{thm:simpleDelta2}]
As a consequence from Lemma~\ref{lem:correspondenceOfModels}, we get  that $\graphFormula$  defines the property $\graphProp$ on the class $\mathcal{C}_3$. Since we constructed the local reduction $f$ in such a way that $f(\mathcal{A})$ is $3$-regular for every $\mathcal{A}\in \classStruc{C}_{\sigma,d}$ by Lemma~\ref{lem:d-regHNF}, we get that $\graphProp$ can be defined by a sentence in $\Pi_2$ on the class $\mathcal{C}_3$. Combining this with Lemma~\ref{lem:localReduction} and Lemma~\ref{lem:local_reduction}, we obtain Theorem~\ref{thm:simpleDelta2}.
\end{proof}
We would like to point out here that while we obtain the non-testability of $\graphProp$ using the local reduction $f$,  we can not conclude that $\graphProp$ is a class of expanders. However, we will show that this is true in the following section.
\subsection{The property of graphs is a class of expanders}\label{sec:expansionOfGraphProperty}

In this subsection we show that $\graphProp $ is a family of expanders and hence prove the following theorem. 
\begin{theorem}\label{thm:expandingClassDefinableInFO}
	There exists a universal constant $\xi>0$ and an (infinite) class of $\xi$-expanders with maximum degree at most $3$ which is definable in FO on undirected graphs.
\end{theorem}
Expansion of $\graphProp$ is not needed for the non-testability results in this paper. However, we think that Theorem~\ref{thm:expandingClassDefinableInFO} is of independent interest since it gives us new insights into the expressibility of first-order logic. Furthermore, for an expanding property of undirected graphs, its non-testability follows from the main result from \cite{fichtenberger2019every}. Details of this are given in Appendix~\ref{app:C}.

\begin{lemma}\label{lemma:undirected_expander}
	The models of $\graphFormula$ is a family of $\xi$-expanders, for some constant $\xi>0$.
\end{lemma}
For $\struc{A}\in \classStruc{C}_{\sigma,d}$ we call vertices of 	$G_A:=f(\struc{A})$  of the form $u_{a,i}$ where  $ 1\leq i\leq d$, $a\in \univ{A}$ \emph{original vertices} (because they correspond to an original element of $\struc{A}$) and vertices of the form $v^k_{a,i}$ where  $ 1\leq i\leq d$, $a\in \univ{A}$ and $1\leq k\leq 6\ell+5$ \emph{auxiliary vertices}, where $\ell$ is the number of relations (the number of edge colours) in $\sigma$. Here $f$ is the local reduction defined in the previous section. Note that $\ell$ is a function of the degree bound $d$.
Now consider that $\struc{A}\in \classStruc{P}_{\zigzag}$.
Our strategy to prove that $G_A$ is an expander is to consider different cases dependent on how the number of original vertices relates to the number of auxiliary vertices  contained in some set $S\subseteq V(G_A)$ of size at most $\frac{V(G_A)}{2}$. Since the connected components of $G_A$ after deleting all auxiliary vertices (or after deleting all original vertices) are of constant size,  we get well connectedness of $S$ if the number of auxiliary vertices in comparison to the number of original vertices contained in $S$ is small (or the number of original vertices in comparison to the number of auxiliary vertices contained in $S$ is small respectively).  On the other hand, in the case that the number of original and the number of auxiliary vertices are relatively close, we can use the expansion of $G_A$ to prove that $S$ is well connected to the rest of $G_A$. We give the detailed proof in the following.
\begin{proof}[Proof of Lemma~\ref{lemma:undirected_expander}]
	Let $\struc{A}\in \classStruc{P}_{\zigzag}$ and $G_A:=f(\struc{A})$. Let $S\subset V(G_A)$ such that $|S|\leq \frac{|V(G_A)|}{2}$. Let $V_{\operatorname{original}}\sqcup V_{\operatorname{auxiliary}}=V(G_A)$ be the partition of $V(G_A)$ into original and auxiliary vertices. Let $S_{\operatorname{original}}:=V_{\operatorname{original}}\cap S$ and $ S_{\operatorname{auxiliary}}:=V_{\operatorname{auxiliary}}\cap S$.
	
	First note that by the above definitions, every element in $\struc{A}$ corresponds to $d$   vertices in $V_{\operatorname{original}}$ and every directed coloured edge in $\struc{A}$ corresponds to a constant number $c:=d(6\ell+5)$ of vertices in $V_{\operatorname{auxiliary}}$ (note that $\ell$ depends on $d$ only). 
	
	Assume $|S_{\operatorname{original}}|> \frac{2}{c} \cdot|S|$. 
	Then there are  $|S|-|S_{\operatorname{original}}|<\frac{c-2}{2}\cdot|S_{\operatorname{original}}|$ vertices in $S_{\operatorname{auxiliary}}$. Hence at most $\frac{c-2}{2c}\cdot|S_{\operatorname{original}}|$ of  arrows consist entirely of vertices from $S$. Since  every arrow which is incident to a vertex in $S_{\operatorname{original}}$ and contains at least one vertex which is not in $S$ contributes at least $1$ distinct edge to $\langle S,V(G)\setminus S \rangle_G $ we get  
	\begin{align*}
	|\langle S,V(G)\setminus S \rangle_G| &\geq  \frac{1}{2}\cdot|S_{\operatorname{original}}|-\frac{c-2}{2c}\cdot|S_{\operatorname{original}}|\\&=\frac{1}{c}\cdot |S_{\operatorname{original}}|\geq \frac{2}{c^2}\cdot|S|.
	\end{align*}

	Assume $\frac{1}{2dc}\cdot|S|< |S_{\operatorname{original}}| \leq  \frac{2}{c}\cdot|S|$. Let $\epsilon$ be as defined in the proof of Theorem \ref{thm:expansionOfModels}. 
	We define two sets $S^{\operatorname{full}},S^{\operatorname{part}}\subseteq \univ{A}$ in the following way. 
	\begin{align*}
	S^{\operatorname{full}}&:=\big\{a\in \univ{A}:\{u_{a,i}:1\leq i\leq d\}\subseteq S\big\} \text{ and}\\
	S^{\operatorname{part}}&:=\big\{a\in \univ{A}:\{u_{a,i}:1\leq i\leq d\}\cap S\not= \emptyset \big\}\setminus S^{\operatorname{full}}
	\end{align*}  
	where $(u_{a,1},\dots,u_{a,d},u_{a,1})$ is the cycle representing $a$ in the construction of the local reduction $f$ from the previous section.
	
	First assume that $|S^{\operatorname{part}}|\leq \frac{\epsilon}{(2+c)d^2}\cdot|S_{\operatorname{original}}|$.
	Observe that if $a\in S^{\operatorname{full}}$ and $a'\notin S^{\operatorname{full}}\cup S^{\operatorname{part}}$, then there exists at least one unique edge in $G$ which contributes to $|\langle S,V(G)\setminus S \rangle_G|$. Since there are $d\cdot |S^{\operatorname{part}}|$ coloured edges (tuples) containing an element from $S^{\operatorname{part}}$, we get $$|\langle S,V(G)\setminus S \rangle_G|\geq |\langle S^{\operatorname{full}},\univ{A}\setminus S^{\operatorname{full}} \rangle_{\underlyingGraph{A}}|-d\cdot |S^{\operatorname{part}}|.$$  Since $\struc{A}$ is $d$-regular, every edge gets replaced by $c$  vertices and every element gets replaced by $d$ vertices, we know that $|V(G)|=(d+\frac{dc}{2})|A|$. Hence
	\begin{displaymath}
	|S_{\operatorname{original}}|\leq \frac{2}{c}\cdot |S|\leq \frac{1}{c}\cdot |V(G)|= \frac{(2+c)d}{2c}\cdot|A|.
	\end{displaymath} 
	Furthermore, by definition $|S^{\operatorname{full}}|\leq \frac{|S_{\operatorname{original}}|}{d}$ and hence we get $$|\univ{A}\setminus S^{\operatorname{full}}|\geq \bigg(\frac{2c}{(2+c)d}-\frac{1}{d}\bigg)\cdot|S_{\operatorname{original}}|=\frac{c-2}{(2+c)d}\cdot|S_{\operatorname{original}}|.$$
	Then from Theorem $\ref{thm:expansionOfModels}$ we directly get 
	\begin{align*}
	|\langle S,V(G)\setminus S \rangle_G|&\geq |\langle S^{\operatorname{full}},\univ{A}\setminus S^{\operatorname{full}} \rangle_{U(\mathcal{A})}|-d\cdot S^{\operatorname{part}}\\
	&\geq \epsilon \cdot \min \{|S_{\operatorname{original}}|,|A\setminus S_{\operatorname{original}}|\}-d\cdot S^{\operatorname{part}}\\
	&\geq \bigg( \frac{\epsilon(c-2)}{(2+c)d}-\frac{\epsilon}{(2+c)d}\bigg)\cdot|S_{\operatorname{original}}|\\
	&\geq  \frac{\epsilon(c-3)}{2(2+c)d^2c}\cdot|S|. 
	\end{align*}
	On the other hand if $|S^{\operatorname{part}}|\geq \frac{\epsilon}{(2+c)d^2}\cdot|S_{\operatorname{original}}|$ then $|\langle S,V(G)\setminus S \rangle_G|\geq \frac{\epsilon}{2(2+c)d^3c}\cdot|S|$. This is because every $a\in S^{\operatorname{part}}$ contributes at least one unique edge to $\langle S,V(G)\setminus S \rangle_G$, \ie one of the edges of the cycle $(u_{a,1},\dots,u_{a,k},u_{a,1})$. 
	
	Now assume $|S_{\operatorname{original}}|\leq  \frac{1}{2dc}\cdot |S|$. Therefore there are $|S|-|S_{\operatorname{original}}|\geq \frac{2dc-1}{2dc}\cdot |S|$ in $S_{\operatorname{auxiliary}}$. Of these at least $\frac{2dc-1}{2dc}\cdot |S|-c|S_{\operatorname{original}}|\geq \frac{2dc-1-c}{2dc}\cdot|S|$ vertices in $S_{\operatorname{auxiliary}}$ are not connected with any element from $S_{\operatorname{original}}$ in the graph $G[S]$. Since any connected component of $G[S]$ with no vertices in $S_{\operatorname{original}}$ contains at most $c$ vertices, we get that 
	\begin{align*}
	\langle S,V(G)\setminus S \rangle_G \geq \frac{2dc-c-1}{2dc^2}\cdot |S|.
	\end{align*}
	By setting $\xi=\min\{\frac{2}{c^2},\frac{\epsilon(c-3)}{2(2+c)d^2c},\frac{\epsilon}{2(2+c)d^3c},\frac{2dc-c-1}{2dc^2}\}>0$ 
	we proved the claimed.
\end{proof}

\section{On the testability of all $\Sigma_2$-properties}\label{sec:testableSigma2}
Let $\sigma=\{R_1,\dots,R_m\}$ be any relational signature and $\classStruc{C}_{\sigma,d}$ the set of $\sigma$-structures of bounded degree $d$. We prove the following.
\begin{theorem}\label{thm:sigma2}
Every first-order property defined by a $\sigma$-sentence in $\Sigma_2$ is testable in the bounded-degree model.
\end{theorem}

We adapt the notion of indistinguishability of~\cite{alon2000efficient} from the dense model to the bounded-degree model.

\begin{definition}\label{def:indistinguishability}
	Two properties $\classStruc{P},\classStruc{Q}\subseteq \classStruc{C}_{\sigma,d}$ are called \emph{indistinguishable} if for every $\epsilon \in (0,1)$ there exists $N=N(\epsilon)$ such that for every structure $\struc{A}\in \classStruc{P}$ with $|\univ{A}|>N$ there is a structure $\tilde{\struc{A}}\in \classStruc{Q}$ with the same universe, that is $\epsilon$-close to $\struc{A}$; and for every $\struc{B}\in \classStruc{Q}$ with $|\univ{B}|>N$ there is a structure $\tilde{\struc{B}}\in \classStruc{P}$ with the same universe, that is $\epsilon$-close to $\struc{B}$.
\end{definition}
The following lemma follows from the definitions, and is similar to~\cite{alon2000efficient}, though we make use of the canonical testers for bounded-degree graphs (\cite{CzumajPS16,goldreich2011proximity}). 
 
\begin{lemma}
	If $\classStruc{P},\classStruc{Q}\subseteq \classStruc{C}_{\sigma,d}$ are indistinguishable properties, then $\classStruc{P}$ is testable on $\classStruc{C}_{\sigma,d}$ if and only if $\classStruc{Q}$ is testable on $\classStruc{C}_{\sigma,d}$.
\end{lemma}
\begin{proof}
We show that if $\classStruc{P}$ is testable, then $\classStruc{Q}$ is also testable. The other direction follows by the same argument. Let $\epsilon>0$. Since $\classStruc{P}$ is testable, there exists an $\frac{\epsilon}{2}$-tester for $\classStruc{P}$ with success probability at least $\frac23$. Furthermore, we can assume that the tester (called canonical tester) behaves as follows (see \cite{CzumajPS16,goldreich2011proximity}): it first uniformly samples a constant number of elements, then explores the union of $r$-balls around all sampled elements for some constant $r>0$, and makes a deterministic decision whether to accept, based on an isomorphic copy of the explored substructure. Let  $C=C(\frac{\epsilon}{2},d)$ denote the number of queries the tester made on the input structure. By repeating this tester and taking the majority, we can have a tester $T$ with $c_1\cdot C$ queries and success probability at least $\frac{5}{6}$ for some integer $c_1>0$.

	Let $N$ be a number such that if a structure $\struc{B}$ with $n>N$ elements satisfies $\classStruc{Q}$, then there exists a  $\tilde{\struc{B}}\in \classStruc{P}$ with the same universe such that $\dist(\struc{B},\tilde{\struc{B}})\leq \min\{\frac{\epsilon}{2},\frac{1}{c_2 C\cdot d^{C+2}}\}dn$ for some large constant $c_2>0$. Now we give an $\epsilon$-tester for $\classStruc{Q}$. If the input structure $\struc{B}$ has size at most $N$, we can query the whole input to decide if it satisfies $\classStruc{Q}$ or not. If its size is larger than $N$, then we use the aforementioned $\frac{\epsilon}{2}$-tester for $\classStruc{P}$ with success probability at least $\frac{5}{6}$. If $\struc{B}$ satisfies $\classStruc{Q}$, then there exists $\tilde{\struc{B}}\in \classStruc{P}$ that differs from $\struc{B}$ in no more than ${1}/(c_2 C\cdot d^{C+2}) dn$ places. Since the algorithm samples at most $c_1\cdot C$ elements and queries the $r$-balls around all these sampled elements, for $r< C$, we have that with probability at least $1-\frac{1}{6}$, the algorithm does not query any part where $\struc{B}$ and $\tilde{\struc{B}}$ differ, and thus its output is correct with probability at least $\frac{5}{6}-\frac{1}{6}=\frac23$. If $\struc{B}$ is $\epsilon$-far from satisfying $\classStruc{Q}$ then it is $\frac{\epsilon}{2}$-far from satisfying $\classStruc{P}$ and with probability at least $\frac{5}{6}>\frac23$, the algorithm will reject $\struc{B}$. Thus $\classStruc{Q}$ is also testable.
\end{proof}

\paragraph{High-level idea of proof of Theorem~\ref{thm:sigma2}.} 
Let $\varphi\in \Sigma_2$. We prove that the property defined by $\varphi$ can be written as the union of properties, each of which is defined by another formula $\varphi'$ in $\Sigma_2$ where the structure induced by the existentially quantified variables is a fixed structure $\struc{M}$ (see Claim \ref{claim:JM}). With some further simplification of $\varphi'$, we obtain a formula $\varphi''$ in $\Sigma_2$ which expresses that the structure has to have $\struc{M}$ as an induced substructure  and every set of elements of fixed size $\ell$ has to induce some structure from a set of structures $\mathfrak{H}$, and -- depending on the structure from $\mathfrak{H}$ -- 
there might be some connections to the elements of $\struc{M}$ (see Claim \ref{claim:non-iso}). We then define a formula $\psi$ in $\Pi_1$ such that the property defined by $\psi$ is indistinguishable from the property defined by $\varphi''$ in the sense that we can transform any structure satisfying  $\psi$, into a structure satisfying $\varphi''$ by modifying no more than a small fraction of the tuples and vice versa (see Claim  \ref{claim:indistinguishable}). The intuition behind this is that every structure satisfying $\varphi''$ can be made to satisfy $\psi$ by removing the structure $\struc{M}$ while on the other hand for every structure which satisfies $\psi$ we can plant the structure $\struc{M}$ to make it satisfy $\varphi''$. Since it is a priori unclear how the existentially and universally quantified variables interact, we have to define $\psi$ very carefully. Here it is important to note that the number  of occurrences of structures in $\mathfrak{H}$ forcing an interaction with $\struc{M}$ is limited because of the degree bound (see Claim \ref{claim:notManyTuples}). Thus such structures can not be allowed to occur for models of $\psi$, as here the number of occurrences can not be limited in any way. Since properties defined by a formula in $\Pi_1$ are testable, this implies with the indistinguishability of $\psi$ and $\varphi''$ that the property defined by $\varphi''$ is testable. Furthermore  by the fact that testable properties are closed under union \cite{goldreich2017introduction}, 
we reach the conclusion that any property defined by a formula in $\Sigma_2$ is testable. 
 
We will not directly give a tester for the property $\classStruc{P}_\varphi$  but decompose $\varphi$ into simpler cases. However, every simplification of $\varphi$ used is computable, and the proof below yields a construction of an $\epsilon$-tester for $\classStruc{P}_\varphi$ for every $\epsilon\in (0,1)$ and every $\varphi\in \Sigma_2$.\\

For the full proof of Theorem~\ref{thm:sigma2}, we use the following definition.
\begin{definition}
Let $\struc{A}$ be a $\sigma$-structure with $\univ{A}=\{a_1,\dots,a_t\}$. Let $\overline{z}=(z_1,\dots,z_t)$ be a tuple of variables.  
 	Then we define $\iota^\struc{A}(\overline{z})$ as follows.
\begin{align*}
	\iota^{\struc{A}}(\overline{z}):=
	&\bigwedge_{R\in\sigma}\Bigg(\bigwedge_{\big(a_{i_1},\dots,a_{i_{\ar(R)}}\big)\in \rel{R}{\struc{A}}}R\big(z_{i_1},\dots,z_{i_{\ar(R)}}\big)
	\land \\
	&
	\bigwedge_{\big(a_{i_1},\dots,a_{i_{\ar(R)}}\big)\in \univ{A}^{\ar(R)}\setminus \rel{R}{\struc{A}}}\neg R\big(z_{i_1},\dots,z_{i_{\ar(R)}}\big)\Bigg)	\land
	\bigwedge_{\substack{i,j\in[t]\\i\neq j}}(\neg z_i=z_j).
\end{align*}
	
\end{definition}
Note that 
	for every $\sigma$-structure $\struc{A'}$ and $\overline{a}'=(a_1',\dots,a_t')\in \univ{A'}^t$ we have that $\struc{A'}\models \iota^\struc{A}(\overline{a}')$ if and only if $a_i\mapsto a_i'$, $i\in \{1,\dots,t\}$ is an isomorphism from $\struc{A}$ to $\struc{A'}[\{a_1',\dots,a_t'\}]$. In particular, if $\struc{A'}\models \iota^\struc{A}(\overline{a}')$, then $\{a_1',\dots,a_t'\}$ induces a substructure isomorphic to $\struc{A}$ in $\struc{A'}$.

\begin{proof}[Proof of Theorem~\ref{thm:sigma2}]
Let $\varphi$ be any sentence in  $\Sigma_2$. Therefore we can assume that $\varphi$ is of the form $\varphi=\exists \overline{x} \,\forall\overline{y} \,\chi(\overline{x},\overline{y})$ where $\overline{x}=(x_1,\dots,x_k)$ is a tuple of $k\in \mathbb{N}$ variables, $\overline{y}=(y_1,\dots,y_\ell)$ is a tuple of $\ell\in \mathbb{N}$ variables and $\chi(\overline{x},\overline{y})$ is a quantifier-free formula. We can further assume 
that $\chi(\overline{x},\overline{y})$ is in disjunctive normal form, and that \begin{eqnarray}
\varphi=\exists \overline{x}\,\forall \overline{y}\bigvee_{i\in I}\Big(\alpha^i(\overline{x})\land \beta^i(\overline{y})\land \operatorname{pos}^i(\overline{x},\overline{y})\land \operatorname{neg}^i(\overline{x},\overline{y})\Big),\label{eqn:phi_1}
\end{eqnarray}
 where $\alpha^i(\overline{x})$ is a conjunction of literals only containing variables from $\overline{x}$, $\beta^i(\overline{y})$ is a conjunction of literals only containing  variables in $\overline{y}$, $\operatorname{neg}^i(\overline{x},\overline{y})$ is a conjunction of negated atomic formulas containing both variables from $\overline{x}$ and $\overline{y}$ and $\operatorname{pos}^i(\overline{x},\overline{y})$ is a conjunction of atomic formulas containing both variables from $\overline{x}$ and $\overline{y}$. 
	Now note that if an expression `$x_j=y_{j'}$' appears in a conjunctive clause, then we can replace every occurrence of $y_{j'}$ by $x_j$ in that clause, which will result in an equivalent formula.
	
	We now write the formula $\varphi$ given in (\ref{eqn:phi_1}) as a disjunction over all possible structures in $\classStruc{C}_{\sigma,d}$ the existentially quantified variables could enforce. Since the elements realising the existentially quantified variables will have a certain structure, it is natural to decompose the formula in this way.

	Let $\mathfrak{M}\subseteq \classStruc{C}_{\sigma,d}$ be a set of models of $\varphi$, such that every model $\struc{A}\in \classStruc{C}_{\sigma,d}$ of $\varphi$ contains an isomorphic copy of some $\struc{M}\in \mathfrak{M}$ as an induced substructure, and $\mathfrak{M}$ is minimal with this property.

	\begin{claim}\label{claim:aboutM}
		Every $\struc{M}\in \mathfrak{M}$ has at most $k$ elements.
	\end{claim}
	\begin{proof}
		Assume there is $\struc{M}\in \mathfrak{M}$ with $|M|>k$. Since every structure in $\mathfrak{M}$ is a model of $\varphi$ there must be a tuple $\overline{a}=(a_1,\dots,a_k)\in \univ{M}^k$ such that $\struc{M}\models \forall \overline{y}\bigvee_{i\in I}\Big(\alpha^i(\overline{a})\land \beta^i(\overline{y})\land \operatorname{pos}^i(\overline{a},\overline{y})\land \operatorname{neg}^i(\overline{a},\overline{y})\Big)$. This implies that for every tuple $\overline{b}\in \univ{M}^\ell$ we have $\struc{M}\models \bigvee_{i\in I}\Big(\alpha^i(\overline{a})\land \beta^i(\overline{b})\land \operatorname{pos}^i(\overline{a},\overline{b})\land \operatorname{neg}^i(\overline{a},\overline{b})\Big)$. Furthermore, since $\{a_1,\dots,a_k\}^\ell\subseteq \univ{M}^\ell$ we have that $\struc{M}[\{a_1,\dots,a_k\}]\models \forall \overline{y}\bigvee_{i\in I}\Big(\alpha^i(\overline{a})\land \beta^i(\overline{y})\land \operatorname{pos}^i(\overline{a},\overline{y})\land \operatorname{neg}^i(\overline{a},\overline{y})\Big)$. This means that $\struc{M}[\{a_1,\dots,a_k\}]\models \varphi$. 
		Hence  $\mathfrak{M}$ contains an induced substructure
		$\struc{M'}$ of $\struc{M}[\{a_1,\dots,a_k\}]$. 
		Since every model of $\varphi$ containing $\struc{M}$ as an induced substructure must also contain $\struc{M}'$ 
		as an induced substructure $\mathfrak{M}\setminus \{\struc{M}\}$ is a strictly smaller set than $\mathfrak{M}$ with all desired properties. This contradicts the minimality
		$\mathfrak{M}$. 
	\end{proof}
	Therefore $\mathfrak{M}$ is finite.		
		For $\struc{M}\in\mathfrak{M}$ let $J_\struc{M}:=\{j\in I\mid \struc{M}\models \alpha^j(\overline{m})\text{ for some }\overline{m}\in \univ{M}^{\ell}\}\subseteq I$.
	\begin{claim}\label{claim:JM}
We have		$\varphi \equiv_d \bigvee_{\struc{M}\in \mathfrak{M}}\Big(\exists \overline{x}\forall\overline{y} \Big[\iota^\struc{M}(\overline{x})\land \bigvee_{j\in J_\struc{M}}\Big( \beta^j(\overline{y})\land \operatorname{pos}^j(\overline{x},\overline{y})\land \operatorname{neg}^j(\overline{x},\overline{y})\Big) \Big] \Big)$.
	\end{claim}
	\begin{proof}
Let $\struc{A}\in \classStruc{C}_{\sigma,d}$ be a model of $\varphi$. Then there is a tuple $\overline{a}=(a_1,\dots,a_k)\in \univ{A}^k$ such that $\struc{A}\models \forall y \chi (\overline{a},\overline{y})$. Since $\{a_1,\dots,a_k\}^\ell\subseteq \univ{A}^\ell$ this implies that $\struc{A}[\{a_1,\dots,a_k\}]\models   \forall y \chi (\overline{a},\overline{y})$ and hence $\struc{A}[\{a_1,\dots,a_k\}]\models \varphi$.
		In addition, we may assume that we picked $\overline{a}$ in such a way that for any tuple $\overline{a}'=(a_1',\dots,a_k')\in \{a_1,\dots,a_k\}^k$  
	 with $\{a_1',\dots,a_k'\}\subsetneq \{a_1,\dots,a_k\}$ we have that $\struc{A}\not\models \forall \overline{y} \chi(\overline{a}',\overline{y})$. (The reason is that if for some tuple $\overline{a}'$ this is not the case then we just replace $\overline{a}$ by $\overline{a}'$ and so on until this property holds). Hence $\struc{A}[\{a_1,\dots,a_k\}]$ cannot have a proper induced substructure in $\mathfrak{M}$, and it follows that there is $\struc{M}\in \mathfrak{M}$ such that $\struc{M}\cong \struc{A}[\{a_1,\dots,a_k\}]$. By choice of $J_{\struc{M}}$ we get $\struc{A}\models \forall\overline{y} \Big[\iota^\struc{M}(\overline{a})\land \bigvee_{j\in J_\struc{M}}\Big( \beta^j(\overline{y})\land \operatorname{pos}^j(\overline{a},\overline{y})\land \operatorname{neg}^j(\overline{a},\overline{y})\Big) \Big]$ and hence \[\struc{A}\models \bigvee_{\struc{M}\in \mathfrak{M}}\Big(\exists \overline{x}\forall\overline{y} \Big[\iota^\struc{M}(\overline{x})\land \bigvee_{j\in J_\struc{M}}\Big( \beta^j(\overline{y})\land \operatorname{pos}^j(\overline{x},\overline{y})\land \operatorname{neg}^j(\overline{x},\overline{y})\Big) \Big] \Big).\]
		
		To prove  the other direction, we now let the structure $\struc{A}\in \classStruc{C}_{\sigma,d}$ be a model of the formula \newline$\bigvee_{\struc{M}\in \mathfrak{M}}\Big(\exists \overline{x}\forall\overline{y} \Big[\iota^\struc{M}(\overline{x})\land \bigvee_{j\in J_\struc{M}}\Big( \beta^j(\overline{y})\land \operatorname{pos}^j(\overline{x},\overline{y})\land \operatorname{neg}^j(\overline{x},\overline{y})\Big) \Big] \Big)$. Consequently there is $\struc{M}\in \mathfrak{M}$ and $\overline{a}\in \univ{A}^k$ such that $\struc{A}\models \forall\overline{y} \Big[\iota^\struc{M}(\overline{a})\land \bigvee_{j\in J_\struc{M}}\Big( \beta^j(\overline{y})\land \operatorname{pos}^j(\overline{a},\overline{y})\land \operatorname{neg}^j(\overline{a},\overline{y})\Big) \Big]$. By choice of $J_\struc{M}$ this implies $\struc{A} \models \forall \overline{y}\bigvee_{j\in J_\struc{M}}\Big(\alpha^j(\overline{a})\land \beta^j(\overline{y})\land \operatorname{pos}^j(\overline{a},\overline{y})\land \operatorname{neg}^j(\overline{a},\overline{y})\Big)$ and hence $\struc{A}\models \varphi$.		
	\end{proof}
	
	Since the union of finitely many testable properties is testable (see e.g.~\cite{goldreich2017introduction}),  it is sufficient to show that the property $\classStruc{P}_\varphi$ is testable  where $\varphi$ is of the form 	
	\begin{eqnarray}
	\varphi &=&\exists \overline{x}\forall \overline{y}\chi(\overline{x},\overline{y}), \label{eqn:phi_2}\\
	& \text{ where } & \chi(\overline{x},\overline{y})=\Big[\iota^\struc{M}(\overline{x})\land \bigvee_{j\in J_\struc{M}}\Big(\beta^j(\overline{y})\land \operatorname{pos}^j(\overline{x},\overline{y})\land \operatorname{neg}^j(\overline{x},\overline{y})\Big)\Big],\nonumber
\end{eqnarray} 
for some $\struc{M}\in \mathfrak{M}$.
	In the following, we will enforce that for every conjunctive clause of the big disjunction of $\chi$, the universally quantified variables induce a specific substructure.

For $j\in J_\struc{M}$ let $\mathfrak{H}_j\subseteq \classStruc{C}_{\sigma,d}$ be a maximal set of pairwise non-isomorphic structures $\struc{H}$ such that   $\struc{H}\models \beta^j(\overline{b})$ for some $\overline{b}=(b_1,\dots,b_\ell)\in \univ{H}^\ell$ with $\{b_1,\dots,b_\ell\}=\univ{H}$. 	
	
	\begin{claim}\label{claim:non-iso}
We have 
\[\varphi \equiv_d \exists \overline{x}\forall\overline{y} \Big[\iota^\struc{M}(\overline{x})\land \bigvee_{{\struc{H}\in \mathfrak{H}_j,}\atop{j\in J_\struc{M}}} \Big(\iota^\struc{H}(\overline{y})\land \operatorname{pos}^j(\overline{x},\overline{y})\land \operatorname{neg}^j(\overline{x},\overline{y})\Big)\Big].
\]
	\end{claim}
	\begin{proof}
		Let $\struc{A}\in \classStruc{C}_{\sigma,d}$ and $\overline{a}=(a_1,\dots,a_k)\in \univ{A}^k$. First assume that $\struc{A}\models \forall \overline{y}\chi(\overline{a},\overline{y})$. Hence for any tuple $\overline{b}\in \univ{A}^\ell$ there is an index $j\in J_\struc{M}$ such that $\struc{A}\models \beta^j(\overline{b})\land \operatorname{pos}^j(\overline{a},\overline{b})\land \operatorname{neg}^j(\overline{a},\overline{b})$. 
		Then $\struc{A}\models \beta^j(\overline{b})$ implies that $\struc{A}[\{b_1,\dots,b_\ell\}]\cong \struc{H}$ for some $\struc{H}\in \mathfrak{H}_j$. Hence $\struc{A}\models \iota^\struc{H}(\overline{b})$ and $\struc{A}\models  \Big[\iota^\struc{M}(\overline{a})\land \bigvee_{{\struc{H}\in \mathfrak{H}_j,}\atop{j\in J_\struc{M}}} \Big(\iota^\struc{H}(\overline{b})\land \operatorname{pos}^j(\overline{a},\overline{b})\land \operatorname{neg}^j(\overline{a},\overline{b})\Big)\Big]$.
		
		For the other direction, we let $\struc{A}\models \forall\overline{y} \Big[\iota^\struc{M}(\overline{a})\land \bigvee_{{\struc{H}\in \mathfrak{H}_j,}\atop{j\in J_\struc{M}}} \Big(\iota^\struc{H}(\overline{y})\land \operatorname{pos}^j(\overline{a},\overline{y})\land \operatorname{neg}^j(\overline{a},\overline{y})\Big)\Big]$. Then for every tuple $\overline{b}\in \univ{A}^\ell$ there is an index $j\in J_\struc{M}$ and $\struc{H}\in \mathfrak{H}_j$ such that $\struc{H}\models \iota^\struc{H}(\overline{b})\land \operatorname{pos}^j(\overline{a},\overline{b})\land \operatorname{neg}^j(\overline{a},\overline{b})$. Therefore $\struc{A}[\{b_1,\dots,b_\ell\}]\cong \struc{H}$ and we know that $\struc{A}\models \beta^j(\overline{b})$. Therefore $\struc{A}\models \beta^j(\overline{b})\land \operatorname{pos}^j(\overline{a},\overline{b})\land \operatorname{neg}^j(\overline{a},\overline{b})$ and since this is true for any $\overline{b}\in \univ{A}^\ell$ we get $\struc{A}\models \varphi$.
	\end{proof}

Thus, it suffices to assume that 
\begin{eqnarray}
\varphi &=&\exists \overline{x}\forall\overline{y} \chi(\overline{x},\overline{y}),\label{eqn:phi_3}\\
& \text{ where }& \chi(\overline{x},\overline{y}):=\Big[\iota^\struc{M}(\overline{x})\land \bigvee_{{\struc{H}\in \mathfrak{H}_j,}\atop{j\in J_\struc{M}}} \Big(\iota^\struc{H}(\overline{y})\land \operatorname{pos}^j(\overline{x},\overline{y})\land \operatorname{neg}^j(\overline{x},\overline{y})\Big)\Big]\nonumber
\end{eqnarray}
for some $\struc{M}\in \mathfrak{M}$.  
	
Next we will define a universally quantified formula $\psi$  and show that $\classStruc{P}_\varphi$ is indistinguishable from the property $\classStruc{P}_\psi$. To do so we will need the two claims below.
Intuitively, Claim~\ref{claim:notManyTuples} says that models of $\varphi$ of bounded degree
do not have many `interactions'
between existential and universal variables -- only a constant number of tuples in relations
combine both types of variables.
Note that for a structure $\struc{A}$ and tuples $\overline{a}\in \univ{A}^k$, $\overline{b}=(b_1,\dots,b_\ell)\in \univ{A}^\ell$ the condition $\struc{A}\models \iota^\struc{H}(\overline{b})\land \operatorname{pos}^j(\overline{a},\overline{b})\land \operatorname{neg}^j(\overline{a},\overline{b})$ can force an element of $\overline{b}$ to be in a tuple (of a relation of $\struc{A}$) with an element of $\overline{a}$, even if $\operatorname{pos}^j(\overline{x},\overline{y})$ only contains literals of the form $x_i=y_{i'}$. (For example,  
it may be the case that for some tuple $\overline{b}'\in \{b_1,\dots,b_\ell\}^\ell$, every clause $\iota^{\struc{H'}}(\overline{y})\land \operatorname{pos}^{j'}(\overline{x},\overline{y})\land \operatorname{neg}^{j'}(\overline{x},\overline{y})$ for which $\struc{A}\models \iota^{\struc{H'}}(\overline{b}')\land \operatorname{pos}^{j'}(\overline{a},\overline{b}')\land \operatorname{neg}^{j'}(\overline{a},\overline{b}')$ forces a tuple to contain some element of $\overline{b}'$ and some element of $\overline{a}$.) We will now define a set $J$ to pick out the clauses that do not force a tuple to contain both an element from $\overline{a}$ and $\overline{b}$. Note that we still allow elements from $\overline{b}$ to be amongst the elements in $\overline{a}$. In Claim~\ref{claim:notManyTuples} we show that for every $\struc{A}\in \classStruc{C}_{\sigma,d}$, $\overline{a}\in \univ{A}^k$ for which $\struc{A}\models \forall \overline{y} \chi(\overline{a},\overline{y})$ there are a constant number of tuples $\overline{b}\in \univ{A}^\ell$ that only satisfy clauses which force a tuple to contain both an element from $\overline{a}$ and from $\overline{b}$.

	Let $j\in J_\struc{M}$, $\struc{H}\in \mathfrak{H}_j$ and $\overline{h}=(h_1,\dots,h_\ell)\in \univ{H}^\ell$ such that $\struc{H}\models \iota^\struc{H}(\overline{h})$. We define the set $P_{j,\struc{H}}:=\{h_i\mid i\in \{1,\dots,\ell\}, \operatorname{pos}^j(\overline{x},\overline{y})\text{ does not contain }y_i=x_{i'} \text{ for any }i'\in \{1,\dots,k\}\}.$
	Now we let $J\subseteq J_\struc{M}\times \classStruc{C}_{\sigma,d}$ be the set of pairs $(j,\struc{H})$, with $\struc{H}\in \mathfrak{H}_j$ such that  
	 the disjoint union $\struc{M}\sqcup \struc{H}[P_{j,\struc{H}}]\models \varphi$. Now $J$ precisely specifies the clauses that can be satisfied by a structure $\struc{A}$ and tuple $\overline{a}\in \univ{A}^k$ and $\overline{b}\in \univ{A}^\ell$ where $\struc{A}$ does not contain any tuples both containing elements from $\overline{a}$ and $\overline{b}$.

	\begin{claim}\label{claim:notManyTuples}
		Let $\struc{A}\in \classStruc{C}_{\sigma,d}$ and $\overline{a}=(a_1,\dots,a_k)\in \univ{A}^k$. If  $\struc{A}\models \forall \overline{y}\,\chi(\overline{a},\overline{y})$ then there are at most $k \cdot d$ tuples $\overline{b}\in \univ{A}^\ell $ such that $\struc{A}\not\models \bigvee_{(j,\struc{H})\in J}(\iota^\struc{H}(\overline{b})\land \operatorname{pos}^j(\overline{a},\overline{b})\land \operatorname{neg}^j(\overline{a},\overline{b}))$.
	\end{claim}
	\begin{proof}
		Since	$\struc{A}\models \forall \overline{y}\,\chi(\overline{a},\overline{y})$, it holds that $\struc{A}\models \forall \overline{y} \bigvee_{{\struc{H}\in \mathfrak{H}_j,}\atop{j\in J_\struc{M}}} \Big(\iota^\struc{H}(\overline{y})\land \operatorname{pos}^j(\overline{a},\overline{y})\land \operatorname{neg}^j(\overline{a},\overline{y})\Big)$ by Equation(\ref{eqn:phi_3}). Now let 
		$B:=\{\overline{b}\in \univ{A}^\ell\mid \struc{A}\not\models \bigvee_{(j,\struc{H})\in J}(\iota^\struc{H}(\overline{b})\land \operatorname{pos}^j(\overline{a},\overline{b})\land \operatorname{neg}^j(\overline{a},\overline{b}))\}\subseteq \univ{A}^\ell$. Then each $\overline{b}\in B$ adds at least one to $\sum_{i=1}^{k}\deg_\struc{A}(a_i)$. 
		Since $\struc{A}\in \classStruc{C}_{\sigma,d}$ implies that $\sum_{i=1}^{k}\deg_\struc{A}(a_i)\leq k\cdot d$ we get that $|B|\leq k \cdot d$.
	\end{proof}	
	\begin{claim}\label{claim:propertiesOfUniversalProperties}
		Let $\psi$ be a formula of the form $\psi =  \forall \overline{z} \chi(\overline{z})$ where $\overline{z}=(z_1,\dots,z_t)$ is a tuple of variables and $\chi(\overline{z})$ is a quantifier-free formula. Let $\struc{A}\in \classStruc{C}_{\sigma,d} $ with $|\univ{A}|> d\cdot \ar(\sigma)\cdot t$ 
		and let $b\in A$ be an arbitrary element. Let $\struc{A}\models \psi$ and let $\struc{A'}$ be obtained from $\struc{A}$ by `isolating' $b$, i.\,e.\ by deleting all tuples containing $b$ from $\rel{R}{\struc{A}}$ for every $R\in \sigma$. Then  $\struc{A'}\models\psi$. 
	\end{claim}
	\begin{proof}
		First note that $\struc{A'}\models\chi(\overline{a})$ for any tuple $\overline{a}=(a_1,\dots,a_t)\in (\struc{A}\setminus \{b\})^t$ as no tuple over the set of elements $\{a_1,\dots,a_t\}$ has been deleted.  
		Let $\overline{a}=(a_1,\dots,a_t)\in \univ{A}^t$ be a tuple containing $b$. Pick $b'\in \univ{A}$ such that $\operatorname{dist}_\struc{A}(a_j,b')>1$ for every $j\in \{1,\dots,t\}$. Such an element exists as $|\univ{A}|> d\cdot \ar(R)\cdot t$. Let $\overline{a}'=(a_1',\dots,a_t')$ be the tuple obtained from $\overline{a}$ by replacing any occurrence of $b$ by $b'$. Hence $a_j\mapsto a_j'$ defines an isomorphism from $\struc{A'}[\{a_1,\dots,a_t\}]$ to $\struc{A}[\{a_1',\dots,a_t'\}]$ since $b$ is an isolated element in $\struc{A'}[\{a_1,\dots,a_t\}]$ and $b'$ is an isolated element in $\struc{A}[\{a_1',\dots,a_t'\}]$. Since $\struc{A}\models \chi(\overline{a}')$, it follows that $\struc{A'}\models \chi(\overline{a})$.
	\end{proof}
Let $J'\subseteq J$ be the set of all pairs $(j,\struc{H})$ for which $\operatorname{pos}^j(\overline{x},\overline{y})$ is the empty conjunction. $J'$ contains $(j,\struc{H})$ for which we want to use $\iota^\struc{H}(\overline{y})$ to define the formula $\psi$.
	\begin{claim}\label{claim:indistinguishable}
		The property $\classStruc{P}_\varphi$ with $\varphi$ as in (\ref{eqn:phi_3}) is indistinguishable from the property $\classStruc{P}_\psi$ where $\psi:=\forall \overline{y} \bigvee_{(j,\struc{H})\in J'}\iota^\struc{H}(\overline{y})$.  
	\end{claim}
	\begin{proof} Let $\epsilon>0$ and  $N(\epsilon)=N:= \frac{k\cdot \ell^2\cdot d\cdot \ar(R)}{\epsilon}$ and $\struc{A}\in \classStruc{C}_{\sigma,d}$ be  any structure with $|\univ{A}|>N$. 
		
		First assume that $\struc{A}\models \varphi$. The strategy is to isolate any element $b$  which is contained in a tuple $\overline{b}\in \univ{A}^\ell$ such that  $\struc{A}\not \models \bigvee_{(j,\struc{H})\in J'}\iota^\struc{H}(\overline{b})$ by deleting all tuples containing $b$. This will result in a structure which is  $\epsilon$-close to $\struc{A}$ and a model of $\psi$. 
		
		Let $\overline{a}\in \univ{A}^k$ be a tuple such that $\struc{A}\models \forall \overline{y}\chi(\overline{a},\overline{y})$. Let $B\subseteq \univ{A}^\ell$ be the set of tuples $\overline{b}\in \univ{A}^\ell $ such that $\struc{A}\not\models \bigvee_{(j,\struc{H})\in J}(\iota^\struc{H}(\overline{b})\land \operatorname{pos}^j(\overline{a},\overline{b})\land \operatorname{neg}^j(\overline{a},\overline{b}))$.  Then $|B|\leq k\cdot d$ by Claim \ref{claim:notManyTuples}.  Hence the structure $\struc{A'}$ obtained from $\struc{A}$ by deleting all tuples containing an element of  $C:=\{a_1,\dots,a_k\}\cup \big\{b\in A\mid\text{ there is }(b_1,\dots,b_\ell)\in B\text{ such that }b\in \{b_1,\dots,b_\ell\}\big\}$ is $\epsilon$-close to $\struc{A}$. 
		Since $\struc{A}\models \forall \overline{y}\chi(\overline{a},\overline{y})$ implies $\struc{A}\models \forall \overline{y}\bigvee_{{\struc{H}\in \mathfrak{H}_j,}\atop{j\in J_\struc{M}}} \iota^\struc{H}(\overline{y})$, by Claim \ref{claim:propertiesOfUniversalProperties} we know that $\struc{A'}\models  \forall \overline{y}\bigvee_{{\struc{H}\in \mathfrak{H}_j,}\atop{j\in J_\struc{M}}} \iota^\struc{H}(\overline{y})$. 
		For any tuple $\overline{b}=(b_1,\dots,b_\ell)\in (\univ{A}\setminus C)^\ell$ we have by definition of $J'$ that $\struc{A}\models \iota^\struc{H}(\overline{b})$ for some $(j,\struc{H})\in J'$. Furthermore $\struc{A}[\{b_1,\dots,b_\ell\}]=\struc{A'}[\{b_1,\dots,b_\ell\}]$ and hence $\struc{A'}\models \bigvee_{(j,\struc{H})\in J'}\iota^\struc{H}(\overline{b})$.
		Let $\overline{b}=(b_1,\dots,b_\ell)\in \univ{A}^\ell$ be any tuple containing elements from $C$ and let $c_1,\dots,c_t\in C$ be those elements.  
		Pick $t$ elements $c_1',\dots,c_{t}'\in \univ{A}\setminus C$ such that $\operatorname{dist}_\struc{A}(a_i,c_{i'}')>1$, $\operatorname{dist}_\struc{A}(c_{i'}',b_{i})>1$ and $\operatorname{dist}_\struc{A}(c'_i,c_{i'}')>1$ for suitable $i,i'$. This is possible as $|\univ{A}|> (k+2\ell)\cdot d\cdot \ar(R)$ which guarantees the existence of $k+2\ell$ elements of pairwise distance $1$.
		Let $\overline{b}'=(b_1',\dots,b_\ell')$ be the vector obtained from $\overline{b}$ by replacing $c_{i}$  with $c'_i$. Since $\overline{b}'\in \univ{A}^\ell$ there must be $j'$, $\struc{H'}\in \mathfrak{H}_j$ such that  $\struc{A}\models \iota^{\struc{H'}}(\overline{b}')\land \operatorname{pos}^{j'}(\overline{a},\overline{b}')\land \operatorname{neg}^{j'}(\overline{a},\overline{b}')$. By choice of $c_1',\dots,c_{t}'$ we have that $\operatorname{pos}_{j'}(\overline{x},\overline{y})$ must be the empty conjunction and hence $(j',\struc{H'})\in J'$. 
		Since additionally $b_{i}\mapsto b_{i}'$ defines an isomorphism of $\struc{A}[\{b_1',\dots,b_\ell'\}]$ and $ \struc{A'}[\{b_1,\dots,b_\ell\}]$ this implies that $\struc{A'}\models  \bigvee_{(j,\struc{H})\in J'}\iota^\struc{H}(\overline{b})$ for all $\overline{b}\in \univ{A}^\ell$ and hence $\struc{A'}\models \psi$.\\

		Now we prove the other direction. Let $\struc{A}\models \psi$ with $|\univ{A}|>N$. The idea here is to plant the structure $\struc{M}$ somewhere in $\struc{A}$. While this takes less then an $\epsilon$-fraction of edge modifications the resulting structure will be a model of $\varphi$.
		
		Take any set $B\subseteq \struc{A}$ of $|\univ{M}|$ elements. Let $\struc{A'}$ be the structure obtained from $\struc{A}$ by deleting all edges incident to any element contained in $B$. Let $\struc{A''}$ be the structure obtained from $\struc{A'}$ by adding all tuples such that the structure induced by $B$ is isomorphic to $\struc{M}$. This takes no more than $2\ell\cdot d\cdot \ar(R)<\epsilon\cdot d \cdot |\univ{A}|$ edge modifications. Let $\overline{a}\in B^k$ be such that $\struc{A}\models \iota^\struc{M}(\overline{a})$.  By Claim \ref{claim:propertiesOfUniversalProperties} we get $\struc{A'}\models \psi$. Therefore pick any tuple $\overline{b}=(b_1,\dots,b_\ell)\in (\univ{A}\setminus B)^\ell$. Since by construction we have that all $b_i$'s are of distance at least two from $\overline{a}$ we have that  $\struc{A''}\models \bigvee_{(j,\struc{H})\in J'}(\iota^\struc{H}(\overline{b})\land \operatorname{neg}^j(\overline{a},\overline{b}))$. By choice of $\struc{M}$ we also know that $\struc{A''}\models  \bigvee_{{\struc{H}\in \mathfrak{H}_j,}\atop{j\in J_\struc{M}}}\Big(\iota^\struc{H}(\overline{b})\land \operatorname{pos}^j(\overline{a},\overline{b})\land \operatorname{neg}^j(\overline{a},\overline{b})\Big)$ for all $\overline{b}\in B^\ell$. 
		Therefore pick $\overline{b}=(b_1,\dots,b_\ell)$ containing both elements from $B$ and from $\univ{A}\setminus B$. Now pick a tuple $\overline{b}'=(b_1',\dots,b_\ell')\in (\univ{A}\setminus B)^\ell$ that equals $\overline{b}$ in all positions containing an element from $\univ{A}\setminus B$. As noted before there is $(j,\struc{H})\in J'$ such that $\struc{A''}\models (\iota^\struc{H}(\overline{b}')\land \operatorname{neg}^j(\overline{a},\overline{b}'))$. Hence	$A''[\{b_1',\dots,b_\ell'\}]$ is isomorphic to $H$ and further because $(j,H)\in J'$ the set $P_{j,H}$ (used in the definition of J) is the entire universe of $H$. Since $J'\subseteq J$ this means  that by the definition of $J$ we get $ \struc{A''}[\{a_1,\dots, a_k,b'_1\dots b_\ell'\}]\cong \struc{A''}[\{a_1,\dots, a_k\}]\sqcup A''[\{b'_1\dots b_\ell'\}]\cong M\sqcup H[P_{j,H}] \models\varphi$. 
		Since $\overline{b}\in \{a_1,\dots, a_k,b'_1\dots b_\ell'\}^\ell$  this implies $$\struc{A''}[\{a_1,\dots, a_k,b'_1\dots b_\ell'\}]\models \bigvee_{{\struc{H}\in \mathfrak{H}_j,}\atop{j\in J_\struc{M}}}\Big(\iota^\struc{H}(\overline{b})\land \operatorname{pos}^j(\overline{a},\overline{b})\land \operatorname{neg}^j(\overline{a},\overline{b})\Big)$$. Then $\struc{A''}\models \bigvee_{{\struc{H}\in \mathfrak{H}_j,}\atop{j\in J_\struc{M}}}\Big(\iota^\struc{H}(\overline{b})\land \operatorname{pos}^j(\overline{a},\overline{b})\land \operatorname{neg}^j(\overline{a},\overline{b})\Big)$ and hence $\struc{A''}\models \varphi$.
	\end{proof}
	Since $\psi \in \Pi_1$ we have that $\classStruc{P}_\psi$ is testable, and hence $\classStruc{P}_\varphi$ is testable by Claim \ref{claim:indistinguishable}.
\end{proof}
\section{GSF-locality is not sufficient for proximity oblivious testing}\label{sec:GSFlocality}
In this section we show that the property $\graphProp$ can be defined by a generalised 
notion of forbidden subgraph introduced in \cite{goldreich2011proximity} (Lemma~\ref{lemma:graphproperty_gsf_local}). 
Here a subgraph is only forbidden if it is connected to the rest of the graph in a predefined way, \ie 
for a vertex in a forbidden subgraph we can specify that it can not have neighbours which are not 
contained in the subgraph itself.  Combining our results we show that not every property definable 
by generalised forbidden subgraphs are testable in the bounded-degree model 
(Theorem~\ref{thm:existenceLocalNonTestableProperty}). This implies a negative 
answer to a question posed by Goldreich and Ron in \cite{goldreich2011proximity} 
(Question~\ref{que:areAllGSFLocalPropertiesPropagating}) which asks whether a small 
number of appearances of generalised forbidden subgraphs can be fixed with a small 
number of edge modification or whether any way of fixing the appearances invokes a 
chain reaction of necessary edge modifications. In the following we introduce the 
notions and results needed from \cite{goldreich2011proximity}.

\subsection{Generalised subgraph freeness}\label{sec:gsf_preliminaries}
In the following, we present the formal definitions of generalised subgraph freeness, GSF-local properties and the notion of non-propagation, which were introduced in \cite{goldreich2011proximity}. 
\begin{definition}[Generalized subgraph freeness (GSF)
	]\label{def:gsf}
	A \emph{marked} graph is a graph with each vertex marked as either \emph{`full'} or \emph{`semifull'} or \emph{`partial'}. An \emph{embedding} of a marked graph $F$ into a graph $G$ is an injective map $f:V(F)\rightarrow V(G)$ such that for every $v\in V(F)$ the following three conditions hold.
	\begin{enumerate}
		\item If $v$ is marked `full', then 
		$N_1^G(f(v))=f(N_1^F(v))$. 
		\item If $v$ is marked `semifull', then  
		$N_1^G(f(v))\cap f(V(F))=f(N_1^F(v))$.
		\item If $v$ is marked `partial', then  
		$N_1^G(f(v))\supseteq f(N_1^F(v))$.
	\end{enumerate}
	The graph $G$ is called $F$-free if  there is no embedding of $F$ into $G$. For a set of marked graphs $\mathcal{F}$, a graph $G$ is called $\mathcal{F}$-free if it is $F$-free for every $F\in \mathcal{F}$. 
\end{definition}
Based on the above definition of generalised subgraph freeness, we can define GSF-local properties. 
\begin{definition}[GSF-local properties	]\label{def:localProp}
	Let $\mathcal{P}=\bigcup_{n\in \mathbb{N}}\mathcal{P}_n$ be a graph property where $\mathcal{P}_n=\{G\in \mathcal{P}\mid |V(G)|=n\}$ and  $\overline{\mathcal{F}} = (\mathcal{F}_n)_{n \in \mathbb{N}}$ a sequence of sets of
	marked graphs. $\mathcal{P}$ is called \emph{$\overline{\mathcal{F}}$-local} if there exists an integer $s$ such that for every $n$ the following conditions hold.
	\begin{enumerate}
		\item $\mathcal{F}_n$ is a set of marked graphs, each of size at most $s$.
		\item $\mathcal{P}_n$ equals the set of $n$-vertex graphs that are $\mathcal{F}_n$-free. 
	\end{enumerate}
	$\mathcal{P}$ is called \emph{GSF-local} if there is a sequence $\overline{\mathcal{F}} = (\mathcal{F}_n)_{n \in \mathbb{N}}$  of sets of
	marked graphs such that $\mathcal{P}$ is $\overline{\mathcal{F}}$-local.
\end{definition}
The following notion of non-propagating condition of a sequence of sets of marked graphs was introduced to study constant-query POTs.
\begin{definition}[Non-propagating	]
	Let $\overline{\mathcal{F}} = (\mathcal{F}_n)_{n \in \mathbb{N}}$ be a sequence of sets of
	marked graphs.
	\begin{itemize}
		\item For a graph $G $, a subset $B \subset V(G)$ \emph{covers} $\mathcal{F}_n$ in $G$ if for every marked
		graph $F \in  \mathcal{F}_n$ and every embedding of $F$ in $G$, at least one vertex of $F$ is mapped to a vertex
		in $B$.
		\item The sequence $\overline{\mathcal{F}}$ is \emph{non-propagating} if there exists a (monotonically non-decreasing) function
		$\tau: (0, 1] \rightarrow (0, 1]$ such that the following two conditions hold.
		\begin{enumerate}
			\item For every $\epsilon > 0$ there exists $\beta > 0$ such that $\tau(\beta) < \epsilon$.
			\item For every graph $G$ and every $B \subset V(G)$ such that $B$ covers $\mathcal{F}_n$ in $G$, either $G$
			is $\tau(|B|/n )$-close to being $\mathcal{F}_n$-free or there are no $n$-vertex graphs that are $\mathcal{F}_n$-free. 
		\end{enumerate}
		A GSF-local property $\mathcal{P}$ is \emph{non-propagating} if there exists a non-propagating sequence $\overline{\mathcal{F}}$ such
		that $\mathcal{P}$ is $\overline{\mathcal{F}}$-local.
	\end{itemize}
\end{definition}
In the above definition, the set $B$ can be viewed as the set involving necessary modifications for repairing a graph $G$ that does not satisfy the property $\mathcal{P}$ that is $\overline{\mathcal{F}}$-local, and the second condition says we do not need to modify $G$ ``much beyond'' $B$. In particular, it implies we can repair 
$G$ without triggering a global ``chain reaction''. 
Goldreich and Ron gave the following characterization for the proximity-oblivious testable properties in the bounded-degree graph model. 
\begin{theorem}[Theorem 5.5 in \cite{goldreich2011proximity}]\label{thm:charOfPOT}
	A graph property $\mathcal{P}$ has a
	constant-query proximity-oblivious tester if and only if $\mathcal{P}$ is GSF-local and non-propagating.
\end{theorem}

The following open question was raised in \cite{goldreich2011proximity}. 
\begin{question}[Are all GSF-local properties non-propagating?]\label{que:areAllGSFLocalPropertiesPropagating}
	Is it the case that for every GSF-local property $\mathcal{P}=\bigcup_{n\in \mathbb{N}}\mathcal{P}_n$,  there is a sequence $\overline{\mathcal{F}} = (\mathcal{F}_n)_{n \in \mathbb{N}}$ that is non-propagating and $\mathcal{P}$ is $\overline{\mathcal{F}}$-local?
\end{question}
We are now able to state our theorem answering Question~\ref{que:areAllGSFLocalPropertiesPropagating}. The rest of this section is dedicated to the proof of Theorem~\ref{thm:existenceLocalNonTestableProperty}.
\begin{theorem}\label{thm:existenceLocalNonTestableProperty}
	There exists a GSF-local property of graphs of bounded degree $3$ that is not testable in the bounded-degree graph model. Thus, not all GSF-local properties are non-propagating.  
\end{theorem}
\subsection{Relating different notions of locality}\label{sec:relatingNotionsOfLocality}
In this section we define properties by prescribing upper and lower bounds on the number of occurrences of  neighbourhood types. These bounds are given by \emph{neighbourhood profiles} which we will define formally below. 
We use these properties to give a natural characterization of FO properties of bounded-degree structures in Lemma~\ref{lemma:FO-neighbourhood}, which is a straightforward consequence of Hanf's Theorem (Theorem~\ref{thm:Hanf}). We use this characterization to establish links between FO definability and GSF-locality. This connection is  the key ingredient in the proof of our main theorem.\\

Observe that for fixed $r,d\in \mathbb N$ and $\sigma$, there are only finitely many $r$-types in structures in $\classStruc{C}_{\sigma,d}$.
For any signature $\sigma$ and  $d,r\in \mathbb{N}$ we let $n_{d,r,\sigma}\in \mathbb{N}$ be the number of different $r$-types of $\sigma$-structures of degree at most $d$. Assuming that for all $d,r\in \mathbb{N}$ the $r$-neighbourhood-types of $\sigma$-structures of degree at most $d$ are ordered, we let $\tau_{d,r,\sigma}^i$ denote the $i$-th 
such neighbourhood type, for $i\in \{1,\dots,n_{d,r,\sigma}\}$. 
With each $\sigma$-structure $\struc{A}\in \classStruc{C}_{\sigma,d}$ we associate its 
\emph{$r$-histogram vector} $\vet{v}_{d,r,\sigma}(\struc{A})$, given by 
\begin{displaymath}
(\vet{v}_{d,r,\sigma}(\struc{A}))_i:=|\{a\in \univ{A}\mid \mathcal{N}_{r}^{\struc{A}}(a)\in \tau_{d,r,\sigma}^i\}|.
\end{displaymath} 
We let 
\begin{displaymath}
\mathfrak{I}:=\{[k,l]\mid k\leq l\in \mathbb{N}\}\cup \{[k,\infty)\mid k\in \mathbb{N}\}
\end{displaymath}
be the set of all closed or half-closed, infinite intervals with natural lower/upper bounds.
\begin{definition}
	Let $\sigma$ be a signature and $d,r\in \mathbb{N}$.
	\begin{enumerate}
		\item An \emph{$r$-neighbourhood profile} 
		of degree $d$ is a function $\rho:\{1,\dots,n_{d,r,\sigma}\}\rightarrow \mathfrak{I}$. 
		
		\item For a structure $\struc{A}\in \classStruc{C}_{\sigma,d}$, we say that $\struc{A}$ obeys $\rho$, denoted by $\struc{A}\sim \rho$, if
		\[
		(\vet{v}_{d,r,\sigma}(\struc{A}))_i\in \rho(i) \text{ for all }i\in \{1,\dots,n_{d,r,\sigma}\}.
		\]
		Let $\classStruc{P}_\rho$ be the set of structures $\struc{A}$ that obey $\rho$, i.e., $\classStruc{P}_\rho=\{\struc{A}\in \classStruc{C}_{\sigma,d}\mid \struc{A}\sim \rho\}$.
		\item We say that a property $\classStruc{P}$ is \emph{defined by a finite union of neighbourhood profiles} if there is $k\in \mathbb{N}$ such that  $\classStruc{P}=\bigcup_{1\leq i \leq k}\classStruc{P}_{\rho_i}$ where $\rho_i$ is an $r_i$-neighbourhood profile and $r_i\in \mathbb{N}$ for every $i\in \{1,\dots,k\}$. 
	\end{enumerate}
\end{definition}

We let $n_{d,r}:=n_{d,r,\sigma_{\operatorname{graph}}}$ denote the total number of $r$-types of  
directed graphs of degree at most~$d$. We fix an odering of the types and let $\tau_{d,r}^i:=\tau_{d,r,\sigma_{\operatorname{graph}}}^i$ be the $i$-th $r$-type of bounded degree~$d$, for any $i\in \{1,\dots,n_{d,r}\}$. Further, for a graph $G$ let $\vet{v}_{d,r}(G)$ denote the $r$-histogram vector of $G$. Note if $G$ is undirected, for any type $\tau_{d,r}^i$ where the edge relation is not 
symmetric we have that $(\vet{v}_{d,r}(G))_i=0$ and therefore 
in any $r$-neighbourhood profile $\rho$ for graphs we have $\rho(i)=[0,0]$ for any type $\tau_{d,r}^i$ which is not symmetric. For convenience, for undirected graphs we will ignore the non-symmetric types. 

Let us consider the following example in which we find a representation by neighbourhood profiles for an FO-property. 
\begin{example}Consider the following FO-sentence.
	\begin{align*}\varphi:=\forall x \forall y \lnot E(x,y) \lor \forall x \exists y_1 \exists y_2 \Big(y_1\not=y_2\land E(x,y_1)\land E(x,y_2)&\\\land \forall z (z\not=y_1\land z\not=y_2)\rightarrow \lnot E(x,z)\Big).&
	\end{align*} 
	The property $P_\varphi$ defined by the sentence $\varphi$ is the property containing all edgeless graphs and all graphs that are disjoint unions of cycles.
	
	For degree bound $2$ all $1$-types are listed in Figure~\ref{fig:oneTypes}.
	\begin{figure*}
		\centering
		\begin{tikzpicture}[scale = 1.7]
		\definecolor{C1}{RGB}{1,1,1}
		\definecolor{C2}{RGB}{0,0,170}
		\definecolor{C3}{RGB}{251,86,4}
		\definecolor{C4}{RGB}{50,180,110}
		\tikzstyle{ns1}=[line width=0.7]
		\tikzstyle{ns2}=[line width=1.2]
		\node[draw,C4,circle,fill=C4,inner sep=0pt, minimum width=5pt] (0) at (-2,0) {};	
		\node[draw,C4,circle,fill=C4,inner sep=0pt, minimum width=5pt] (1) at (0,0) {};
		\node[draw,circle,fill=black,inner sep=0pt, minimum width=5pt] (2) at (0,0.6) {};
		\node[draw,C4,circle,fill=C4,inner sep=0pt, minimum width=5pt] (3) at (2,0) {};
		\node[draw,circle,fill=black,inner sep=0pt, minimum width=5pt] (4) at (1.7,0.6) {};
		\node[draw,circle,fill=black,inner sep=0pt, minimum width=5pt] (5) at (2.3,0.6) {};
		\node[draw,C4,circle,fill=C4,inner sep=0pt, minimum width=5pt] (6) at (4,0) {};
		\node[draw,circle,fill=black,inner sep=0pt, minimum width=5pt] (7) at (3.7,0.6) {};
		\node[draw,circle,fill=black,inner sep=0pt, minimum width=5pt] (8) at (4.3,0.6) {};
		\path[ns1]          (1)  edge   (2);
		\path[ns1]          (3)  edge   (4);
		\path[ns1]          (3)  edge   (5);
		\path[ns1]          (6)  edge   (7);
		\path[ns1]          (6)  edge   (8);
		\path[ns1]          (7)  edge   (8);
		\node[minimum height=10pt,inner sep=0] at (-2,1) {$\tau_1$}; 
		\node[minimum height=10pt,inner sep=0] at (0,1) {$\tau_2$}; 
		\node[minimum height=10pt,inner sep=0] at (2,1) {$\tau_3$}; 
		\node[minimum height=10pt,inner sep=0] at (4,1) {$\tau_4$};

		\end{tikzpicture}
		\caption[One types of bounded degree $2$.]{One types of bounded degree $2$, where the centres are marked in green.}\label{fig:oneTypes}
	\end{figure*}
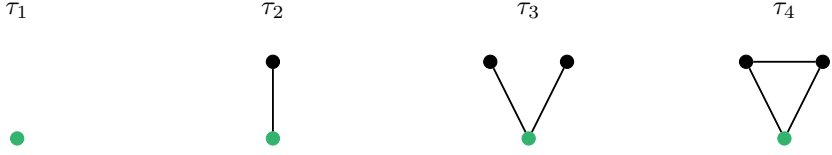 
	Let $\rho_1:\{1,\dots,4\}\rightarrow \mathfrak{I}$ be the neighbourhood profile defined by $\rho_1(1)=[0,\infty)$ and $\rho_1(i)=[0,0]$ for $i\in \{2,3,4\}$. Furthermore, let $\rho_2:\{1,\dots,4\}\rightarrow \mathfrak{I}$ be the neighbourhood profile defined by $\rho_2(i)=[0,\infty)$ for $i\in \{3,4\}$ and $\rho_2(j)=[0,0]$ for $j\in \{1,2\}$. It is easy to observe that the properties $P_\varphi$ and $P_{\rho_1}\cup P_{\rho_2}$ are equal.   
	
\end{example}
Indeed representing FO-properties by neighbourhood profiles works in general.
We now give a lemma showing that bounded-degree FO properties can be equivalently defined as finite unions of properties defined by neighbourhood profiles. Here the technicalities  that arise are due to Hanf normal form not requiring the locality-radius  of all Hanf-sentences to be the same. 

\begin{lemma}\label{lemma:FO-neighbourhood}
	For every non-empty property $\classStruc{P}\subseteq \classStruc{C}_{\sigma,d}$, $\classStruc{P}$ is FO definable on $\classStruc{C}_{\sigma,d}$ if and only if
	$\classStruc{P}$ can be obtained as a finite union of properties defined by neighbourhood profiles. 
\end{lemma}
\begin{proof}
	For the first direction assume $\varphi$ is an FO-sentence. Then by Hanf's Theorem (Theorem~\ref{thm:Hanf}) there is a sentence $\psi$ in Hanf normal form such that $\classStruc{P}_\varphi=\classStruc{P}_\psi$. 
	
	We will first convert $\psi$ into a sentence in Hanf normal form where every Hanf sentence appearing has the same locality radius. Let $r\in \mathbb{N}$ be the maximum locality radius appearing in $\psi$, and let $\varphi^{\geq m}_\tau:=\exists ^{\geq m} x \phi_{\tau}(x)$ be a Hanf sentence, where $\tau$ is an $r'$-type for some $r'\leq r$. Let $\tau_1,\dots,\tau_k$ be a list of all  $r$-types of bounded degree $d$ for which $(\mathcal{N}_{r'}^{\struc{B}}(b),b)\in \tau$ for  $(\struc{B},b)\in \tau_i$, for every $1\leq i\leq k$. Let $\Pi$ be the set of all partitions of $m$ into $k$ parts. Let 
	\begin{displaymath}
	\tilde{\varphi}^{\geq m}_\tau :=\bigvee_{(m_1,\dots,m_k)\in \Pi}\phantom{ii}\bigwedge_{i=1}^k \exists ^{\geq m_i} x \phi_{\tau_i}(x).
	\end{displaymath} 
	\begin{claim}\label{claim:increasingRadius}
		$\varphi^{\geq m}_\tau$ is $d$-equivalent to $\tilde{\varphi}^{\geq m}_\tau$.
	\end{claim}
	\begin{proof}
		Assume that $\struc{A}\in \classStruc{C}_{d}$ satisfies $\varphi^{\geq m}_\tau$, and assume that $a_1,\dots,a_m$ are $m$ distinct elements with $(\mathcal{N}_{r'}^{\struc{A}}(a_j),a_j)\in \tau$, for every $1\leq j\leq m$. Let $\tilde{\tau}_j$ be the $r$-type for which $(\mathcal{N}_r^{\struc{A}}(a_j),a_j)\in \tilde{\tau}_j$. By choice of $\tau_1,\dots,\tau_k$, we get that there are indices $i_1,\dots,i_m$ such that $\tilde{\tau}_j=\tau_{i_j}$. For $i\in \{1,\dots,k\}$ let $m_i=|\{j\in \{1,\dots,m\}\mid i_j=i \}|$. Hence $\struc{A}\models \bigwedge_{i=1}^k \exists ^{\geq m_i} x \phi_{\tau_i}(x)$ and since additionally $(m_1,\dots,m_k)\in\Pi$ this implies $\struc{A}\models \tilde{\varphi}^{\geq m}_\tau$.
		
		On the other hand, let $\struc{A}\in \classStruc{C}_{d}$ satisfy $\tilde{\varphi}^{\geq m}_\tau$, and let $(m_1,\dots,m_k)\in \Pi$ be a partition of $m$ such that $\struc{A}\models \bigwedge_{i=1}^k \exists ^{\geq m_i} x \phi_{\tau_i}(x)$. For every $1\leq i\leq k$, let  $a_{1}^i,\dots, a_{m_i}^i$  be $m_i$ distinct elements such that $(\mathcal{N}_r^{\struc{A}}(a_{j}^i),a_{j}^i)\in \tau_i$, for every $1\leq j\leq m_i$. By choice of $\tau_1,\dots,\tau_k$, we get that $(\mathcal{N}_{r'}^{\struc{A}}(a_{j}^i),a_{j}^i)\in \tau$, for every pair $1\leq i\leq k$, $1\leq j\leq m_i$. But since $m_1+\dots+m_k=m$ this implies that $\struc{A}\models\varphi^{\geq m}_\tau$. This proves that $\varphi^{\geq m}_\tau$ and $\tilde{\varphi}^{\geq m}_\tau$ are $d$-equivalent. 
	\end{proof}
	Let $\psi'$ be the formula in which every Hanf-sentence $\varphi^{\geq m}_\tau$ for which $\tau$ is an $r'$-type for some $r'<r$ gets replaced by $\tilde{\varphi}^{\geq m}_\tau$. By a simple inductive argument using Claim~\ref{claim:increasingRadius}, we get that $\psi$ is $d$-equivalent to $\psi'$, and hence $\classStruc{P}_\varphi=\classStruc{P}_{\psi}=\classStruc{P}_{\psi'}$. Furthermore since $\tilde{\varphi}^{\geq m}_\tau$ is a Boolean combination of Hanf-sentences for every $\varphi^{\geq m}_\tau$, and any Boolean combination of Boolean combinations is a Boolean combination itself, $\psi'$ is in Hanf normal form. 
	Furthermore, every Hanf-sentence appearing in $\psi'$ has locality radius $r$ by construction.
	
	Since any Boolean combination can be converted into disjunctive normal form,
	we can assume that $\psi'$ is a disjunction of sentences $\xi$ of the form
	\begin{displaymath}
	\xi=\bigwedge_{j=1}^k \exists ^{\geq m_j} x \phi_{\tau_j}(x)\land \bigwedge_{j=k+1}^\ell \lnot \exists ^{\geq m_j+1} x \phi_{\tau_j}(x),
	\end{displaymath}
	where $\ell\in \mathbb{N}_{\geq 1}$, $1\leq k \leq \ell$, $m_i\in \mathbb{N}_{\geq 1}$ and $\tau_i$ is an $r$-type for every $1\leq i\leq \ell$. We can further assume that every sentence in the disjunction $\psi'$ is satisfiable by some $\struc{A}\in \classStruc{C}_{d}$, as any sentence with no bounded degree $d$ model can be removed from $\psi'$.
	
	Let $\tilde{\tau}_1,\dots,\tilde{\tau}_t$ be a list of all $r$-types of bounded degree $d$ in the order we fixed.
	Let $k_i:=\max(\{m_j \mid 1\leq j\leq k, \tau_j=\tilde{\tau}_i \}\cup\{0\})$ and $\ell_i:=\min(\{m_j \mid k+1\leq j\leq \ell, \tau_j=\tilde{\tau}_i \}\cup\{\infty\})$ for every $i\in \{1,\dots,t\}$. Since $\xi$ has at least one bounded-degree model,  $k_i\leq \ell_i$ for every $i\in \{1,\dots,t\}$.
	Let $\rho: \{1,\dots,t\}\rightarrow \mathfrak{I}$ be the neighbourhood profile defined by
	$\rho(i):=[k_i,\ell_i]$ if $\ell_i<\infty$ and $\rho(i):=[k_i,\ell_i)$ otherwise. Then by construction, we get that $\classStruc{P}_\rho=\classStruc{P}_\xi$. Since $\psi'$ is a disjunction of formulas, each of which defines a property which can be defined by some neighbourhood profile, we get that $\classStruc{P}_{\psi'}$ must be a finite  union of properties defined by some neighbourhood profile. \\
	
	On the other hand, for every $r$-neighbourhood profile $\rho$ of degree $d$, $\tau_1,\dots,\tau_t$ a list of all $r$-types of bounded degree $d$ in the order fixed  and the formula
	\begin{displaymath}
	\varphi_\rho:=\bigwedge_{i\in \{1,\dots,t\},\atop{\rho(i)=[k_i,\ell_i]}}\Big(\exists ^{\geq k_i} x \phi_{\tau_i}(x)\land \lnot \exists ^{\geq \ell_i+1} x \phi_{\tau_i}(x)\Big)\land \bigwedge_{i\in \{1,\dots,t\},\atop{\rho(i)=[k_i,\infty)}}\exists ^{\geq k_i} x \phi_{\tau_i}(x)
	\end{displaymath}
	it clearly holds that $\classStruc{P}_\rho=\classStruc{P}_{\varphi_\rho}$. Hence every finite union of properties defined by neighbourhood profiles can be defined by the disjunction of the formulas $\varphi_\rho$ of all $\rho$ in the finite union.
\end{proof}

\subsubsection{Relating FO properties to GSF-local properties}  
We now prove that FO properties which arise as unions of neighbourhood profiles of a particularly simple form are GSF-local. 
\begin{displaymath}
	\mathfrak{I}_{0}:=\{[0,k]\mid k\in \mathbb{N} \}\cup\{[0,\infty)\}\subset \mathfrak{I}. 
\end{displaymath}
We call any neighbourhood profile $\rho$ with codomain $\mathfrak{I}_{0}$ 
a \emph{$0$-profile}, as all lower bounds for the occurrence of types are $0$.
\begin{observation}\label{obs:expressingExOfType}
	Let $\rho$ be a\, $0$-profile. If two structures $\struc{A},\struc{A}'\in \classStruc{C}_{\sigma,d}$ satisfy $(\vet{v}_{d,r,\sigma}(\struc{A}))_i\leq (\vet{v}_{d,r,\sigma}(\struc{A}'))_i$ for every $i\in\{1,\dots,n_{d,r,\sigma}\}$ and $\struc{A}'\sim\rho$, then $\struc{A}\sim \rho$. 
	
	In particular, the existence of an $r$-type cannot be expressed by a\, $0$-profile. 
\end{observation}

\begin{theorem}\label{thm:subsetOfFOIsLocal}
	Every finite union of  properties of undirected graphs defined by $0$-profiles is GSF-local.
\end{theorem} 
\begin{proof}
	We prove this in two parts (Claim~\ref{claim:GSFProfileIsGSFLocal} and Claim~\ref{claim:GSFLocalClosedUnderUnion}). We first argue that every property $\classStruc{P}_{\rho}$ 
	defined by some $0$-profile $\rho:\{1,\dots,n_{d,r}\}\rightarrow\mathfrak{I}_{0}$ 
	is GSF-local. For this it is important to note that we can express a forbidden $r$-type $\tau$ by a forbidden generalised subgraph. 
	For $(B,b)\in \tau$, the set of all graphs with no vertex of neighbourhood type $\tau$ is the set of all $B$-free graphs where every vertex in $V(B)$ of distance less than $r$ to $b$ is marked `full' and every vertex in $V(B)$ of distance $r$ to $b$ is marked `semifull'.
	Since a profile of the form $\rho:\{1,\dots,n_{d,r,\sigma}\}\rightarrow\mathfrak{I}_{0}$ can express that some neighbourhood type $\tau$ can appear at most $k$ times for some fixed $k\in \mathbb{N}$, we need to forbid all marked graphs in which type $\tau$ appears $k+1$ times. We will formalise this in the following claim. 
	\begin{claim}\label{claim:GSFProfileIsGSFLocal}
		For every $r$-neighbourhood profile $\rho: \{1,\dots,n_{d,r}\}\rightarrow \mathfrak{I}_{0}$, there is a finite set $\mathcal{F}$ of marked graphs such that $\classStruc{P}_\rho$ is exactly the property of $\mathcal{F}$-free graphs. 
	\end{claim}
	\begin{proof}
		Assume $\tau$ is an $r$-type and $k\in \mathbb{N}_{>0}$. Then we say that a marked graph $F$ is a \emph{$k$-realisation} of $\tau$ if $F$ has the following properties.
		\begin{enumerate}
			\item There are $k$ distinct vertices $v_1,\dots,v_k$ in $F$ such that $(\mathcal{N}_r^F(v_i),v_i)\in \tau$ for every $i=1,\dots,k$. 
			\item Every vertex $v$ in $F$ has distance less  or equal to $r$ to at least one vertex $v_i$.
			\item Every vertex $v$ in $F$ of distance less than $r$ to at least one $v_i$ is marked as `full'.
			\item Every vertex $v$ in $F$ of distance greater or equal to $r$ to every $v_i$ is marked as `semifull'.
		\end{enumerate}  
		We denote by $S^k(\tau)$ the set of all $k$-realisations of $\tau$.
		
		Now we can define the set $\mathcal{F}$  of forbidden subgraphs to be
		\[
		{\mathcal F}:=\bigcup_{k\in \mathbb{N}, 1\leq i\leq n_{d,r,\sigma}: \rho(i)=[0,k]} 
		S^{k+1}(\tau_{d,r}^i).
		\]
		
		Let $\mathcal{P}$ be the property of all $\mathcal{F}$-free graphs. We first prove that the property $\mathcal{P}$ is contained in  $\classStruc{P}_\rho$. Towards a contradiction  assume that $G\in \mathcal{C}_d$ is ${\mathcal F}$-free  but not contained in $\classStruc{P}_{\rho}$. As $G$ is not contained in $\classStruc{P}_{\rho}$ there must be an index $i\in \{1,\dots,n_{d,r}\}$ such that $(\vet{v}_{d,r}(G))_i\notin \rho(i)$. Since $\rho(i) \in \mathfrak{I}_{0}$ there is $k\in \mathbb{N}$ such  that $\rho(i)=[0,k]$ and hence $(\vet{v}_{d,r}(G))_i> k$. Hence there must be $k+1$ vertices $v_1,\dots,v_{k+1}$ in $G$ such that $(\mathcal{N}_r^G(v_i),v_i)\in \tau_{d,r}^i$. 
		We define the  marked graph  $F$ to be the subgraph of $G$ induced by the $r$-neighbourhoods 
		of $v_1,\dots,v_{k+1}$, \ie $G[\cup_{1\leq i\leq k+1}N_r^G(v_i)]$, in which every vertex of distance less than $k$ to at least one of the $v_i$ is marked as `full' and every other vertex is marked as `semifull'.  
		Then $F$ is by definition a $(k+1)$-realisation of $\tau_{d,r}^i$ and hence $F\in {\mathcal F}$.
		We now argue that $F$ can be embedded into $G$. Since $F$ is an induced subgraph of $G$ the identity map gives us a natural embedding $f:F\rightarrow G$. Let $v$ be any vertex marked `full' in $F$. By construction of $F$, there is $i\in \{1,\dots,k+1\}$ such that $f(v)$ is of distance less than $r$ to $v_i$ in $G$. But then $N_1^G(f(v))$ is a subset of 
		$N_r^G(v_i)$. As $F$ without the marking is the subgraph of $G$ 
		induced by $\cup_{1\leq i\leq k+1}N_r^G(v_i)$ this implies that $f(N_1^F(v))=N_1^G(f(v))$.  Furthermore, assume $v$ is a vertex marked `semifull' in $F$. Then $f(N_1^F(v))= N_1^G(f(v))\cap f(V(F))$ holds as $F$ without the markings is an induced subgraph of $G$. 
		This proves that $G$ is not $F$-free by Definition~\ref{def:gsf}. This is a contradiction to our assumption that $G$ is $\mathcal{F}$-free and $F\in \mathcal{F}$.

		Similarly, we can show that $\classStruc{P}_\rho\subseteq \mathcal{P}$ by assuming $G\in \mathcal{C}_d$ is in $\classStruc{P}_\rho$  but not ${\mathcal F}$-free, and showing that the embedding of any graph of $\mathcal{F}$ into $G$ yields an amount of vertices of a certain type contradicting containment in $\classStruc{P}_\rho$. 
	\end{proof}
	Next we prove that 
	classes defined by excluding finitely many marked graphs are closed under finite unions.	
	\begin{claim}\label{claim:GSFLocalClosedUnderUnion}
		Let $\mathcal{F}_1, \mathcal{F}_2$ be two finite sets of marked graphs. 
		For $i\in\{1,2\}$, let $\mathcal{P}_i$ be the property of $\mathcal{F}_i$-free graphs. 
		Then there is a set $\mathcal{F}$ of generalised subgraphs such that $\mathcal{P}_1\cup\mathcal{P}_2$ is the property of $\mathcal{F}$-free graphs. 
	\end{claim}
	\begin{proof}
		We say that a marked graph $F$ is a (not necessarily disjoint) union of  marked graphs $F_1, F_2$ if 
		\begin{enumerate}
			\item  there is an embedding $f_i$ of $F_i$ into the graph $F$ without its markings as in Definition~\ref{def:gsf} for every $i\in \{1,2\}$. 
			\item for every vertex $v$ in $F$ there is $i\in\{1,2\}$ and a vertex $w$ in $F_i$ such that $f_i(w)=v$.
			\item every vertex $v$ in $F$ is marked `full', if there is  $i\in\{1,2\}$ and a `full' vertex $w$ in  $F_i$ such that $f_i(w)=v$.
			\item every vertex $v$ in $F$ is marked `semifull', if there is  $i\in \{1,2\}$ and a `semifull' vertex $w$ in  $F_i$ such that $f_i(w)=v$ and $f_i(u)\not=v$ for every $i\in\{1,2\}$ and every `full' vertex $u$.
			\item every vertex $v$ in $F$ is marked `partial' if $f_i(u)\not=v$ for every $i\in\{1,2\}$ and every `full' or `semifull' vertex $u$.
		\end{enumerate}  
		We define $S(F_1,F_2)$ to be the set of all possible (not necessarily disjoint) unions of $F_1,F_2$.  
		We can now define the set $\mathcal{F}$ to be 
		\[\mathcal{F}:=\bigcup_{F_1\in \mathcal{F}_1,F_2\in \mathcal{F}_2} S(F_1,F_2).\] 
		
		Let $\mathcal{P}$ be the property of all $\mathcal{F}$-free graphs.
		Now we prove $\mathcal{P}\subseteq \mathcal{P}_1\cup\mathcal{P}_2$. Towards a contradiction assume $G$ is $\mathcal{F}$-free but $G$ is in neither $\mathcal{P}_1$ nor in $\mathcal{P}_2$.  
		Then for every $i\in \{1,2\}$ there is a graph $F_i\in \mathcal{F}_i$ such that $G$ is not $F_i$-free. It is easy to see that  there is a  union $F_\cup$ of $F_1$ and $F_2$  
		such that $G$ is not $F_\cup$-free, which contradicts that $G$ is $\mathcal{F}$-free. 
		
		Conversely, in order to prove $\mathcal{P}_1\cup \mathcal{P}_2 \subseteq \mathcal{P}$, if $G$ is $\mathcal{F}_i$ free for some 
		$i\in\{1,2\}$ then $G$ must be $\mathcal{F}$-free by construction of $\mathcal{F}$.
	\end{proof}
	Combining the two claims above proves the  Theorem~\ref{thm:subsetOfFOIsLocal}.
\end{proof}
\paragraph{Further discussion of the relation between FO and GSF-locality} First let us remark that it 
is neither true that every FO definable property
is GSF-local, nor that every GSF-local property is FO definable. 
\begin{example}\label{exa:FONotContainedInGSF}
	The property of bounded-degree  graphs containing a triangle is FO definable but not GSF-local.
\end{example}
Indeed, the existence  of 
a fixed number of vertices of certain neighbourhood types can be expressed in FO, while in general, this cannot be expressed by forbidding generalised subgraphs. 
If a formula has a $0$-profile 
(and hence does not require the existence of any types) 
then the property defined by that formula is GSF-local, as shown in Theorem~\ref{thm:subsetOfFOIsLocal}.
\begin{example}\label{exa:GSFNotContainedInFO}
	The class of all bounded-degree graphs with an even number of vertices is GSF-local but not FO definable.
\end{example}

Let us remark that Theorem~\ref{thm:subsetOfFOIsLocal} combined with Lemma~\ref{lemma:FO-neighbourhood} proves that every finite union of properties definable by $0$-profiles is both FO definable and GSF-local. 
Hence it is natural to ask whether the intersection of FO definable properties and GSF-local properties is precisely the set of finite unions of properties definable by $0$-profiles. However, this is not the case.
The following example shows that there are properties which are both FO definable and GSF-local but cannot be expressed by $0$-profiles.
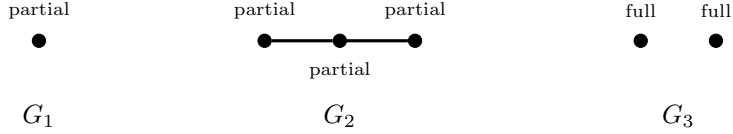
\begin{figure*}
	\centering
	\begin{tikzpicture}
	\tikzstyle{ns1}=[line width=1.2]
	\tikzstyle{ns2}=[line width=1.2]
	\definecolor{C1}{RGB}{1,1,1}
	\definecolor{C2}{RGB}{0,0,170}
	\definecolor{C3}{RGB}{251,86,4}
	\definecolor{C4}{RGB}{50,180,110}
	\def \dist {0.7}
	\def \heightWriting {0.4}
	\def \heightName {-1}
	\node[circle,fill=white,inner sep=0pt, minimum width=0pt] (15) at (0,0) {};
	\node[draw,circle,fill=black,inner sep=0pt, minimum width=5pt] (0) at (0,0) {};
	\node[draw,circle,fill=black,inner sep=0pt, minimum width=5pt] (5) at (3,0) {};
	\node[draw,circle,fill=black,inner sep=0pt, minimum width=5pt] (6) at (4,0) {};
	\node[draw,circle,fill=black,inner sep=0pt, minimum width=5pt] (7) at (5,0) {};
	\node[draw,circle,fill=black,inner sep=0pt, minimum width=5pt] (9) at (8,0) {};
	\node[draw,circle,fill=black,inner sep=0pt, minimum width=5pt] (10) at (9,0) {};
	
	\draw[ns1] (5)--(6)--(7);
	
	\node[minimum height=10pt,inner sep=0,font=\scriptsize] at (0,\heightWriting) {partial};
	\node[minimum height=10pt,inner sep=0,font=\scriptsize] at (3,\heightWriting) {partial};
	\node[minimum height=10pt,inner sep=0,font=\scriptsize] at (4,-\heightWriting) {partial};
	\node[minimum height=10pt,inner sep=0,font=\scriptsize] at (5,\heightWriting) {partial};
	\node[minimum height=10pt,inner sep=0,font=\scriptsize] at (8,\heightWriting) {full};
	\node[minimum height=10pt,inner sep=0,font=\scriptsize] at (9,\heightWriting) {full};
	\node[minimum height=10pt,inner sep=0] at (0,\heightName) {$G_1$};
	\node[minimum height=10pt,inner sep=0] at (4,\heightName) {$G_2$};
	\node[minimum height=10pt,inner sep=0] at (8.5,\heightName) {$G_3$};

	\end{tikzpicture} 
	\caption{Marked graphs for Example~\ref{ex:0ProfilesNotEntireIntersecion}.}\label{fig:setOfMarkedGraphs}
	
\end{figure*}
\begin{example}\label{ex:0ProfilesNotEntireIntersecion} We let $d\geq 2$ and let  $B_1:=(\{v\},\{\})$, $B_2=(\{v,w\},\{\{v,w\}\})$ be two graphs. We further let $\tau_1,\tau_2$ be the $1$-types of degree $d$ such that $(B_1,v)\in \tau_1$ and $(B_2,v)\in \tau_2$. Consider the property $\mathcal{P}$ defined by the following FO formula
	\begin{displaymath}
	\varphi:=\lnot \exists x (x=x)\lor \exists^{=1}x\big(\varphi_{\tau_1}(x)\land \forall y (x\not=y\rightarrow \varphi_{\tau_2}(y))\big). 	
	\end{displaymath}
	$\mathcal{P}$ contains, besides the empty graph, unions of an arbitrary amount  of disjoint edges and one isolated vertex. To define a sequence of forbidden subgraphs we let $G_1,G_2,G_3$ be the marked graphs in Figure~\ref{fig:setOfMarkedGraphs}. Let $\mathcal{F}_{\operatorname{even}}:=\{G_1\}$ and $\mathcal{F}_{\operatorname{odd}}:=\{G_2,G_3\}$ and let $\overline{\mathcal{F}}=(\mathcal{F}_n)_{n \in \mathbb{N}}$ where $\mathcal{F}_i=\mathcal{F}_{\operatorname{even}}$ if $i$ is even and $\mathcal{F}_i=\mathcal{F}_{\operatorname{odd}}$ if $i$ is odd. Note that every graph on more than one 
	vertex with an odd number of vertices which is $\mathcal{F}_{\operatorname{odd}}$-free must contain a vertex of neighbourhood type $\tau_1$, 
	and that the set of $\mathcal{F}_{\operatorname{even}}$-free graphs contains only the empty graph. Hence $\mathcal{P}$ is $\overline{\mathcal{F}}$-local. Now assume towards a contradiction that $\mathcal{P}=\bigcup_{1\leq i\leq k}\mathcal{P}_{\rho_i}$ for $0$-profiles $\rho_i$. Let $G_m$ be the graph consisting of $m$ disjoint edges and one isolated vertex and $H_m$ the graph consisting of $m$ disjoint edges. Since $G_m\in\mathcal{P}$ there is $i\in \{1,\dots,k\}$ such that $G_m\sim \rho_i$. By choice of $G_m$ and $H_m$ we have $0\leq (\vet{v}_{d,r}(H_m))_j\leq (\vet{v}_{d,r}(G_m))_j\in \rho_i(j)$ for every $j\in \{1,\dots,n_{d,r}\}$. Since additionally $\rho_i(j)\in \mathfrak{I}_0$ this implies that $(\vet{v}_{d,r}(H_m))_j\in \rho_i(j)$. But then $H_m\sim \rho_i$ which yields a contradiction as $H_m\notin \mathcal{P}$. Hence $\mathcal{P}$ can not be defined as a finite union of $0$-profiles. 
\end{example}
Figure~\ref{fig:overview} gives a schematic overview of all classes of properties discussed here and their relationship.
\begin{figure}
	\centering
	\begin{tikzpicture}[scale=1,use Hobby shortcut,closed=true]
	\tikzstyle{LW1}=[line width=1.3]
	\tikzstyle{LW3}=[line width=1]
	\definecolor{C4}{RGB}{251,86,4}
	\definecolor{C3}{RGB}{170,0,0}
	\definecolor{C2}{RGB}{0, 153, 51}
	\definecolor{C1}{RGB}{0,0,170}
	
	\begin{scope}
	\clip[postaction={fill=white,fill opacity=0.2}] ([closed]0,0)..(3,2)..(6,0)..(3,-2);
	\foreach \x in {-7.1,-7,...,15.2}%
	\draw[C1!30](\x, -5)--+(-12,14.4);
	\end{scope}
	\fill[white]([closed]3,0)..(6,2)..(8,0)..(6,-2);
	\begin{scope}
	\clip[postaction={fill=white,fill opacity=0.2}] ([closed]3,0)..(6,2)..(8,0)..(6,-2);
	\foreach \x in {-7.1,-7,...,15.2}%
	\draw[C2!30](\x, 5)--+(-12,-14.4);
	\end{scope}
	\fill[white]([closed]2,1)..(2.5,1.92)..(3,1.975)..(5.62,1)..(5.9,0.4)..(3,0);
	\begin{scope}
	\clip[postaction={fill=white,fill opacity=0.2}] ([closed]2,1)..(2.5,1.92)..(3,1.975)..(5.62,1)..(5.9,0.4)..(3,0);
	\foreach \x in {-7.1,-7,...,15.2}%
	\draw[C3!30](\x, 5)--+(-12,-14.4);
	\end{scope}
	\fill[white]([closed]4,0)..(5.05,1)..(5.7,0.8)..(5.97,-0.1)..(4,-1.2);
	\begin{scope}
	\clip[postaction={fill=white,fill opacity=0.2}] ([closed]4,0)..(5.05,1)..(5.7,0.8)..(5.97,-0.1)..(4,-1.2);
	\foreach \x in {-7.1,-7,...,15.2}%
	\draw[C4!30](\x, -5)--+(-12,14.4);
	\end{scope}
	\draw[LW1,C1] ([closed]0,0)..(3,2)..(6,0)..(3,-2);
	\draw[LW1,C2] ([closed]3,0)..(6,2)..(8,0)..(6,-2);
	\draw[LW3,C3] ([closed]2,1)..(2.5,1.92)..(3,1.975)..(5.62,1)..(5.9,0.4)..(3,0);
	\draw[LW3,C4] ([closed]4,0)..(5.05,1)..(5.7,0.8)..(5.97,-0.1)..(4,-1.2);

	\node[C1,minimum height=10pt,inner sep=0,font=\scriptsize] at (1,0) {GSF-local};
	\node[C2,minimum height=10pt,inner sep=0,font=\scriptsize] at (7,1) {FO};
	\node[C3,minimum height=10pt,inner sep=0,font=\scriptsize] at (2.7,1.4) {POT};
	\node[C4,minimum height=10pt,inner sep=0,font=\scriptsize] at (5,0) {$0$-profiles};
	\node[draw,circle,fill=black,inner sep=0pt, minimum width=3pt] (1) at (5,-1.2) {};
	\node[minimum height=10pt,inner sep=0,font=\scriptsize] at (5,-1) {$\graphProp$};
	\node[draw,circle,fill=black,inner sep=0pt, minimum width=3pt] (1) at (4.2,0.95) {};
	\node[minimum height=10pt,inner sep=0,font=\scriptsize] at (4.2,1.2) {$\mathcal{P}_{\ref{ex:0ProfilesNotEntireIntersecion}}$};
	\node[draw,circle,fill=black,inner sep=0pt, minimum width=3pt] (1) at (2.7,0.9) {};
	\node[minimum height=10pt,inner sep=0,font=\scriptsize] at (2.7,0.65) {$\mathcal{P}_{\ref{exa:GSFNotContainedInFO}}$};
	\node[draw,circle,fill=black,inner sep=0pt, minimum width=3pt] (1) at (7,-0.55) {};
	\node[minimum height=10pt,inner sep=0,font=\scriptsize] at (7,-0.8) {$\mathcal{P}_{\ref{exa:FONotContainedInGSF}}$};
	\node[draw,circle,fill=black,inner sep=0pt, minimum width=3pt] (1) at (5.3,0.6) {};
	\node[minimum height=10pt,inner sep=0,font=\scriptsize] at (5.1,0.45) {$\mathcal{C}_d$};
	
	\end{tikzpicture}
	\caption{Overview of the classes of properties, here $\mathcal{P}_i$ refers to the property from Example $i$, $\mathcal{C}_d$ refers to the property of all graphs of bounded degree $d$ and $\graphProp$ is the property defined in Section~\ref{sec:localreduction}.}\label{fig:overview}
\end{figure}
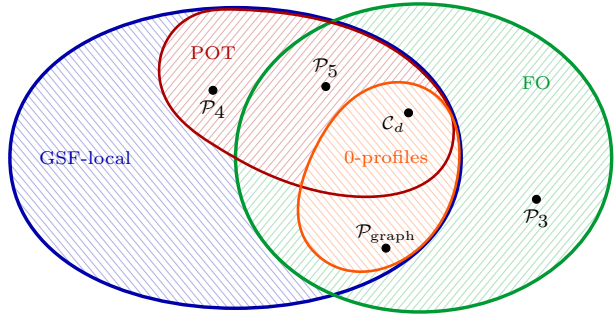

\subsection{Proving the existence of a GSF-local non-testable property}\label{sec:proofMainThm}

In this section we prove Theorem~\ref{thm:existenceLocalNonTestableProperty}.  We show that the property $\classStruc{P}_{\zigzag}$ from Section~\ref{sec:FOnontestability}
can be expressed by a union of $0$-profiles. We then show that the local reduction from $\classStruc{P}_{\zigzag}$ to $\graphProp$ given in Section~\ref{sec:localreduction} preserves the expressibility by $0$-profiles, and hence by Theorem~\ref{thm:subsetOfFOIsLocal} $\graphProp$ is GSF-local. 

Let $\sigma$ be the signature, $d\in \mathbb{N}$  and $\classStruc{P}_{\zigzag}$ be the property of $d$ $\sigma$-structures of bounded-degree from Section~\ref{sec: definitionFormula}. 
\subsubsection{Characterisation  of the relational structure property by neighbourhood profiles}\label{sec:charBy0Profiles}
Our aim in this section is to prove that the property $\classStruc{P}_{\zigzag}$ of relational structures can be written as a finite union of properties defined by $0$-profiles
of radius $2$.

For all $\sigma$-structures in $\classStruc{P}_{\zigzag}$ (excluding $\struc{A}_{\emptyset}$) it is crucial that they are allowed to contain precisely one root element. Hence the neighbourhood profile describing $\classStruc{P}_{\zigzag}$ must restrict the number of occurrences of the $2$-type of the root element. But since in $\classStruc{P}_{\zigzag}\setminus \{\struc{A}_{\emptyset}\}$, the root elements in different structures may have different $2$-types, 
we partition $\classStruc{P}_{\zigzag}\setminus \{\struc{A}_{\emptyset}\}$ into parts 
$\classStruc{P}_1,\dots,\classStruc{P}_m$  by the $2$-type of the root element. 
Note that the number $m$ of parts is constant as there are at most $n_{d,2,\sigma}$ $2$-types in total. For each of these parts we then define a neighbourhood profile  $\rho_k$ such that $\classStruc{P}_k\cup\{\struc{A}_{\emptyset}\}=\classStruc{P}_{\rho_k}$. 
We would like to remark here that the roots of all but one structure in  $\classStruc{P}_{\zigzag}$ actually have the same $2$-types. Hence the partition only contains two parts and one of the two parts only contains one structure. 
We now define the parts and corresponding profiles formally.

Assume without loss of generality that the $2$-types $\tau_{d,2,\sigma}^1,\dots, \tau_{d,2,\sigma}^{n_{d,2,\sigma}}$ of degree $d$ are ordered in such a way that for  $(\struc{B},b)\in \tau_{d,2,\sigma}^{k}$, it holds that $\struc{B}\models \varphi_{\operatorname{root}}(b)$  if and only if $k\in \{1,\dots,m\}$ for some $m\leq n_{d,2,\sigma}$.
For $k\in \{1,\dots,m\}$, let 
\[
\classStruc{P}_k:=\{\struc{A}\in \classStruc{P}_{\zigzag}\setminus\{\struc{A}_{\emptyset}\}\mid \text{ there is }a\in \univ{A}\text{ such that }(\mathcal{N}_2^{\struc{A}}(a),a)\in \tau_{d,2,\sigma}^k\}.
\] 
Since by Lemma~\ref{lem:connected} every $\struc{A}\in \classStruc{P}_{\zigzag}\setminus \{\struc{A}_{\emptyset}\}$ must contain exactly one root  we get that 
\[
\classStruc{P}_{\zigzag}=\bigcup_{1\leq k\leq m}\classStruc{P}_k\cup \{\struc{A}_{\emptyset}\}
\]
and this union is disjoint.   
Furthermore, for $k\in \{1,\dots,m\}$, let $I_k\subseteq \{1,\dots,n_{d,2,\sigma}\}$ be the set of indices $j$ such that there is a structure $\struc{A}\in \classStruc{P}_k$ and $a\in \univ{A}$ with $(\mathcal{N}_2^{\struc{A}}(a),a)\in\tau_{d,2,\sigma}^j$. 
For every $k\in \{1,\dots,m\}$ we define the $2$-neighbourhood profile   $\rho_k:\{1,\dots,n_{d,2,\sigma}\}\rightarrow \mathfrak{I}_{0}$ by
\begin{align*}
\rho_k(i):=
\begin{cases}
[0,1] &  \text{ if }i=k, \\
[0,\infty) & \text{ if }i\in I_k\setminus\{k\},\\
[0,0] & \text{ otherwise}.
\end{cases}
\end{align*}
To prove that these $0$-profiles of radius $2$ define the property $\classStruc{P}_{\zigzag}$, the crucial observation is that for every element $a$ of some structure in $\classStruc{C}_{\sigma,d}$, the FO-formula $\varphi_{\zigzag}$ only talks about elements of distance at most $2$ to $a$ (\ie $\varphi_{\zigzag}$ is $2$-local). Hence  
the $2$-histogram vector of a structure already captures whether the structure satisfies $\varphi_{\zigzag}$. We will now formally prove this.
\begin{lemma}\label{lem:neighbouhoodProfilOfPZigZag}
	It holds that $\classStruc{P}_{\zigzag}=\bigcup_{1\leq k \leq m}\classStruc{P}_{\rho_k}$. 
\end{lemma}
\begin{proof}
	We first prove that $\classStruc{P}_{\zigzag}\subseteq \bigcup_{1\leq k \leq m}\classStruc{P}_{\rho_k}$. First note that trivially $\struc{A}_{\emptyset}\in \bigcup_{1\leq k \leq m}\classStruc{P}_{\rho_k}$. Now assume $\struc{A}\in \classStruc{P}_{\zigzag}\setminus \{\struc{A}_{\emptyset}\}$. This implies that there is $k\in \{1,\dots,m\}$ such that $\struc{A}\in \classStruc{P}_k$.  By construction we have that for every $a\in \struc{A}$, there is  $i\in I_k$  such that $(\mathcal{N}_2^{\struc{A}}(a),a)\in \tau_{d,2,\sigma}^i$. Furthermore, since $\struc{A}\models \varphi_{\zigzag}$ and $\univ{A}\not= \emptyset$, we have by Lemma~\ref{lem:connected} that $\struc{A}\models \exists^{=1}x\varphi_{\operatorname{root}}(x)$, and that there can be at most one $a\in \univ{A}$   such that $(\mathcal{N}_2^{\struc{A}}(a),a)\in \tau_{d,2,\sigma}^k$.   Therefore $\struc{A}\in \classStruc{P}_{\rho_k}$. \\ 
	
	To prove $\bigcup_{1\leq k \leq m}\classStruc{P}_{\rho_k}\subseteq \classStruc{P}_{\zigzag}$, we  prove that  
	every structure in $\bigcup_{1\leq k \leq m}\classStruc{P}_{\rho_k}$ 
	must satisfy $\varphi_{\zigzag}$. We will prove that every $\struc{A}\in \bigcup_{1\leq k \leq m}\classStruc{P}_{\rho_k}$ satisfies $\varphi_{\operatorname{recursion}}$, and refer for the proof that $\struc{A}$ satisfies  $\varphi_{\operatorname{tree}}\land \varphi_{\operatorname{rotationMap}}\land \varphi_{\operatorname{base}}$ to  Claim~\ref{claim:satisfyingTree}, Claim~\ref{claim:satisfyingRotationMap} and Claim~\ref{claim:satisfyingBase} in Appendix~\ref{app:B}. Note that $\struc{A}_{\emptyset}\models \varphi_{\zigzag}$ by Lemma~\ref{lem:exactFormOfModels} and hence we exclude $\struc{A}_{\emptyset}$ in the following.
	
	\begin{claim}\label{claim:satisfyingRecursion}
		Every structure $\struc{A}\in \bigcup_{1\leq k \leq m}\classStruc{P}_{\rho_k}\setminus \{\struc{A}_{\emptyset}\}$ satisfies $\varphi_{\operatorname{recursion}}$.
	\end{claim}
	\begin{proof}
		Let $\struc{A}\in \bigcup_{1\leq k \leq m}\classStruc{P}_{\rho_k}\setminus \{\struc{A}_{\emptyset}\}$. Then there is a $k\in \{1,\dots,m\}$ such that $\struc{A}\in \classStruc{P}_{\rho_k}$. 
		
		By definition, $\varphi_{\operatorname{recursion}}:= \forall x \forall z\big(\varphi(x,z)\lor \psi(x,z)\big)$ (see Section~\ref{sec: definitionFormula}),  where 
		\begin{align*}
		\varphi(x,z):=&\lnot \exists y F(x,y)\land \lnot \exists y F(z,y) \text{ and }\\
		\psi(x,z):=&\bigwedge_{\substack{k_1',k_2'\in \indexSetRotation\\\ell_1',\ell_2'\in \indexSetRotation}}\bigg(\exists y \big[E_{k_1',\ell_1'}(x,y)\land E_{k_2',\ell_2'}(y,z)\big]\rightarrow \\&
		\bigwedge_{\substack{i,j,i',j'\in [D], k,\ell\in \indexSetH\\\rot_H(k,i)=((k_1', k_2'),i')\\ \rot_H((\ell_2', \ell_1'),j)=(\ell,j')}}\exists x'\exists z'\big[ F_k(x,x')\land F_\ell(z,z')\land
		E_{(i,j),(j',i')}(x',z')\big]\bigg).
		\end{align*}
		Let $a,c\in \univ{A}$. Assume first that there is $b\in \univ{A}$ with $(a,b)\in \rel{F}{\struc{A}}$.  Hence $\struc{A}\not\models \varphi(a,c)$. Since $\varphi_{\operatorname{recursion}}:= \forall x \forall z\big(\varphi(x,z)\lor \psi(x,z)\big)$ we aim to prove $\struc{A}\models \psi(a,c)$.  
		By construction of $\rho_k$, there is an $i\in I_{k}$ such that $(\mathcal{N}_2^{\struc{A}}(a),a)\in \tau_{d,2,\sigma}^i$. Therefore there is a structure $\other{\struc{A}}\models \varphi_{\zigzag}$ and $\other{a}\in \univ{\other{A}}$ such that $(\mathcal{N}_2^{\struc{A}}(a),a)\cong (\mathcal{N}_2^{\other{\struc{A}}}(\other{a}),\other{a})$.  Let $f$ be an isomorphism from $(\mathcal{N}_2^{\struc{A}}(a),a)$ to $(\mathcal{N}_2^{\other{\struc{A}}}(\other{a}),\other{a})$. Since $b\in N_2^{\struc{A}}(a)$, we get that $f(b)$ is defined. Since $f$ is an isomorphism mapping $a$ onto $\other{a}$, we have that $(a,b)\in \rel{F}{\struc{A}}$ implies that $(\other{a},f(b))\in \rel{F}{\other{\struc{A}}}$. Hence $\other{\struc{A}}\not\models \varphi(\other{a},\other{c})$, for every $\other{c}\in \univ{\other{A}}$. But since $\other{\struc{A}}\models \varphi_{\operatorname{recursion}}$, as $\other{\struc{A}}\models \varphi_{\zigzag}$, this shows that $\other{\struc{A}}\models \psi(\other{a},\other{c})$ for every $\other{c}\in \univ{\other{A}}$. 
		
		Let $k_1',k_2'\in \indexSetRotation$ and $\ell_1',\ell_2'\in \indexSetRotation$ be indices such that there is  $b'\in \univ{A}$ with $(a,b')\in \rel{E_{k_1',\ell_1'}}{\struc{A}}$ and $(b',c)\in \rel{E_{k_2',\ell_2'}}{\struc{A}}$. 
		Since $b',c\in N_2^{\struc{A}}(a)$, by assumption we get that $f(b')$ and $f(c)$ are defined. Furthermore, $(a,b')\in \rel{E_{k_1',\ell_1'}}{\struc{A}}$ and $(b',c)\in \rel{E_{k_2',\ell_2'}}{\struc{A}}$ imply that $(\other{a},f(b'))\in \rel{E_{k_1',\ell_1'}}{\other{\struc{A}}}$ and $(f(b'),f(c))\in \rel{E_{k_2',\ell_2'}}{\other{\struc{A}}}$, since $f$ is an isomorphism mapping $a$ onto $\other{a}$. We proved in the previous paragraph that $\other{\struc{A}}\models \psi(\other{a},f(c))$. Hence we can conclude that for all indices $i,j,i',j'\in [D]$,  $k,\ell\in \indexSetH$ for which $\rot_H(k,i)=((k_1', k_2'),i')$ and  $\rot_H((\ell_2', \ell_1'),j)=(\ell,j')$, there are elements $\other{a}',\other{c}'\in \univ{\other{A}}$ such that $ (\other{a},\other{a}')\in \rel{F_k}{\other{\struc{A}}}$, $(f(c),\other{c}')\in \rel{F_\ell}{\other{\struc{A}}}$, and $(\other{a}',\other{c}')\in \rel{E_{(i,j),(j',i')}}{\other{\struc{A}}}$. Since  $\other{a}',\other{c}'\in N_2^{\other{\struc{A}}}(\other{a})$, we get that $a':=f^{-1}(\other{a}')$ and $c':=f^{-1}(\other{c}')$ are defined. Furthermore, we get that $ (a,a')\in \rel{F_k}{\struc{A}}$, $(c,c')\in \rel{F_\ell}{\struc{A}}$ and $(a',c')\in \rel{E_{(i,j),(j',i')}}{\struc{A}}$. This proves that $\struc{A}\models \psi(a,c)$.\\
		
		In the case that there is  $b\in \univ{A}$ with $(c,b)\in \rel{F}{\struc{A}}$, we can prove similarly that $\struc{A}\models \psi(a,c)$, by considering that there exist $\other{\struc{A}}\models \varphi_{\zigzag}$ and $\other{c}\in \univ{\other{A}}$ such that $(\mathcal{N}_2^{\struc{A}}(a),c)\cong (\mathcal{N}_2^{\other{\struc{A}}}(\other{c}),\other{c})$ by construction of $\rho_k$. Finally if there is no $b\in \univ{A}$ such that $(a,b)\in \rel{F}{\struc{A}}$ or $(c,b)\in \rel{F}{\struc{A}}$ then $\struc{A}\models \varphi(a,c)$. Since this covers every case we get that $\struc{A}\models \varphi_{\operatorname{recursion}}$.
	\end{proof}
	Assume $\struc{A}\in \bigcup_{1\leq k \leq m}\classStruc{P}_{\rho_k}$. As proved in Claims~\ref{claim:satisfyingTree}, \ref{claim:satisfyingRotationMap}, \ref{claim:satisfyingBase} and \ref{claim:satisfyingRecursion} this implies that $\struc{A}\models \varphi_{\operatorname{tree}}$, $\struc{A}\models \varphi_{\operatorname{rotationMap}}$, $\struc{A}\models \varphi_{\operatorname{base}}$ and $\struc{A}\models \varphi_{\operatorname{recursion}}$. Since $\varphi_{\zigzag}$ is a conjunction of these formulas, we get $\struc{A}\models \varphi_{\zigzag}$ and hence $\struc{A}\in \classStruc{P}_{\zigzag}$.
\end{proof}

\subsubsection{The graph property is GSF-local} 
Let $\graphProp$ be the graph property as defined in Section~\ref{sec:localreduction} and let $f:\classStruc{C}_{\sigma,d}\rightarrow \mathcal{C}_3$ be the local reduction from $\classStruc{P}_{\zigzag}$ to $\graphProp$. We now use this local reduction and the expressibility of $\classStruc{P}_{\zigzag}$ by $0$-profiles to show that $\graphProp$ is GSF-local. 
\begin{lemma}\label{lemma:graphproperty_gsf_local}
	The graph property $\graphProp$ is GSF-local.
\end{lemma}
\begin{proof}
	For this we will prove that $\graphProp$ is equal to a finite union of properties defined by $0$-profiles, and then use Theorem~\ref{thm:subsetOfFOIsLocal} to prove that $\graphProp$ is GSF-local. We define the $0$-profiles for $\graphProp$ in a very similar way to the relational structure case, and then use the description of $\classStruc{P}_{\zigzag}$ by $0$-profiles shown in Lemma~\ref{lem:neighbouhoodProfilOfPZigZag}. To this end, let $\ell':=24\ell+18+d$ and assume that the $\ell'$-types $\tau_{d,\ell'}^1,\dots,\tau_{d,\ell'}^{n_{d,\ell'}}$ are ordered in such a way that $(\mathcal{N}_{\ell'}^{f(\struc{B})}(u_{b,1}),u_{b,1})\in \tau_{d,\ell'}^{k}$, for every $k\in \{1,\dots,m\}$ and $(\struc{B},b)\in \tau_{d,2,\sigma}^{k}$, where $m$ is the number of parts of the partition of $\classStruc{P}_{\zigzag}$ defined in Subsection~\ref{sec:charBy0Profiles}. For $k\in \{1,\dots,m\}$, let  $\hat{I}_{k}$ be the set of indices $i$ such that there is $\struc{A}\in \classStruc{P}_k$, and $v\in V(f(\struc{A}))$ for which $(\mathcal{N}_{\ell'}^{f(\struc{A})}(v),v)\in \tau_{d,\ell'}^i$. Let $\hat{\rho_k}:\{1,\dots,n_{d,\ell'}\}\rightarrow \mathfrak{I}_{0}$ be defined by
	\begin{align*}
	\hat{\rho_k}(i):=
	\begin{cases}
	[0,1] &  \text{ if }i=k, \\
	[0,\infty) & \text{ if }i\in \hat{I}_{k}\setminus\{k\},\\
	[0,0] & \text{ otherwise}.
	\end{cases}
	\end{align*}
	\begin{claim}\label{claim:profilOfGraphProperty}
		It holds that $\graphProp=\bigcup_{1\leq k \leq m}\classStruc{P}_{\hat{\rho}_k}$.
	\end{claim}
	\begin{proof}
		First we prove $\graphProp\subseteq \bigcup_{1\leq k \leq m}\classStruc{P}_{\hat{\rho}_k}$. Assume $G\in \graphProp$ and let $\struc{A}\in \classStruc{P}_{\zigzag}$ be a structure such that $G=f(\struc{A})$. If $\struc{A}=\struc{A}_{\emptyset}$ then clearly $G\in \bigcup_{1\leq k \leq m}\classStruc{P}_{\hat{\rho}_k}$. Hence assume $\struc{A}\not=\struc{A}_{\emptyset}$. Then $\struc{A}\in \classStruc{P}_k$ for some $k\in \{1,\dots,m\}$. By the construction of $\hat{I}_{k}$ we know that for every $v\in V(G)$ we have  $(\mathcal{N}_{\ell'}^G(v),v)\in \tau_{d,\ell'}^i$ for some $i\in \hat{I}_{k}$. Furthermore, since $\struc{A}\in \classStruc{P}_k$ there is at most  one $a\in \univ{A}$ with $(\mathcal{N}_{2}^{\struc{A}}(a),a)\in \tau_{d,2,\sigma}^{k}$. This implies directly that there can be at most one vertex $v\in V(G)$ with  $(\mathcal{N}_{\ell'}^G(v),v)\in \tau_{d,\ell'}^{k}$ and hence $G\in \classStruc{P}_{\hat{\rho}}$. \\
		
		Now we prove that $\bigcup_{1\leq k \leq m}\classStruc{P}_{\hat{\rho}_k}\subseteq \graphProp$. Let $G\in \bigcup_{1\leq k \leq m}\classStruc{P}_{\hat{\rho}_k}$ and let $k\in \{1,\dots,m\}$ be an index such that $G\in \classStruc{P}_{\hat{\rho}_k}$. Further assume that $G$ is not the empty graph, as $f(\struc{A}_{\emptyset})\in \graphProp$ is the empty graph. 
		
		Since for every $i$ for which $\hat{\rho}(i)\not=[0,0]$, there is a graph $G'\in \graphProp$ and $v\in V(G')$ such that $(\mathcal{N}_{\ell'}^{G'}(v'),v')\in \tau_{d,\ell'}^i$, we get that the $\ell'$-neighbourhood of every vertex in $G$ appears in some graph $G'\in \graphProp$. By choice of $\ell'$ we get that every vertex $v\in V(G)$ either is contained in a cycle of length $d$ and is the endpoint of some $k$-arrow, $k$-loop or non-arrow or $v$ is an internal vertex of a $k$-arrow, $k$-loop or non-arrow. Hence, we obtain a $\sigma$-structure $\struc{A}$ with $f(\struc{A})\cong G$ by replacing any cycle $C$ of length $d$ by an element $a_C$ and adding a tuple $(a_C,a_{C'})$ to the relation $\rel{R_k}{\struc{A}}$ if there are vertices $u$ on $C$ and $v$ on $C'$ such that $u\xrightarrow{k}v$ in $G$. Let $g$ be an isomorphism from $f(\struc{A})$ to  $G$.  
		
		Now we argue that $\struc{A}\in \classStruc{P}_{\rho_k}$. First assume that there are two elements $a,b\in \univ{A}$ with $(\mathcal{N}_2^{\struc{A}}(a),a)\in \tau_{d,2,\sigma}^{k}$ and $(\mathcal{N}_2^{\struc{A}}(b),b)\in \tau_{d,2,\sigma}^{k}$. By definition, we get that $(\mathcal{N}_{\ell'}^{f(\struc{A})}(u_{a,1}),u_{a,1})\in\tau_{d,\ell'}^{k}$ and $(\mathcal{N}_{\ell'}^{f(\struc{A})}(u_{b,1}),u_{b,1})\in\tau_{d,\ell'}^{k}$. Since $g$ is an isomorphism, the restriction of $g$ to $N_{\ell'}^{f(\struc{A})}(u_{a,1})$ must be an isomorphism from $\mathcal{N}_{\ell'}^{f(\struc{A})}(u_{a,1})$ to $\mathcal{N}_{\ell'}^{G}(g(u_{a,1}))$, and hence $(\mathcal{N}_{\ell'}^{G}(g(u_{a,1})),g(u_{a,1}))\cong (\mathcal{N}_{\ell'}^{f(\struc{A})}(u_{a,1}),u_{a,1})\in \tau_{d,\ell'}^{k}$. But the same holds for the $\ell'$-ball of $g(u_{b,1})$, and hence we contradict the assumption that $G\in \classStruc{P}_{\hat{\rho}_k}$ since $\hat{\rho}_k(k)=[0,1]$.
		Let us further assume that there is an $a\in \univ{A}$ such that $(\mathcal{N}_2^{\struc{A}}(a),a)\in \tau_{d,2,\sigma}^i$ for some $i\notin I_{k}$.  Since $G\in\classStruc{P}_{\hat{\rho}_k}$   we get $(\mathcal{N}_{\ell'}^G(g(u_{a,1})),g(u_{a,1}))\in \tau_{d,\ell'}^j$ for some $j\in \hat{I}_k$. 
		But then by construction of $\hat{\rho}_k$, there must be $G'\in \graphProp$, and a vertex $v\in V(G')$ such that $(\mathcal{N}_{\ell'}^{G'}(v),v)\in \tau_{d,\ell'}^j$. Furthermore, since $\ell'>d$ vertex $v$ must be contained in cycle of length $d$. By construction of $\graphProp$, there is a structure  $\struc{A}\in \classStruc{P}_{\zigzag}$ such that $f(\struc{A}')=G'$. Since $v$ is contained in a cycle of length $d$, $v$ must be an element-vertex corresponding to some element $a'\in \univ{A'}$. Since we picked $\ell'$ in such a way, that $f(\mathcal{N}_2^{\struc{A}'}(a'))\subseteq \mathcal{N}_{\ell'}^{G'}(v)$ we get $(\mathcal{N}_2^{\struc{A}'}(a'),a')\in \tau_{d,2,\sigma}^i$  by choice of $i$ and $j$. Hence  $\struc{A}'\notin \classStruc{P}_{\rho_k}$. But this contradicts Lemma~\ref{lem:neighbouhoodProfilOfPZigZag}. 
		
		Hence we have shown that $\struc{A}\in P_{\rho_k}$. Then by Lemma~\ref{lem:neighbouhoodProfilOfPZigZag} $\struc{A}\in \classStruc{P}_{\zigzag}$, and by construction $G\in \graphProp$. 
	\end{proof}
	
	Since by Claim~\ref{claim:profilOfGraphProperty} we can express $\graphProp$ as a finite union of properties each defined by a $0$-profile,  Theorem~\ref{thm:subsetOfFOIsLocal} implies that $\graphProp$ is GSF-local.
\end{proof}

\subsubsection{Putting everything together}
Now we prove Theorem~\ref{thm:existenceLocalNonTestableProperty}.

\begin{proof}[Proof of Theorem~\ref{thm:existenceLocalNonTestableProperty}]
	 Combining Theorem~\ref{thm:nonTestabilityForStructures}, Lemma~\ref{lem:localReduction} and Lemma~\ref{lem:local_reduction} we obtain that the graph property $\graphProp$  is not testable. Lemma \ref{lemma:graphproperty_gsf_local} shows that $\graphProp$ is also a GSF-local property. Hence there exists a GSF-local property of bounded-degree graphs which is not testable. 
	Furthermore, since having a POT implies being testable, this proves that there is a GSF-local property which has no POT. By Theorem~\ref{thm:charOfPOT} this implies that not all GSF-local  
	properties are non-propagating.
\end{proof}
\subsection{GSF-local properties of graphs of bounded degree 1 and 2 are non-propagating}\label{subsec:degree1-2}
In this section, we show that the degree $3$ from Theorem~\ref{thm:existenceLocalNonTestableProperty} of the example of a GSF-local property which is propagating is optimal, in the sense that all GSF-local properties of graphs of bounded degree $1$ and $2$ are non-propagating. We note that Ito et al. \cite{ito2019characterization} claimed that every GSF-local sequence of bounded degree at most $2$ is non-propagating in the appendix of their paper. However, there is one subtle issue in their proof, as they only considered \emph{connected} forbidden generalized subgraphs (which are called forbidden configurations in \cite{ito2019characterization}). In the following, we resolve this issue. Indeed, the extension from connected forbidden generalised subgraphs to arbitrary forbidden generalized subgraphs is non-trivial and requires an involved proof which we present in this section.

We first observe that even for graphs of bounded degree $1$, not every sequence of marked graphs $\overline{\mathcal{F}}$ is non-propagating as the following example shows. A similar example was given in \cite{goldreich2011proximity}.
\begin{example}\label{ex:propagatingSequence}
	Let $\mathcal{P}\subseteq \mathcal{C}_1$ be the property of $\overline{\mathcal{F}}$-free graphs, where $F$ is the marked graph depicted in Figure~\ref{fig:propagatingSequence},			
	$\mathcal{F}_n=\{F\}$ and   $\overline{\mathcal{F}}=(\mathcal{F}_n)_{n\in \mathbb{N}}$. Let $G_k$ be the graph consisting of $k$ edges and one isolated vertex. Then the set $B$ containing the one isolated vertex of $G_k$ covers all embeddings of $F$ (see Figure~\ref{fig:propagatingSequence}). But the only way to make $G_k$ $F$-free is to remove all $k$ edges of $G_k$. Hence $G_k$ is $1/2$-far from being $F$-free, which implies that  $\mathcal{P}$ is propagating for $\overline{\mathcal{F}}$.

	However, the property $\mathcal{P}$ is non-propagating, as we show in the proof of Theorem~\ref{thm:degreeTwoCase}. Indeed, consider the alternative sequence of marked graphs $\overline{\mathcal{F}}=(\mathcal{F}_n)_{n\in \mathbb{N}}$, where $\mathcal{F}_n=\{F\}$ for $n$ even and $\mathcal{F}_n=\{F,\tilde{F}\}$ for $n$ odd. Clearly, in $G_{2k+1}$  any set $\tilde{B}$ 
	covering $\mathcal{F}_{2k+1}$ must contain one incident vertex of every edge. Hence the number of necessary modifications is at most $|B|$, suggesting that $\mathcal{P}$ is non-propagating.  
	\end{example}
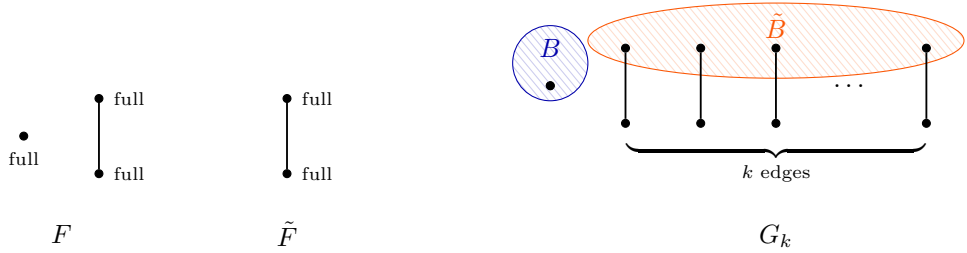
\begin{figure*}
	\centering
	\begin{tikzpicture}[scale = 1]
	\definecolor{C1}{RGB}{1,1,1}
	\definecolor{C2}{RGB}{0,0,170}
	\definecolor{C3}{RGB}{251,86,4}
	\definecolor{C4}{RGB}{50,180,110}
	\tikzstyle{ns1}=[line width=0.7]
	\tikzstyle{ns2}=[line width=1.2]
	\node[draw,C1,circle,fill=C1,inner sep=0pt, minimum width=3pt] (0) at (-1,0.5) {};	
	\node[draw,C1,circle,fill=C1,inner sep=0pt, minimum width=3pt] (1) at (0,0) {};
	\node[draw,circle,fill=black,inner sep=0pt, minimum width=3pt] (2) at (0,1) {};
	\path[ns1]          (1)  edge   (2);
	\node[minimum height=10pt,inner sep=0,font=\scriptsize] at (-1,0.2) {full}; 
	\node[minimum height=10pt,inner sep=0,font=\scriptsize] at (0.4,0) {full}; 
	\node[minimum height=10pt,inner sep=0,font=\scriptsize] at (0.4,1) {full}; 
	\node[minimum height=10pt,inner sep=0] at (-0.5,-0.8) {$F$};
	
	\node[draw,C1,circle,fill=C1,inner sep=0pt, minimum width=3pt] (3) at (2.5,0) {};
	\node[draw,circle,fill=black,inner sep=0pt, minimum width=3pt] (4) at (2.5,1) {};
	\path[ns1]          (3)  edge   (4);
	\node[minimum height=10pt,inner sep=0,font=\scriptsize] at (2.9,0) {full}; 
	\node[minimum height=10pt,inner sep=0,font=\scriptsize] at (2.9,1) {full}; 
	\node[minimum height=10pt,inner sep=0] at (2.5,-0.8) {$\tilde{F}$};
	
	\begin{scope}[xshift=7cm,yshift=0.67cm]
	\begin{scope}
	\clip[postaction={fill=white,fill opacity=0.2}](2,1.1) ellipse (2.5cm and 0.5cm);
	\foreach \x in {-7.1,-7,...,15.2}%
	\draw[C3!20](\x, -5)--+(-12,14.4);
	\end{scope}
	\draw [C3](2,1.1) ellipse (2.5cm and 0.5cm);
	\node[draw,C1,circle,fill=C1,inner sep=0pt, minimum width=3pt] (0) at (-1,0.5) {};	
	\node[draw,C1,circle,fill=C1,inner sep=0pt, minimum width=3pt] (1) at (0,0) {};
	\node[draw,circle,fill=black,inner sep=0pt, minimum width=3pt] (2) at (0,1) {};
	\node[draw,C1,circle,fill=C1,inner sep=0pt, minimum width=3pt] (3) at (1,0) {};
	\node[draw,circle,fill=black,inner sep=0pt, minimum width=3pt] (4) at (1,1) {};
	\node[draw,C1,circle,fill=C1,inner sep=0pt, minimum width=3pt] (5) at (2,0) {};
	\node[draw,circle,fill=black,inner sep=0pt, minimum width=3pt] (6) at (2,1) {};
	\node[draw,C1,circle,fill=C1,inner sep=0pt, minimum width=3pt] (7) at (4,0) {};
	\node[draw,circle,fill=black,inner sep=0pt, minimum width=3pt] (8) at (4,1) {};
	\path[ns1]          (1)  edge   (2);
	\path[ns1]          (3)  edge   (4);
	\path[ns1]          (5)  edge   (6);
	\path[ns1]          (7)  edge   (8);
	\node[minimum height=10pt,inner sep=0] at (3,0.5) {$\dots$};
	\begin{scope}
	\clip[postaction={fill=white,fill opacity=0.2}](-1,0.8) circle (0.5);
	\foreach \x in {-7.1,-7,...,15.2}%
	\draw[C2!20](\x, -5)--+(-12,14.4);
	\end{scope}
	\draw [C2](-1,0.8) circle (0.5);
	\node[draw,C1,circle,fill=C1,inner sep=0pt, minimum width=3pt] (0) at (-1,0.5) {};
	\node[C2,minimum height=10pt,inner sep=0] at (-1,1) {$B$};
	\node[C3,minimum height=10pt,inner sep=0] at (2,1.3) {$\tilde{B}$};
	\node[minimum height=10pt,inner sep=0] at (2,-0.4) {$\underbrace{\phantom{textllllllllllllllllllllllllllllll}}_{k\text{ edges}}$};
	\node[minimum height=10pt,inner sep=0] at (2,-1.5) {$G_k$};
	\end{scope}
	
	\end{tikzpicture}
	\caption{Marked graphs $F$ and $\tilde{F}$ and graph $G_k$ from Example~\ref{ex:propagatingSequence}.}\label{fig:propagatingSequence}
\end{figure*}
Indeed, adding certain redundant marked graphs to the sequence $\overline{\mathcal{F}}=(\mathcal{F}_n)_{n\in \mathbb{N}}$ to control the behaviour of sets covering $\mathcal{F}_n$ as in Example~\ref{ex:propagatingSequence}  works in general both in the degree $1$ and degree $2$ case and will be our proof strategy for the following theorem.
More precisely, for a property $\mathcal{P}$ of graphs of bounded degree $2$, a sequence $\overline{\mathcal{F}}$ of marked graphs such that $\mathcal{P}$ is $\overline{\mathcal{F}}$-local and a bound $k$ on the size of any graph appearing in $\overline{\mathcal{F}}$, we add forbidden generalized subgraphs to $\overline{\mathcal{F}}$ in the following way. In case there is no graph in $\mathcal{P}$ on $n$ vertices containing a set of different small connected components (connected components with at most $k$ vertices) each with frequency at least $k$, we add a generalised subgraph forbidding precisely this combination of connected components to $\mathcal{F}_n$. Additionally, if no graph in $\mathcal{P}$ on $n$ vertices contains a set of different small connected components each with frequency at least $k$ and one large component (connected component with at least $k+1$ vertices), we add a generalised subgraph forbidding precisely this combination of connected components to $\mathcal{F}_n$. Now for a graph $G$ which is not in $\mathcal{P}$ and a set $B$ covering all forbidden generalised subgraphs in $G$, we look at what types of connected components appear in the part of $G$ not containing vertices from $B$. In case $G$ is large enough and $B$ is small enough, we observe that some types of connected components have to appear with high frequency, or there must be a large component in the part of $G$ which is not covered by $B$. By adding redundant subgraphs as described earlier, this now implies that there must be a graph $G'$ in $\mathcal{P}$, which has the same structure as $G$ on a large subset of the part of $G$ which is not covered by $B$. Hence we can modify $G$ to obtain a graph satisfying the property $\mathcal{P}$ (by changing $G$ to $G'$) without modifying $G$ much beyond $B$. The restriction to bounded degree at most $2$ is crucial in this argument as it gives us the necessary control over large connected components. 

\begin{theorem}\label{thm:degreeTwoCase}
	Any GSF-local property $\mathcal{P}\subseteq \mathcal{C}_d$ for $d\leq 2$ is non-propagating.
\end{theorem}
\begin{proof}
	We only consider the case that $\mathcal{P}\subseteq \mathcal{C}_2$. We can consider any property $\mathcal{P}\subseteq \mathcal{C}_1$ as a property in $\mathcal{C}_2$ by forbidding any vertex to have degree $2$, \ie adding a path of length 2 in which both degree 1 vertices are marked `partial' and the degree 2 vertex is marked `full' to every set of forbidden marked graphs in any sequence defining $\mathcal{P}$, and adjusting constants in the following argument to account for the degree being $1$ instead of $2$. 
	
	Let $\mathcal{P}=\bigcup_{n\in \mathbb{N}}\mathcal{P}_n$ and $\overline{\mathcal{F}}=(\mathcal{F}_n)_{n\in \mathbb{N}}$ be a sequence of marked graphs such that $\mathcal{P}$ is $\overline{\mathcal{F}}$-local. By definition there exists $k\in \mathbb{N}$ such that every marked graph appearing in $\overline{\mathcal{F}}$ contains at most $k$ vertices.
	
	For two sets $I\subseteq [k]:=\{0,\dots,k-1\}$, $J\subseteq \{3,\dots,k\}$ such that $I\cup J \not= \emptyset$ let $F_{I,J}$ be the marked graph which is the disjoint union  of $k$ paths of length $i$ for every $i\in I$ and $k$ cycles of length $j$ for every $j\in J$ in which every vertex is marked as `full'. Be aware that 
	a path of length $i$ contains $i+1$ vertices and a cycle of length $j$ contains $j$ vertices. Note that graphs that are  $F_{I,J}$-free can not contain at the same time $k$ connected components that are paths of length $i$ for every $i\in I$ and $k$ connected components which are cycles of length $j$ for every $j\in J$. We let $F_{\emptyset,\emptyset}^{\operatorname{large}}$ be a path of length $k+1$ in which both  vertices of degree $1$ are marked as `partial' and every other vertex is marked `full'. We further let $F_{I,J}^{\operatorname{large}}$ be the disjoint union of $F_{I,J}$ and  $F_{\emptyset,\emptyset}^{\operatorname{large}}$ for $I\subseteq [k]$, $J\subseteq \{3,\dots,k\}$ with $I\cup J \not= \emptyset$. Note that graphs that are  $F_{I,J}^{\operatorname{large}}$-free can not contain at the same time $k$ connected components that are paths of length $i$ for every $i\in I$ and $k$ connected components which are cycles of length $j$ for every $j\in J$ and one connected component containing at least $k+1$ vertices.
	
	We obtain a sequence $\overline{\mathcal{F}}'=(\overline{\mathcal{F}}'_n)_{n\in \mathbb{N}}$ by setting
\begin{align*}\mathcal{F}'_n:= \mathcal{F}_n\cup & \Big\{F\in \{F_{\emptyset,\emptyset}^{\operatorname{large}},F_{I,J},F_{I,J}^{\operatorname{large}}: I\subseteq [k], J\subseteq \{3,\dots,k\}, \\
&I\cup J\not=\emptyset\}: \text{ every }G\in \mathcal{P}_n \text{ is } F\text{-free}\Big\}.
\end{align*}
	First observe that by construction $\mathcal{P}$ must be $\overline{\mathcal{F}}'$-local.
	
	We use the following notation. For a graph $G\in \mathcal{C}_2$, $i\in [k]$ and $j\in \{3,\dots,k\}$ we let 
	\begin{itemize}
		\item $p_i(G)$ be the number of connected components of $G$ that are path of length $i$.
		\item $c_j(G)$ be the number of connected components of $G$ that are cycles of length $j$.
		\item $cc^{\operatorname{large}}(G)$ be the number of connected components of $G$ with more than $k$ vertices.
	\end{itemize} 
	We choose the following (monotonically non-decreasing) function
	$\tau(\epsilon):= \min(1,8k^3\epsilon) $ for $\epsilon\in (0, 1]$. 	
	Let $G$ be a graph on $n$ vertices which is not $\mathcal{F}_n'$-free and $\mathcal{P}_n$ 
	is not empty. Let $B\subseteq V(G)$ be any set covering $\mathcal{F}_n'$. 
	To show that $\mathcal{F}'$ is non-propagating 
	it is sufficient to show that  $G$ is 
	$\tau(|B|/n)$-close to $\mathcal{P}$. By choice of $\tau$ this means that we have to argue 
	that we can make  $G$ have property $\mathcal{P}_n$ by modifying at most $16 k^3 |B|$ edges.
	Hence for the remainder of this proof we argue that $G$ is $\tau(|B|/n)$-close to $\mathcal{P}$. 
	 
	Assuming $n< 8k^3$, we get that $\tau(|B|/n)= 1$ (since $G$ is not 
	$\mathcal{F}_n'$-free we know that $|B|\geq 1$), 
	which means $G$ is $\tau(|B|/n)$-close  to $\mathcal{P}$, as in this case we can modify all edges of $G$ and hence we can make $G$ into any graph 
	in $\mathcal{P}_n$. Hence we now assume that $n\geq 8k^3$.
	 
	Now consider the case that $|B|\geq  \frac{n}{8k}$. In this case $\tau(|B|/n)=1$ and $G$ is 
	$\tau(|B|/n)$-close to $\mathcal{P}$ again because we are allowed to modify all edges of $G$ which allows us to make $G$ into any graph 
	in $\mathcal{P}_n$. Hence from now on we only consider the case that $|B|\leq  \frac{n}{8k}$.

	Let $S$ be the set of vertices for which the $k$-neighbourhood does not contain any vertex from $B$.
	Let $I\subseteq [k]$, $J\subseteq \{3,\dots,k\}$ be the sets of indices such that $i\in I$ if and only 
	if $p_i(G[S])\geq k$  and $j\in J$ if and only if $c_j(G[S])\geq k$. Note that $I\cup J$ could be empty.
	\bigskip
	
	\textbf{Case 1:} Assume that $F_{I,J}^{\operatorname{large}}\notin \mathcal{F}'_n$. 
	\bigskip
	
	First note that every component of size at most $k$ which contains a vertex from $S$ cannot contain a vertex from $B$ by definition of $S$. Hence every connected component of $G$ of size at most $k$ is either fully contained in $S$ or disjoint from $S$. Since there are at most $2k$ isomorphism types of connected  
	components of size at most $k$ we know that there are at most $2k^2$ connected components $X$ of $G[S]$ 
	 such that  	there are at most $k-1$ other connected components of $G[S]$ isomorphic to $X$. In other words, there are
	 at most $2k^2$ components $X$ of $G$ containing no element from $B$ such that if $X$ is a path of length $i$ then $i\notin I$ and if 
	 $X$ is a cycle of length $j$ then $j\notin J$. We now obtain $G'$ 
	by the following edge modifications from $G$. For every cycle of length $j$ where 
	$j\notin J$,  
	we delete one edge (at most $2k^2+|B|$ edge by our previous argument). Then we add edges 
	connecting all path (including the paths obtained in the last step) of length $i$ for 
	$i\notin I$ to one long cycle $C$ (at most $2k^2+|B|$ edge additions). If $C$ has length 
	less or equal to $k$ there must be $i\in I$ or $j\in J$ such that $p_i(G)>k$ or $c_j(G)>k$, 
	in which case we include one respective component in $C$ and repeat this until $C$ has 
	length at least $k+1$ (at most $2k$ modifications). Since in total we did at most 
	$4k^2+2|B|+2k\leq 16 k^2|B|$ edge modifications, $G$ is $\tau(|B|/n)$-close to $G'$. 
	The following claim completes the proof of Case 1 by showing that $G'\in \mathcal{P}_n$.

	\begin{claim}\label{claim:existanceOfGraphWithLargeCycle}
		Let $I\subseteq [k]$, $J\subseteq \{3,\dots,k\}$ and  $a_i,b_j\geq k$ where $i\in I$, $j\in J$ be any
		selection of integers  such that 
		\begin{equation}\label{eq:conditionQuantitiesWithLargeCycle}
		\sum_{i\in I}i\cdot a_i+\sum_{j\in J}j\cdot b_j\leq n-(k+1).
		\end{equation}
		If  
		$F_{I,J}^{\operatorname{large}}\notin \mathcal{F}'_n$ then
		any graph $H\in \mathcal{P}_n$ with $p_i(H)=a_i$, $c_j(H)=b_j$ for $i\in I$, $j\in J$ 
		and  one additional connected component which is a cycle is $\mathcal{F}'_n$-free. 
	\end{claim}
	\begin{proof}
		Assume there is a graph $H\in \mathcal{P}_n$ as given in the statement which is not 
		$\mathcal{F}_n'$-free and let $C$ be the cycle in $H$ of length larger than $k$.
		Then there is $F\in \mathcal{F}_n'$ such that there is an embeddings $f:V(F)\rightarrow V(H)$.
		Since $F_{I,J}^{\operatorname{large}}\notin \mathcal{F}'_n$, by construction there is a graph 
		$H'\in \mathcal{P}_n$ with $p_i(H')\geq k$, $c_j(H')\geq k$ for $i\in I$, $j\in J$ and 
		$cc^{\operatorname{large}}(H')\geq 1$. We let $C'$ be a connected component of $H'$ of size larger than $k$.
		To find an embedding of $F$ into $H'$, for  every connected component $X$ of $H$ of size 
		at most $k$ which contains a vertex from $f(V(F))$, we pick a unique connected component $X'$ 
		of $H'$ which is isomorphic to $X$. Note that 
		because $|f(V(F))|\leq k$ and $p_i(H')\geq k$, $c_j(H')\geq k$ we can pick the connected component 
		in $H'$ uniquely. For  every connected component $X$ of $H$ of size 
		at most $k$ which contains a vertex from $f(V(F))$, we now define $f_X$ to be an isomorphism from $X$ to $X'$. 
		Furthermore, we pick an injective  graph homomorphism  $f^{\operatorname{large}}: f(V(F))\cap C \rightarrow C'$. 
		Again, this is possible because  $|f(V(F))|\leq k$.
		We now let $f'(v):=f_X(f(v))$ if $f(v)$ is in the connected component $X$ and $f'(v):=f^{\operatorname{large}}(f(v))$ if 
		$f(v)$ is in $C$. Note that $f'$ is injective by construction. Furthermore,  as a consequence of picking $f_X$ to 
		to be isomorphisms and $f^{\operatorname{large}}$ to be a homomorphism we get that $f'$ is an embedding of $F$ into $H'$. 
		To see this we observe that for any vertex $v\in V(F)$ which is marked as `full' and for which $f'(v)$ is in a connected 
		component $X$ with at most $k$ vertices 
		we obtain the condition $N_1^{H'}(f'(v))=f'(N_1^F(v))$ from $f_X$ being an isomorphism. 
		On the other hand, in case $f'(v)$ is in $C'$ and $v$ is marked `full' we get that $v$ has 
		two neighbours $w_1,w_2$ in $F$ and $f(w_1)$, $f(w_2)$ are neighbours of $f(v)$ 
		(since $f(v)$ must be on $C$) which implies that $f'(w_1)$ and $f'(w_2)$ are neighbours of $f'(v)$ (since $f^{\operatorname{large}}$ is a homomorphism).
		Since $f'$ is an embedding of $F$ into $H'$ we obtain a contradiction to $H'\in \mathcal{P}_n$ and hence $H$ is $\mathcal{F}_n'$-free. 
		Therefore $H$ must be $\mathcal{F}_n'$-free as claimed.
	\end{proof}
	\bigskip
	
	\textbf{Case 2:} Assume that $F_{I,J}^{\operatorname{large}}\in \mathcal{F}'_n$. In this case our strategy 
	is to modify the connected components of $G$ containing a vertex from $B$ into paths and cycles of length $i$ 
	for $i\in I$ or $i\in J$, respectively. 
	
	Since the $k$-neighbourhood of every vertex contains no more than $2k+1$ vertices, 
	$|B|\leq  \frac{n}{8k}$  implies that  $|S|\geq n/2$.  Furthermore, since 
	$F_{I,J}^{\operatorname{large}}\in \mathcal{F}'_n$ no vertex in $S$ can be contained in 
	a connected component of size larger than $k$ as otherwise there would be an embedding of 
	$F_{I,J}^{\operatorname{large}}$ into $G$ which is not covered by $B$. Hence $G[S]$ is the 
	disjoint union of paths
	of length at most $k-1$ and cycles of length at most $k$.
	Since $|S|\geq 4k^3$ and $G[S]$ contains at most $2k$ different isomorphism types of 
	connected components and each of the connected components has at most $k$ vertices we 
	conclude that at least $2k\geq k+1$ of the connected components of $G[S]$ are pairwise isomorphic. Hence 
	$I\cup J\not= \emptyset$. Furthermore,  $F_{I,J}$ is defined and not in $\mathcal{F}_n'$ 
	since $B$ covers $\mathcal{F}_n'$. 
	
	The next claim is the key to showing that we can modify $G$ into having property $\mathcal{P}$ without
	modifying more than a constant number of edges in $G[S]$.
	
	\begin{claim}\label{claim:existanceOfGraphWithProperty}
		If for $I\subseteq [k]$, $J\subseteq \{3,\dots,k\}$ with $I\cup J \not= \emptyset$ we have 
		that $F_{I,J}^{\operatorname{large}}\in \mathcal{F}'_n$ and $F_{I,J}\notin \mathcal{F}'_n$ then 
		for any selection of integers $a_i,b_j\geq k$ where $i\in I$, $j\in J$ such that 
		\begin{equation}\label{eq:conditionQuantities}
		\sum_{i\in I}i\cdot a_i+\sum_{j\in J}j\cdot b_j\leq n-k^3
		\end{equation}  there is an $\mathcal{F}_n'$-free graph $H\in \mathcal{P}_n$ such that $p_i(H)\geq a_i$ and $c_j(H)\geq b_j$.
	\end{claim}
\begin{proof}
	We set $a_i=0$ for $i\in [k]\setminus I$ and $b_j=0$ for $j\in \{3,\dots,k\}\setminus J$.
	Since $F_{I,J}^{\operatorname{large}}\in \mathcal{F}'_n$ and $F_{I,J}\notin \mathcal{F}'_n$ 
	by construction of $\overline{\mathcal{F}}$ there must be a graph in $\mathcal{P}_n$ whose 
	connected components include at least $k$ paths of length $i$ for every $i\in I$, $k$ cycles 
	of length $j$ for every $j\in J$ and no connected component containing more than $k$ vertices.  
	Pick $H$ amongst all graphs in $\mathcal{P}_n$ with these properties such that 
	$$(\ast):=\sum_{i\in [k] \atop{p_i(H) < a_i}}a_i-p_i(H)+\sum_{j\in \{3,\dots,k\}\atop{c_j(H)<b_j}}c_j(H)-b_j$$ 
	is minimal. In case $(\ast)>0$ there is $i\in [k+1]$ such that either $p_i(H)<a_i$ or $c_i(H)<b_i$.
	Combining this with Equation~\ref{eq:conditionQuantities} we obtain that there must be $j\in [k+1]$ 
	such that either $p_j(H)-a_j>k$ or $c_j(H)-b_j>k$.  We let $H'$ be the graph obtained from $H$ by  
	replacing $i$ connected components which are paths of length $j$ or $i$ connected components which 
	are cycles of length $j$, respectively, and  adding $j$ disjoint paths of length $i$ or $j$ disjoint 
	cycles of length $i$, respectively. By choice of $i,j$ we get that 
	$$(\ast)>\sum_{i\in [k] \atop{p_i(H') < a_i}}a_i-p_i(H')+\sum_{j\in \{3,\dots,k\}\atop{c_j(H')<b_j}}c_j(H')-b_j.$$ 
	Furthermore, $H'$ must be $\mathcal{F}_n'$-free which we will argue in the following. Assume this is 
	not the case and there is $F\in \mathcal{F}_n'$ and an embedding $f:V(F)\rightarrow V(H')$. 
	We obtain a map $f':V(F)\rightarrow V(H)$ from $f$ as follows. For every connected component $X$ in $H'$ 
	which has been altered we pick a unique connected component $X'$ of $H'$ which is isomorphic to $X$ and 
	contains no vertex in the image $f(V(F))$. This is possible as the assumption that  $X$ was altered implies 
	that either $|X|-1\in I$ or $|X|\in J$ and hence there are at least $k$ connected components isomorphic to 
	$X$ in $H'$ which were not altered. Since further $|f(V(F))|\leq k$ we can pick the $X'$ uniquely.
	We now let $f_X$ be an isomorphism from $X$ to $X'$ for every connected component $X$ which has been altered and $f_X$ the identity
	for every connected component $X$ which has not been altered. We define $f'(v):=f_X(v)$ for $v\in X$. By construction $f'$ is obviously an embedding of $F$ into $H$.
	Since $H\in \mathcal{P}_n$ this yields a contradiction.
	Hence the existence of $H'$ contradicts the assumption that $(\ast)>0$ which implies that $H$ has the claimed properties.
\end{proof}
	First observe that $n\geq 8k^3$ allows us to chose  $a_i$ and $b_j$ for every $i\in I$ and $j\in J$ in such a way that $k\leq a_i \leq p_i(G[S])$, $k\leq b_j \leq c_j(G[S])$ 
	 and $\sum_{i\in I}i\cdot a_i+\sum_{j\in J}j\cdot b_j\leq  n-k^3$. Amongst all such choices we pick $a_i$ and $b_j$ such that 
	$\sum_{i\in I}i\cdot a_i+\sum_{j\in J}j\cdot b_j$ is maximum. 
	Let $M$ be a set of vertices containing all connected components of $G$ apart from  $a_i$ 
	paths and $b_j$ cycles from $G[S]$ for every $i\in I$, $j\in J$. Then $|M|\leq  2k|B|+|B|+4k^3$ 
	since $M$ consists of $N_k^G(B)$ (at most $2k|B|+|B|$ vertices), all vertices in a connected 
	component which is either a path of length $i$ for $i\notin I$ or a cycle of length $j$ for 
	$j\notin J$ (since there are at most $2k^2$ such paths and cycles (as argued in Case 1) and each contains at most $k$ vertices) 
	or in case $a_i\not= p_i(G)$ or $b_j\not=c_j(G)$ for some $i\in I$, 
	$j\in J$, $M$ consist of at most $k^3+k$ vertices as we picked $a_i$, $b_j$ to maximise 
	$\sum_{i\in I}i\cdot a_i+\sum_{j\in J}j\cdot b_j$. 
	
	Now we use Claim~\ref{claim:existanceOfGraphWithProperty} and obtain an $\mathcal{F}_n'$-free graph $H\in \mathcal{P}_n$ such that $p_i(H)\geq a_i$ and $c_j(H)\geq b_j$. Hence we can modify $G$ into a graph $G'$ which is isomorphic to $H$ by only modifying $G[M]$. Since we can modify $G[M]$ into any graph with no more than $4k|B|+2|B|+8k^3\leq 16 k^3|B|$ modifications we showed that $G$ is $\tau(|B|/n)$-close to having $\mathcal{P}$.
\end{proof}

\section{Testing properties of neighbourhoods}
\label{sec:freeness}
In this section we only consider simple graphs, \ie undirected graphs without self-loops and without parallel edges, and  for any  $d\in \mathbb{N}$ let $\mathcal{C}_d$  be the class of simple graphs of bounded degree $d$.
We view simple graphs as structures over the signature $\sigma_{\operatorname{graph}}:=\{E\}$, where $E$ encodes a binary, symmetric and irreflexive relation. This allows transferring the notions from Section~\ref{sec:preliminaries} to graphs. 

Let $r\geq 1$ and let $\tau$ be an $r$-type and let $\varphi_{\tau}(x)$ be a FO formula saying that $x$ has $r$-type $\tau$. 
We say that a graph $G$ is \emph{$\tau$-neighbourhood regular}, if $G\models \forall x\varphi_{\tau}(x)$. We say that a graph $G$ is \emph{$\tau$-neighbourhood free}, if $G \models \lnot \exists x \varphi_{\tau}(x)$. Let $\tau_1,\dots,\tau_t$ be a list of all $r$-types in $\mathcal{C}_d$. If $F\subseteq \{\tau_1,\dots,\tau_t\}$ we say that $G$ is $F$-free, if $G$ is $\tau$-neighbourhood free for all $\tau\in F$. 

Observe that both  $\tau$-neighbourhood-freeness and $\tau$-neighbourhood regularity can be defined by formulas in $\Pi_2$ for any neighbourhood type $\tau$. Hence the next Lemma shows that there exist neighbourhood properties that are in $\Pi_2$, but not in $\Sigma_2$. 
\begin{lemma}\label{lem:existencesigma2}
There exist $1$-types $\tau,\tau'$ such that neither $\tau$-neighbourhood freeness nor $\tau'$-neighbourhood regularity can be defined by a formula in $\Sigma_2$. 
\end{lemma}

Note that the above lemma implies that we cannot simply invoke the testers for testing $\Sigma_2$ properties from Theorem \ref{thm:sigma2} to test these two properties. 

\begin{proof}[Proof of Lemma \ref{lem:existencesigma2}]
	For $n\in \mathbb{N}$, let $C_n$ be the cycle graph with vertex set $[n]:=\{0,1,\dots,n-1\}$. 
	Let $P_{n-1}$ be the path graph with vertex set $[n-1]$. We first show the following claim. 
	 \begin{claim}\label{claim:sigma2DistinguishingCycleFromPath}
Let $\varphi= \exists \overline{x}\forall \overline{y} \chi(\overline{x},\overline{y})$ where $\overline{x}=(x_1,\dots,x_k)$, $\overline{y}=(y_1,\dots,y_\ell)$ are tuples of variables and $\chi(\overline{x},\overline{y})$ is a quantifier-free formula. If $C_n\models \varphi$ then  $P_{n-1}\models \varphi$ for any $n>k$.
	\end{claim}
	\begin{proof}
		
		Assume on the contrary that for some $n>k$, 
		 it holds that $C_n \models \varphi$, while $P_{n-1}\not\models \varphi$. Since $C_n\models \varphi$ there are $k$ vertices $v_1,\dots,v_k$ in $C_n$ such that $C_n\models \forall \overline{y} \chi((v_1,\dots,v_k),\overline{y})$. Since $n>k$, there exists at least one vertex $i\in [n]$ that is not amongst $v_1,\dots,v_k$. Let $v_j':= (v_j+n-1-i)\mod n$ be a vertex of $P_{n-1}$. Since $P_{n-1}\not\models \varphi$ and $v_j'\in [n-1]$, we have that $P_{n-1}\not\models \forall \overline{y}\chi((v_1'\dots,v_k'),\overline{y})$. Hence there must be vertices $w_1',\dots,w_\ell'$ in $P_{n-1}$ such that $P_{n-1}\not\models \chi((v_1',\dots,v_k'),(w_1',\dots,w_\ell'))$. Now let $w_j:=(w_j'+i+1)\mod n$. Then $v_j\mapsto v_j'$ and $w_j\mapsto w_j'$ defines an isomorphism from $C_n[\{v_1,\dots,v_k,w_1,\dots,w_\ell\}]$ and $P_{n-1}[\{v_1',\dots,v_k',w_1',\dots,w_\ell'\}]$. Hence $C_n\not \models \chi((v_1,\dots,v_k),(w_1,\dots,w_\ell))$ which contradicts that $C_n\models \varphi$.
	\end{proof}
	
Now we let $\tau$ be the $1$-neighbourhood type saying that the center vertex $x$ has exactly one neighbour. Let $\tau'$ be the $1$-neighbourhood type saying that the center vertex has two non-adjacent vertices. Since $C_n$ is $\tau$-neighbourhood free and $\tau'$-neighbourhood regular, while $P_{n-1}$ is neither, the statement of the lemma follows from Claim~\ref{claim:sigma2DistinguishingCycleFromPath}. 
\end{proof}

Now we state our main algorithmic results in this section. The first result shows that if $\tau$ is an $r$-type with degree smaller than the degree bound of the class of graphs, then the $\tau$-neighbourhood-freeness is testable.
\begin{theorem}\label{thm:dNeighbourhoodFreeness}
	Let $\tau$ be an $r$-type, where $r\geq 1$. If $\tau \subseteq \mathcal{C}_{d'}$ and $d'<d$, then $\tau$-neighbourhood freeness is uniformly testable on the class $\mathcal{C}_{d}$ with constant running time.
\end{theorem}
The second result shows if $\tau$ is a $1$-type, then $\tau$-neighbourhood-freeness is testable.
\begin{theorem}\label{thm:1NeighbourhoodFreeness}
	For every $1$-type $\tau$, $\tau$-neighbourhood freeness is uniformly testable on the class $\mathcal{C}_{d}$ with constant time.
\end{theorem}
The third result says that $\tau$-neighbourhood regularity is testable for every $1$-type $\tau$ consisting of cliques, which only overlap in the centre vertex. 
\begin{theorem}\label{thm:neighbourhoodRegularity}
	Let $\tau$ be a $1$-type such that  vertex $a$ having $1$-type $\tau$ in $B$  implies that $B\setminus \{a\}$ is a union of disjoint cliques for every $1$-ball $B$ with centre $a$. Then $\tau$-neighbourhood regularity is uniformly testable on $\mathcal{C}_{d}$ in 
	constant time. 
\end{theorem}

By previous discussions, the above theorems imply that there are formulas in $\Pi_2\setminus\Sigma_2$ which are testable.

\subsection{Proofs of Theorem \ref{thm:dNeighbourhoodFreeness}, \ref{thm:1NeighbourhoodFreeness} and \ref{thm:neighbourhoodRegularity}}
Consider the following algorithm $\sampler_{r,s}$: given access to a graph $G\in \mathcal{C}_{d}$,
it samples a set $S$ of $s$ vertices of $G$ uniformly and independently, and then explores their $r$-balls; it then returns
the \emph{distribution vector} $\bar v$ of length $t$ of the $r$-types of this sample, \ie $\bar v_i=|\{v\in S\mid \mathcal{N}_r^G(v)\in \tau_i\}|/s$.

\begin{lemma}[Lemma 5.1 in~\cite{NewmanSohler2013}]\label{lem:estimate-frequencies}
	Let $\lambda\in (0,1]$, $r\in \mathbb{N}$ and let $G\in \mathcal{C}_{d}$ be a graph with $n$ vertices.
	Let $s\geq ({t^2}/{\lambda^2})\ln(t+40)$.
	Then the vector $\bar v$ returned
	by $\sampler_{r,s}(G)$ satisfies
	$\sum_{i=1}^{t}|\rho_{G,r}(\{\tau_i\})-\bar v_i|  \leq \lambda$  with probability at least $19/20$ .
	
If  $\rho_{G,r}({\tau_i})$ is $0$, then $\Pr[v_i=0]=1$.
\end{lemma}

The following Lemma provides a framework that will be used in Theorems~\ref{thm:dNeighbourhoodFreeness},\ref{thm:1NeighbourhoodFreeness} and \ref{thm:neighbourhoodRegularity}.
\begin{lemma}\label{lemma:testerFramework}
	Let $\F$  
	be a finite set of $r$-types of bounded maximum degree $d$ and let
	$\mathcal{P}\subseteq \mathcal{C}_{d}$ be the set of all graphs being $\F$-free.  
	Let
	$M\subseteq \mathbb{N}$  be a decidable set such that
	$G\in \mathcal{P}$ implies that $|V(G)|\notin M$. Let $f_M:\mathbb{N}\rightarrow \mathbb{N}$
	be a function such that $M$ can be decided in time $f_M$. Assume for every
	$\epsilon\in (0,1]$ there exist $\lambda:=\lambda(\epsilon) \in (0,1]$ and
	$n_0:=n_0(r,\epsilon)\in \mathbb{N}$  such that every graph $G\in \mathcal{C}_{d}$ on $n\geq
	n_0$, $n\notin M$ vertices, which is $\epsilon$-far from $\mathcal{P}$, contains
	more than $\lambda n$ elements $v$ with
	$\mathcal{N}_r^G(v)\in \tau\in \F$. Then $\mathcal{P}$ is uniformly
	testable on $\mathcal{C}_{d}$ in time $\mathcal{O}(f_M)$.  
\end{lemma}
\begin{proof}
	Consider the following probabilistic algorithm  $T$, which is given direct access to a graph $G\in \mathcal{C}_{d}$ and gets the number of vertices $n$ as input.  Let $s=({t^2}/{\lambda^2})\ln(t+40)$.
	\medskip
	\begin{enumerate}
		\item Reject if $n\in M$.
		\item If $n < n_0$: use a precomputed table to decide exactly if $G\in \mathcal{P}$.
		\item Otherwise run $\sampler_{r,s}(G)$ to get $\bar v$, which satisfies that 	
		\[\sum_{i=1}^{t}|\rho_{G,r}(\{\tau_i\})-\bar v_i| \leq \lambda\]
		 with probability at least ${19}/{20}$.
		\item Reject $G$ if $\sum_{\tau_i\in \F}\bar v_{i}>0$. Accept otherwise.
	\end{enumerate}
	
	The query complexity of $T$ is clearly constant, since $s$ is constant and
	the number of vertices in any $r$-neighbourhood is bounded by $d^{r+1}+1$
	for graphs in $\mathcal{C}_{d}$. The running time of the first step is $f_M(n)$ and
	for the other steps it is constant.
	
	To prove that $T$ is an $\epsilon$-tester, first assume that $G\in \mathcal{P}$. Then $n\notin M$ and $\mathcal{N}_r^G\in \tau \notin \F$ for all vertices $v$ . Hence $\sum_{\tau_i\in \F}\bar v_{i}=0$ and $T$ accepts $G$.
	Now consider that $G$ is $\epsilon$-far from $\mathcal{P}$. If $n\in M$ then
	$G$ is rejected in the first step. Hence let $n\notin M$, and assume
	$\sum_{i=1}^{t}|\rho_{G,r}(\{\tau_i\})-\bar v_i|  \leq \lambda$, which occurs with
	probability at least ${19}/{20}\geq {2}/{3}$. Then
	\begin{align*}
	\sum_{\tau_i\in \F}\bar v_{i}&=\sum_{\tau_i\in \F}\rho_{G,r}(\{\tau_i\})-\sum_{\tau_i\in \F}\big(\rho_{G,r}(\{\tau_i\})-\bar v_{i}\big)\\&>
	 \lambda-\Big|\sum_{\tau_i\in \F}\big(\rho_{G,r}(\{\tau_i\})-\bar v_{i}\big)\Big|\geq\lambda-\sum_{\tau_i\in \F}\big|\rho_{G,r}(\{\tau_i\})-\bar v_{i} \big|\geq 0,
	\end{align*}
	where 
	the first inequality holds by the  assumption that in graphs that are $\epsilon$-far from $\mathcal{P}$ there are more than $\lambda n$ vertices whose type is in $F$. 
	Hence $T$ rejects $G$.
\end{proof}
To illustrate the use of the set $M$ in Lemma~\ref{lemma:testerFramework}, let
$\mathcal{P}$ be the property of being $K_4$-neighbourhood regular. Let $G_m$ be the
graph consisting of $m$ disjoint copies of $K_4$ and one isolated vertex. First
note that $G_m$ contains $4m+1$ vertices. Being $K_4$-regular implies that
every vertex has degree $3$. But because every graph contains an even number of
vertices of odd degree, $G_m$ cannot be made $K_4$-neighbourhood regular by
edge modifications. Therefore $G_m$ is $\epsilon$-far from $\mathcal{P}$. But for
$m\rightarrow \infty$ the probability of sampling the isolated vertex in $G_m$
tends to $0$ meaning that with high probability the tester with $M=\emptyset$
will accept $G_m$. We will show in Theorem~\ref{thm:neighbourhoodRegularity}
that $\mathcal{P}$ is testable if we set $M=\mathbb{N} \setminus \{4m\mid m\in \mathbb{N}\}$.
\begin{lemma}\label{lemma:farImpliesManyCounterexamples-forbidden$r$-ball-oneDegreeMissing}
	For $r\geq 1$ let  $\tau$ be an $r$-type. Let $B$ be an $r$-ball with constant $a$ of type $\tau$. Let $\tilde{d}< d$, $d\not= 0$ be integers and
	assume that $\mathcal{N}_{r-1}^B(a)$ contains a vertex $b$ with
	$\deg_B(b)=\tilde{d}$ and that $\deg_B(v)\not=\tilde{d}+1$ for all
	vertices $v$ in $\mathcal{N}_{r-1}^B(a)$. Let $\epsilon \in (0,1]$ be fixed,
	$n_0={2d^2}/{\epsilon}$ and $\lambda={\epsilon
		d}/(14(1+d^{2r+1}))$. Then every graph $G\in \mathcal{C}_{d}$ on
	$n\geq n_0$ vertices which is $\epsilon$-far from being
	$\tau$-neighbourhood free contains more than $\lambda n$ vertices
	of $r$-type $\tau$.
\end{lemma}
\begin{proof}We proceed by contraposition. Assume $G\in
	\mathcal{C}_{d}$ is a graph on $n\geq n_0$ vertices containing no more than $\lambda n
	$ vertices $v$ of $r$-type $\tau$. 
	
	\textbf{Case $\tilde{d}=0$, $d>1$.} 
	In this case the forbidden type $\tau$ must be an isolated vertex. Since every vertex of degree $0$ must be of the forbidden type,  we add one edge to every pair of vertices of degree $0$.  If there is only one vertex $v$ of degree $0$ left, we  add an edge from $v$ to any other vertex of degree $<d$. If there is no such vertex then there must be vertex $u$ contained in two edges and we replace one edge $\{u,w\}$ by $\{v,w\}$. That way we obtain $G'$ which is $2\lambda n\leq \epsilon dn$ close to $G$.
	
	\textbf{Case $\tilde{d}\geq 1$.} Let us pick a set
	$\{v_1,\dots,v_k\}$ of $k\leq \lambda n$ vertices of degree $\tilde{d}$ such
	that for every vertex $v$ of $r$-type $\tau$ there is an
	index $1\leq i\leq k$ with $v_i\in N_{r-1}^G(v)$. We will
	distinguish the following two cases.
	
	First assume that there are less than $\lambda n$ vertices of degree
	$\tilde{d}$, of pairwise distance greater than $2r$ and of distance greater
	than $2r$ from $\{v_1,\dots,v_k\}$. Then there are less than $2\lambda n
	(1+d^{2r+1})$ vertices of degree $\tilde{d}$ in total. We distinguish two cases. First consider the case that $\tilde{d}=1$. In this case
	we  add edges between pairs of degree $1$ vertices. If there are two vertices of degree $1$ left   who are adjacent, we delete the edge between them. If there is only one vertex $v$ of degree $1$ left, then there is another vertex $u$ of odd degree. By removing an edge $\{u,w\}$ and adding $\{v,w\}$ we get that $\deg_G(v),\deg{w}\geq 1$. We obtain $G'$  which is $2\lambda n(1+d^{2r+1})\leq \epsilon dn$ close to $G$ and  contains no vertex of degree $\tilde{d}$ and therefore has to be $\tau$-free. 
	
	Now consider the case that $\tilde{d}\geq 2$. In this case we let $G'$ be a
	graph obtained from $G$ by the following modifications. For every vertex
	$v$ of degree $\tilde{d}$ we pick edges
	$\{v,v_1\},\{v,v_2\},\{w,w'\},\{u,u'\}$ such that $v,w,u$ have pairwise
	distance at least $3$. We delete the  edges $\{v,v_1\},\{v,v_2\}$, $\{w,w'\},\{u,u'\}$ and insert the edges
	$\{v_1,w\},\{v_2,u\},\{w',u'\}$, reducing the degree of $v$ while
	maintaining the degrees of all other vertices. The resulting graph has no
	vertex of degree $\tilde{d}$. Note that if such edges do not exist at any
	point during the iteration the graph contains no more than $2d^3\leq
	\epsilon dn$ edges, and we delete them all resulting in a graph with no
	vertex of degree $\tilde{d}$.
	In total we did no more than $7\cdot
	2\lambda n (1+d^{2r+1})\leq \epsilon dn$ edge modifications which implies that
	$G'$ is $\epsilon$-close to $G$. In addition, $G'$ is
	$\tau$-neighbourhood free, because a neighbourhood
	of type $\tau$ would imply having a vertex of degree
	$\tilde{d}$.\\
	
	Now assume that there are at least $\lambda n$ vertices of degree
	$\tilde{d}$, of pairwise distance greater than $2r$ and of distance greater
	than $2r$ from $\{v_1,\dots,v_k\}$. Let $\{v_1',\dots,v_k'\}$ be a set of
	vertices of degree $\tilde{d}$ such that $\dist_G(v_i,v_j')> 2r$ for all
	$1\leq i,j\leq k$ and $\dist(v_i',v_j')> 2r$ for all $1\leq i<j\leq k$. Let
	$G'$ be the graph obtained from $G$ by inserting the edges $\{v_i,v_i'\}$.
	First note that this takes no more than $\lambda n \leq \epsilon d n$ edge
	modifications which implies that $G$ is $\epsilon$-close to $G'$.  Further
	assume that $v'$ is a vertex in $G'$ of $r$-type $\tau$.
	By choice of the set $\{v_1,\dots,v_k\}$  we altered the isomorphism type of each vertex of type $\tau$ in $G$. Therefore
	$\mathcal{N}_r^{G'}(v')\not=\mathcal{N}_{r}^G(v')$.
	Therefore $\mathcal{N}_r^{G'}(v')$  contains an
	inserted edge $(v_i,v_i')$. We will first prove that either
	$\dist_{G'}(v',v_i)<r$ or $\dist_{G'}(v',v_i')<r$. Assume that this is not the case. Then there is a path 
	$P=$\linebreak $(v_i=w_{-r},w_{-r+1},\dots,w_{-1},w_0=v',w_1,\dots,w_{r-1},w_r=v_i')$
	such that $w_j\not=v_i$ and $w_j\not=v_i'$ for all $-r<j<r$. Let $-r\leq
	j<r$ be the largest index such that $w_j\in
	\{v_1,\dots,v_k,v_1',\dots,v_k'\}$. Then the path $(w_j,\dots,w_r=v_i')$ is
	a path in $G$ of length $\leq 2r$, which contradicts the choice of  $v_1,\dots,v_k,v_1',\dots,v_k'$. Since
	$\deg_{G'}(v_i)=\deg_{G'}(v_i')=\tilde{d}+1$, this implies that
	$\mathcal{N}_{r-1}^G(v')$ contains a vertex of degree $\tilde{d}+1$,
	which contradicts that $v'$ has $r$-type
	$\tau$. Hence $G'$ is
	$\tau$-neighbourhood free.
\end{proof}
\begin{lemma}\label{lemma:farImpliesManyCounterexamples-forbidden$r$-ball-OnlyDegree$d$}
	For $r\geq 1$ let $\tau$ be an $r$-type. Let $B$ be an $r$-ball with constant $a$ of type $\tau$. Assume $\deg_B(v)=d$ for all vertices $v\in N^{B}_{r-1}(a)$. Let $\epsilon\in (0,1]$ be fixed and let $\lambda=\epsilon$. Then every graph $G\in \mathcal{C}_{d}$ on $n\geq 1$ vertices which is $\epsilon$-far from being $\tau$-neighbourhood free contains more than $\lambda n$ vertices of $r$-type $\tau$.
\end{lemma}
\begin{proof}
	If $d=0$, then the Lemma holds. 
	Hence we can assume that $B$ is not just an isolated vertex.
	We proceed by contraposition. Assume $G\in \mathcal{C}_{d}$ is a graph on $n\geq 1$
	vertices containing no more than $\lambda n $ vertices $v$ of $r$-type
	$\tau$. Let $G'$ be the graph obtained from $G$ by
	isolating all vertices $v$ of $r$-type $\tau$. First note
	that $G'$ is $\epsilon$-close to $G$ since we did no more than $d\lambda n
	\leq \epsilon dn$ edge modifications. Now assume that $v'$ is a vertex of
	$r$-type $\tau$. Since we isolated all vertices having
	$r$-type $\tau$ we know that
	$\mathcal{N}_{r}^{G'}(v')\not=
	\mathcal{N}_{r}^G(v')$.  Therefore there
	must be a vertex $v$ in $N_{r}^G(v')$ such that
	$v$ has type $\tau$, because otherwise the
	$r$-ball of $v'$ could not witness any of the edge modifications. This means
	that there is a path $(v'=v_0,v_1,\dots,v_{k-1},v_k=v)$ of length $k\leq r$
	in $G$. Now pick the maximum index $i$ such that
	$\dist_{G'}(v',v_i)<\infty$. First observe that because $v=v_k$ is isolated
	in $G'$ we get that $i<k$ and therefore $\dist_{G'}(v',v_i)<r$.  Since
	$\dist_{G'}(v',v_{i+1})=\infty$ by construction and $\{v_i,v_{i+1}\}\in
	E(G)$, we get
	$\deg_{\mathcal{N}_{r}^{G'}(v')}(v_i)=\deg_{G'}(v_i)<\deg_G(v_i)\leq
	d$. Since $\mathcal{N}_{r}^{G'}(v')\in \tau$ this yields a
	contradiction to our previous assumption that all vertices in
	$\mathcal{N}^{B}_{r-1}{a}$ have degree $d$. Hence the graph $G'$ can not
	contain a vertex $v'$ of $r$-type $\tau$ and is therefore
	$\tau$-neighbourhood free. 
\end{proof}
The next Lemma follows from Lemmas~\ref{lemma:farImpliesManyCounterexamples-forbidden$r$-ball-OnlyDegree$d$} and~\ref{lemma:farImpliesManyCounterexamples-forbidden$r$-ball-oneDegreeMissing}
for $r=1$, since the $0$-ball has one vertex.
\begin{lemma}\label{lemma:farImpliesManyCounterexamples-forbidden1-ball}
	Let  $\tau$ be a $1$-type. Let $\epsilon\in (0,1]$ be fixed, $n_0={2d^2}/{\epsilon}$ and $\lambda={\epsilon d}/(14(1+d^{3}))$. Then every graph $G\in \mathcal{C}_{d}$ on $n\geq n_0$ vertices which is $\epsilon$-far from being $\tau$-neighbourhood free contains more than $\lambda n$ vertices of $1$-type $\tau$.\qed
\end{lemma}

\begin{proof}[Proof of Theorem~\ref{thm:dNeighbourhoodFreeness}]
	Lemma~\ref{lemma:testerFramework} with $\F=\{\tau\}$   and $M=\emptyset$ combined with Lemma~\ref{lemma:farImpliesManyCounterexamples-forbidden$r$-ball-oneDegreeMissing} proves Theorem~\ref{thm:dNeighbourhoodFreeness} in all cases apart from when $d=1$. In case $d=1$ we have $\tilde{d}=0$. In this case we set $M:=\{n\in\mathbb{N}\mid n\equiv 1\mod 2 \}$ and get that for $\epsilon\in (0,1]$ and  $\lambda=\epsilon$ we have that every graph $G\in \mathcal{C}_{d}$ on $n\equiv 0\mod 2$ vertices which is $\epsilon$-far from being $\tau$-neighbourhood free contains more than $\lambda n$ vertices of $r$-type $\tau$. This is the case as assuming the number of vertices of $r$-type $\tau$ is no more than $\lambda n$ we can add an edge between any pair of vertices of degree $0$, obtaining a graph $G'$ which is $\lambda n\leq \epsilon dn$ close to $G$. 
\end{proof}

Theorem~\ref{thm:1NeighbourhoodFreeness} follows from Lemma~\ref{lemma:testerFramework} and Lemma~\ref{lemma:farImpliesManyCounterexamples-forbidden1-ball} where in Lemma~\ref{lemma:testerFramework} we use either $\F=\{\tau\}$ or $\F=\emptyset$ depending on whether $\tau$ has degree bounded by $d$.
\begin{proof}[Proof of Theorem \ref{thm:neighbourhoodRegularity}] Let $\tau$ be a $1$-type such that   $B\setminus \{a\}$ is a union of disjoint cliques for all $(B,a)\in \tau$ as in the statement of the theorem.
	We define $\mathcal{P}$ to be the property of being $\tau$-neighbourhood regular and let $\setOfMaxCliques$ be the set of maximal cliques in $G$. Let us define the function $\maxcl^G:V(G)\times\mathbb{N}\rightarrow \mathbb{N}$ where $\maxcl^G(v,i):=|\{K\in \setOfMaxCliques\mid |K|=i,v\in K\}|$  is the number of maximal $i$-cliques containing $v$.
	\begin{claim}\label{claim:DefiningMFor1NieghbourhoodRegularity}
		If $G\in \mathcal{P}$ then $\maxcl^B(a,i)\cdot n\equiv 0\mod i$. 
	\end{claim}
	\begin{proof}
		First note that $G\in \mathcal{P}$ implies that $\mathcal{N}_1^G(v)\in \tau$ for all $v\in V(G)$. Then $\maxcl^B(a,i)=\maxcl^G(v,i)$ for all $v\in V(G)$ and 
		\[\maxcl^B(a,i)\cdot n  =\sum_{v\in V}\maxcl^G(v,i) =|\{K\in\setOfMaxCliques\mid |K|=i \}|\cdot i\equiv 0 \mod i.\]
	\end{proof}
	Let $M:=\{n\in \mathbb{N}\mid \text{there is }1\leq i \leq d \text{ such that }\maxcl^B(a,i)\cdot n\not\equiv 0 \mod i\}$. 
	Note that deciding whether $n\in M$ 
	only requires standard arithmetic operations.
	\begin{claim}\label{claim:farImpliesManyCounterexamples-negihbourhoodRegularity}
		For $\epsilon\in (0,1]$ let $\lambda={\epsilon}/(20d^6) $ and $n_0=20 d^8$.  Than any graph $G\in \mathcal{C}_{d}$ on $n\geq n_0$, $n\notin M$ vertices, which is $\epsilon$-far from $\mathcal{P}$, contains more than $\lambda n$ vertices $v$ which are not of $1$-type $\tau$. 
	\end{claim}
	\begin{proof}
		We proceed by contraposition. Let $G\in \mathcal{C}_{d}$ be a graph on $n\geq n_0$, $n\notin M$ vertices and assume that $G$ contains $\leq \lambda n$ vertices which are not of $1$-type $\tau$. We will now describe an algorithmic procedure which takes $<\epsilon d n$ edge modifications and transforms $G$ into a graph $G^{(4)}\in \mathcal{P}$, which will prove the claim. 
		
		Let $\tilde{E}^{(1)}:=\{e\in E(G)\mid \text{there are distinct } K, K'\in \setOfMaxCliques, |K\cap K'|>1,e\subseteq K\}$.  Let $G^{(1)}$ be the graph $G^{(1)}=(V(G), E^{(1)})$, where $E^{(1)}=E\setminus\tilde{E}^{(1)}$. First note that $G^{(1)}$ has no distinct maximal cliques $K,K'$ with $|K\cap K'|>1$. Furthermore 
		\begin{align*}
		|\tilde{E}^{(1)}|&\leq {d\choose 2}\cdot |\{K \in \setOfMaxCliques \mid \text{exists } K'\in \setOfMaxCliques , |K\cap K'|>1 \}|\leq \frac{d^3\lambda n}{2},
		\end{align*}
		where the second inequality holds because every $K\in \setOfMaxCliques$ such that there is $K'\in\setOfMaxCliques$ with $|K\cap K'|>1$ and $K\neq K'$ must contain one of the $\lambda n$ vertices $v$ which are not of $1$-type $\tau$ and there are $\leq d\lambda n$ maximal cliques containing such a vertex. In addition, the removal of the edges in $\tilde{E}^{(1)}$ will affect no more than $d^4\lambda n$ vertices because there are no more than $d^3\lambda n$ vertices contained within an edge of $\tilde{E}^{(1)}$, each of their $1$-neighbourhoods contains $d$ vertices and any vertex, whose $1$-neighbourhood is affected, must be of distance $1$ to one of the vertices contained in an edge in $\tilde{E}^{(1)}$. Hence $G^{(1)}$ contains $\leq (d^4+1)\lambda n<2d^4\lambda n$ vertices $v$ which are not of $1$-type $\tau$.

		Note that in $G^{(1)}$ for all vertices $v$ the graph $\mathcal{N}^{G^{(1)}}_{1}(v)\setminus \{v\}$ is a disjoint union of cliques but there might be $K\in \setOfMaxCliques[G^{(1)}]$ such that $\maxcl^B(a,|K|)=0$. We define the edge set  $\tilde{E}^{(2)}:=\{e\in E^{(1)}\mid \text{exists }K\in \setOfMaxCliques[G^{(1)}], e\subseteq K,\maxcl^B(a,|K|)=0\}$ and let $G^{(2)}$ be the graph $G^{(2)}=(V(G), E^{(2)}))$, where $E^{(2)}=E^{(1)}\setminus\tilde{E}^{(2)}$.
		Furthermore 
		\begin{align*}
		|\tilde{E}^{(2)}|&\leq d\cdot |\{v\mid \text{exists }K\in \setOfMaxCliques[G^{(1)}], v\in K,\maxcl^B(a,|K|)=0\}|\leq d\cdot 2d^4\lambda n,
		\end{align*} 
		where the first inequality holds because every clique in $G^{(1)}$ has size $\leq d$ and the second because $\mathcal{N}^{G^{(1)}}_{1}(v)\notin \tau$ for every $v\in \{v\mid \text{exists }K\in \setOfMaxCliques[G^{(1)}], v\in K,\maxcl^B(a,|K|)=0\}$.
		
		Note that $\maxcl^B(a,|K|)\not= 0$ for all $K\in
		\setOfMaxCliques[G^{(2)}]$, but there might be $v\in V(G)$ and
		$i\leq d$ with $\maxcl^B(a,i)\not=\maxcl^{G^{(2)}}(v,i)$. Moreover,
		note that because $n\geq n_0$ there are at least $2d$ $4$-balls in
		$G^{(2)}$ which are completely disjoint from the $4$-balls of any vertex
		$v$ of $1$-type $\tau$. $G^{(3)}$ will also have this
		property.
		Let $G^{(3)}=(V(G),E^{(3)})$ be the graph obtained from $G^{(2)}$ by the following operations.
		For every pair $v,v'$ such that there is $i\leq d$ with
		$\maxcl^B(a,i)>\maxcl^{G^{(2)}}(v,i)$ and
		$\maxcl^B(a,i)<\maxcl^{G^{(2)}}(v',i)$, let $w$ be a vertex of type
		$\tau$ which has at least distance $4$ to $v$ and to $v'$.
		Let $ K'=\{v'_1,\dots,v'_{i-1},v'\}\in \setOfMaxCliques[G^{(2)}]$ and
		$K=\{v_1,\dots,v_{i-1},w\}\in \setOfMaxCliques[G^{(2)}]$. Delete the
		edges $\{\{v',v'_j\},\{w,v_j\}\mid  j\in [i-1]\}$ and add the edges
		$\{\{v,v_j\}, \{w,v'_j\}\mid j\in [i-1]\}$. Note that the vertices
		$v_1,\dots,v_{i-1},v'_1,\dots,v'_{i-1},w$ are still contained in the same
		number of cliques as before, while $v$ is contained in one additional
		$i$-clique and $v'$ is contained in one less. 
		
		Note that in $G^{(3)}$, it holds that either $\maxcl^B(a,i)\geq
		\maxcl^{G^{(3)}}(v,i)$ for all vertices $v$, or $\maxcl^B(a,i)\leq
		\maxcl^{G^{(3)}}(v,i)$ for all $v$ for every $i\in \{1,\dots,d\}$.
		Let $G^{(4)}$ be the graph obtained from $G^{(3)}$ by the following
		operations. For every $i$ such that there is a vertex $v$ with
		$\maxcl^B(a,i)<\maxcl^{G^{(3)}}(v,i)$, we pick $i$ not necessarily
		distinct vertices 
		$v_1,\dots,v_i$ with the following property. In case a vertex $v$ appears $k$ times amongs $v_1,\dots, v_i$ then 
		$\maxcl^{G^{(3)}}(v,i)-\cdot\maxcl^B(a,i)\geq k$. Note that
		these choices are possible because $\sum_{v\in
			V(G^{(3)})}\maxcl^{G^{(3)}}(v,i)\equiv 0 \mod i$ and $\maxcl^B(a,i)\cdot
		n\equiv 0 \mod i$ by assumption $n\notin M$ and hence we have $\sum_{v\in
			E^{(3)}}(\maxcl^{G^{(3)}}(v,i)-\maxcl^B(a,i))\equiv 0 \mod i$. Let
		$K_1,\dots,K_i\in \setOfMaxCliques[G^{(3)}]$  be distinct cliques such that $v_j\in K_j$ for every
		$1\leq j\leq i$. Let $K=\{w_1,\dots,w_i\}\in \setOfMaxCliques[G^{(3)}]$
		such that the distance between any pair $v_j,w_k$ is at least $4$. Remove
		the set of edges $\{\{w_j,w_k\},\{v_j,v\}\mid v\in K_j, j,k\in [i]\}$ and
		add the set of edges $\{\{w_j,v\}\mid v\in K_j,j\in [i]\}$. Note that
		this reduces the number of maximal $i$-cliques $v_1,\dots,v_i$ are in by
		one, while leaving the number of cliques $w_1,\dots,w_i$ are in the
		same.
		Similarly, for  every $i$  such that there is a vertex $v$ with
		$\maxcl^B(a,i)>\maxcl^{G^{(3)}}(v,i)$ we pick $i$ not necessarily
		distinct vertices 
		$v_1,\dots,v_i$ with the following property. In case $v$ appears $k$ time amongst $v_1,\dots,v_i$ we have
		$\maxcl^B(a,i)-\maxcl^{G^{(3)}}(v,i)\leq k$. Let
		$w_1,\dots,w_i$ be vertices with $\maxcl^B(a,i)=\maxcl^{G^{(3)}}(w_j,i)$
		such that $w_1,\dots,w_i$ are of distance at least $4$ from every $v_j$, $1\leq j\leq i$, and $w_1,\dots,w_i$ are pairwise of distance at least
		$4$. 
		Let $K_j\in \setOfMaxCliques[G^{(3)}]$ with $w_j\in K_j$ for $ j\in
		\{1,\dots,i\}$. Remove the set of edges  $\{\{w_j,w\}\mid w\in K_j,1\leq j\leq i\}$
		and add the set of edges $\{\{v_j,w\}\{w_j,w_k\}\mid w\in K_j,j,k\in
		\{1,\dots,i\}\}$. Note that this adds one to the number of $i$-cliques
		$v_1,\dots,v_i$ are in, while leaving the number of $i$-cliques
		$w_1,\dots,w_i$ are in the same.
		
		By construction $G^{(4)}\in \mathcal{P}$. The number of edge modifications in
		total  is $|E^{(1)}|+|E^{(2)}|$ plus the number of modifications it takes
		to transform $G^{(2)}$ into $G^{(4)}$. First note that 
		\begin{displaymath}
		\sum_{i=3}^{d}\sum_{v\in
			V(G)}|\maxcl^B(a,i)-\maxcl^{G^{(2)}}(v,i)|\leq 2d \cdot
		2d^4\lambda n
		\end{displaymath}
		 since every of the $2d^4\lambda n$ vertices $v$ in
		$G^{(2)}$ of $1$-type $\tau$ can contribute at most $2d$ to
		the sum above. Since transforming $G^{(2)}$ into $G^{(4)}$ we proceed
		greedily, meaning we reduce the number $\sum_{i=3}^{d}\sum_{v\in
			V(G)}|\maxcl^B(a,i)-\maxcl^{G^{(2)}}(v,i)|$ by at least one in
		every step, and every such reduction takes a maximum of $4d^2$ edge
		modifications in total we need less than 
		\begin{displaymath}
		|E^{(1)}|+|E^{(2)}|+4d^2\cdot 2d\cdot 2d^4\lambda n\leq 20 d^7\lambda n=\epsilon d n
		\end{displaymath}
		edge modifications. 
	\end{proof}
	Let $\mathcal{F}:=\{\tau'\mid \tau \text{ is a }1\text{-type },\tau\not=\tau'\}$. Note that $|\mathcal{F}|\leq t<\infty$, where equality occurs when $B\notin \mathcal{C}_{d}$. Then Claim~\ref{claim:farImpliesManyCounterexamples-negihbourhoodRegularity} combined with Lemma~\ref{lemma:testerFramework} for $M$ and $\mathcal{F}$ defined as above proves the Theorem. 
\end{proof}

\section{Conclusion}\label{sec:conclusion}
We studied testability of properties definable in first-order logic in the bounded-degree model of property testing for graphs and relational structures, where \emph{testability} of a property means if it is testable with constant query complexity. We showed that all properties in $\Sigma_2$ are testable (Theorem~\ref{thm:sigma2}), 
and we complemented this by exhibiting a property (of relational structures) in $\Pi_2$ that is not testable (Theorem~\ref{thm:pi2}). 
Using a hardness reduction, we also exhibit a property of undirected, 3-regular graphs in $\Pi_2$ that is not testable (Theorem~\ref{thm:simpleDelta2}). The question whether first-order definable properties are testable with
a \emph{sublinear} number of queries  (e.g. $\sqrt{n}$) in the bounded-degree model is left open. 

Similar results (on the separation between $\Sigma_2$ and $\Pi_2$ properties) were obtained in the dense graph model in~\cite{alon2000efficient}, albeit with very different methods. Indeed, non-testability of first-order logic in the bounded-degree model is somewhat unexpected: Testing algorithms proceed by sampling vertices and then exploring their local neighbourhoods, and
it is well-known that first-order logic can only express `local' properties. On
graphs and structures of bounded degree this is witnessed by Hanf's strong normal form of first-order logic~\cite{Hanf1965}, which is built around the absence and presence of different isomorphism types of local neighbourhoods.
However, our negative result shows that locality of first-order logic is not sufficient for testability. This also answers an open question from~\cite{AdlerH18}. 

We obtained our non-testable properties by 
encoding the zig-zag construction of bounded-degree expanders into first-order logic on
relational structures (Theorem~\ref{thm:nonTestabilityForStructures}) and then extending this to undirected graphs (Theorem~\ref{thm:simpleDelta2}). 
We believe that this will be of independent interest.
We remark that it might also be possible to use the iterative construction of replacement product graphs of~\cite{Reingold00entropywaves}  
instead of the zig-zag construction to obtain a similar example.

We then used our non-testable graph property to answer a question on \emph{proximity oblivious} testers in the bounded-degree model, asked 
 by Goldreich and Ron more than 10 years ago~\cite{goldreich2011proximity}. Such a tester is particularly simple: it performs a basic test a number of times that may depend on the proximity parameter, 
whereas the basic test is oblivious of the parameter. In~\cite{goldreich2011proximity}, the properties that are testable in this model have been characterised as those that are both \emph{GSF-local}, and \emph{non-propagating}. Roughly speaking, \emph{GSF-local} means that the graph class omits a family of \emph{generalised subgraphs} (\ie subgraphs with constraints on how the subgraphs interact with the rest of the graph), and
\emph{non-propagating} means that 
graphs in which a forbidden generalised subgraph is unlikely to be detected by sampling vertices
are actually close to having the property in terms of edge modifications. In other words, no `chain reactions' occur, where repairing one edge will produce new unwanted configurations that again need repairing, etc.
Goldreich and Ron asked, whether `non-propagating' is necessary.
We showed that this is the case.
Our proof is based on relating first-order definable properties to GSF-local properties, via a notion that we call neighbourhood profiles, which captures first-order definability.
    Finally, we took an approach suggested by Hanf's normal form, and we proved testability of some 
	 first-order properties that speak about isomorphism types of neighbourhoods, including testability 
	 of $1$-neighbourhood-freeness, and  $r$-neighbourhood-freeness under a mild assumption on the 
	 degrees (Theorem~\ref{thm:dNeighbourhoodFreeness}, Theorem~\ref{thm:1NeighbourhoodFreeness}, and Theorem~\ref{thm:neighbourhoodRegularity}). 
In particular, these theorems imply that there are properties defined by formulas in $\Pi_2\setminus \Sigma_2$ which are testable. Since subgraph-freeness and subgraph containment are testable, Hanf's normal form suggests 
studying testability of (negated) Hanf sentences, i.\,e.\ neighborhood properties, as a next step. 
While Hanf sentences are trivially testable, we pose as an open problem whether every negated Hanf sentence is testable.

Let us remark that not all $\Sigma_2$-properties are GSF-local. Indeed, subgraph containment can be expressed in $\Sigma_1$,  
 while it is not a GSF-local property.
Also observe that our testers for neighbourhood properties in Section~\ref{sec:freeness} have \emph{one-sided error}, i.\,e.\ the testers always accept the graphs that satisfy the property. Finally, note that, in contrast to subgraph-freeness and induced subgraph-freeness, the properties $\tau$-neighbourhood regularity and $\tau$-neighbourhood-freeness are neither \emph{monotone} nor \emph{hereditary}, which are properties that are closed under edge deletion and closed under vertex deletion, respectively.  
As we mentioned before, Ito et al. \cite{ito2019characterization} recently characterised one-sided error (constant-query) testable monotone and hereditary graph properties in the bounded-degree (directed and undirected) graph model. In order to give a full characterisation of one-sided error testable properties in the bounded-degree graph model, it is important to take a step beyond monotone and hereditary graph properties. 

\vspace{1em}

\appendix
\section{Deferred Proofs from Section~\ref{sec:expansionOfGraphProperty}: Alternative proof of Theorem \ref{thm:simpleDelta2}}
\label{app:C}
Now we make use of a result from \cite{fichtenberger2019every} and our Lemma \ref{lemma:undirected_expander} that the models of $\graphFormula$ is a family of expanders to give an alternative proof of Theorem \ref{thm:simpleDelta2}.
We first introduce a definition of ``hyperfinite graphs''.
\begin{definition}\label{def:hyperfinite}
	Let $\varepsilon\in (0,1]$ and $k\geq 1$. A graph $G$ with maximum degree bounded by $d$ is called \emph{$(\varepsilon,k)$-hyperfinite} if one can remove at most $\varepsilon d |V|$ edges from $G$ so that each connected component of the resulting graph has at most $k$ vertices. For a function $\rho: (0,1]\rightarrow\mathbb{N}^+$, a graph $G$ is called \emph{$\rho$-hyperfinite} if $G$ is $(\varepsilon,\rho(\varepsilon))$-hyperfinite for every $\varepsilon>0$. A set (or property) $\Pi$ of graphs is called \emph{$\rho$-hyperfinite} if every graph in $\Pi$ is $\rho$-hyperfinite. A set (or property) $\Pi$ of graphs is called \emph{hyperfinite} if it is $\rho$-hyperfinite for some function $\rho$. 
\end{definition}

Now we recall that a graph property is a set of graphs that is invariant under graph isomorphism. A subproperty of a property $P$ is a subset of graphs in $P$ that is also invariant under graph isomorphism.
\begin{lemma}[Corollary 1.1 in \cite{fichtenberger2019every}]\label{lem:FPS19}
	Let $C_d$ be the class of graphs of bounded maximum degree $d$. Let $P\subseteq C_d$ be a property that does not contain an infinite hyperfinite subproperty, and let $P'\subseteq C_d$ be arbitrary property such that $P\cap P'$ is an infinite set. Then $P\cap P'$ is not testable.    
\end{lemma}

Now we are ready to give another proof of Theorem \ref{thm:simpleDelta2}.
\begin{proof}[(An Alternative) Proof of Theorem \ref{thm:simpleDelta2}]
	We show that the property $\graphProp$ does not contain an infinite hyperfinite subproperty. If this is true, then by applying Lemma \ref{lem:FPS19} with $P=\graphProp$ and $P'=C_d$, we have that $\graphProp$ is not testable. This will then finish the proof of Theorem \ref{thm:simpleDelta2}.
	
	Suppose towards contradiction that $\graphProp$ contains an infinite hyperfinite subproperty. That is, there exists an infinite subset $Q\subseteq \graphProp$ and a function $\rho : (0, 1] \rightarrow \mathbb{N}$ such that $Q$ is $(\varepsilon, \rho(\varepsilon))$-hyperfinite for every $\varepsilon > 0$. 	
	That is, for any graph $G=(V,E)\in Q$, for any $\varepsilon>0$, we can remove $\varepsilon d |V|$ edges from $G$ so that each connected component of the resulting graph has at most $\rho(\varepsilon)$ vertices. Now let $\varepsilon$ be an arbitrarily small constant such that $\rho(\varepsilon)\ll |V|$ and that $\varepsilon\leq \frac{\xi}{100d}$, where $\xi$ is the constant from Lemma \ref{lemma:undirected_expander}. Let $V_1,V_2,\dots$ be a vertex partitioning of $V$ such that $|V_i|\leq \rho(\varepsilon)$ and the number of edges crossing different parts is at most $\varepsilon d |V|$. Let $S$ be a vertex subset that is a union of the first $j$ parts $V_1,\cdots,V_j$ such that $|\cup_{i\leq j-1} V_i|< \frac{|V|}{3}$ and $|\cup_{i\leq j} V_i|\geq\frac{|V|}{3}$. Note that such a set always exists as $|V_i|\leq \rho(\varepsilon)\ll |V|$ and furthermore, $|S|=|\cup_{i\leq j} V_i|<\frac{|V|}{2}$. Now on one hand, $|\langle S, \bar{S}\rangle|$ is at most the number of edges crossing different parts and thus at most $\varepsilon d |V|$. On the other hand, since $G\in \graphProp$, $G$ is a $\xi$-expander for some constant $\xi$ from Lemma \ref{lemma:undirected_expander}. Thus, $|\langle S, \bar{S}\rangle|\geq \xi \frac{|V|}{3} > \varepsilon d |V|$, which is a contradiction by our setting of $\varepsilon$. Therefore,  $\graphProp$ does not contain an infinite hyperfinite subproperty. This finishes the proof of the theorem. 
\end{proof}

\section{Deferred Proofs from Section~\ref{sec:charBy0Profiles}}
\label{app:B}
\begin{claim}\label{claim:satisfyingTree}
	Every structure $\struc{A}\in \bigcup_{1\leq k \leq m}\classStruc{P}_{\rho_k}\setminus \{\struc{A}_{\emptyset}\}$ satisfies $\varphi_{\operatorname{tree}}$.
\end{claim}
\begin{proof}
	Let $\struc{A}\in \bigcup_{1\leq k \leq m}\classStruc{P}_{\rho_k}\setminus \{\struc{A}_{\emptyset}\}$. Then there is $k\in \{1,\dots,m\}$ such that $\struc{A}\in \classStruc{P}_{\rho_k}$.

	By definition, $\varphi_{\operatorname{tree}}:= \exists^{\leq 1} x \varphi_{\operatorname{root}}(x)\land \varphi \land \forall x (\psi(x)\lor \chi(x))$,  where
	\begin{align*}
		\varphi:= &\forall x \Big(\big(\varphi_{\operatorname{root}}(x)\land R(x,x)\big)\lor  
		\big(\exists^{=1} y F(y,x)\land \lnot \exists y R(x,y)\land \lnot \exists y R(y,x)\big)\Big),
		\\\psi(x):= &\neg\exists y F(x,y)\land \bigwedge_{k\in \indexSetH} L_k(x,x)\land \forall y \Big(y\not= x \rightarrow 
		\\& \bigwedge_{k\in \indexSetH}\lnot L_k(x,y) \wedge \bigwedge_{k\in \indexSetH}\lnot L_k(y,x)\Big)
	\end{align*}
and
\begin{align*}	
		\chi(x):=&\lnot\exists y \bigvee_{k\in \indexSetH}\big(L_k(x,y)\lor L_k(y,x)\big)\land 
		\bigwedge_{k\in \indexSetH}\exists y_{k} \Big(x\not=y_{k}\land F_{k}(x,y_{k})\\
		&\land 
		(\bigwedge_{k'\in \indexSetH,k'\not=k}\lnot F_{k'}(x,y_k))\land \forall y(y\not=y_k\rightarrow \lnot F_{k}(x,y))\Big).
	\end{align*}
	Thus, it is sufficient to prove that $\struc{A}\models \exists^{\leq 1} x \varphi_{\operatorname{root}}(x)$, $\struc{A}\models \varphi$ and $\struc{A}\models\forall x (\psi(x)\lor \chi(x))$. 
	
	To prove $\struc{A}\models \exists^{\leq 1} x \varphi_{\operatorname{root}}(x)$ we note that by construction of $\rho_k$ we have $\struc{A}\not\models \varphi_{\operatorname{root}}(a)$ 
	for any $a\in \univ{A}$ for which $(\mathcal{N}_2^{\struc{A}}(a),a)\notin \tau_{d,2,\sigma}^k$. Since $\rho_k$ restricts the number of occurrences of elements of neighbourhood type $\tau_{d,2,\sigma}^{k}$ to at most one, this proves that there is at most one $a\in \univ{A}$ with $\struc{A}\models \varphi_{\operatorname{tree}}(a)$ and hence $\struc{A}\models \exists^{\leq 1} x \varphi_{\operatorname{root}}(x)$. 
	
	To prove $\struc{A}\models \varphi$, let $a\in \univ{A}$ be an arbitrary element. 
	Since $\struc{A}\in \classStruc{P}_{\rho_k}$, there is an $i\in I_{k}$ such that $(\mathcal{N}_2^{\struc{A}}(a),a)\in \tau_{d,2,\sigma}^i$. 
	But then by definition, there exist $\other{\struc{A}}\models \varphi_{\zigzag}$ and $\other{a}\in \univ{\other{A}}$ such that $(\mathcal{N}_2^{\struc{A}}(a),a)\cong (\mathcal{N}_2^{\other{\struc{A}}}(\other{a}),\other{a})$. 
	Assume $f$ is an isomorphism from  $(\mathcal{N}_2^{\struc{A}}(a),a)$ to  $(\mathcal{N}_2^{\other{\struc{A}}}(\other{a}),\other{a})$. First consider the case that $\struc{A}\models \varphi_{\operatorname{root}}(a):=\forall y \lnot F(y,a)$. Assume there  is $\other{b}\in \univ{\other{A}}$ such that $(\other{b},\other{a})\in \rel{F}{\other{\struc{A}}}$. Since $\other{b}\in N_2^{\other{\struc{A}}}(\other{a})$, there must be an element $b\in N_2^{\struc{A}}(a)$ such that $f(b)=\other{b}$. Since $f$ is an isomorphism mapping $a$ to $\other{a}$, this implies $(b,a)\in \rel{F}{\struc{A}}$, which contradicts $\struc{A}\models \varphi_{\operatorname{root}}(a)$. 
	Hence $\other{\struc{A}}\models \varphi_{\operatorname{root}}(\other{a})$. Since $\other{\struc{A}}\models \varphi_{\operatorname{tree}}'$, it holds that $\other{\struc{A}}\models \varphi$, which means that $(\other{a},\other{a})\in \rel{R}{\other{\struc{A}}}$. But since $f$ is an isomorphism mapping $a$ onto $\other{a}$, this implies $(a,a)\in \rel{R}{\struc{A}}$. Now consider the case that  
	$\struc{A}\not\models \varphi_{\operatorname{root}}(a)$. Then there is $b\in \univ{A}$ with $(b,a)\in \rel{F}{\struc{A}}$. Since $f$ is an isomorphism, this implies $(f(b),\other{a})\in \rel{F}{\other{\struc{A}}}$. Hence $\other{\struc{A}}\models \exists^{=1} y F(y,\other{a})\land \lnot \exists y R(\other{a},y)\land \lnot \exists y R(y,\other{a})$, as $\other{\struc{A}}\models \varphi$. Now assume that  there is $b'\not=b$ such that $(b',a)\in \rel{F}{\struc{A}}$. Then $f(b)\not=f(b')$ and $(f(b),\other{a}),(f(b'),\other{a})\in \rel{F}{\other{\struc{A}}}$. Since this contradicts $\other{\struc{A}}\models \exists^{=1} y F(y,\other{a})$ we have $\struc{A}\models \exists^{=1} y F(y,a)$. Furthermore, assume that there is $b'\in \univ{A}$ such that either $(a,b')\in \rel{R}{\struc{A}}$ or $(b',a)\in \rel{R}{\struc{A}}$. Then either $(\other{a},f(b'))\in \rel{R}{\other{\struc{A}'}}$ or $(f(b'),\other{a})\in \rel{R}{\other{\struc{A}}}$, which contradicts $\other{\struc{A}}\models \lnot \exists R(\other{a},y)\land \lnot \exists y R(y,\other{a})$. Therefore $\struc{A}\models \lnot \exists y R(a,y)\land \lnot \exists y R(y,a)$ which completes the proof of $\struc{A}\models\varphi$.  
	
	We prove $\struc{A}\models \forall x (\psi(x)\lor \chi(x))$ by considering the two cases $\struc{A}\models \neg\exists y F(a,y)$ and $\struc{A}\models \exists yF(a,y)$ for each element $a\in \univ{A}$. For this, let $a\in \univ{A}$ be any element. By the construction of $\rho_k$ there is $\other{\struc{A}}\models \varphi_{\zigzag}$ and $\other{a}\in \univ{\other{A}}$ such that $(\mathcal{N}_2^{\struc{A}}(a),a)\cong (\mathcal{N}_2^{\other{\struc{A}}}(\other{a}),\other{a})$. Let $f$ be an isomorphism from  $(\mathcal{N}_2^{\struc{A}}(a),a)$ to $(\mathcal{N}_2^{\other{\struc{A}}}(\other{a}),\other{a})$. First consider the case that $\struc{A}\models \neg\exists y F(a,y)$. If there was $\other{b}\in \univ{\other{A}}$ with $(\other{a},\other{b})\in \rel{F}{\other{\struc{A}}}$ then $(a,f^{-1}(\other{b}))\in \rel{F}{\struc{A}}$ contradicting our assumption. Hence $\other{\struc{A}}\models \neg\exists y F(\other{a},y)$  which implies that $\other{\struc{A}}\not\models \chi(\other{a})$. But since $\other{\struc{A}}\models \varphi_{\zigzag}$, it holds that $\other{\struc{A}}\models \forall x (\psi(x)\lor \chi(x))$, which implies that $\other{\struc{A}}\models \psi(\other{a})$. Hence $(\other{a},\other{a})\in \rel{L_k}{\other{\struc{A}}}$ for every $k\in \indexSetH$. Since $f$ is an isomorphism and $f(a)=\other{a}$, it holds that $(a,a)\in \rel{L_k}{\struc{A}}$ for every $k\in \indexSetH$, and hence $\struc{A}\models \bigwedge_{k\in \indexSetH} L_k(a,a)$. Furthermore, assume that there is $b\in \univ{A}$, $b\not=a$ and $k\in \indexSetH$ such that either $(a,b)\in \rel{L_k}{\struc{A}}$ or $(b,a)\in \rel{L_k}{\struc{A}}$. Since $f$ is an isomorphism this implies that either $(\other{a},f(b))\in \rel{L_k}{\other{\struc{A}}}$ or $(f(b),\other{a})\in \rel{L_k}{\other{\struc{A}}}$ which contradicts $\other{\struc{A}}\models \chi(\other{a})$. Hence $\struc{A}\models\forall y \Big(y\not= a \rightarrow \bigwedge_{k\in \indexSetH}\lnot L_k(a,y) \wedge \bigwedge_{k\in \indexSetH}\lnot L_k(y,a)\Big)$ proving that $\struc{A}\models \psi(a)$. 
	
	Now consider the case that there is an element $b\in \univ{A}$ such that $(a,b)\in \rel{F}{\struc{A}}$. Since this implies that $(\other{a},f(b))\in \rel{F}{\other{\struc{A}}}$, we get that $\other{\struc{A}}\not\models \psi(\other{a})$, and hence $\other{\struc{A}}\models \chi(\other{a})$. Now assume that there is  $b\in \univ{A}$ and $k\in \indexSetH$ such that either $(a,b)\in \rel{L_k}{\struc{A}}$ or $(b,a)\in \rel{L_k}{\struc{A}}$. But then either $(\other{a},f(b))\in \rel{L_k}{\other{\struc{A}}}$ or $(f(b),\other{a})\in \rel{L_k}{\other{\struc{A}}}$, which contradicts $\other{\struc{A}}\models \chi(\other{a})$. Hence $\struc{A}\models  \lnot\exists y \bigvee_{k\in \indexSetH}\big(L_k(a,y)\lor L_k(y,a)\big)$. For each $k\in \indexSetH$, let $\other{b}_k\in \univ{\other{A}}$ be an element such that $\other{\struc{A}}\models \other{a}\not=\other{b}_{k}\land F_{k}(\other{a},\other{b}_{k})\land 
	(\bigwedge_{k'\in \indexSetH,k'\not=k}\lnot F_{k'}(\other{a},\other{b}_k))\land \forall y(y\not=\other{b}_k\rightarrow \lnot F_{k}(\other{a},y))$. Since $f$ is an isomorphism,  this implies that $a\not= b_k:=f^{-1}(\other{b}_k)$, $(a,b_k)\in \rel{F_k}{\struc{A}}$ and $(a,b_k)\notin \rel{F_{k'}}{\struc{A}}$, for each $k'\in \indexSetH,k'\not=k$. Furthermore, assume there is  $b\in \univ{A}$, $b\not=b_k$ such that $(a,b)\in \rel{F_k}{\struc{A}}$. Since $f$ is an isomorphism, this implies $f(b)\not=f(b_k)=\other{b}_k$ and $(\other{a},\other{b})\in \rel{F_k}{\other{\struc{A}}}$, which contradicts $\other{\struc{A}}\models \forall y(y\not=\other{b}_k\rightarrow \lnot F_{k}(\other{a},y))$. Hence $\struc{A}\models \forall y(y\not=b_k\rightarrow \lnot F_{k}(a,y))$ and therefore concluding that $\struc{A}\models \chi(a)$. This proves that in either case $\struc{A}\models \psi(a)\lor\chi(a)$ and therefore $\struc{A}\models \forall x(\psi(x)\lor \chi(x))$. 
\end{proof}
\begin{claim}\label{claim:satisfyingRotationMap}
	Every structure $\struc{A}\in \bigcup_{1\leq k \leq m}\classStruc{P}_{\rho_k}\setminus \{\struc
	{A}_{\emptyset}\}$ satisfies $\varphi_{\operatorname{rotationMap}}$.
\end{claim}
\begin{proof}
	Let $\struc{A}\in \bigcup_{1\leq k \leq m}\classStruc{P}_{\rho_k}\setminus \{\struc{A}_{\emptyset}\}$. Then there is a $k\in \{1,\dots,m\}$ such that $\struc{A}\in \classStruc{P}_{\rho_k}$.
	
	By definition, $\varphi_{\operatorname{rotationMap}}=\varphi\land \psi$, where 
	\begin{align*}
		&\varphi:= \forall x \forall y \Big(\bigwedge_{i,j\in \indexSetRotation}(E_{i,j}(x,y)\rightarrow E_{j,i}(y,x))\Big) \text{ and }\\
		&\psi:= \forall x \Big(\bigwedge_{i\in \indexSetRotation}\Big(\bigvee_{j\in \indexSetRotation}\big(\exists^{=1}y E_{i,j}(x,y)\land \bigwedge_{\substack{j'\in \indexSetRotation\\ j'\not=j}}\lnot \exists y E_{i,j'}(x,y)\big)\Big)\Big).
	\end{align*}
	Thus, it is sufficient to prove that $\struc{A}\models\varphi$ and $\struc{A}\models \psi$. 
	
	To prove $\struc{A}\models\varphi$, assume towards a contradiction that there are $a,b\in \univ{A}$ such that for some pair $i,j\in \indexSetRotation$, we have that $(a,b)\in \rel{E_{i,j}}{\struc{A}}$, but $(b,a)\notin \rel{E_{j,i}}{\struc{A}}$. By construction of  $ \classStruc{P}_{\rho_k}$, there is a structure $\other{\struc{A}}\models \varphi_{\zigzag}$ and $\other{a}\in \univ{\other{A}}$ such that $(\mathcal{N}_2^{\struc{A}}(a),a)\cong (\mathcal{N}_2^{\other{\struc{A}}}(\other{a}),\other{a})$. Assume $f$ is an isomorphism from $(\mathcal{N}_2^{\struc{A}}(a),a)$ to  $(\mathcal{N}_2^{\other{\struc{A}}}(\other{a}),\other{a})$. Note that $f(b)$ is defined since $b$ is in the $2$-neighbourhood of $a$. Furthermore since $f$ is an isomorphism, $(a,b)\in \rel{E_{i,j}}{\struc{A}}$ implies $(\other{a},f(b))\in \rel{E_{i,j}}{\other{\struc{A}}}$, and $(b,a)\notin \rel{E_{j,i}}{\struc{A}}$ implies $(f(b),\other{a})\notin \rel{E_{j,i}}{\other{\struc{A}}}$. Hence $\other{\struc{A}}\not\models \varphi$, which contradicts $\other{\struc{A}}\models \varphi_{\operatorname{rotationMap}}$. 
	
	To prove $\struc{A}\models\psi$, assume towards a contradiction that there is an $a\in \univ{A}$ and $i\in \indexSetRotation$ such that  $\struc{A}\not\models \exists^{=1}y E_{i,j}(a,y)\land \bigwedge_{\substack{j'\in \indexSetRotation\\ j'\not=j}}\lnot \exists y E_{i,j'}(a,y)$ for every  $j\in \indexSetRotation$. We know that there is a structure $\other{\struc{A}}\models \varphi_{\zigzag}$ and $\other{a}\in \univ{\other{A}}$ such that $(\mathcal{N}_2^{\struc{A}}(a),a)\cong (\mathcal{N}_2^{\other{\struc{A}}}(\other{a}),\other{a})$. Let $f$ be an isomorphism from $(\mathcal{N}_2^{\struc{A}}(a),a)$ to $(\mathcal{N}_2^{\other{\struc{A}}}(\other{a}),\other{a})$. Since $\other{\struc{A}}\models \psi$, there must be $j\in \indexSetRotation$ such that $\other{\struc{A}}\models \exists^{=1}y E_{i,j}(\other{a},y)\land \bigwedge_{\substack{j'\in \indexSetRotation\\ j'\not=j}}\lnot \exists y E_{i,j'}(\other{a},y)$. Hence there must be $\other{b}\in \univ{\other{A}}$ such that $(\other{a},\other{b})\in \rel{E_{i,j}}{\other{\struc{A}}}$, which implies that $(a,f^{-1}(\other{b}))\in \rel{E_{i,j}}{\struc{A}}$. Since we assumed that $\struc{A}\not\models \exists^{=1}y E_{i,j}(a,y)\land \bigwedge_{\substack{j'\in \indexSetRotation\\ j'\not=j}}\lnot \exists y E_{i,j'}(a,y)$, there must be either $b\not=f^{-1}(\other{b})$ with $(a,b)\in \rel{E_{i,j}}{\struc{A}}$, or there must be $j'\in \indexSetRotation$, $j'\not=j$ and $b'\in \univ{A}$ such that $(a,b')\in \rel{E_{i,j'}}{\struc{A}}$.  In the first case  $(\other{a},f(b))\in \rel{E_{i,j}}{\other{\struc{A}}}$, since $f$ is an isomorphism. But then $\other{\struc{A}}\not\models \exists^{=1}y E_{i,j}(\other{a},y)$, which is a contradiction. In the second case, we get that $(\other{a},f(b'))\in \rel{E_{i,j'}}{\other{\struc{A}}}$. But then $\other{\struc{A}}\not\models \bigwedge_{\substack{j'\in \indexSetRotation\\ j'\not=j}}\lnot \exists y E_{i,j'}(\other{a},y)$, which is a contradiction. Hence $\struc{A}\models \varphi \land \psi$.
\end{proof}
\begin{claim}\label{claim:satisfyingBase}
	Every structure $\struc{A}\in \bigcup_{1\leq k \leq m}\classStruc{P}_{\rho_k}\setminus \{\struc{A}_{\emptyset}\}$ satisfies $\varphi_{\operatorname{base}}$.
\end{claim}
\begin{proof}
	Let $\struc{A}\in \bigcup_{1\leq k \leq m}\classStruc{P}_{\rho_k}\setminus \{\struc{A}_{\emptyset}\}$. Then there is a $k\in \{1,\dots,m\}$ such that $\struc{A}\in \classStruc{P}_{\rho_k}$.
	
	By definition, $\varphi_{\operatorname{base}}:=\forall x  \big(\varphi_{\operatorname{root}}(x)\rightarrow (\varphi(x)\land\psi(x))\big)$, where 
	\begin{align*}
		\varphi(x):=&\bigwedge_{i,j\in \indexSetRotation}\Big(E_{i,j}(x,x)\land \forall y \Big(x\not= y\rightarrow \big(\lnot E_{i,j}(x,y)\land \lnot E_{i,j}(y,x)\big)\Big)\Big)\text{ and }\\
		\psi(x):=& \bigwedge_{\substack{ \rot_{H^2}(k,i)=(k',i')\\k,k'\in \indexSetH\\i,i'\in \indexSetRotation }}\exists y \exists y'\big(F_k(x,y)\land F_{k'}(x,y')\land E_{i,i'}(y,y')\big).
	\end{align*}
	Thus, it is sufficient to prove that $\struc{A}\models \varphi(a)$ and $\struc{A}\models \psi(a)$ for every $a\in \univ{A}$ for which $\struc{A}\models \varphi_{\operatorname{root}}(a)$. Therefore assume $a\in \univ{A}$ is any element such that $\struc{A}\models \varphi_{\operatorname{root}}(a)$. Because $\struc{A}\in \classStruc{P}_{\rho_k}$ there is an $i\in I_{k}$ such that $(\mathcal{N}_2^{\struc{A}}(a),a)\in \tau_{d,2,\sigma}^i$. Then by definition there is a structure $\other{\struc{A}}\models \varphi_{\zigzag}$ and $\other{a}\in \univ{\other{A}}$ such that $(\mathcal{N}_2^{\struc{A}}(a),a)\cong (\mathcal{N}_2^{\other{\struc{A}}}(\other{a}),\other{a})$.  Let $f$ be an isomorphism from $(\mathcal{N}_2^{\struc{A}}(a),a)$ to $(\mathcal{N}_2^{\other{\struc{A}}}(\other{a}),\other{a})$. Assume that there is an element $\other{b}\in \univ{\other{A}}$ such that $(\other{b},\other{a})\in \rel{F}{\other{\struc{A}}}$. Since $f$ is an isomorphism and $\other{b}\in N_2^{\other{\struc{A}}}(\other{a})$ we get that $(f^{-1}(\other{b}),a)\in \rel{F}{\struc{A}}$ which contradicts that $\struc{A}\models\varphi_{\operatorname{root}}(a)$ as $\varphi_{\operatorname{root}}(x):= \forall y \lnot F(y,x)$. Hence there is no element  $\other{b}\in \univ{\other{A}}$ such that $(\other{b},\other{a})\in \rel{F}{\other{\struc{A}}}$ which implies that $\other{\struc{A}}\models \varphi_{\operatorname{root}}(\other{a})$. But since $\other{\struc{A}}\models \varphi_{\zigzag}$ we have that $\other{\struc{A}}\models \varphi_{\operatorname{base}}$ and hence $\other{\struc{A}}\models \varphi(\other{a})$ and $\other{\struc{A}}\models \psi(\other{a})$.
	
	To prove $\struc{A}\models\varphi(a)$ first observe that $(a,a)\in \rel{E_{i,j}}{\struc{A}}$ for every $i,j\in \indexSetRotation$ since $\other{\struc{A}}\models \varphi(\other{a})$ implies that $(\other{a},\other{a})\in \rel{E_{i,j}}{\other{\struc{A}}}$ for every $i,j\in \indexSetRotation$ and $f$ is an isomorphism mapping $a$ onto $\other{a}$. Assume that there is an element $b\in \univ{A}$, $b\not=a$ and indices $i,j\in \indexSetRotation$ such that either $(a,b)\in \rel{E_{i,j}}{\struc{A}}$ or $(b,a)\in \rel{E_{i,j}}{\struc{A}}$. Since $b\in N_2^{\struc{A}}(a)$ and $f$ is an isomorphism we get that $f(b)\not=f(a)=\other{a}$ and either $(\other{a},f(b))\in \rel{E_{i,j}}{\other{\struc{A}}}$ or $(f(b),\other{a})\in \rel{E_{i,j}}{\other{\struc{A}}}$. But this contradicts $\other{\struc{A}}\models \varphi(\other{a})$ and hence $\struc{A}\models \varphi(a)$.
	
	We now prove that $\struc{A}\models \psi(a)$. Let $k,k'\in \indexSetH$ and $i,i'\in \indexSetRotation$ such that $\rot_{H^2}(k,i)=(k',i')$. Since $\other{\struc{A}}\models \psi(\other{a})$ there must be elements $\other{b},\other{b}'\in \univ{\other{A}}$ such that $(\other{a},\other{b})\in \rel{F_k}{\other{\struc{A}}}$, $(\other{a},\other{b}')\in \rel{F_{k'}}{\other{\struc{A}}}$ and $(\other{b},\other{b}')\in \rel{E_{i,i'}}{\other{\struc{A}}}$. But since $\other{b},\other{b}'\in N_2^{\other{\struc{A}}}(\other{a})$ we get that $f^{-1}(\other{b})$ and $f^{-1}(\other{b}')$ are defined  and since $f$ is an isomorphism we get that $(a,f^{-1}(\other{b}))\in \rel{F_k}{\struc{A}}$, $(a,f^{-1}(\other{b}'))\in \rel{F_{k'}}{\struc{A}}$ and $(f^{-1}(\other{b}),f^{-1}(\other{b}'))\in \rel{E_{i,i'}}{\struc{A}}$. Hence $\struc{A}\models \exists y \exists y'\big(F_k(a,y)\land F_{k'}(a,y')\land E_{i,i'}(y,y')$ for any $k,k'\in \indexSetH$ and $i,i'\in \indexSetRotation$ such that $\rot_{H^2}(k,i)=(k',i')$ which implies that $\struc{A}\models \psi(a)$.
\end{proof}

\section*{Acknowledgments}
We thank Sebastian Ordyniak for inspiring discussions.
The second author thanks Micha\l{} Pilipczuk and Pierre Simon for inspiring discussions at the docCourse on Sparsity in Prague 2018.

\bibliographystyle{siamplain}
\bibliography{testingFO}

\newpage

\printunsrtglossary[type=symbols,title={List of Notation}]

\end{document}